\documentclass[12pt,oneside,reqno]{amsart}
\usepackage{geometry}
\usepackage[utf8]{inputenc}
\usepackage[T1]{fontenc}
\usepackage{verbatim}
\usepackage{amsmath}
\usepackage{amsthm}
\usepackage{mathabx}
\usepackage{amssymb}
\usepackage{amsfonts}
\usepackage{mathtools}
\usepackage{nccmath}
\usepackage{caption}
\usepackage{tikz}
\usepackage{xcolor}
\usepackage{enumitem}
\usepackage{braket}
\usepackage{mathrsfs}
\usepackage{verbatim}
\usepackage[english]{babel}
\usepackage{hyphenat}
\usepackage{bbm}
\usepackage{calligra}
\usepackage{empheq}
\usepackage{mathtools}
\usepackage[T1]{fontenc}
\usepackage{lmodern}
\usepackage[nopatch=footnote]{microtype}
\usepackage{fmtcount}
\usepackage[dvipsnames]{xcolor}
\usepackage{comment}
\usepackage{soul}
\usepackage[numbers]{natbib}
\usepackage[utf8]{inputenc}
\usepackage{mathtools}
\usepackage{amssymb}
\usepackage{amsfonts}
\usepackage{amsthm}
\usepackage{scalerel,stackengine}
\stackMath
\newcommand\rwh[1]{%
\savestack{\tmpbox}{\stretchto{%
  \scaleto{%
    \scalerel*[\widthof{\ensuremath{#1}}]{\kern-.6pt\bigwedge\kern-.6pt}%
    {\rule[-\textheight/2]{1ex}{\textheight}}%WIDTH-LIMITED BIG WEDGE
  }{\textheight}% 
}{0.5ex}}%
\stackon[1pt]{#1}{\tmpbox}%
}

\newcommand{\RN}[1]{%
  \textup{\uppercase\expandafter{\romannumeral#1}}%
}

\newif\ifshowdetails
\showdetailsfalse           % <-- set \showdetailsfalse before submission
\ifshowdetails
  \newenvironment{details}{%
    \par\medskip\noindent
    \stepcounter{details}%
    {\color{Red}\bfseries Details \thedetails:}\ %
    \begingroup\color{RoyalBlue}\small
  }{%
    \endgroup\par\medskip
  }
\else
  \excludecomment{details}
\fi

\newcounter{details}[section]
\renewcommand{\thedetails}{\thesection.\arabic{details}}

\usepackage[english]{babel}

\allowdisplaybreaks
\usepackage{bm} %% special bold-face characters

\usepackage{graphicx,url}

\usepackage[svgnames]{xcolor}

\usepackage[colorlinks=true, linkcolor=Maroon, urlcolor=Maroon]{hyperref}

\hypersetup{citecolor=blue}

\hypersetup{
  linkcolor=Maroon,
  citecolor=DarkBlue,
  urlcolor=red
}

\usepackage{dsfont}
\usepackage{suetterl} %%frak letters%%
\usepackage{pgfplots}	
\pgfplotsset{compat=1.18} 

\usepackage{caption}
\usepackage{subcaption}		%%subfigures, etc.%%
\usepackage{braket}
\usepackage{stmaryrd} %%double brackets

\usepackage[T1]{fontenc}
\usepackage{currvita}

\usepackage{csquotes} %% quotes

\mathtoolsset{showonlyrefs} %%% Only number referenced equations

\newcommand{\assumref}[1]{\hyperref[#1]{\textcolor{Maroon}{Assumption~\ref*{#1}}}}

\newcommand{\secref}[1]{\hyperref[#1]{\textcolor{Maroon}{Section~\ref*{#1}}}}
\newcommand{\hypref}[1]{\hyperref[#1]{\textcolor{Maroon}{\ref*{#1}}}}
\newcommand{\appref}[1]{\hyperref[#1]{\textcolor{Maroon}{Appendix~\ref*{#1}}}}

\newcommand{\appsubref}[1]{\hyperref[#1]{\textcolor{Maroon}{Appendix~\ref*{#1}}}}

\newcommand{\thmref}[1]{\hyperref[#1]{\textcolor{Maroon}{Theorem~\ref*{#1}}}}
\newcommand{\lemref}[1]{\hyperref[#1]{\textcolor{Maroon}{Lemma~\ref*{#1}}}}
\newcommand{\defref}[1]{\hyperref[#1]{\textcolor{Maroon}{Definition~\ref*{#1}}}}
\newcommand{\propref}[1]{\hyperref[#1]{\textcolor{Maroon}{Proposition~\ref*{#1}}}}
\newcommand{\corref}[1]{\hyperref[#1]{\textcolor{Maroon}{Corollary~\ref*{#1}}}}
\newcommand{\remref}[1]{\hyperref[#1]{\textcolor{Maroon}{Remark~\ref*{#1}}}}
\newcommand{\figref}[1]{\hyperref[#1]{\textcolor{Maroon}{Figure~(\ref*{#1})}}}

\newcommand{\refhyp}[1]{\hyperref[#1]{\ref*{#1}}}

\newcommand{\hyporef}[1]{%
    \hyperref[#1]{%
        \textcolor{Maroon}{Hypothesis~\ref*{#1}}%
    }%
}

\usepackage{simpler-wick} 	%% wick contraction %%
\newlength{\wdth}

\usepackage{tikz}
\usepgfplotslibrary{fillbetween}
\usetikzlibrary{patterns}

\usetikzlibrary{arrows}
\usetikzlibrary{arrows,positioning,arrows.meta} 

\allowdisplaybreaks		%% page break in equation arrays %%

\def\C{{\mathbb C}}
\def\N{{\mathbb N}}	

\def\R{{\mathbb R}}

\def\H{{\mathbb H}}

\DeclareMathAlphabet{\mathdutchcal}{U}{dutchcal}{m}{n}
\DeclareMathAlphabet{\mathbdutchcal}{U}{dutchcal}{b}{n} % Optional: bold version

\def\cA{{\mathcal A}}   
\def\cB{{\mathcal B}}

\def\cD{{\mathcal D}}

\def\cF{{\mathcal F}}
\def\cH{{\mathcal H}}

\def\cM{{\mathcal M}}    
\def\cN{{\mathcal N}}

\def\cQ{{\mathcal Q}}     

\def\cS{{\mathcal S}}
\def\cT{{\mathcal T}}
\def\cU{{\mathcal U}}
    
\def\cW{{\mathcal W}}

\def \cW{{\mathcal W}}

\def\dcE{{\mathdutchcal E}}

\def\dcH{{\mathdutchcal H}}

\def\dcQ{{\mathdutchcal Q}}

\def\dcV{{\mathdutchcal V}}

\def\dcd{{\mathdutchcal d}}

\def\dch{{\mathdutchcal h}}

\def\dcq{{\mathdutchcal q}}

\def\fkH{{\mathfrak H}}

\def\fkS{{\mathfrak S}}

\def\ux{{\underline{x}}}

\def\bf{{\bf f}}

\def\h{{\mathfrak h}}

\def\cVec{\mathbf{V}}

\def\sym{{\operatorname{sym}}}

\def\vacf{{\Omega_f}}

\def\mr #1{\mathrm{#1}}		
\def \ms #1 {\mathscr{#1}}	%%%%%%%%%%%%%%

\newcommand{\op}[1]{\operatorname{#1}}

\newcommand{\pt}{\partial}

\newcommand{\wh}{\widehat}

\def\rchi{\raisebox{0.3ex}{$\chi$}}

\renewcommand{\Im}[1]{ \mathfrak{Im}\prc{ #1 }}
\renewcommand{\Re}[1]{ \mathfrak{Re}\prc{ #1 }}
\newcommand{\re}{\mathfrak{Re}}
\newcommand{\im}{\mathfrak{Im}}

\newcommand{\Tr}{\operatorname{Tr}}   %%%%%%%%%%%%%%

\def\1{\mathord{\text{\usefont{U}{dsrom}{m}{n}1}}}

\def\ad{a^\dagger}

\makeatletter
\def\namedlabel#1#2{\begingroup					
	#2%
	\def\@currentlabel{#2}%
	\phantomsection\label{#1}\endgroup
}
\makeatother

\theoremstyle{plain}
\newtheorem{theorem}{Theorem}[section]		
\newtheorem{definition}[theorem]{Definition}				%%%%%%%%%
\newtheorem{proposition}[theorem]{Proposition}		   %% Numbering%
\newtheorem{hypothesis}{Hypothesis}[section]		  %%%%%%%%%

\newtheorem{lemma}[theorem]{Lemma}
\newtheorem{corollary}[theorem]{Corollary}
\newtheorem{remark}[theorem]{Remark}

\newtheorem*{remark*}{Remark}
		
\numberwithin{equation}{section}

\def\wyl{\cW}

\def\wyld{\cW^\dagger}

\def\UN{\cU_N}
\def\UNd{\cU_N^\dagger}

\def\UN{\cU_N}
\def\UNd{\cU_N^\dagger}

\def\HMN{\cH_N}
\def\HN{\dcH_N}
\def\HMNbos{{\mathcal{H}_N^{bos}}}
\def\Fock{\mathcal{F}}
\def\Hf{\cH_f}

\def\Hkp{\mathfrak{H}_{k+1}}

\def\OA{{A}}
\def\OB{{B}}
\def\OT{{T}}
\def\OJ{{J}}

\def\Hf{\cH_f}

\def\HNb{\mathfrak{H}_N}
\def\Hkb{\mathfrak{H}_k}

\def\vfi{\varphi}

\def\eps{\epsilon}

\def\lam{\lambda}

\def\omg{\omega}

\makeatother

\makeatletter
\def\paragraph{\@startsection{paragraph}{4}%
  \z@\z@{-\fontdimen2\font}%
  {\normalfont\bfseries}}
\makeatother

\makeatletter
\def\@tocline#1#2#3#4#5#6#7{\relax
  \ifnum #1>\c@tocdepth % then omit
  \else
    \par \addpenalty\@secpenalty\addvspace{#2}%
    \begingroup \hyphenpenalty\@M
    \@ifempty{#4}{%
      \@tempdima\csname r@tocindent\number#1\endcsname\relax
    }{%
      \@tempdima#4\relax
    }%
    \parindent\z@ \leftskip#3\relax \advance\leftskip\@tempdima\relax
    \rightskip\@pnumwidth plus4em \parfillskip-\@pnumwidth
    #5\leavevmode\hskip-\@tempdima
      \ifcase #1
       \or\or \hskip 2em \or \hskip 3em \else \hskip 4em \fi%
      #6\nobreak\relax
    \hfill\hbox to\@pnumwidth{\@tocpagenum{#7}}\par% <---- \dotfill -> \hfill
    \nobreak
    \endgroup
  \fi}
\makeatother

\usepackage[affil-it]{authblk}
\usepackage{orcidlink}
\usepackage{amsaddr}
 
\newcommand\abs[1]{\left|#1 \right|}
\newcommand\norm[1]{\left\lVert#1 \right\rVert}
\newcommand{\inp}[1]{ \left \langle #1 \right \rangle}

\newcommand{\prc}[1]{ {\left( #1 \right)}}

\newcommand{\brs}[1]{ \left\{ #1 \right\}}
\def\pr{{\partial}}

\newcommand{\cc}[1]{\overline{#1}}

\newcommand*\mcap{\mathbin{\mathpalette\mcapinn\relax}}
\newcommand*\mcapinn[2]{%
  \vcenter{\hbox{%
    \scalebox{0.7}{$\mathsurround=0pt
      \ifx\displaystyle#1\textstyle\else#1\fi\bigcap$}}}}

\newcommand*\mcup{\mathbin{\mathpalette\mcupinn\relax}}
\newcommand*\mcupinn[2]{%
  \vcenter{\hbox{%
    \scalebox{0.7}{$\mathsurround=0pt
      \ifx\displaystyle#1\textstyle\else#1\fi\bigcup$}}}}

\begin{document} 

\title[Dynamical BEC formation and vanishing entropy for trapped bosons in non-relativistic QED]{Dynamical BEC formation and vanishing entropy for trapped bosons in non-relativistic QED via nonlinear resonance cascades}
\author{Thomas Chen$^\dagger$\orcidlink{0000-0003-2704-1454},\quad Ali Mezher$^\dagger$\orcidlink{0000-0002-9215-1804}}

\address{$^\dagger$Department of Mathematics, University of Texas at Austin, Austin, TX 78712, USA.}

\email{tc@math.utexas.edu} 
\email{alimezher@utexas.edu}

% \subjclass[2020]{Primary 35Q55, 35B40; Secondary 81Q15, 47F05, 82C10}
% \keywords{\hangindent=14pt Bose-Einstein condensation, mean-field limit, BBGKY hierarchy, von Neumann entropy, dynamical purification, quantum de Finetti theorem, Gross-Pitaevskii hierarchy}

\maketitle
\begin{abstract}

The dynamical formation of a complete Bose-Einstein condensate and the asymptotic purification of trapped bosons coupled to a quantized radiation field are rigorously established. By tracking the macroscopic one-particle density matrix $\kappa_T$ through a nonlinear resonance cascade, it is demonstrated that radiative relaxation irreversibly depletes all excited trap modes, even when the initial macroscopic state $\kappa_0$ is genuinely mixed. Here $\kappa_T$ is an average, with respect to a de Finetti measure $\mu_0$, over mean-field particle-field orbits, so that the mixedness of $\kappa_0$ reflects a statistical ensemble of orbits. Under the assumptions of a non-vanishing initial ground state component and strictly positive downward Fermi-Golden-Rule rates, strong trace-norm convergence to a pure condensate is proven,
\begin{equation*}
    \lim_{T\to\infty} \| \kappa_T-\ket{\rchi_0}\bra{\rchi_0} \|_{\mathfrak{S}_1(\h)} = 0 \, .
\end{equation*}
Furthermore, the von Neumann entropy of the macroscopic one-particle density matrix is shown to vanish asymptotically,
\begin{equation*}
    \lim_{T\to\infty}S[\kappa_T]=0 \, ,
\end{equation*}
without any claim of monotonicity in $T$. Thus, any initially mixed, positive-entropy state ($0<S[\kappa_0]<\infty$) is purified to the ground state $\ket{\rchi_0}\bra{\rchi_0}$. For the proof, we derive the Hartree hierarchy of the coupled system of bosons and photons from the finite-$N$ BBGKY hierarchy of the microscopic quantum many-body system, establishing trace-norm convergence in the mean-field limit as $N\to\infty$. The limits
are taken successively; first $N\to\infty$, then the weak-coupling limit, and finally $T\to\infty$. Although the global particle-field evolution remains unitary, the macroscopic one-particle state converges to a pure state as energy and correlations are continuously radiated into the dispersive field.

\end{abstract}

\setcounter{tocdepth}{2}
\tableofcontents
\section{Introduction and main results}
\label{sec:introduction}

Most dynamical theories of Bose gases begin where the formation problem ends. Initially, complete condensation is imposed, and the primary objective is to show that it persists. In the reduced-density-matrix formulation introduced by Penrose and Onsager, Bose--Einstein condensation is
characterized by the macroscopic occupation of a one-particle orbital \cite{penrose1956bose}.
Complete condensation means that the fraction of particles
occupying that orbital tends to one in the many-particle limit. With the trace-one normalization used here, complete condensation into a specified normalized orbital is expressed by convergence of the one-particle density matrix to the corresponding rank-one projector in trace norm. This criterion forms the foundation of the rigorous equilibrium theory of dilute trapped gases \cite{lieb2002proof,liseyn,liseyn2,lisesoyn,narose,deuchert2019bose} and the derivation of effective Hartree or Gross-Pitaevskii evolution from many-body quantum dynamics \cite{chen2025derivation,chen2011quintic,chen2014derivation,esy4,knowles2010mean,dietzelee2023,kirkpatrick2021large,kiscst,petrat2016new,pickl2011simple,rodnianski2009quantum,boers2016mean,fantechi2025quantum, AIF_2026__76_2_701_0,chen2019mean,golse2017schrodinger,rougerie2021scaling,chen2008cauchy,chen2010energy}.

These results elucidate the evolution of an already existing condensate. However, they do not, in isolation, account for the dynamical process by which a genuinely mixed macroscopic state sheds all excited-mode occupations, selects a single orbital, and attains purity. The formation problem is considerably more subtle because closed quantum dynamics are reversible. Known formation mechanisms often operate at a reduced level of detail. For example, a condensate may emerge as a singular mass in the bosonic Nordheim equation \cite{esve2015}, or it may be prepared by a prescribed Markovian reservoir \cite{diehl2008quantum,kraus2008preparation}.
The rigorous framework for entropy production and return to equilibrium in subsystem-reservoir interactions at positive temperature is a cornerstone of quantum statistical mechanics \cite{jakvsic2001entropy,jakvsic2002non,frohlich2004statistical}; the effective dynamics is typically linear and does not lead to complete ground state occupation in the equilibrium state. 

In the present work, there is no external dissipative generator or heat bath; instead, the underlying condensation mechanism is genuinely {\em nonlinear}. The confined particles and the radiation field form a single closed Hamiltonian system; the evolution is unitary, although energy and correlations may transfer from the particle subsystem to the dispersive field. As a result, the reduced bosonic state may asymptotically become pure without violating global unitarity. This process remains consistent with the standard entropy inequalities for quantum subsystems \cite{araki1970entropy,carlen2012bounds}. The central mathematical question is whether this physical picture persists through the transition from a microscopic many-body model to controlled macroscopic, long-time behavior.

This study provides a rigorous affirmative answer to this question. The principal finding is the identification of a radiative formation mechanism for complete ground-state Bose-Einstein condensation. Beginning with a macroscopic one-particle density matrix \(\kappa_0\), which may be genuinely mixed, the radiation-induced nonlinear resonance cascade evolution transfers the entire occupation of the excited trap modes to the ground state. The resulting rescaled density matrix
\(\kappa_T\) converges in trace norm to \(\ket{\rchi_0}\bra{\rchi_0}\) where $\chi_0$ is the ground state of the boson subsystem. The von Neumann entropy $S[\kappa_T]$ of this density matrix also vanishes asymptotically. That is, 
\begin{equation}
\label{eq:intro-headline-result}
\lim_{T\to\infty}
\left\|
\kappa_T-\ket{\rchi_0}\bra{\rchi_0}
\right\|_{\mathfrak S_1(\h)}=0,
\end{equation}
and
\begin{equation}
    \lim_{T\to\infty}S[\kappa_T]=0.
\end{equation}
Therefore, when \(0<S[\kappa_0]<\infty\), this result constitutes both dynamical condensation and purification. In other words, a positive-entropy bosonic state is irreversibly driven to the unique pure ground-state
condensate. It is important to note that this statement concerns asymptotic entropy vanishing and does not assert that the mapping
\(T\mapsto S[\kappa_T]\) is strictly monotone.
We note that the decrease of entropy in flocking phenomena with emergent order formation has been rigorously proven in recent work of E. Tadmor \cite{tadmor}.

\subsection{Ground-state condensation and entropy vanishing}The principal results are presented in precise form by considering the trapping potential \(V\) as specified in \hyporef{hyp:standing-hypotheses},
\begin{equation}
\label{eq:intro-trap-spectrum}
    \dch_0=-\Delta+V,
    \qquad
    \dch_0\rchi_j=E_j\rchi_j,
    \qquad
    E_0<E_1<E_2<\cdots,
    \qquad j\in\N_0,
\end{equation}
where \(\set{\rchi_j}_{j\in\N_0}\) is a complete orthonormal basis and \(\rchi_0\) is the ground state. Then, following \cite{chen2026trapped}, we expand the one-particle wave function in the trap eigenbasis
\begin{equation}
\label{eq:intro-trap-expansion}
    \vfi_t
    =\sum_{j=0}^{\infty}A_{j,t}(\vfi) e^{-itE_j}\rchi_j,
    \qquad
    \sum_{j=0}^{\infty}|A_{j,t}(\vfi)|^2=1,
\end{equation}
and we define the rescaled amplitudes
\begin{equation}
\label{eq:intro-rescaled-amplitudes}
    F_{j,T}^\eta=A_{j,T/\eta^2},
    \qquad
    T=\eta^2t.
\end{equation} 
Thus, the wave function in terms of the rescaled amplitudes reads
\begin{equation}
\label{eq:intro-rescaled-wave-function}
    \vfi_{T/\eta^2}
    =\sum_{j=0}^{\infty}F_{j,T}^\eta(\vfi)e^{-iT E_j/\eta^2}\rchi_j.
\end{equation}

The initial limiting bosonic hierarchy is represented by a de Finetti measure \(\mu_0\) supported on
\(B_R(0)\mcap\mathbb S\h\), where
\begin{equation}
    B_R(0)
    :=
    \left\{
        \vfi\in Y^1:
        \|\OT\vfi\|_{\h}<R
    \right\},
    \qquad R>\sqrt{C_E}.
\end{equation}
Here, \(C_E\) is the constant in \corref{cor:k-particle-energy-bound}, and the support
property follows from \corref{cor:bounded-energy-support-de-finetti}. See \defref{def:Y1-and-density-matrix-space} and \thmref{thm:strong-quantum-de-Finetti} for the relevant definitions and representation theorem. Its initial one-particle state is given by
\begin{equation}
\label{eq:intro-kappa-zero}
    \kappa_0
    :=
    \int_{B_R(0)\mcap\mathbb S\h}
        \ket{\vfi}\bra{\vfi}\,d\mu_0(\vfi).
\end{equation}
Here, \(\mu_0\) is the de Finetti measure associated with the full initial limiting hierarchy. The subsequent evolution is defined by evolving each orbital in this ensemble under the nonlinear cascade and then averaging. In general, the one-particle density matrix \(\kappa_0\) alone does not determine \(\mu_0\) or the subsequent state \(\kappa_T\).

For each \(\vfi\) in the support of \(\mu_0\), let the sequence \(F_T(\vfi)=\bigl(F_{j, T}(\vfi)\bigr)_{j\geq0}\) denote the solution to the effective nonlinear resonance cascade equation derived in \cite{chen2026trapped}, where the initial amplitudes correspond to the coefficients of \(\vfi\) in the basis \(\set{\rchi_j}_{j\in\N_0}\). The rescaled interaction-picture one-particle state is defined by
\begin{equation}
\label{eq:intro-kappa-T}
    \kappa_T
    :=\sum_{j,\ell=0}^{\infty}
    \left(
        \int_{B_R(0)\mcap\mathbb S\h}
        F_{j, T}(\vfi)\cc{F_{\ell, T}(\vfi)}\,d\mu_0(\vfi)
    \right)
    \ket{\rchi_j}\bra{\rchi_\ell},
\end{equation}
and the diagonalized state is given by
\begin{equation}
      \kappa_T^{\mr{diag}}
    :=
    \sum_{j=0}^{\infty}
    \left(
        \int_{B_R(0)\mcap\mathbb S\h}
        |F_{j, T}(\vfi)|^2\,d\mu_0(\vfi)
    \right)
    \ket{\rchi_j}\bra{\rchi_j}.
\end{equation}
This definition preserves the essential aspects of the physical problem: the radiative relaxation of each effective trajectory and the statistical mixing represented by \(\mu_0\). Furthermore, it clarifies the rationale for starting the modal sums at \(j=0\), thereby justifying the selection of the limiting state as the ground state. We assume the regularity and non-resonance hypotheses of \cite{chen2026trapped}. Additionally, assume that for \(\mu_0\)-almost every initial \(\vfi\), the ground state component is non-vanishing
\begin{equation}
\label{eq:intro-assumptions-F0}
    |F_{0,0}(\vfi)|>0.
\end{equation}
Furthermore, assume the strict positivity of the downward Fermi-Golden-Rule rates
\begin{equation}
\label{eq:intro-assumption-FGRjj}
    \Gamma_{0,j'}^{FGR}>0 \quad\text{for every }j'\geq1.
\end{equation}
The path--wise condensation theorem of \cite{chen2026trapped} guarantees
\begin{equation}
\label{eq:intro-pathwise-BEC}
    \lim_{T\to\infty}
    \sum_{j=1}^{\infty}|F_{j,T}|^2=0,
    \qquad
    \lim_{T\to\infty}|F_{0,T}|^2=1.
\end{equation}
The statement is intentionally formulated in terms of occupations. The phase of \(F_{0,T}\) does not require convergence, as it is inherently absent from the associated rank-one density matrix. Because the cascade conserves the \(\ell^2\)-mass of the boson subsystem, \eqref{eq:intro-pathwise-BEC} and the dominated convergence theorem together imply
\begin{equation}
\label{eq:intro-ensemble-ground-occupation}
    1-\langle\rchi_0,\kappa_T\rchi_0\rangle
    =
   \int_{B_R(0)\mcap\mathbb S\h}
    \sum_{j=1}^{\infty}|F_{j,T}|^2\,d\mu_0(\vfi)
    \longrightarrow0.
\end{equation}
For density matrices, the convergence in \eqref{eq:intro-ensemble-ground-occupation} directly implies the trace-norm result in \eqref{eq:intro-headline-result}. Because $\kappa_T$ is a positive trace-class operator with unit trace and $\ket{\rchi_0}\bra{\rchi_0}$ is a pure state, the convergence of the matrix element $\langle\rchi_0, \kappa_T \rchi_0\rangle \to 1$ forces the trace distance to vanish. This finding demonstrates complete ground-state condensation, characterized by the absence of excited occupation and the entire de Finetti ensemble consistently selecting the same rank-one state. 

The von Neumann entropy of $\kappa_T$ is bounded above by the Shannon entropy of its diagonal part,
\begin{equation}
\label{eq:intro-entropy-dephasing}
    0\leq S[\kappa_T]
    \leq S[\kappa_T^{\mr{diag}}].
\end{equation}
This result expresses the pinching inequality in the current basis \cite{lindblad1975completely,zhang2016quantifying}. The contribution of each fixed finite block in \(S[\kappa_T^{\mr{diag}}]\) vanishes according to \eqref{eq:intro-pathwise-BEC}, while the spectral tightness in \cite[Lemma~4.2]{chen2026trapped} provides a uniform tail control. Therefore,
\begin{equation}
    \lim_{T\to\infty}S[\kappa_T]=0.
\end{equation}
The result of \cite[Lemma~4.2]{chen2026trapped} is essential, as it reflects the established requirement for an energy- or entropy-tightness constraint in infinite-dimensional continuity theorems \cite{winter2016tight,shirokov2010continuity,shirokov2018adaptation,shirokov2019uniform,de2021quantum}. 

The physical interpretation of this result depends on the nature of the initial state. When \(\kappa_0\) is rank one, the theorem establishes radiative ground-state selection for a pure condensate. In contrast, if \(0<S[\kappa_0]<\infty\), then \(\kappa_0\) is genuinely mixed and \eqref{eq:intro-headline-result} demonstrates true dynamical formation and purification. This distinction cannot be determined solely by whether a particular representation of \(\mu_0\) is a Dirac measure, since global phases do not affect \(\ket{\vfi}\bra{\vfi}\). The state-level condition fundamentally depends on the mixedness of \(\kappa_0\) itself.

\subsection{The radiative mechanism and its relation to previous work}
The physical mechanism underlying \eqref{eq:intro-headline-result} involves transferring discrete-mode energy into dispersive radiation. Fermi--Golden--Rule damping, excited-state instability, and nonlinear ground-state selection have precedents in Hamiltonian dispersive equations \cite{soffer1999resonances,tsai2002relaxation,soffer2004selection,gang2007relaxation}. Related results on relaxation, decoherence, and return to equilibrium for finite-level matter coupled to massless bosonic fields are presented in \cite{bach2000return,de2013approach,hasler2011ground,merkli2008resonance}. These foundational studies demonstrate the same general principle: a localized degree of freedom can relax by radiating into a continuum. However, they do not address the collapse of a de Finetti mixture, the entropy reduction in a macroscopic bosonic subsystem from condensation, or the dynamics of complete Bose--Einstein condensation in a confined boson system, under conservation of the total boson mass.

In this context, the effective particle and field variables are \(\vfi_t\) and \(u_t\), which solve the coupled system
\begin{equation}
\label{eq:intro-effective-system}
\begin{cases}
    i\pt_t\vfi_t(x)
    =
    \cH[u_t,\vfi_t]\vfi_t(x),
    \\[3mm]
    i\pt_tu_t(y)
    =
    \omg(-i\nabla_y)u_t(y)
    +
    \eta\bigl(w*|\vfi_t|^2\bigr)(y).
\end{cases}
\end{equation}
Applying the Duhamel formula to the field equation separates \(u_t\) into the freely propagating coherent field and the radiation generated by \(|\vfi_t|^2\). Substituting this result into the particle equation produces an effective particle dynamics that is nonlocal in time, as the state of the trapped gas at time \(t\) depends on the radiation generated during its prior evolution. This memory term, rather than an artificially imposed dissipator, drives the irreversible mode transfer observed after weak-coupling averaging. 

To elucidate this transfer, the particle wave function is expanded in the trap eigenbasis,
\begin{equation}
\label{eq:intro-modal-expansion}
    \vfi_t
    =\sum_{j=0}^{\infty}A_j(t)e^{-itE_j}\rchi_j,
\end{equation}
where the \(A_j(t)\) are the time-dependent amplitudes associated with the eigenfunctions \(\rchi_j\). We set \(\lam=\eta^2\), introduce the slow macroscopic time \(T=\eta^2t\), and define the rescaled amplitudes as in \cite{chen2026trapped}, namely, for any fixed \(\vfi_t\) and for each \(j\geq0\), the rescaled amplitude is defined by
\begin{equation}
    F_{j,T}^\eta =A_j(T/\eta^2),
\end{equation}
where the index $j$ corresponds to the $j$-th eigenfunction of the one-particle Hamiltonian $\dch_0=-\Delta+V$. At fixed \(T\in[0,T_\eta]\) where $T_\eta=C_1\log\langle1/\eta\rangle$ for a sufficiently small $C_1>0$, the weak-coupling limit \(\eta\to0\) averages out the non-resonant oscillations and yields
\begin{equation}
    F_{j,T}^\eta \to F_{j,T}
\end{equation}
strongly in \(\ell^2\). The rational independence of the energy eigenvalues in
\hypref{hyp:spectral-genericity} excludes nontrivial four-mode resonances. Under the hypotheses of
\cite[Theorem~1.2]{chen2026trapped}, the limiting
amplitudes satisfy the diagonal nonlinear cascade equation
\begin{equation}
\label{eq:intro-diagonal-cascade}
    \pt_T F_{j,T}
    =
    \sum_{j'\geq0}
        M_{j,j'}|F_{j',T}|^2F_{j,T}.
\end{equation}
The non-negative symmetric Fermi--Golden--Rule rates \(\Gamma_{j,j'}^{FGR}\), combined with the directional ordering of the eigenvalues, dictate the evolution of the occupations
\begin{equation}
\label{eq:intro-occupation-cascade}
    \pt_T|F_{j,T}|^2
    =\sum_{j'\geq0}
    2\Gamma_{j,j'}^{FGR}
    \left(\1_{\{j'>j\}}-\1_{\{j'<j\}}\right)
    |F_{j',T}|^2|F_{j,T}|^2.
\end{equation}
Consequently, the \(\ell^2\)-mass is perfectly conserved, whereas the particle energy decays
\begin{equation}
\label{eq:intro-energy-decay}
    \pt_T\sum_{j=0}^{\infty}E_j|F_{j,T}|^2
    =-\sum_{j<j'}
    2\Gamma_{j,j'}^{FGR}(E_{j'}-E_j)
    |F_{j',T}|^2|F_{j,T}|^2
    \leq0.
\end{equation}
Equation \eqref{eq:intro-assumption-FGRjj} formalizes radiative cooling by expressing mass conservation alongside the irreversible transfer of particle energy to the dispersive field. The strict positivity of the downward rates in \eqref{eq:intro-assumption-FGRjj}, together with the monotonicity of the high-mode mass tails, ensures \eqref{eq:intro-pathwise-BEC}. Averaging these independent trajectories with respect to \(\mu_0\) yields the conclusions in \eqref{eq:intro-headline-result}. Our previous work \cite{chen2026trapped} established the weak-coupling nonlinear resonance cascade equation and demonstrated complete ground state occupation for a single effective particle-field trajectory. The principal advance of the present paper is to extend this dynamical mechanism from pure wave functions to macroscopic states characterized by a de Finetti measure. In particular, every ensemble member converges to the same ground-state projector, ensuring condensation and purification at the ensemble level. The second key contribution, logically supporting but distinct from the long-time purification theorem, is the rigorous derivation of the effective particle-field dynamics directly from the microscopic \(N\)-body evolution.

\subsection{The microscopic model and the mean-field bridge}
Under the conditions of \hyporef{hyp:standing-hypotheses}, the microscopic system acting on $\HN$ is generated by the Hamiltonian 
\begin{equation}
\label{eq:intro-microscopic-Hamiltonian}
    \HMN
    =\HMNbos\otimes\1_{\Fock}
    +{\1}_{\HNb}\otimes\Hf
    +\frac{\eta}{\sqrt N}W_N,
\end{equation}
where
\begin{equation}
\label{eq:intro-components}
\begin{cases}
    \HMNbos
    &=\displaystyle \sum_{j=1}^{N}\dch_0^{(j)}
      +\frac{\lam}{N}\sum_{1\leq i<j\leq N}v(x_i-x_j),\vspace{0.5 em}\\
    \Hf
    &=\displaystyle \int_{\R^3}\omg(\xi)\ad_\xi a_\xi\,d\xi,
    \qquad
    W_N
    =\sum_{j=1}^{N}\phi\bigl(w(x_j-\cdot)\bigr).
\end{cases}
\end{equation}
The scaling factor \(N^{-1}\) balances the \(O(N^2)\) interacting particle pairs. The factor \(N^{-1/2}\) is matched
to an initially coherent field,
\begin{equation}
\label{eq:intro-initial-state}
    \Psi_{N,0}
    =\phi_{N,0}\otimes\wyl[\sqrt N\,u_0]\vacf,
    \qquad
    \phi_{N,0}\in\HNb,
    \qquad
    u_0\in\mathcal D(\omg^{1/2}),
\end{equation}
ensuring that both the particle interaction and the field back-reaction remain dynamically visible in the limit \(N\to\infty\).
Mean-field and classical-field limits for related Nelson- and Pauli--Fierz-type systems have been extensively studied \cite{ammari2014wigner,leopold2017mean,AIF_2026__76_2_701_0}. Furthermore, coupled mean-field limits for other two-component quantum systems provide an important comparison
\cite{cardenas2023effective, cardenas2025emergence}. The present result, however, is distinguished by its subsequent weak-coupling and long-time analysis of asymptotic condensation and entropy drop. For the special factorized particle data \(\phi_{N,0}=\vfi_0^{\otimes N}\), our secondary main
theorem proves that, for every fixed \(k\) and on bounded time intervals,
\begin{equation}
\label{eq:intro-mean-field-convergence}
    \lim_{N\to\infty}
    \left\|
        \gamma_{N,t}^{(k)}
        -\ket{\vfi_t^{\otimes k}}\bra{\vfi_t^{\otimes k}}
    \right\|_{\mathfrak S_1(\Hkb)}=0,
\end{equation}
where $\vfi_t$ solves the effective system \eqref{eq:intro-effective-system}. 

\smallskip
\textit{Note} that it is important to distinguish factorization in the particle variables from localization in the trap spectrum. The assumption $\phi_{N,0}=\vfi_0^{\otimes N}$ dictates that all particles initially occupy the same one-particle orbital; however, it does not require $\vfi_0$ to equal the ground state $\rchi_0$, to be close to it, or to represent a mere microscopic fluctuation around it. Indeed, one may expand the initial state as
\begin{equation}
    \vfi_0=\sum_{j=0}^{\infty}A_{j,0}\rchi_j, \qquad \sum_{j=0}^{\infty}|A_{j,0}|^2=1,
\end{equation}
This analysis does not assume that the excited-state mass $\sum_{j\geq1}|A_{j,0}|^2$ is small. The initial wave function may exhibit an order-one excited-mode occupation and may involve infinitely many trap modes. Therefore, we use factorization exclusively to isolate a single deterministic trajectory of the fully nonlinear mean-field particle-field system, rather than to linearize the dynamics or confine them to perturbative fluctuations. 

For a general initial state, the de Finetti measure $\mu_0$ evolves each admissible orbit within its support under the same effective flow, and the resulting trajectories are subsequently averaged. As a result, the nonlinear resonance cascade equation governs the complete modal occupations $F_{j,T}$ of general admissible excited initial states, not merely the fluctuations around the ground state. The long-time condensation result remains subject to the stated non-resonance and Fermi-Golden-Rule hypotheses, including the requirement that the initial ground-state amplitude is nonzero, $|F_0(0)|>0$. Thus, for general limiting data, the BBGKY hierarchy is represented as
\begin{equation}
\label{eq:intro-de-Finetti-hierarchy}
    \gamma_{\infty,t}^{(k)}
    =\int_{B_R(0)\subset Y^1\mcap\mathbb S\h}
      \ket{\vfi_t^{\otimes k}}\bra{\vfi_t^{\otimes k}}
      \,d\mu_0(\vfi).
\end{equation}
The proof crucially depends on uniform particle and field estimates, the compactness and admissibility of the limiting marginals, the identification of the transported radiation source via multi-time observables, and the uniqueness of the limiting hierarchy. The phase-space and hierarchy techniques supporting this analysis are grounded in quantum de Finetti and Wigner-measure methods \cite{lions1993mesures,ammari2011mean,lewin2014derivation,chen2015unconditional}. The hierarchy-Liouville duality used to establish uniqueness was introduced in \cite{ammari2020well}. Related mean-field Egorov formulations are discussed in \cite{frgrsc,frohlich2019microscopic}.

\bigskip

The three limits considered in this paper are executed successively. First, we take the limit \(N\to\infty\) at fixed \(t\) and fixed \(\eta\), yielding the effective mean-field particle–field system. Second, by setting \(\lam=\eta^2\) and \(T=\eta^2t\), the weak-coupling limit \(\eta\to0\) is taken at a fixed macroscopic time \(T\in[0,T_\eta ]\); the macroscopic time interval is defined with $T_\eta = C_1\log\langle1/\eta\rangle$ for a sufficiently small constant $C_1>0$ determined in \cite{chen2026trapped}, and tends to $\R_+$ as $\eta\rightarrow0$. This yields nonlinear resonance-cascade dynamics for \(F_T\) and the state \(\kappa_T\) in the interaction picture. Finally, the long-time limit \(T\to\infty\) determines the complete ground-state occupation and the asymptotic vanishing of the entropy. We do not claim a simultaneous limit in \(N\), \(\eta\), and \(T\). 
The remainder of the paper follows this logical sequence. First, we establish the estimates, compactness, and identification necessary for the mean-field limit. Next, we use the de Finetti representation and the uniqueness of the limiting hierarchy to identify the effective particle–field flow. The subsequent weak-coupling analysis rigorously isolates the nonlinear resonance cascade. Finally, the Fermi–Golden-Rule transfer, uniform control of high-energy modes, and entropy estimates are combined to demonstrate complete Bose–Einstein condensation and the convergence of the boson subsystem to a pure state as $T\rightarrow\infty$.

\subsection{Main results}
In this section, we give the main results of the paper. The first theorem demonstrates that the von Neumann entropy of the particle subsystem vanishes in the long-time limit, a direct consequence of radiative relaxation and complete Bose-Einstein condensation. The second theorem rigorously derives the effective particle-field dynamics from the microscopic many-body evolution, confirming that the mean-field limit correctly describes the system's behavior.

\begin{theorem}[From microscopic dynamics to condensation]
\label{cor:microscopic-to-cascade}
Fix \(\eta_*>0\) and set \(\lambda=\eta^2\) for
\(0<\eta\leq\eta_*\). Assume that the trapping and interaction potentials
satisfy the hypotheses of the weak-coupling cascade theorem in
\cite[Theorem~1.2]{chen2026trapped}, and that
\hyporef{hyp:standing-hypotheses} holds for each fixed
\(\eta\). Let \(H_N^\eta\) denote the corresponding microscopic Hamiltonian.
Let
\begin{equation}
  \Psi_{N,0}^{\eta}
  =\varphi_{N,0}^{\eta}\otimes
    W[\sqrt{N}\,u_0]\Omega_f,
  \qquad
  \|\varphi_{N,0}^{\eta}\|_{\mathcal H_N}=1,
\end{equation}
where the position-space field profile
\begin{equation}
  u_0\in H^{1/2}(\mathbb R^3)\mcap W^{3,1}(\mathbb R^3)
\end{equation}
is independent of \(N\) and \(\eta\).
The constants and lower thresholds on \(N\) in the initial moment
assumption \hypref{hyp:bounded-energy-per-particle} may depend on \(\eta\). Suppose that there is a phase-invariant Borel probability measure
\(\mu_0\), independent of \(\eta\), supported on
\begin{equation}
  \mathcal E_R
  :=\{\phi\in Y^1:\|\phi\|_h=1,\ \|\phi\|_{Y^1}\leq R\},
  \qquad R<\infty,
\end{equation}
such that, for every fixed \(\eta\) and every integer \(k\geq1\), the weak-* limit holds
\begin{equation}
\label{eq:bridge-initial-hierarchy}
  \gamma_{N,0}^{(k),\eta}
  \mathop{\rightharpoonup}^{*}
  \int_{\mathcal E_R}
       |\phi^{\otimes k}\rangle\langle\phi^{\otimes k}|
       \,d\mu_0(\phi)
  \qquad\text{as }N\to\infty, 
\end{equation}
against compact operators.
Here \(\gamma_{N,t}^{(k),\eta}\) is the $k$-particle marginal of
\(\Psi_{N,t}^{\eta}=e^{-itH_N^\eta}\Psi_{N,0}^{\eta}\). Let \(F_T(\phi)=(F_{j,T}(\phi))_{j\geq0}\) solve the limiting
resonance cascade \eqref{eq:intro-diagonal-cascade} with
\(F_{j,0}(\phi)=\langle\chi_j,\phi\rangle_h\), and define the macroscopic time density matrix
\begin{equation}
  \kappa_T:=
  \int_{\mathcal E_R}
       |\Phi_T(\phi)\rangle\langle\Phi_T(\phi)|
       \,d\mu_0(\phi),
\end{equation}
where 
\begin{equation}
      \Phi_T(\phi):=\sum_{j\geq0}F_{j,T}(\phi)\chi_j.
\end{equation}
Then, for every \(\tau<\infty\),
\begin{equation}
\label{eq:bridge-ordered-convergence}
  \lim_{\eta\downarrow0}\lim_{N\to\infty}
  \sup_{0\leq T\leq\tau}
  \left\|
    e^{iT \dch_0/\eta^2}\gamma_{N,T/\eta^2}^{(1),\eta}
    e^{-iT \dch_0/\eta^2}-\kappa_T
  \right\|_{\mathfrak S_1(h)}
  =0.
\end{equation}
If, furthermore,
\begin{equation}
  |\langle\chi_0,\phi\rangle_h|>0
  \quad\text{for }\mu_0\text{-almost every }\phi,
  \qquad
  \Gamma_{0,j}^{\mathrm{FGR}}>0
  \quad\text{for every }j\geq1,
\end{equation}
then
\begin{equation}
\label{eq:bridge-ordered-condensation}
  \lim_{T\to\infty}\lim_{\eta\downarrow0}\lim_{N\to\infty}
  \left\|
    e^{iT \dch_0/\eta^2}\gamma_{N,T/\eta^2}^{(1),\eta}
    e^{-iT \dch_0/\eta^2}
    -|\chi_0\rangle\langle\chi_0|
  \right\|_{\mathfrak S_1(h)}
  =0.
\end{equation}
In particular, the von Neumann entropy of $\kappa_T$,
\begin{equation}
    S(\kappa)=-\operatorname{Tr}_h(\kappa\log\kappa) \,,
\end{equation} 
satisfies
\begin{equation}
\label{eq:bridge-cascade-entropy}
  \lim_{T\to\infty}S(\kappa_T)=0
  \,.
\end{equation}
The microscopic limits in these statements are taken in the displayed
order; first \(N\to\infty\), then \(\eta\downarrow0\), and finally
\(T\to\infty\) in \eqref{eq:bridge-ordered-condensation}.
\end{theorem}

\begin{proof}
For each initial orbital \(\phi\), let
\((\phi_t^\eta(\phi),u_t^\eta(\phi))\) be the coupled mean-field
solution of \eqref{eq:coupled-system-intro} with initial data
\((\phi,u_0)\) and coupling
\(\lambda=\eta^2\). Set
\begin{equation}
  \Phi_T^\eta(\phi)
  :=e^{iT \dch_0/\eta^2}\phi_{T/\eta^2}^\eta(\phi),
  \qquad
  \kappa_T^\eta
  :=\int_{\mathcal E_R}
       |\Phi_T^\eta(\phi)\rangle\langle\Phi_T^\eta(\phi)|
       \,d\mu_0(\phi).
\end{equation}

For fixed \(\eta>0\), \thmref{thm:uniqueness-of-gp-hierarchy} applies on the finite physical time interval \([0,\tau/\eta^2]\). Unitary invariance of the trace
norm therefore yields
\begin{equation}
\label{eq:bridge-fixed-coupling}
  \lim_{N\to\infty}
  \sup_{0\leq T\leq\tau}
  \left\|
    e^{iTh_0/\eta^2}\gamma_{N,T/\eta^2}^{(1),\eta}
    e^{-iTh_0/\eta^2}-\kappa_T^\eta
  \right\|_{\mathfrak S_1(h)}
  =0.
\end{equation}

The weak-coupling convergence in \cite[Theorem~1.2]{chen2026trapped},
expressed in the trap eigenbasis, gives
\begin{equation}
  \sup_{0\leq T\leq\tau}
  \|\Phi_T^\eta(\phi)-\Phi_T(\phi)\|_h
  \longrightarrow0
  \qquad\text{as }\eta\downarrow0
\end{equation}
for \(\mu_0\)-almost every \(\phi\).
Both vectors have norm one. For unit vectors \(f,g\in h\),
\begin{equation}
  \bigl\||f\rangle\langle f|-|g\rangle\langle g|\bigr\|_{\mathfrak S_1(h)}
  \leq 2\|f-g\|_h.
\end{equation}
Consequently,
\begin{equation}
  \sup_{0\leq T\leq\tau}
  \|\kappa_T^\eta-\kappa_T\|_{\mathfrak S_1(h)}
  \leq
  2\int_{\mathcal E_R}
       \sup_{0\leq T\leq\tau}
       \|\Phi_T^\eta(\phi)-\Phi_T(\phi)\|_h
       \,d\mu_0(\phi)
  \longrightarrow0.
\end{equation}
The last step uses dominated convergence; the integrand is bounded by
\(2\), and it is measurable because the trajectories are continuous
and the supremum can be taken over a countable dense set of times.
Combining this limit with \eqref{eq:bridge-fixed-coupling} proves
\eqref{eq:bridge-ordered-convergence}. 

Under the additional ground state and transition-rate assumptions,
\thmref{thm:entropy-boson-subsystem} gives
\begin{equation}
  \|\kappa_T-|\chi_0\rangle\langle\chi_0|\|_{\mathfrak S_1(h)}
  \longrightarrow0,
  \qquad
  S(\kappa_T)\longrightarrow0.
\end{equation}
For each fixed \(T\), the reverse triangle inequality and
\eqref{eq:bridge-ordered-convergence} identify the first two limits
in \eqref{eq:bridge-ordered-condensation} with
\(\|\kappa_T-|\chi_0\rangle\langle\chi_0|\|_{\mathfrak S_1(h)}\).
Letting \(T\to\infty\) proves the remaining assertions.
\end{proof}

\begin{theorem}[Entropy and pure state formation in the the boson subsystem]
\label{thm:entropy-boson-subsystem}
Suppose that \hyporef{hyp:standing-hypotheses} holds, along with the regularity and non-resonance assumptions from \cite{chen2026trapped}. Let $\set{\rchi_j}_{j\geq0}$ be the orthonormal complete basis of the one-particle Hamiltonian $\dch_0=-\Delta+V$, where $\rchi_0$ corresponds to the ground state. By \corref{cor:bounded-energy-support-de-finetti}
\begin{equation}
    \mu_0\!\left(
        \left\{
            \varphi\in Y^1\mcap\mathbb S\h:
            \|\OT\varphi\|^2\leq C_E
        \right\}
    \right)=1.
\end{equation}
Thus, the bounded initial energy-support condition below holds with $R>\sqrt{C_E}$.
Let $\mu_0$ be the initial de Finetti measure, supported on a ball of finite radius \(B_R(0)\subset Y^1\mcap\mathbb S\h\), and let $\kappa_T$ be the non-diagonal density matrix of the subsystem defined by
\begin{equation}
\label{eq:def-kappa-T}
\kappa_T := \sum_{j,\ell=0}^{\infty} \left( \int_{B_R(0)\mcap \, \mathbb{S}\h} F_{j,T}(\vfi)\cc{F_{\ell,T}(\vfi)} \,d\mu_0(\vfi) \right) \ket{\rchi_j}\bra{\rchi_\ell},
\end{equation}
where $F_{j,T}(\vfi)$ denotes the $j$-th rescaled amplitude at the macroscopic time $T$ satisfying the nonlinear cascade equation \eqref{eq:intro-diagonal-cascade} introduced by the authors in \cite{chen2026trapped}, with initial value $F_{j,0}(\vfi)=\inp{\rchi_j,\vfi}_{\h}$. 
Assume, in addition, that for $\mu_0$-almost every initial state $\vfi$, the ground state component is non-vanishing, such that $|F_{0,0}(\vfi)| > 0$. Furthermore, assume the positivity of all required downward Fermi-Golden-Rule rates, $\Gamma_{0,j'}^{\mr{FGR}} > 0$ for all $j' \ge 1$.
Then, as $T \to \infty$, the one-particle density matrix at time $T$ converges in trace norm to the pure state corresponding to the ground state,
\begin{equation}
\label{eq:trace-norm-convergence}
\lim_{T \to \infty} \left\| \kappa_T - \ket{\rchi_0}\bra{\rchi_0} \right\|_{\mathfrak S_1(\h)} = 0,
\end{equation}
under conservation of the trace,
\begin{eqnarray}
    \Tr(\kappa_T) = \Tr(\kappa_0) 
    \;\;\forall\,T>0 \,,
\end{eqnarray}
and the von Neumann entropy of the subsystem vanishes asymptotically,
\begin{equation}
    \lim_{T \to \infty} S[\kappa_T] = 0.
\end{equation}
This establishes dynamical Bose-Einstein formation in the boson subsystem with 100\% ground-state occupation, under conservation of the boson mass and entropy decrease to zero.
\end{theorem}

\begin{proof}
    We recall some properties of the von Neumann entropy. For a state $\gamma$ defined as the weak-$\ast$ limit of the sequence $\gamma_{N,t}^{(1)}$ in $\mathcal{L}^1(\mathfrak{H})$, the entropy is given by $S[\gamma] = -\Tr[\gamma \log \gamma]$. If $\gamma$ is diagonal in some orthonormal basis, that is, $\gamma = \sum_{j=0}^{\infty} p_j \ket{\rchi_j} \bra{\rchi_j}$, the entropy simplifies to the Shannon entropy of its eigenvalues
    \begin{equation}
        S[\gamma] = -\sum_{j=0}^{\infty} p_j \log p_j.
    \end{equation}
    In the extreme case where $\gamma$ is a pure state (e.g., $p_0=1$ and $p_j=0$ for all $j \ge 1$), the entropy vanishes completely.
    
    Let $p_j(T)$ denote the ensemble-averaged occupation probability of the $j$-th mode at time $T$, defined as
    \begin{equation}
        p_j(T) := \int_{Y^1 \mcap \mathbb{S}\h} |F_{j,T}(\vfi)|^2 \, d\mu_0(\vfi).
    \end{equation}
    By the pinching inequality \cite{zhang2016quantifying, carlen2010trace, wehrl1978general, nielsen2001quantum, yau1991relative}, the von Neumann entropy is bounded above by the Shannon entropy of its diagonal part
    \begin{equation}
        0\leq S[\kappa_T] \leq S[\kappa_T^{\mr{diag}}] = \sum_{j=0}^\infty (-p_j(T) \log p_j(T)) \,.
    \end{equation}
    
    \noindent
    The nonlinear cascade equation derived in \cite{chen2026trapped} guarantees the pathwise convergence of occupations; for $\mu_0$-almost every initial state $\vfi$, the mass transfers entirely to the ground state such that $\sum_{j\ge 1}|F_{j,T}(\vfi)|^2 \to 0$ and $|F_{0,T}(\vfi)|^2 \to 1$ as $T \to \infty$. Because the amplitude functions are uniformly bounded by $1$ and $\mu_0$ is a probability measure, dominated convergence allows us to pass the limit inside the integral to obtain
    \begin{equation}
        \lim_{T \to \infty} p_0(T) = \int_{B_R(0)\subset Y^1 \mcap \mathbb{S}\h} \lim_{T \to \infty} |F_{0,T}(\vfi)|^2 \, d\mu_0(\vfi) = 1.
    \end{equation}

    \noindent
    To establish trace-norm convergence, we evaluate the state's fidelity with respect to the ground state. Since $p_0(T) = \inp{\rchi_0, \kappa_T \rchi_0}$, we can utilize the standard trace-distance bound for pure states
    \begin{equation}
        \left\| \kappa_T - \ket{\rchi_0}\bra{\rchi_0} \right\|_{\mathfrak S_1(\h)} \le 2\sqrt{1 - p_0(T)}.
    \end{equation}
    Because $p_0(T) \to 1$, the right-hand side vanishes, proving that $\kappa_T$ converges in trace norm to $\ket{\rchi_0}\bra{\rchi_0}$ without requiring the convergence of the $U(1)$ phase of the ground state amplitude.
    
    To prove the asymptotic vanishing of the entropy, we analyze the diagonal entropy sum. 
    %Let $\{E_k\}_{k\geq0}$ denote the spectrum of $-\Delta+V$ with $V(x)\geq -C_0+c_0|x|^{1+\beta}$ with $\beta>0$.
    The quantum de Finetti measure $\mu_0$ is supported in a ball of finite radius $B_R(0)\subset Y^1$. This implies that $\mu_0$-almost surely, $\|\phi\|_{Y^1}\leq R$. Therefore, the spectral tightness bounds proven in \cite[Lemma~4.2]{chen2026trapped} implies that for all $K>K_0\gg1$ sufficiently large, and uniformly in $T>0$,
    \begin{eqnarray}
        P_K(T):=\sum_{j>K}p_j(T) \leq C K^{-\theta} 
    \end{eqnarray}
    for a constant $\theta>0$ and $C=C(R)$ (in fact, $\theta=\frac{2(1+\beta)}{3(3+\beta)}$ in dimension $d=3$). For convenience, we omit the dependence of $p_j$ and $P_K$ on $T$ from the notation in what follows. 

If $P_K=0$, the entropy tail vanishes. We therefore assume
$P_K>0$ in what follows.
    Let $q_j:=\frac{p_j}{P_K}$ so that $\sum_{j>K}q_j=1$. We now consider 
    \begin{eqnarray}
        \sum_{j>K}(p_j \log(1/p_j) = P_K\sum_{j>K}(q_j \log(1/ q_j)) + P_K\log(1/P_K) \,
    \end{eqnarray}
    where $P_K\log(1/P_K)\leq C K^{-\theta}\log K$.

    Next, we observe that since $p_j=P_{j-1}-P_j$,
    \begin{eqnarray}
        \sum_{j>K} p_j\log(j/K) &=& \sum_{j\geq K}P_j \Big(\log\Big(\frac{j+1}K\Big)-\log\Big(\frac jK\Big)\Big)
        \;=\; \sum_{j\geq K}P_j \log\Big(1+\frac1 j\Big)
        \nonumber\\
        &\leq& C\sum_{j\geq K}j^{-1-\theta} 
        \;\leq\; C K^{-\theta} \,.
    \end{eqnarray}
    Thus, dividing by $P_K$ yields
    \begin{eqnarray}\label{eq:qlogr-sum-1-0}
        \sum_{j>K} q_j\log j  
        \;\leq\; \log K + C \frac{K^{-\theta} }{P_K}\,.
    \end{eqnarray}
    
    Moreover, we observe that for any $\{r_j\}_{j>K}$ with $\sum_{j>K}r_j\leq1$,
    \begin{eqnarray}\label{eq:qlogr-sum-1-1}
        \sum_{j>K}q_j \log(1/ q_j)
        &=&\sum_{j>K}q_j \log(1/r_j) + \sum_{j>K}q_j \log(r_j/ q_j)
        \nonumber\\
        &\leq&\sum_{j>K}q_j \log(1/r_j)
    \end{eqnarray}
    because the last term on the first line is non-positive. To see this, consider the random variable $X=\frac{r_j}{q_j}$ on the set $S:=\{j\,|\,q_j>0\}$. Then, by the concavity of the logarithm and Jensen's inequality,
    \begin{eqnarray}
        \sum_{j>K}q_j \log(r_j/ q_j) &=& \mathbb{E}_q[\log X] 
        \leq \log(\mathbb{E}_q[X] ) 
        \nonumber\\
        &=& \log\sum_{j>K}  r_j \leq 0 \,.
    \end{eqnarray}

    We choose $r_j:=\frac{K}{j^2}$ so that $\sum_{j>K}r_j\leq\int_K^\infty dx\;K/x^2\leq 1$. Then, 
    \begin{eqnarray}
        \sum_{j>K}q_j \log(1/ q_j)
        &\leq& \sum_{j>K}q_j (2\log j-\log K)
        \nonumber\\
        &\leq& 2 \log K + C\frac{K^{-\theta}}{P_K}-\log K \,,
    \end{eqnarray}
    using \eqref{eq:qlogr-sum-1-0} in the last step.
    Multiplying by $P_K$ yields
    \begin{eqnarray}\label{eq:tailentropy-bd-1-0-0}
        \sum_{j>K}p_j \log(1/ p_j)
        &\leq&  P_K\log (1/P_K) + P_K \log K + C K^{-\theta}
        \nonumber\\
        &\leq & C K^{-\theta}\log K\,,
    \end{eqnarray}
    which tends to $0$ as $K\rightarrow\infty$. The constant neither depends on $K$ nor on the macroscopic time $T>0$.
    
    Given an arbitrary $\epsilon>0$, we choose $K=K(\epsilon)$ large enough so that $\eqref{eq:tailentropy-bd-1-0-0}<\epsilon$. Then,
    we split the entropy sum into
    \begin{eqnarray}\label{eq:entropy-sum-1-0-0}
        S[\kappa_T^{\mr{diag}}] = \sum_{j=0}^K (-p_j(T) \log p_j(T) )
       +
       \sum_{j > K} (-p_j(T) \log p_j(T))
    \end{eqnarray}
    Because $p_j(T) \to 0$ for all $j \ge 1$ and $p_0(T) \to 1$ as $T \to \infty$, the finite sum for $j \le K$ in \eqref{eq:entropy-sum-1-0-0} vanishes as $T \to \infty$. Therefore, we conclude that 
 \begin{eqnarray}
       0\leq \lim_{T\rightarrow\infty} S[\kappa_T^{\mr{diag}}]  
       &< &
    \epsilon \,,
    \end{eqnarray} 
    due to \eqref{eq:tailentropy-bd-1-0-0}.
    Since $\epsilon>0$ is arbitrary, this proves that 
    \begin{equation}
        0\leq \lim_{T \to \infty} S[\kappa_T] \le \lim_{T \to \infty} S[\kappa_T^{\mr{diag}}] = 0 \,.
    \end{equation}
    Thus, we conclude that dynamical Bose-Einstein condensation with asymptotically complete ground-state occupation occurs with the convergence of $\kappa_T$ to a pure state and the decay of the von Neumann entropy of the macroscopic state to zero.
\end{proof}

For a factorized initial boson state, the phase-invariant initial de Finetti measure is the uniform probability measure on the phase orbit of the initial orbital. The corresponding limiting particle marginals remain factorized, as stated in the following theorem.

\begin{theorem}[Derivation of the mean-field limit]
\label{thm:main-derivation-mean-field}
Suppose that \hyporef{hyp:standing-hypotheses} holds. Let $\Psi_{N,0} \in \HN$ be the sequence of initial microscopic states given by the coherent factorized ansatz
\begin{equation}
    \Psi_{N,0} = \vfi_0^{\otimes N} \otimes \wyl[\sqrt{N} u_0] \vacf,
\end{equation}
with $\vfi_0 \in Y^1 \mcap \mathbb{S}\h$ and $u_0 \in \cD(\omg^{1/2})$. Let $\Psi_{N,t} = \UN(t,0) \Psi_{N,0}$ be the microscopic time evolution under the Hamiltonian $\HMN$, and let $\gamma_{N,t}^{(k)}$ denote the corresponding $k$-particle reduced density matrices. Assume that the effective coupled mean-field system defined in \eqref{eq:coupled-system}, namely
\begin{equation}
    \label{eq:coupled-system-intro}
    \begin{cases}
        i\partial_t \vfi_t(x) = \dch_0(x) \vfi_t(x) + \lambda (v * |\vfi_t|^2)(x) \vfi_t(x) + \eta F_{\vfi_t}(x) \vfi_t(x), \vspace{8pt}\\
        i\partial_t u_t(y) = \omg (-i\nabla_y) u_t(y) + \eta (w * |\vfi_t|^2)(y)
    \end{cases}
\end{equation}
admits a unique global mild solution $(\vfi_t, u_t) \in C([0,T]; Y^1 \times H^{1/2})$. Then, for every $k \in \N$ and every $t \in [0,T]$, the sequence of reduced density matrices converges in the trace-norm topology to the purely factorized state
\begin{equation}
    \lim_{N \to \infty} \norm{ \gamma_{N,t}^{(k)} - \ket{\vfi_t^{\otimes k}}\bra{\vfi_t^{\otimes k}} }_{\mathfrak{S}_1(\Hkb)} = 0.
\end{equation}
Consequently, the mean-field behavior of the $N$-boson subsystem is completely described by the nonlinear flow of the effective Hartree system.
\end{theorem}

\begin{proof}
    The proof synthesizes the compactness, identification, and uniqueness results established in the subsequent sections of the paper.

    \noindent
    \textbf{Step 1: Compactness.} By the uniform first-energy estimate in
\corref{cor:one-particle-energy-bound}
and the source bound in
\lemref{lem:uniform-field-source-bound}, \thmref{thm:compactness-bbgky-sequence} guarantees the existence of a subsequence $N_\ell \to \infty$ such that the reduced density matrices $\gamma_{N_\ell, t}^{(k)}$ converge in the weak-$\ast$ topology of $\mathfrak{S}_1(\Hkb)$ to a limit point $\gamma_{\infty, t}^{(k)}$ for every $k \geq 1$. Furthermore, by \propref{prop:admissibility-limit}, this sequence of limiting marginals $\Gamma_{\infty,t}:=\set{\gamma_{\infty, t}^{(k)}}_{k \geq 1}$ is admissible.

    \smallskip
    \noindent
    \textbf{Step 2: Identification.} Passing to the limit $N \to \infty$ in the BBGKY hierarchy, \corref{cor:explicit-gp-hierarchy} establishes that the limit points $\Gamma_{\infty,t}$ explicitly satisfy the infinite Hartree hierarchy, where the radiation source is transported classically.

    \smallskip
\noindent
\textbf{Step 3: Joint uniqueness and factorization.}
The factorized initial particle state has marginals
\begin{equation}
    \gamma_{N,0}^{(k)}
    =
    |\vfi_0^{\otimes k}\rangle
    \langle\vfi_0^{\otimes k}|,
    \qquad N\geq k.
\end{equation}
In the phase-invariant de Finetti convention, their
representing measure is
\begin{equation}
    \mu_0
    =
    \frac{1}{2\pi}
    \int_0^{2\pi}
        \delta_{e^{i\theta}\vfi_0}\,d\theta.
\end{equation}
Indeed, for every $k\geq1$,
\begin{equation}
    \int_{\mathbb S\h}
        |\vfi^{\otimes k}\rangle
        \langle\vfi^{\otimes k}|
        \,d\mu_0(\vfi)
    =
    |\vfi_0^{\otimes k}\rangle
    \langle\vfi_0^{\otimes k}|.
\end{equation}
The coherent initial field gives the initial joint law
\begin{equation}
    \nu_0=\mu_0\otimes\delta_{\widehat u_0}.
\end{equation}
By \thmref{thm:joint-measure-transport}, every extracted
joint limit satisfying the established Liouville
equation and moment bounds is
\begin{equation}
    \nu_t
    =
    (\mathcal U_t)_{\#}
    \bigl(\mu_0\otimes\delta_{\widehat u_0}\bigr).
\end{equation}
The coupled flow is equivariant under multiplication
of the particle state by a constant phase.
Consequently, the trajectory starting at
$(e^{i\theta}\vfi_0,\widehat u_0)$ has particle
component $e^{i\theta}\vfi_t$, and its field
component is independent of $\theta$.
Here the flow agrees with the solution assumed
in the theorem by uniqueness. Thus, as also stated in
\corref{cor:coupled-mean-field-system},
\begin{equation}
\begin{aligned}
    \gamma_{\infty,t}^{(k)}
    &=
    \frac{1}{2\pi}
    \int_0^{2\pi}
       |(e^{i\theta}\vfi_t)^{\otimes k}\rangle
       \langle(e^{i\theta}\vfi_t)^{\otimes k}|
       \,d\theta=
    |\vfi_t^{\otimes k}\rangle
    \langle\vfi_t^{\otimes k}|.
\end{aligned}
\end{equation}
The full-sequence initial convergence required by
\thmref{thm:uniqueness-of-gp-hierarchy} holds because
the initial particle marginals are factorized.
That theorem therefore identifies the same limiting
particle hierarchy for the entire microscopic
sequence.

    \smallskip
    \noindent
    \textbf{Step 4: Strong convergence.} Finally, since the limiting density matrices $\gamma_{\infty, t}^{(k)}$ are non-negative and preserve the unit trace, i.e. $\Tr_{\Hkb}(\gamma_{\infty, t}^{(k)}) = 1$, \lemref{lem:weak-str-conv-and-traceclass-imply-trace-norm-conv} dictates that the weak-$\ast$ convergence automatically upgrades to strong trace-norm convergence. This yields
    \begin{equation}
        \lim_{N \to \infty} \norm{ \gamma_{N,t}^{(k)} - \ket{\vfi_t^{\otimes k}}\bra{\vfi_t^{\otimes k}} }_{\mathfrak{S}_1(\Hkb)} = 0,
    \end{equation} which completes the proof.
\end{proof}

\subsection{Strategy of the proof}
\smallskip

The argument first identifies the mean-field limit at fixed
coupling on finite time intervals. The condensation and entropy
conclusions then follow from the resonance cascade under the
additional assumptions stated in the main results.

\begin{enumerate}
    \item \textbf{Compactness and de Finetti representation.}
    The uniform first-energy and photon-source estimates give
    compactness of the particle marginals.
    Confinement prevents loss of trace, and the limiting
    hierarchy is admissible.
    The quantum de Finetti theorem therefore represents its
    marginals by phase-invariant probability measures $\mu_t$.
    The multiplicative particle-energy estimates further imply
    the bounded energy support established in
    \corref{cor:bounded-energy-support-de-finetti}.

    \item \textbf{Joint limits and the Liouville equation.}
    To identify the radiation source, we retain particle--field
    correlations through joint observables consisting of
    particle tests and scaled field Weyl operators.
    Their limits define a joint probability measure $\nu_t$
    whose particle marginal is $\mu_t$.
    The moment bounds allow recovery of the field insertions
    appearing in the source, as proved in
    \propref{prop:joint-radiation-source-limit}.
    Passing to the limit in the microscopic mixed-observable
    identities then yields the joint Liouville equation in
    the interaction picture, by
    \propref{prop:joint-compactness-and-Liouville-equation}.

    \item \textbf{Transport, uniqueness, and convergence.}
    We construct the global coupled particle--field flow
    $\mathcal U_t$ and apply the superposition principle to
    the joint Liouville equation.
    Uniqueness of the characteristic trajectories identifies
    every extracted joint measure as
    \begin{equation}
        \nu_t
        =
        (\mathcal U_t)_{\#}
        \bigl(\mu_0\otimes\delta_{\widehat u_0}\bigr),
    \end{equation}
    as established in
    \thmref{thm:joint-measure-transport}.
    This determines the limiting particle marginals and
    radiation source.
    When the initial particle hierarchy converges along the
    full sequence, uniqueness in the stated joint-measure
    class yields full-sequence convergence, uniformly on
    compact time intervals in particle trace norm, by
    \thmref{thm:uniqueness-of-gp-hierarchy}.
    Pure initial condensates give factorized limiting
    particle marginals.

    \item \textbf{Condensation and vanishing of entropy.}
    Under the additional weak-coupling and cascade assumptions,
    we set $\lambda=\eta^2$ and combine the mean-field limit
    with the resonance-cascade limit.
    The limits are taken in the order $N\to\infty$,
    then $\eta\downarrow0$, and finally $T\to\infty$.
    Pathwise transfer of mass to the ground state gives
    trace-norm convergence of the ensemble-averaged cascade
    state to the ground-state projector.
    A uniform bound on the diagonal entropy tails supplies
    the additional control needed to prove
    $S(\kappa_T)\to0$, as shown in
    \thmref{thm:entropy-boson-subsystem}.
\end{enumerate}

\subsection{Organization of the paper}
The remainder of the paper is organized as follows to reflect this successive logical progression. 
In \secref{sec:finite-N-BBGKY-hierarchy}, we introduce the BBGKY hierarchy and define the rigorous framework, including the relevant Hamiltonians and coherent initial data. We then establish the \textit{a priori} energy estimates and uniform bounds on the coherent photon source necessary for the mean-field limit, as carried out in \secref{sec:existence-and-uniqueness-of-GP-hierarchy}. We prove the weak-$*$ compactness of the BBGKY sequence and utilize quantum de Finetti theorems to represent the admissible limit points in phase space. Finally, we identify the radiation source and establish the unconditional uniqueness of the resulting infinite Gross--Pitaevskii hierarchy.

\section{The BBGKY hierarchy}
\label{sec:finite-N-BBGKY-hierarchy}

\subsection{Hilbert space and Hamiltonians}

Let $\h_x = L^2(\R^3, dx)$ denote the Hilbert space of square-integrable functions on $\R^3$ with respect to the Lebesgue measure, where $x$ refers to the position of the particles. Define $\HNb$, the Hilbert space of the $N$-particle sector of the bosonic subsystem, as the symmetric tensor product of $N$ copies of $\h_x$
\begin{equation}
    \label{eq:definition-of-HNb}
    \HNb = \bigotimes_\sym^N \prc{\h_x} \cong L^2_\sym(\R^{3N}, dx_1\cdots dx_N).
\end{equation}
The subscript $\sym$ indicates that $L^2_\sym(\R^{3N}, dx_1\cdots dx_N)$ is the subspace of $L^2(\R^{3N}, dx_1\cdots dx_N)$ comprising symmetric functions---namely, wave functions that are invariant under the permutation of particle coordinates. Consequently, we define the bosonic Fock space for the particles as
\begin{equation}
    \label{eq:definition-of-bosonic-Fock-space}
    \fkH := \bigoplus_{N=0}^{\infty} \HNb.
\end{equation}
The inner product on $\fkH$ is naturally induced by the sequence of $N$-particle inner products, that is,
\begin{equation}
    \label{eq:inner-product-on-HN}
    \inp{\Phi, \Psi}_{\fkH} := \sum_{N=0}^{\infty} \inp{\Phi_N, \Psi_N}_{\HNb},
\end{equation}
where
\begin{equation}
    \label{eq:inner-product-on-HNb}
    \inp{\Phi_N, \Psi_N}_{\HNb} := \int_{\R^{3N}} \cc{\Phi_N(x_1,\ldots,x_N)} \Psi_N(x_1,\ldots,x_N) \, dx_1\cdots dx_N.
\end{equation}

We model the radiation field with scalar photons quantized in the Fock representation. We denote the one-photon Hilbert space by $\h_f$, where the subscript $f$ distinguishes the field from the particle degrees of freedom. This space consists of square-integrable functions in either position or momentum space
\begin{equation}
    \h_f := L^2(\R^3_y, dy) \cong L^2(\R^3_\xi, d\xi).
\end{equation} 
The Hilbert space of the $n$-photon sector is the symmetric tensor product
\begin{equation}
    \cF_n := \bigotimes_\sym^n \h_f.
\end{equation} 
Then, the Fock space of the quantized radiation field is defined as the direct sum of all $n$-photon sectors
\begin{equation}
    \label{eq:Fock-space-of-scalar-photons}
    \Fock = \bigoplus_{n=0}^{\infty} \cF_n.
\end{equation}
The inner product on $\Fock$ is defined as
\begin{equation}
    \label{eq:inner-product-on-Fock-space}
    \inp{F, G}_{\Fock} := \sum_{n=0}^{\infty} \inp{F_n, G_n}_{\cF_n},
\end{equation}
where
\begin{equation}
    \inp{F_n, G_n}_{\cF_n} := \int_{\R^{3n}} \cc{F_n(y_1,\ldots,y_n)} G_n(y_1,\ldots,y_n) \, dy_1\cdots dy_n.
\end{equation}
We adopt the standard convention that the zero-photon sector is $\C$. While physical photons possess two transverse polarizations, incorporating both would necessitate replacing $\h_f$ with $\h_f \otimes \C^2$. For notational convenience and to reduce complexity that would cloud the main principles, we restrict our attention to scalar photons linearly coupled to the bosons. The Fock vacuum is denoted by $\vacf \in \Fock$. Hereafter, we drop the subscripts $f$ and $x$ from $\h$ when the context unambiguously identifies the relevant Hilbert space.

\smallskip

We now consider the composite system of $N$ confined bosons coupled to the quantized radiation field. The Hilbert space of the full system is the tensor product of the $N$-boson space and the radiation Fock space
\begin{equation}
    \label{eq:full-Hilbert-space-of-boson-field-system}
    \HN = \HNb \otimes \Fock.
\end{equation}
The Hamiltonian governing the dynamics of this interacting system is given by
\begin{equation}
    \label{eq:defination-of-Hamiltonian-of-particle-field-system}
    \HMN := \HMNbos \otimes \1_{\Fock} + \1_{\HNb} \otimes \Hf + \frac{\eta}{\sqrt{N}} W_N.
\end{equation}
Here, $\HMNbos$ is the Hamiltonian of the $N$-boson subsystem acting on $\HNb$, namely
\begin{equation}
    \label{eq:Hamiltonian-of-N-bosons}
    \HMNbos = \sum_{j=1}^{N} \dch_0^{(j)} + \frac{\lam}{N} \sum_{1\leq i<j\leq N} v(x_i-x_j),
\end{equation}
where $\dch_0 := -\Delta + V$ is the one-particle Hamiltonian with an external trapping potential $V:\R^3\to[0,\infty)$. The operator $\dch_0^{(j)}$ represents $\dch_0$ acting on the $j$-th particle, and $v_{ij} = v(x_i-x_j)$ is a real-valued two-body interaction potential.

\smallskip
The field Hamiltonian $\Hf$ acts on $\Fock$ and is defined as
\begin{equation}
    \label{eq:definition-of-Hf}
    \Hf = \int_{\R^3} \omg(\xi) \, \ad_\xi \, a_\xi \, d\xi,
\end{equation}
where the dispersion relation $\omg:\R^3\to[0,\infty)$ will be assumed to be given by $\omega(\xi)=|\xi|$. The creation and annihilation operators, $\ad_\xi$ and $a_\xi$, are operator-valued distributions on $\Fock$
%Formally, $a_\xi := a(\delta_\xi)$ and $\ad_\xi := \ad(\delta_\xi)$ represent the operators evaluated at the momentum $\xi \in \R^3$, where $\delta_\xi$ is the Dirac delta distribution. 
satisfying the distributional canonical commutation relations:
\begin{equation}
    [a_\xi, \ad_{\xi'}] = \delta(\xi-\xi'), \qquad [a_\xi, a_{\xi'}] = [\ad_\xi, \ad_{\xi'}] = 0
\end{equation}
and $a_\xi\vacf=0$ for all $\xi\in\R^3$, where $\vacf\in\Fock$ is the Fock vacuum.
Smeared out with test functions $h \in C_0^\infty(\R^3)$, they define annihilation and creation operators 
\begin{equation}
    \label{eq:definition-of-creation-and-annihilation-operators-of-the-field}
    a(h) = \int_{\R^3} \cc{h(\xi)} a_\xi \, d\xi, \qquad \ad(h) = \int_{\R^3} h(\xi) \ad_\xi \, d\xi
\end{equation}
on $\Fock$.
The photon number operator is given by
\begin{equation}
    \label{eq:number-operator-of-field-bosons}
    \cN_f = \int_{\R^3} \ad_\xi \, a_\xi \, d\xi.
\end{equation}
Throughout this work, we employ the Fourier transform conventions
\begin{equation}
    d\xi := (2\pi)^{-3} d^3\xi, \qquad dy := d^3y,
\end{equation}
with the position-space operators defined as
\begin{equation}
    \label{eq:Fourier-transform-of-creation-and-annihilation-operators}
    a_y = \int_{\R^3} e^{i\xi\cdot y} a_\xi \, d\xi, \qquad a_\xi = \int_{\R^3} e^{-i\xi\cdot y} a_y \, dy.
\end{equation}
The particle-field interaction operator is defined by
\begin{equation}
    \label{eq:W(N) interaction Hamiltonian}
    W_N := \sum_{j=1}^{N} \phi\prc{w(x_j-\cdot)} = \sum_{j=1}^{N} \int_{\R^3} w(x_j-y) \prc{\ad_y + a_y} \, dy,
\end{equation}
where
\begin{equation}
    \phi(h) := \ad(h) + a(h)
\end{equation}
is the field operator smeared with $h \in \h_f$, and $w:\R^3\to\R$ represents the particle-field interaction form factor. The coupling constants $ \lambda>0$ and $\eta>0$ account for the strength of the two-body and particle-field interactions, respectively. The  $1/\sqrt{N}$ scaling in the interaction term in \eqref{eq:full-Hilbert-space-of-boson-field-system} is introduced to achieve a well-defined mean-field limit as $N \to \infty$.

\subsection{Hypothesis}
For our main theorems, we impose the following assumptions, which hold throughout this paper.
\begin{hypothesis}
\label{hyp:standing-hypotheses}
Throughout the work, the following assumptions are imposed

\begin{enumerate}[
label=\textup{$(\mathrm{H}_{\arabic*})$},
ref=\textup{$(\mathrm{H}_{\arabic*})$}
]
\item
\label{hyp:trapping-potential}
\textbf{Trapping potential:}
Let $A$ be a fixed real symmetric positive definite $3\times3$ matrix.
Assume that
\begin{equation}
\label{eq:moment-trap-decomposition}
    V(x)=x\cdot Ax+U(x)\geq0,
    \qquad U:\R^3\longrightarrow\R,
\end{equation}
where, for some $a_*>0$,
\begin{equation}
\label{eq:moment-trap-fourier}
    \int_{\R^3} e^{a_*|\xi|}
       |\widehat U(\xi)|\,d\xi<\infty.
\end{equation}
Here and below, $d\xi=(2\pi)^{-3}d^3\xi$, in accordance with the
Fourier convention of the manuscript. In particular, $U$ is smooth and
bounded. For any fixed $0<\beta\leq1$, the potential $V$ satisfies
\begin{equation}
    V(x)\geq c_0|x|^{1+\beta}-C_0
\end{equation}
for suitable $c_0,C_0>0$. Define $\dch_0=-\Delta+V$ by its
nonnegative closed quadratic form. Its resolvent is compact; write its
eigenvalues in nondecreasing order, counting multiplicity, as
$\{E_j\}_{j\geq0}$, and choose normalized eigenfunctions
$\{\chi_j\}_{j\geq0}$ forming an orthonormal basis.
For the initial-data condition, introduce the reference operators
\begin{equation}
\label{eq:moment-reference-oscillator}
    h_{\mathrm{ref}}:=-\Delta+x\cdot Ax,
    \qquad
    \mathbb K_N:=\sum_{j=1}^N h_{\mathrm{ref}}^{(j)}.
\end{equation}
The reference operator is distinct from $\dch_0$.

\item
\label{hyp:spectral-genericity}
\textbf{Spectral assumptions:}
The eigenvalues and eigenfunctions of the actual operator
$\dch_0=-\Delta+V$ satisfy the following conditions.
\begin{enumerate}
\item
The eigenvalues are rationally independent. More precisely, every
finitely supported integer sequence $(n_j)_{j\geq0}$ satisfies
\begin{equation}
    \sum_{j\geq0}n_jE_j=0
    \quad\Longrightarrow\quad
    n_j=0\quad\hbox{for every }j.
\end{equation}
Moreover, for every integer $M\geq1$, the nonzero four-mode gaps
\begin{equation}
\label{eq:moment-diophantine-gaps}
    \delta_M:=
    \min_{\substack{0\leq i,j,\ell,m\leq M\\
        E_i-E_j-E_\ell+E_m\ne0}}
       |E_i-E_j-E_\ell+E_m|
    \geq c_D(1+M)^{-\alpha}
\end{equation}
obey a Diophantine lower bound with fixed $c_D,\alpha>0$.
\item
For every $j,j'\geq0$ and every $R>0$,
$\widehat{\chi_j\overline{\chi_{j'}}}$ does not vanish identically
on the sphere $\{\xi\in\R^3:|\xi|=R\}$.
\end{enumerate}
These are additional assumptions on $V$, not conclusions of
\eqref{eq:moment-trap-decomposition}--\eqref{eq:moment-trap-fourier}.

\item
\label{hyp:bounded-two-body-interaction}
\textbf{Two-body interaction and coupling constants:}
The interaction kernels $v,w:\R^3\to\R$ are real-valued and even.
The coupling constants $\lam,\eta>0$ satisfy the standing weak-coupling
conditions. The kernels satisfy the following assumptions.
\begin{enumerate}
\item
The classical interaction satisfies $v\in W^{1,\infty}(\R^3)$ and
\begin{equation}
\label{eq:moment-pair-fourier}
    \int_{\R^3}e^{a_*|\xi|}|\widehat v(\xi)|\,d\xi<\infty.
\end{equation}
\item
For a fixed $\frac12<s<1$ with $s<\frac{1+\beta}{2}$,
\begin{equation}
\begin{aligned}
    &w\in L^2_s(\R^3)\mcap H^{1/2}(\R^3)
             \mcap W^{3,1}(\R^3)\mcap L^\infty(\R^3),\\
    &|\nabla|^{-1/2}w\in L^2_s(\R^3),
      \qquad |\widehat w(\xi)|>0\quad(\xi\in\R^3),
\end{aligned}
\end{equation}
where $L^2_s=\{f:\langle x\rangle^sf\in L^2\}$.
In addition,
\begin{equation}
\label{eq:moment-coupling-fourier}
    \int_{\R^3}e^{a_*|\xi|}|\widehat w(\xi)|^2\,d\xi<\infty.
\end{equation}
\end{enumerate}
The same positive $a_*$ may be used in these assumptions by decreasing
it if necessary. The potentials and their bounds are independent of $N$.

\item
\label{hyp:field-dispersion-relation}
\textbf{Dispersion relation of the radiation field:}
The dispersion relation function $\omega:\R^3 \to [0,\infty)$ is non-negative, measurable, and even. In particular, we assume that $\omg(\xi) = |\xi|$ for almost every $\xi\in\R^3$. We do not assume a strictly positive lower bound for $\omg$.

\item
\label{hyp:field-form-factor-and-coherent-profile}
\textbf{Energy regularity of the particle-field form factor:}
The assumptions on $w$ in
\hypref{hyp:bounded-two-body-interaction} imply
\begin{equation}
\label{eq:standing-field-form-factor-energy-condition}
    w\in H^1(\R^3)
       \mcap\dot H^{-1/2}(\R^3)
       \mcap\dot H^{1/2}(\R^3).
\end{equation}
Indeed, $W^{3,1}(\R^3)\subset H^1(\R^3)$, and the weighted
bound on $|\nabla|^{-1/2}w$ implies
$w\in\dot H^{-1/2}(\R^3)$.
For $\omg(\xi)=|\xi|$, we also have
\begin{equation}
    \norm{\widehat{\nabla w}/\sqrt{\omg}}_{L^2}
    =
    \norm{w}_{\dot H^{1/2}}.
\end{equation}

\item
\label{hyp:weak-coupling-condition}
\textbf{Weak-coupling condition:}
Define
\begin{equation}
\label{eq:hyp-def-of-Cw}
\begin{aligned}
    C_w
    &:=
    \norm{w}_{L^2}
    +\norm{\widehat w/\sqrt{\omg}}_{L^2},\\
    C_w'
    &:=
    \norm{\nabla w}_{L^2}
    +\norm{\widehat{\nabla w}/\sqrt{\omg}}_{L^2}.
\end{aligned}
\end{equation}
Both constants are finite under the preceding assumptions.
The coupling constants $\lam,\eta>0$ satisfy
\begin{equation}
\label{eq:weak-coupling-condition}
    2\eta(C_w+1)<1,
    \qquad
    2\lam\norm{v}_{L^\infty}+4\eta C_w<1.
\end{equation}
\begin{remark}
In the scaling $\lam=\eta^2$, the weak-coupling conditions become
\begin{equation}
    2\eta(C_w+1)<1,
    \qquad
    2\eta^2\norm{v}_{L^\infty}+4\eta C_w<1.
\end{equation}
For fixed interaction kernels satisfying the preceding
assumptions, both conditions hold for sufficiently small
$\eta>0$.
\end{remark}

\item
\label{hyp:bounded-energy-per-particle}
\textbf{Bounded energy per particle:} 
For any given $N\gg1$, the family of $N$-particle initial wave functions $\brs{\Psi_{N,0}}_{N\geq 1} \subset \HN$ admits the bounds 
\begin{equation}
    \frac{1}{N^k} \langle\Psi_{N,0}, (\HMN+N)^k \Psi_{N,0}\rangle < C_0^k
    \;\;\;,\;1\leq k\leq K(N) \,,
\end{equation}
for a constant $C_0$ independent of $N$, and a continuous, positive, monotone increasing function $K\in C(\R_+)$ satisfying $1\leq K(N)\leq N$ and $\lim_{N\rightarrow\infty}K(N)=\infty$.

\item
\label{hyp:microscopic-Hamiltonian-and-initial-data}
\textbf{Microscopic Hamiltonian and coherent initial data:}
The microscopic dynamics satisfy
\begin{equation}
\label{eq:Schrodinger-def-1-0}
    \begin{cases}
    i\partial_t\Psi_{N,t}=\HMN\Psi_{N,t},\\[0.6em]
    \Psi_{N,0}=\phi_{N,0}\otimes\wyl[\sqrt N\,u_0]\vacf.
    \end{cases}
\end{equation}
Here $\vacf$ is the photon vacuum, $\wyl$ is the Weyl operator, and
the Hamiltonian, with identity factors suppressed, is
\begin{equation}
\label{eq:moment-microscopic-Hamiltonian}
\begin{aligned}
    \HMN&=\HMNbos+\Hf+\frac{\eta}{\sqrt N}W_N,\\
    \HMNbos&=\sum_{j=1}^N\dch_0^{(j)}
          +\frac{\lam}{N}\sum_{1\leq i<j\leq N}v(x_i-x_j),\\
           \Hf&=\int_{\R^3}\omg(\xi)\ad_\xi a_\xi\,d\xi,\qquad \omg(\xi)=|\xi|,\\
    W_N&=\sum_{j=1}^N\int_{\R^3}
             w(x_j-y)(\ad_y+a_y)\,dy.
\end{aligned}
\end{equation}
The normalized symmetric particle vectors $\phi_{N,0}\in\HNb$
satisfy, for some fixed $E_*>0$ independent of $N$, and $\mathbb K_N$ as in \eqref{eq:moment-reference-oscillator},
\begin{equation}
\label{eq:moment-initial-reference-cutoff}
    \phi_{N,0}
       =\1_{[0,E_*N]}(\mathbb K_N)\phi_{N,0},
    \qquad \norm{\phi_{N,0}}=1.
\end{equation}
The position-space coherent profile is independent of $N$ and satisfies
\begin{equation}
\label{eq:moment-coherent-profile}
    u_0\in L^2(\R^3)\mcap W^{3,1}(\R^3),
    \qquad
    \int_{\R^3}e^{a_*|\xi|}(1+|\xi|^2)
             |\widehat u_0(\xi)|^2\,d\xi<\infty
\end{equation}
for a constant $a_*>0$.
In particular, $u_0\in H^1(\R^3)$ and the Weyl transform of the
field Hamiltonian is well-defined on its operator domain. 
\end{enumerate}
\end{hypothesis}

\begin{lemma}[Initial energy moments for coherent product states]
\label{lem:coherent-oscillator-moments}
Assume \hypref{hyp:trapping-potential},
\hypref{hyp:bounded-two-body-interaction}, and
\hypref{hyp:microscopic-Hamiltonian-and-initial-data} above.
It suffices, for positivity, that the coupling constants satisfy
\begin{equation}
\label{eq:moment-positivity-condition}
    \frac{\lam}{2}\norm{v}_{L^\infty}
      +2\eta^2\norm{\widehat w/\sqrt{\omega}}_{L^2}^2\leq1,
    \qquad \omega(\xi)=|\xi|;
\end{equation}
this condition follows from the given weak-coupling assumptions.
Set
\begin{equation}
\label{eq:moment-Weyl-conjugation}
\begin{aligned}
    A_N&:=\HMN+N,
    &\mathcal W_N&:=\1_{\HNb}\otimes\wyl[\sqrt N\,u_0],\\
    G_N&:=\mathcal W_N^*A_N\mathcal W_N,
    &\Phi_N&:=\phi_{N,0}\otimes\vacf.
\end{aligned}
\end{equation}
Then $A_N,G_N\geq0$ and
$\Psi_{N,0}\in\bigcap_{m\geq1}\cD(A_N^m)$.
There are constants $C,D\geq1$, independent of $N$ and $k$, such that
for all integers $N,k\geq1$,
\begin{equation}
\label{eq:coherent-oscillator-moment-bound}
\begin{aligned}
    \norm{A_N^{k/2}\Psi_{N,0}}^2
      &=\langle\Phi_N,G_N^k\Phi_N\rangle\\
      &\leq C^k(N+k^2)^k
       \leq D^k\bigl(N^k+(k!)^2\bigr).
\end{aligned}
\end{equation}
In particular, for $k\geq2$,
\begin{equation}
\label{eq:coherent-oscillator-N-factorial-bound}
    \norm{A_N^{k/2}\Psi_{N,0}}^2
    \leq D^k\bigl(N^k+N^{k-2}(k!)^2\bigr)
    \leq D^k\sum_{i=0}^2N^{k-i}(k!)^i.
\end{equation}
Moreover,
\begin{equation}
\label{eq:coherent-oscillator-normalized-bound}
    N^{-k}\norm{A_N^{k/2}\Psi_{N,0}}^2
       \leq C^k\left(1+\frac{k^2}{N}\right)^k
       \leq(2C)^k,
    \qquad 1\leq k\leq\sqrt N \,.
\end{equation}
%Consequently, for every fixed $r>0$ and all sufficiently large $N$, the initial energy-moment condition holds for $1\leq k\leq\lfloor\log N/(r\log\log N)\rfloor$.
Thus, in particular, \hyporef{hyp:bounded-energy-per-particle} is satisfied with 
\begin{eqnarray}
    K(N)=\sqrt{N} \,,
\end{eqnarray}
as a consequence of \eqref{eq:coherent-oscillator-normalized-bound}.
\end{lemma}

The proof is given in \appref{sec:coherent-oscillator-moments}.

\begin{remark}[Self-adjointness and unitary dynamics]
Under \hyporef{hyp:standing-hypotheses}, the one-particle
operator $\dch_0$ is nonnegative and self-adjoint.
The two-body interaction is bounded for each fixed $N$.
Moreover, $w\in L^2$ and
$\widehat w/\sqrt{\omg}\in L^2$, so the field estimate gives
\begin{equation}
    \left\|
    \frac{\eta}{\sqrt N}\sum_{j=1}^N
    \phi\bigl(w(x_j-\cdot)\bigr)\Psi
    \right\|
    \leq
    2\eta C_w\sqrt N\,\norm{(1+\Hf)^{1/2}\Psi}.
\end{equation}
Consequently, for each fixed $N$, the particle-field interaction
is infinitesimally operator bounded with respect to
$\sum_{j=1}^N\dch_0^{(j)}+\Hf$.
The Kato--Rellich theorem therefore yields self-adjointness
of $\HMN$ on
\begin{equation}
    \cD\left(\sum_{j=1}^N\dch_0^{(j)}\right)
    \mcap\cD(\Hf).
\end{equation}
Its restriction to the invariant bosonic subspace is
self-adjoint. By Stone's theorem,
$\UN(t,s)=e^{-i(t-s)\HMN}$ defines a strongly continuous
unitary evolution and preserves $\norm{\Psi_{N,t}}$.
\end{remark}

\subsection{Energy spaces and coherent initial data}
\label{sec:energy-spaces-and-coherent-initial-data}
 In this section, we introduce the weighted energy spaces that control the regularity of the particle sector, define the relevant Schatten-class operators for reduced density matrices, and formulate the Cauchy problem for the coupled system. Furthermore, we briefly review the essential algebraic properties of Weyl operators, which are central to constructing the coherent initial data for the radiation field.

\begin{definition}[The energy space and the $Y^1$-based hierarchy space]
\label{def:Y1-and-density-matrix-space}
Recall that $\dch_0=-\Delta+V$ is nonnegative and
self-adjoint. Fix $\kappa>0$ and define
\begin{equation}
\label{eq:definition-of-T}
    \OT:=(\dch_0+\kappa)^{1/2}.
\end{equation}
Then
\begin{equation}
    \OT\geq\sqrt{\kappa}\,\1,
    \qquad
    \|\OT^{-1}\|_{\mathrm{op}}
    \leq\kappa^{-1/2}.
\end{equation}
The one-particle energy space is
\begin{equation}
    Y^1:=\mathcal D(\OT),
    \qquad
    \langle f,g\rangle_{Y^1}
    :=
    \langle\OT f,\OT g\rangle_{\h}.
\end{equation}
In particular,
\begin{equation}
    \|f\|_{Y^1}:=\|\OT f\|_{\h},
    \qquad
    \|f\|_{\h}\leq\kappa^{-1/2}\|f\|_{Y^1}.
\end{equation}
The map $\OT: Y^1 \to \h$ is a surjective isometry. Consequently, $Y^1$ is a separable Hilbert space. 

For $k\geq1$, define the coordinate operators
on the full tensor product $\h^{\otimes k}$ by
\begin{equation}
    \OT_j
    :=
    \1_{\h}^{\otimes(j-1)}
    \otimes\OT
    \otimes\1_{\h}^{\otimes(k-j)},
    \qquad 1\leq j\leq k.
\end{equation}
These operators are self-adjoint and strongly
commute. Their product, defined by joint spectral
calculus, is $\OT^{\otimes k}$.
This product commutes with particle permutations,
so the bosonic subspace $\Hkb$ is reducing for it.
Define
\begin{equation}
    \OT_k
    :=
    \left.\OT^{\otimes k}\right|_{\Hkb}.
\end{equation}
Then $\OT_k$ is self-adjoint and
\begin{equation}
    \OT_k\geq\kappa^{k/2}\1_{\Hkb},
    \qquad
    \|\OT_k^{-1}\|_{\mathrm{op}}
    \leq\kappa^{-k/2}.
\end{equation}
The $k$-particle energy space is
\begin{equation}
    Y_k^1:=\mathcal D(\OT_k),
    \qquad
    \langle f,g\rangle_{Y_k^1}
    :=
    \langle\OT_k f,\OT_k g\rangle_{\Hkb}.
\end{equation}
Thus
\begin{equation}
    \|f\|_{Y_k^1}:=\|\OT_k f\|_{\Hkb}.
\end{equation}
As above, $\OT_k: Y_k^1 \to \Hkb$ is a surjective isometry, and $Y_k^1$ is a separable Hilbert space. 

We define the weighted trace-class space associated
with $Y_k^1$ by
\begin{equation}
    \mathcal L_k^1(Y_k^1)
    :=
    \left\{
        \OT_k^{-1}\OB^{(k)}\OT_k^{-1}:
        \OB^{(k)}\in\mathfrak S_1(\Hkb)
    \right\}.
\end{equation}
Because $\OT_k^{-1}$ is bounded, this is a
linear subspace of $\mathfrak S_1(\Hkb)$. For every $\OA^{(k)}\in\mathcal L_k^1(Y_k^1)$,
there is a unique $\OB^{(k)}\in\mathfrak S_1(\Hkb)$
such that
\begin{equation}
\label{eq:definition-of-B-k-in-L-k-1-Y-1}
    \OA^{(k)}
    =
    \OT_k^{-1}\OB^{(k)}\OT_k^{-1}.
\end{equation}
Equivalently,
\begin{equation}
    \langle\OT_k f,\OA^{(k)}\OT_k g\rangle_{\Hkb}
    =
    \langle f,\OB^{(k)}g\rangle_{\Hkb},
    \qquad f,g\in\mathcal D(\OT_k).
\end{equation}
We write
\begin{equation}
    \OT_k\OA^{(k)}\OT_k=\OB^{(k)}
\end{equation}
for this bounded trace-class realization of the
weighted sesquilinear form. Equip $\mathcal L_k^1(Y_k^1)$ with the norm
\begin{equation}
    \|\OA^{(k)}\|_{\mathcal L_k^1(Y_k^1)}
    :=
    \|\OB^{(k)}\|_{\mathfrak S_1(\Hkb)}
    =
    \|\OT_k\OA^{(k)}\OT_k\|_{\mathfrak S_1(\Hkb)}.
\end{equation}
The correspondence
$\OA^{(k)}\mapsto\OB^{(k)}$ is an isometric
bijection onto $\mathfrak S_1(\Hkb)$.
Hence $\mathcal L_k^1(Y_k^1)$ is a Banach space.
Moreover,
\begin{equation}
    \|\OA^{(k)}\|_{\mathfrak S_1(\Hkb)}
    \leq
    \kappa^{-k}
    \|\OA^{(k)}\|_{\mathcal L_k^1(Y_k^1)}.
\end{equation}
Finally, $\OA^{(k)}\geq0$ if and only if
$\OB^{(k)}\geq0$. Therefore, for positive
$\OA^{(k)}\in\mathcal L_k^1(Y_k^1)$,
\begin{equation}
    \|\OA^{(k)}\|_{\mathcal L_k^1(Y_k^1)}
    =
    \Tr_{\Hkb}\!\left(
        \OT_k\OA^{(k)}\OT_k
    \right).
\end{equation}
\end{definition}

\begin{definition}[Cauchy problem and coherent initial data]
\label{def:cauchy_problem_dynamics}
Fix $N\geq 1$, and let $\HN$ denote the $N$-particle-field Hilbert space, whose elements can be viewed fiber-wise (with respect to the spectrum of the number operator) as Fock-space-valued functions $\Psi_{N,t}(\ux_N) \in \Fock$ for almost every spatial configuration $\ux_N:=(x_1,\ldots,x_N) \in\R^{3N}$. The quantum dynamics of the coupled system is governed by the Cauchy problem for the Schr\"odinger equation
\begin{equation}
    \label{eq:main-Schrodinger-Eq-of-Phi}
    \begin{cases}
    i\pr_t\Psi_{N,t} = \HMN\Psi_{N,t}, \vspace{0.5 em}\\
    \Psi_{N,t}\vert_{t=0} = \Psi_{N,0},
    \end{cases}
\end{equation}
where the initial state $\Psi_{N,0} \in \HN$ assumes a tensor product form
\begin{equation}
    \label{eq:initial-state-of-the-system-Psi-0}
    \Psi_{N,0} = \phi_{N,0}\otimes \wyl[\sqrt N u_0]\vacf.
\end{equation}
Here, $\phi_{N,0} \in \HNb$ represents the normalized $N$-particle wave function, $\vacf$ denotes the vacuum vector in the Fock space $\Fock$, and $\wyl[h]$ is the Weyl operator generated by the field profile $h \in \h$, defined as
\begin{equation}
    \label{eq:defination-of-Weyl-Operator}
    \wyl[h] = e^{\ad(h)-a(h)}.
\end{equation}
The action of $\wyl[h]$ on the vacuum vector $\vacf$ results in the coherent state
\begin{equation}
    \wyl[h]\vacf = e^{-\frac{1}{2}\norm{h}_{\h}^2} \sum_{n=0}^{\infty} \frac{\prc{\ad(h)}^n}{n!}\vacf,
\end{equation}
which, under the mean-field scaling $h = \sqrt{N} u_0$, evaluates to
\begin{equation}
    \wyl[\sqrt N u_0]\vacf = e^{-\frac{N}{2}\norm{u_0}_{\h}^2} \sum_{n=0}^{\infty} \frac{N^{n/2}\prc{\ad(u_0)}^n}{n!}\vacf.
\end{equation}
\end{definition}

\begin{lemma}[Weyl operators and coherent states]
\label{lem:basic-properties-of-Weyl-operators}
Let $\h$ be the one-particle Hilbert space, with inner
product linear in its second argument, and let
$\vacf$ denote the normalized Fock vacuum.
For $g\in\h$, define
\begin{equation}
    \wyl[g]:=e^{\ad(g)-a(g)},
\end{equation}
where the exponential is understood through the
skew-adjoint closure of $\ad(g)-a(g)$, initially
defined on finite-particle vectors.
Then, for every $g,h\in\h$, the following properties hold.

\begin{enumerate}[label=(\roman*)]
    \item The Weyl operators satisfy
    \begin{equation}
        \wyl[g]\wyl[h]
        =
        e^{-i\Im{\braket{g,h}}}\wyl[g+h].
    \end{equation}
    Equivalently,
    \begin{equation}
        \wyl[g]\wyl[h]
        =
        e^{-2i\Im{\braket{g,h}}}\wyl[h]\wyl[g].
    \end{equation}

    \item The Weyl operators are unitary, and
    \begin{equation}
        \wyld[g]
        =
        \wyl[g]^{-1}
        =
        \wyl[-g].
    \end{equation}

    \item The coherent state $\wyl[h]\vacf$ belongs
    to $\mathcal D(a(g))$ and satisfies
    \begin{equation}
        a(g)\wyl[h]\vacf
        =
        \braket{g,h}\wyl[h]\vacf.
    \end{equation}

    \item The coherent state $\wyl[h]\vacf$ belongs
    to $\mathcal D(\cN_f^m)$ for every integer $m\geq1$.
    Its photon-number expectation and variance are
    \begin{equation}
        \braket{\wyl[h]\vacf,\cN_f\wyl[h]\vacf}
        =
        \norm{h}_{\h}^2
    \end{equation}
    and
    \begin{equation}
    \begin{aligned}
        &\braket{\wyl[h]\vacf,\cN_f^2\wyl[h]\vacf}\\
        &\qquad-
        \braket{\wyl[h]\vacf,\cN_f\wyl[h]\vacf}^{\,2}
        =
        \norm{h}_{\h}^2.
    \end{aligned}
    \end{equation}
    More precisely, its photon-number distribution
    is Poissonian with parameter $\norm{h}_{\h}^2$.

    \item Coherent states are normalized, and
    \begin{equation}
    \begin{aligned}
        \braket{\wyl[g]\vacf,\wyl[h]\vacf}
        =
        \exp\left(
            -\frac12\norm{g}_{\h}^2
            -\frac12\norm{h}_{\h}^2
            +\braket{g,h}
        \right).
    \end{aligned}
    \end{equation}
    In particular,
    \begin{equation}
        \abs{\braket{\wyl[g]\vacf,\wyl[h]\vacf}}
        =
        \exp\left(-\frac12\norm{g-h}_{\h}^2\right).
    \end{equation}
\end{enumerate}
\end{lemma}

\begin{proof}
Initially define $A_h:=\ad(h)-a(h)$ on the dense
subspace of finite-particle vectors.
If $\Psi$ has no components above the $m$-particle
sector, the creation and annihilation estimates give
\begin{equation}
    \norm{A_h^n\Psi}
    \leq
    (2\norm{h}_{\h})^n
    \sqrt{\frac{(m+n)!}{m!}}\,
    \norm{\Psi},
    \qquad n\geq0.
\end{equation}
Consequently, every finite-particle vector is an
entire analytic vector for $A_h$.
Since $A_h$ is skew-symmetric on this domain, the
analytic-vector theorem implies that it is
essentially skew-adjoint.
We use the same notation for its skew-adjoint closure.

The canonical commutation relations yield, on
finite-particle vectors,
\begin{equation}
\begin{aligned}
    [A_g,A_h]
    &=
    [\ad(g)-a(g),\ad(h)-a(h)]\\
    &=
    \bigl(\braket{h,g}-\braket{g,h}\bigr)\1\\
    &=
    -2i\Im{\braket{g,h}}\,\1.
\end{aligned}
\end{equation}
This commutator is a scalar multiple of the identity.
The corresponding estimates for mixed products of
$A_g$ and $A_h$ justify the
Baker--Campbell--Hausdorff calculation on
finite-particle analytic vectors. Since
$A_g+A_h=A_{g+h}$ on that domain, we obtain
\begin{equation}
    \wyl[g]\wyl[h]
    =
    e^{-i\Im{\braket{g,h}}}\wyl[g+h].
\end{equation}
Both sides are bounded operators, so the identity
extends to the full Fock space by density.
Interchanging $g$ and $h$ gives
\begin{equation}
    \wyl[h]\wyl[g]
    =
    e^{i\Im{\braket{g,h}}}\wyl[g+h],
\end{equation}
which proves the equivalent commutation relation
in (i).

Skew-adjointness gives
\begin{equation}
    \wyld[g]
    =
    e^{-A_g}
    =
    \wyl[-g]
    =
    \wyl[g]^{-1},
\end{equation}
proving (ii).

We next establish the coherent-state expansion.
The identity
\begin{equation}
    [\ad(h),-a(h)]
    =
    \norm{h}_{\h}^2\1
\end{equation}
and $a(h)\vacf=0$ give, by normal ordering applied
to the vacuum,
\begin{equation}
\begin{aligned}
    \wyl[h]\vacf
    &=
    e^{-\norm{h}_{\h}^2/2}
    \sum_{n=0}^{\infty}
        \frac{(\ad(h))^n}{n!}\vacf\\
    &=
    e^{-\norm{h}_{\h}^2/2}
    \bigoplus_{n=0}^{\infty}
        \frac{h^{\otimes n}}{\sqrt{n!}}.
\end{aligned}
\end{equation}
Here the normal-ordering calculation is justified
by the absolutely norm-convergent power series
on the vacuum; no product of unbounded exponentials
on arbitrary Fock vectors is required.

For every integer $m\geq1$,
\begin{equation}
    e^{-\norm{h}_{\h}^2}
    \sum_{n=0}^{\infty}
        n^{2m}\frac{\norm{h}_{\h}^{2n}}{n!}
    <\infty.
\end{equation}
Thus $\wyl[h]\vacf\in\mathcal D(\cN_f^m)$.
In particular, its finite-sector truncations
converge in the graph norm of $\cN_f^{1/2}$.
The estimate
\begin{equation}
    \norm{a(g)\Psi}
    \leq
    \norm{g}_{\h}\norm{\cN_f^{1/2}\Psi}
\end{equation}
therefore justifies applying $a(g)$ term by term
to the coherent-state expansion.

For $n\geq1$, the canonical commutation relations give
\begin{equation}
    a(g)(\ad(h))^n\vacf
    =
    n\braket{g,h}(\ad(h))^{n-1}\vacf.
\end{equation}
Consequently,
\begin{equation}
\begin{aligned}
    a(g)\wyl[h]\vacf
    &=
    e^{-\norm{h}_{\h}^2/2}
    \sum_{n=1}^{\infty}
        \frac{n\braket{g,h}}{n!}
        (\ad(h))^{n-1}\vacf\\
    &=
    \braket{g,h}\,
    e^{-\norm{h}_{\h}^2/2}
    \sum_{n=0}^{\infty}
        \frac{(\ad(h))^n}{n!}\vacf\\
    &=
    \braket{g,h}\wyl[h]\vacf.
\end{aligned}
\end{equation}
This proves (iii).

By orthogonality of the particle-number sectors,
the probability of observing $n$ photons is
\begin{equation}
    P(n)
    =
    e^{-\norm{h}_{\h}^2}
    \frac{\norm{h}_{\h}^{2n}}{n!},
    \qquad n\geq0.
\end{equation}
Hence the photon-number distribution is Poissonian.
Direct summation gives
\begin{equation}
    \sum_{n=0}^{\infty}nP(n)
    =
    \norm{h}_{\h}^2,
    \qquad
    \sum_{n=0}^{\infty}n(n-1)P(n)
    =
    \norm{h}_{\h}^4.
\end{equation}
Since $n^2=n(n-1)+n$, it follows that
\begin{equation}
    \braket{\wyl[h]\vacf,\cN_f^2\wyl[h]\vacf}
    =
    \norm{h}_{\h}^4+\norm{h}_{\h}^2.
\end{equation}
Subtracting the square of the mean proves the
variance formula in (iv).

Finally, the coherent-state expansion and
orthogonality of the number sectors imply
\begin{equation}
\begin{aligned}
    \braket{\wyl[g]\vacf,\wyl[h]\vacf}
    &=
    e^{-\norm{g}_{\h}^2/2-\norm{h}_{\h}^2/2}
    \sum_{n=0}^{\infty}
        \frac{\braket{g,h}^{\,n}}{n!}\\
    &=
    \exp\left(
        -\frac12\norm{g}_{\h}^2
        -\frac12\norm{h}_{\h}^2
        +\braket{g,h}
    \right).
\end{aligned}
\end{equation}
Taking $g=h$ confirms normalization.
Taking absolute values and using
\begin{equation}
    \norm{g-h}_{\h}^2
    =
    \norm{g}_{\h}^2+\norm{h}_{\h}^2
    -2\Re{\braket{g,h}}
\end{equation}
proves the final assertion.
\end{proof}

\subsection{Reduced density matrices and BBGKY hierarchy}
In this section, we derive the Bogoliubov--Born--Green--Kirkwood--Yvon (BBGKY) hierarchy, a coupled sequence of integro-differential equations that governs the time evolution of these reduced density matrices. Crucially, the interaction between the confined bosons and the radiation field manifests as a novel source term in this hierarchy, coupling the particle marginals to the coherent radiation field.

\begin{definition}[Schatten class operators]
\label{def:schatten-classes}
For a given separable Hilbert space $\mathcal{H}$, we denote the algebra of bounded linear operators on $\mathcal{H}$ by $\mathcal{B}(\mathcal{H})$. For $p \geq 1$, we denote by $\mathfrak{S}_p(\mathcal{H})$ the two-sided ideal of $p$-th Schatten-class operators, equipped with the norm
\begin{equation}
    \norm{A}_{\mathfrak{S}_p} := \Tr(\abs{A}^p)^{1/p}.
\end{equation}
In particular, $\mathfrak{S}_1(\mathcal{H})$ and $\mathfrak{S}_2(\mathcal{H})$ denote the spaces of trace-class and Hilbert--Schmidt operators, respectively.
\end{definition}

\begin{definition}[Reduced density matrices and their integral kernels]
\label{def:reduced-density-matrix}
Suppose the assumptions in \hyporef{hyp:standing-hypotheses} hold. Fix $N \geq 1$, and let $\Psi_{N,t} \in \HN$ be a solution to the Cauchy problem \eqref{def:cauchy_problem_dynamics}. Define the trace-class, rank-one projection $\Gamma_{N,t} \in \mathfrak{S}_1(\HN)$ on $\HN$ as
\begin{equation}
    \Gamma_{N,t} := \ket{\Psi_{N,t}}\bra{\Psi_{N,t}}.
\end{equation}
Then, we define the $k$-particle reduced density matrix as the partial trace over the Fock space and the $N-k$ particles
\begin{equation}
    \label{eq:definition-gamma-k-BBGKY-operator-form}
    \gamma_{N,t}^{(k)} := \Tr_{k+1,\ldots,N,\Fock} \Gamma_{N,t} \in \mathfrak{S}_1(\Hkb).
\end{equation}
Equivalently, for every $1\leq k\leq N$ and
$\Phi_{k,t}\in\Hkb$, the action of
$\gamma_{N,t}^{(k)}$ is given by
\begin{equation}
    \prc{\gamma_{N,t}^{(k)}\Phi_{k,t}}(\ux_k)
    =
    \int_{\R^{3k}}
        \gamma_{N,t}^{(k)}(\ux_k;\ux_k')
        \Phi_{k,t}(\ux_k')
        \,d\ux_k',
\end{equation}
for almost every $\ux_k\in\R^{3k}$, where $\gamma_{N,t}^{(k)}(\ux_k;\ux_k')$ is the integral kernel of $\gamma_{N,t}^{(k)}$, namely, the $k$-particle reduced density matrix in kernel form is defined as
\begin{equation}
\label{eq:definition-gamma-k-BBGKY}
    \gamma_{N,t}^{(k)}(\ux_k;\ux_k')
    :=
    \int_{\R^{3(N-k)}}
    \left\langle
    \Psi_{N,t}(\ux_k',\ux_{N-k}),
    \Psi_{N,t}(\ux_k,\ux_{N-k})
    \right\rangle_{\Fock}
    d\ux_{N-k}.
\end{equation}
\end{definition}

\begin{lemma}[The   BBGKY hierarchy]
\label{lem:finite-N-BBGKY-hierarchy}
Fix $N \geq 1$ and suppose that \hyporef{hyp:standing-hypotheses} holds. For every $1\leq k\leq N$, the   BBGKY hierarchy governing the time evolution of the $k$-particle reduced density matrices with respect to the Hamiltonian $\HMN$ is given, in operator form, by
\begin{equation}
\label{eq:finite-N-BBGKY-operator-form-clean}
\begin{aligned}
    i\pr_t\gamma_{N,t}^{(k)}
    &=
    \left[ \sum_{j=1}^{k}\dch_0^{(j)}, \gamma_{N,t}^{(k)} \right]
    +
    \frac{\lam}{N} \sum_{1\leq i<j\leq k} \left[ v(x_i-x_j), \gamma_{N,t}^{(k)} \right]
    \\
    &\quad +
    \lam\frac{N-k}{N} \sum_{j=1}^{k} \operatorname{Tr}_{k+1} \left[ v(x_j-x_{k+1}), \gamma_{N,t}^{(k+1)} \right]
    +
    \eta \, \mathcal S_{N,t}^{(k)}.
\end{aligned}
\end{equation}
For $k=N$, the term involving $\gamma_{N,t}^{(N+1)}$ is absent. Equivalently, in kernel form, the hierarchy reads
\begin{equation}
\label{eq:finite-N-BBGKY-kernel-form-clean}
\begin{aligned}
    i\pr_t \gamma_{N,t}^{(k)}(\ux_k;\ux_k')
    &=
    \sum_{j=1}^{k} \left( \dch_0(x_j) - \dch_0(x_j') \right) \gamma_{N,t}^{(k)}(\ux_k;\ux_k')
    \\
    & +
    \frac{\lam}{N} \sum_{1\leq i<j\leq k} \left[ v(x_i-x_j) - v(x_i'-x_j') \right] \gamma_{N,t}^{(k)}(\ux_k;\ux_k')
    \\
    & +
    \lam\frac{N-k}{N} \sum_{j=1}^{k} \int_{\R^3} \left[ v(x_j-z) - v(x_j'-z) \right]\\
    &\qquad \qquad \times  \gamma_{N,t}^{(k+1)}(\ux_k,z;\ux_k',z) \, dz+ 
    \eta \, \mathcal S_{N,t}^{(k)}(\ux_k;\ux_k'),
\end{aligned}
\end{equation}
where $\cS_{N,t}^{(k)}(\ux_k;\ux_k')$ is the kernel of the photon source term. Explicitly, the source term is defined as
\begin{equation}
\label{eq:def-source-term-Sk-finite-N-BBGKY}
    \mathcal S_{N,t}^{(k)}(\ux_k;\ux_k')
    :=
    \sum_{j=1}^{k} \mathcal Q_{N,t}^{(k)} \left(\ux_k;\ux_k'; f_{j} \right).
\end{equation}
Here, $\mathcal Q_{N,t}^{(k)}$ denotes the field operator smeared with $f_j$, where
\begin{equation}
\label{eq:first-definition-Q-k-smeared-field-in-BBGKY}
    \mathcal Q_{N,t}^{(k)}(\ux_k;\ux_k';f)
    := \frac{1}{\sqrt N}
    \int_{\R^{3(N-k)}}
    \left\langle
    \Psi_{N,t}(\ux_k',\ux_{N-k}),
    \phi(f)\Psi_{N,t}(\ux_k,\ux_{N-k})
    \right\rangle_{\Fock}
    d\ux_{N-k},
\end{equation}
and $f_j$ defined for $j=1,\ldots,k$ as the difference of form factors
\begin{equation}
\label{eq:f-j-k-BBGKY}
    f_{j}(y) := w(x_j-y) - w(x_j'-y).
\end{equation} 
The hierarchy is understood in the time-integrated
weak sense. More precisely, it is tested against
finite-rank operators $\OJ^{(k)}$ satisfying
\begin{equation}
    \operatorname{Ran}(\OJ^{(k)})
    \cup\operatorname{Ran}((\OJ^{(k)})^*)
    \subset
    \cD\left(\sum_{j=1}^k\dch_0^{(j)}\right).
\end{equation}
For each such test, the scalar function
$t\mapsto\Tr(\OJ^{(k)}\gamma_{N,t}^{(k)})$
is locally absolutely continuous, and the tested
equation holds after integration over every
compact time interval. The kernel equation is interpreted distributionally.
All kernel expressions involving partial traces
refer to the representatives induced by the
underlying wavefunction and the operator partial trace.
\end{lemma}

\begin{proof}
For labeled operator calculations, we work on
$L^2(\R^{3N})\otimes\Fock$ and use the self-adjoint
Hamiltonian with the same expression as $\HMN$.
Its restriction to the invariant bosonic subspace
gives the prescribed evolution.

The coherent initial state belongs to
$\cD(\cN_f^{1/2})$.
The independent bounded-number-cutoff argument
used to prove \eqref{eq:joint-source-uniform-number}
also gives, for every
$\zeta\in\cD(\cN_f^{1/2})$ and $t\geq0$,
\begin{equation}
\begin{aligned}
    \|(\cN_f+1)^{1/2}\UN(t,0)\zeta\|
    &\leq
    \|(\cN_f+1)^{1/2}\zeta\|
    \\
    &\quad+
    \eta\sqrt N\,\|w\|_{L^2}\,t\,\|\zeta\|.
\end{aligned}
\end{equation}
Indeed, the same cutoff differentiation applies
to an arbitrary initial vector, with its conserved
norm retained in the Cauchy--Schwarz estimate.
This argument is independent of the BBGKY hierarchy.

For integers $L\geq1$, define
\begin{equation}
\begin{aligned}
    \Psi_{N,0}^{(L)}
    &:=
    \1_{[0,L]}
       \left(\sum_{j=1}^N\dch_0^{(j)}+\Hf\right)
    \1_{[0,L]}(\cN_f)\Psi_{N,0},
    \\
    \Psi_{N,t}^{(L)}
    &:=
    \UN(t,0)\Psi_{N,0}^{(L)}.
\end{aligned}
\end{equation}
The two spectral projections commute and preserve
particle symmetry.
The operator-domain description of $\HMN$ shows that
$\Psi_{N,0}^{(L)}\in\cD(\HMN)$.
Moreover,
\begin{equation}
    \|(\cN_f+1)^{1/2}
       (\Psi_{N,0}^{(L)}-\Psi_{N,0})\|
    \longrightarrow0.
\end{equation}
Applying the preceding number estimate to the
difference of the initial vectors yields
\begin{equation}
\label{eq:BBGKY-number-graph-approximation}
    \sup_{0\leq t\leq T}
    \|(\cN_f+1)^{1/2}
       (\Psi_{N,t}^{(L)}-\Psi_{N,t})\|
    \longrightarrow0
\end{equation}
for every fixed $N$ and $T<\infty$.

For each $L$, the Schr\"odinger equation holds
strongly for $\Psi_{N,t}^{(L)}$.
Its rank-one density matrix is therefore
differentiable in trace norm, with derivative
given by the Hamiltonian commutator.
Partial trace preserves this identity.

We first perform the following algebraic calculation
for these regularized states and temporarily
suppress the superscript $(L)$.
The calculation applies to unnormalized states
because every expression is quadratic in the state.
The kernel identities below are interpreted
distributionally.
The strong Schr\"odinger equation then yields
\begin{equation}
\label{eq:BBGKY-starting-differentiation-clean}
\begin{aligned}
    i\pr_t\gamma_{N,t}^{(k)}(\ux_k;\ux_k')
    &=
    \int_{\R^{3(N-k)}}
    \left\langle
    \Psi_{N,t}(\ux_k',\ux_{N-k}),
    \HMN\Psi_{N,t}(\ux_k,\ux_{N-k})
    \right\rangle_{\Fock}
    d\ux_{N-k}
    \\
    & -
    \int_{\R^{3(N-k)}}
    \left\langle
    \HMN\Psi_{N,t}(\ux_k',\ux_{N-k}),
    \Psi_{N,t}(\ux_k,\ux_{N-k})
    \right\rangle_{\Fock}
    d\ux_{N-k}.
\end{aligned}
\end{equation}
We systematically analyze the contribution of each component of $\HMN$. For the one-particle Hamiltonian, if $1\leq j\leq k$, the operator $\dch_0^{(j)}$ acts on the unprimed untraced variable $x_j$ in the first integral and on the primed untraced variable $x_j'$ in the second integral, yielding the commutator term
\begin{equation}
    \left( \dch_0(x_j) - \dch_0(x_j') \right) \gamma_{N,t}^{(k)}(\ux_k;\ux_k').
\end{equation}
Conversely, if $j > k$, the self-adjoint operator $\dch_0^{(j)}$ acts exclusively on the integrated (traced) variables, causing the two integral contributions to cancel. Thus, the one-particle operator contributes only for $1 \leq j \leq k$. Similarly, the contribution from the field Hamiltonian $\Hf$ vanishes identically because it is self-adjoint and acts on the traced Fock-space sector. 

Next, we consider the following two-body particle interaction
\begin{equation}
    \frac{\lam}{N} \sum_{1\leq i<j\leq N} v(x_i-x_j).
\end{equation}
To compute the contributions from each term in the sum, we decompose the sum as follows:
\begin{equation}
    \sum_{1\leq i<j\leq N}
    =
    \sum_{1\leq i<j\leq k}
    +
    \sum_{\substack{1\leq i\leq k\\ k+1\leq j\leq N}}
    +
    \sum_{k+1\leq i<j\leq N}.
\end{equation}
For $1\leq i<j\leq k$, both variables are untraced, immediately yielding the internal interaction term
\begin{equation}
    \frac{\lam}{N} \sum_{1\leq i<j\leq k} \left[ v(x_i-x_j) - v(x_i'-x_j') \right] \gamma_{N,t}^{(k)}(\ux_k;\ux_k').
\end{equation}
When only one variable belongs to the traced sector (e.g., $1\leq i\leq k$ and $k+1\leq j\leq N$), bosonic symmetry ensures that all $N-k$ traced particles yield identical contributions, producing the collision term
\begin{equation}
    \lam\frac{N-k}{N} \sum_{j=1}^{k} \int_{\R^3} \left[ v(x_j-z) - v(x_j'-z) \right] \gamma_{N,t}^{(k+1)}(\ux_k,z;\ux_k',z) \, dz.
\end{equation}
When both variables are traced ($k+1\leq i<j\leq N$), the real multiplication operator $v(x_i-x_j)$ acts identically on the integrated variables on both sides of the commutator; thus, the two contributions cancel each other out.  

Finally, we process the particle-photon interaction,
\begin{equation}
    \frac{\eta}{\sqrt N} \sum_{j=1}^{N} \phi(w(x_j-\cdot)).
\end{equation}
For an untraced particle $1\leq j\leq k$, we define the field operators
\begin{equation}
    \Phi_j := \phi(w(x_j-\cdot)), \qquad \Phi_j' := \phi(w(x_j'-\cdot)).
\end{equation}
The $j$-th untraced particle contributes
\begin{equation}
\begin{aligned}
     &\frac{\eta}{\sqrt N}
    \int_{\R^{3(N-k)}}
    \left\langle
    \Psi_{N,t}(\ux_k',\ux_{N-k}),
    \Phi_j\Psi_{N,t}(\ux_k,\ux_{N-k})
    \right\rangle_{\Fock}
    d\ux_{N-k}
    \\
    &\quad -
    \frac{\eta}{\sqrt N}
    \int_{\R^{3(N-k)}}
    \left\langle
    \Phi_j'\Psi_{N,t}(\ux_k',\ux_{N-k}),
    \Psi_{N,t}(\ux_k,\ux_{N-k})
    \right\rangle_{\Fock}
    d\ux_{N-k}.
\end{aligned}
\end{equation}
Because the form factor $w$ is real-valued, $\Phi_j'$ is self-adjoint. Moving $\Phi_j'$ to the right side of the inner product allows us to combine the terms using the linearity of the field operator
\begin{equation}
    \Phi_j - \Phi_j' = \phi\big(w(x_j-\cdot) - w(x_j'-\cdot)\big) = \phi(f_j).
\end{equation}
This yields the $j$-th contribution
\begin{equation}
    \cQ_{N,t}^{(k)}(\ux_k;\ux_k';f_j) = \frac{1}{\sqrt N}
    \int
    \left\langle
    \Psi_{N,t}(\ux_k',\ux_{N-k}),
    \phi(f_j)
    \Psi_{N,t}(\ux_k,\ux_{N-k})
    \right\rangle_{\Fock}
    d\ux_{N-k}.
\end{equation}
Decomposing the source term into its annihilation and creation parts, we obtain
\begin{equation}
    \label{eq:Qa-and-Qad-jth-particle}
    \mathcal Q_{N,t}^{(k)}(\ux_k;\ux_k'; f_j) =\cQ^a + \cQ^{\ad},
\end{equation}
where, for notational convenience, we denote
\begin{equation}
    \cQ^a := \cQ_{N,t}^{a,(k)}(\ux_k;\ux_k';f_j) \quad \text{and} \quad  \cQ^{\ad} := \cQ_{N,t}^{\ad,(k)}(\ux_k;\ux_k';f_j)\,.
\end{equation}
Explicitly, the annihilation part is defined as
\begin{equation}
    \label{eq:definition-of-smeared-field-Qa}
    \cQ^a := \frac{1}{\sqrt N} \iint \cc{f_j(y)} \inp{ \Psi_{N,t}(\ux_k',\ux_{N-k}), a_y \Psi_{N,t}(\ux_k,\ux_{N-k}) }_{\Fock} dy \, d\ux_{N-k},
\end{equation}
and similarly for the creation part,
\begin{equation}
    \label{eq:definition-of-smeared-field-Qad}
    \cQ^{\ad}:= \frac{1}{\sqrt N} \iint f_j(y) \inp{ \Psi_{N,t}(\ux_k',\ux_{N-k}), \ad_y \Psi_{N,t}(\ux_k,\ux_{N-k}) }_{\Fock} dy \, d\ux_{N-k}.
\end{equation}
Summing over the untraced particles
$j=1,\ldots,k$ gives
\eqref{eq:def-source-term-Sk-finite-N-BBGKY}.
The field terms with $j>k$ act entirely on the
traced variables and cancel.
This establishes the BBGKY identities for the
regularized states.

We now restore the superscript $(L)$ and pass
to the original initial data.
Test the regularized hierarchy against a finite-rank
$\OJ^{(k)}$ satisfying the range conditions in the
statement, and integrate in time.
The kinetic test commutator
\begin{equation}
    \left[
       \OJ^{(k)},\sum_{j=1}^k\dch_0^{(j)}
    \right]
\end{equation}
extends to a bounded operator.
All particle-interaction test commutators are
bounded because $v\in L^\infty$.
Here the tests and density matrices are extended
by zero when full labeled tensor-product spaces
are required.

To justify convergence of the field contribution,
the source has the trace-class realization
\begin{equation}
\begin{aligned}
    \mathcal S_{N,t}^{(k)}
    &=
    \frac1{\sqrt N}\sum_{j=1}^k
    \operatorname{Tr}_{k+1,\ldots,N;\Fock}
    \Big(
       |\phi(w(x_j-\cdot))\Psi_{N,t}\rangle
          \langle\Psi_{N,t}|
    \\
    &\hspace{42mm}
       -
       |\Psi_{N,t}\rangle
          \langle\phi(w(x_j-\cdot))\Psi_{N,t}|
    \Big).
\end{aligned}
\end{equation}
This agrees with the kernel definition above.
The creation and annihilation estimates give
\begin{equation}
    \|\phi(w(x_j-\cdot))\zeta\|
    \leq
    2\|w\|_{L^2}
      \|(\cN_f+1)^{1/2}\zeta\|,
    \qquad
    \zeta\in\cD(\cN_f^{1/2}).
\end{equation}
Together with
$\||u\rangle\langle v|\|_{\mathfrak S_1}
=\|u\|\|v\|$
and the contractivity of partial trace, this shows
that the source depends continuously in trace
norm on the state in the photon-number graph norm.

Consequently,
\eqref{eq:BBGKY-number-graph-approximation}
implies
\begin{equation}
    \sup_{0\leq t\leq T}
    \|\mathcal S_{N,t}^{(k),L}
       -\mathcal S_{N,t}^{(k)}\|_{\mathfrak S_1}
    \longrightarrow0,
\end{equation}
where $\mathcal S_{N,t}^{(k),L}$ denotes the source
constructed from $\Psi_{N,t}^{(L)}$.
The same rank-one estimate gives
\begin{equation}
    \sup_{0\leq t\leq T}
    \|\gamma_{N,t}^{(r),L}
       -\gamma_{N,t}^{(r)}\|_{\mathfrak S_1}
    \longrightarrow0,
    \qquad 1\leq r\leq N.
\end{equation}

Every term of the time-integrated tested hierarchy
therefore passes to the limit as $L\to\infty$.
The resulting right-hand side is locally integrable
in time, proving the asserted absolute continuity
and the weak BBGKY identity.
The kernel formula follows in the distributional
sense described in the statement.
\end{proof}

\subsection{A priori bounds and energy estimates}
\label{sec:field-energy-estimate}
This section establishes the a priori bounds and energy estimates required to confirm that the Hartree hierarchy possesses the same regularity as the BBGKY hierarchy. These estimates are crucial for controlling the field operators and ensuring that the limiting dynamic is well-defined.

\begin{lemma}[Energy bounds for the field operators]
\label{lem:field-energy-operator-bound}
Let $\omg:\R^3\to[0,\infty)$ be measurable, and let
$\Hf=d\Gamma(\omg)$ be the corresponding field energy
operator on the bosonic Fock space $\Fock$, expressed
in the momentum representation.

Let $g\in\h$ satisfy
\begin{equation}
    \frac{\widehat g}{\sqrt{\omg}}\in L^2.
\end{equation}
Here this condition includes
$\widehat g=0$ almost everywhere on $\{\omg=0\}$,
and the quotient is defined to be zero on that set.
All momentum integrals and norms use the Fourier
measure fixed in the manuscript. Then
\begin{equation}
    \cD(\Hf^{1/2})
    \subset
    \cD(a(g))\mcap\cD(\ad(g)).
\end{equation}
For every $\psi\in\cD(\Hf^{1/2})$,
\begin{equation}
\begin{aligned}
    \|a(g)\psi\|_{\Fock}
    &\leq
    \left\|
        \frac{\widehat g}{\sqrt{\omg}}
    \right\|_{L^2}
    \|\Hf^{1/2}\psi\|_{\Fock},
    \\
    \|\ad(g)\psi\|_{\Fock}
    &\leq
    \left\|
        \frac{\widehat g}{\sqrt{\omg}}
    \right\|_{L^2}
    \|\Hf^{1/2}\psi\|_{\Fock}
    +
    \|g\|_{L^2}\|\psi\|_{\Fock}.
\end{aligned}
\end{equation}
Consequently, the field operator
$\phi(g)=\ad(g)+a(g)$ satisfies
\begin{equation}
    \|\phi(g)\psi\|_{\Fock}
    \leq
    2\left\|
        \frac{\widehat g}{\sqrt{\omg}}
    \right\|_{L^2}
    \|\Hf^{1/2}\psi\|_{\Fock}
    +
    \|g\|_{L^2}\|\psi\|_{\Fock}.
\end{equation}
Moreover, for every $\delta>0$,
\begin{equation}
    \left|
        \inp{\psi,\phi(g)\psi}_{\Fock}
    \right|
    \leq
    \delta\|\Hf^{1/2}\psi\|_{\Fock}^2
    +
    \frac{1}{\delta}
    \left\|
        \frac{\widehat g}{\sqrt{\omg}}
    \right\|_{L^2}^2
    \|\psi\|_{\Fock}^2.
\end{equation}
\end{lemma}

\begin{proof}
First suppose that
$\psi\in\cD(\Hf^{1/2})$ has only finitely many
nonzero photon sectors.
In the momentum representation,
\begin{equation}
    a(g)\psi
    =
    \int_{\R^3}
        \overline{\widehat g(\xi)}
        a_\xi\psi\,d\xi.
\end{equation}
The weighted Cauchy--Schwarz inequality and the
sectorwise expression for the field energy give
\begin{equation}
\begin{aligned}
    \|a(g)\psi\|_{\Fock}^2
    &\leq
    \left(
        \int_{\{\omg>0\}}
        \frac{|\widehat g(\xi)|^2}{\omg(\xi)}
        \,d\xi
    \right)
    \left(
        \int_{\R^3}
        \omg(\xi)\|a_\xi\psi\|_{\Fock}^2
        \,d\xi
    \right)
    \\
    &=
    \left\|
        \frac{\widehat g}{\sqrt{\omg}}
    \right\|_{L^2}^2
    \|\Hf^{1/2}\psi\|_{\Fock}^2.
\end{aligned}
\end{equation}
The canonical commutation relations imply
\begin{equation}
    \|\ad(g)\psi\|_{\Fock}^2
    =
    \|a(g)\psi\|_{\Fock}^2
    +
    \|g\|_{L^2}^2\|\psi\|_{\Fock}^2.
\end{equation}
Combining these two identities and using
$\sqrt{x+y}\leq\sqrt{x}+\sqrt{y}$ for $x,y\geq0$
proves the creation estimate for such vectors.

Now let $\psi\in\cD(\Hf^{1/2})$ be arbitrary and set
\begin{equation}
    \psi_L
    :=
    \mathbf{1}_{[0,L]}(\cN_f)\psi,
    \qquad L\in\mathbb N.
\end{equation}
The photon number projections commute with $\Hf$,
and therefore
\begin{equation}
    \psi_L\longrightarrow\psi,
    \qquad
    \Hf^{1/2}\psi_L
    \longrightarrow
    \Hf^{1/2}\psi
    \quad\text{in }\Fock.
\end{equation}
Applying the preceding estimates to
$\psi_L-\psi_M$ shows that
$\{a(g)\psi_L\}_L$ and $\{\ad(g)\psi_L\}_L$
are Cauchy sequences in $\Fock$.
Since the annihilation and creation operators are
closed, it follows that
\begin{equation}
    \psi\in\cD(a(g))\mcap\cD(\ad(g)),
\end{equation}
and
\begin{equation}
    a(g)\psi_L\longrightarrow a(g)\psi,
    \qquad
    \ad(g)\psi_L\longrightarrow\ad(g)\psi.
\end{equation}
Passing to the limit proves the annihilation and
creation estimates on all of $\cD(\Hf^{1/2})$.

The closed field operator agrees with
$\ad(g)+a(g)$ on this domain, as follows by the same
approximation. Hence the triangle inequality gives
\begin{equation}
\begin{aligned}
    \|\phi(g)\psi\|_{\Fock}
    &\leq
    \|\ad(g)\psi\|_{\Fock}
    +
    \|a(g)\psi\|_{\Fock}
    \\
    &\leq
    2\left\|
        \frac{\widehat g}{\sqrt{\omg}}
    \right\|_{L^2}
    \|\Hf^{1/2}\psi\|_{\Fock}
    +
    \|g\|_{L^2}\|\psi\|_{\Fock}.
\end{aligned}
\end{equation}

Finally, adjointness yields
\begin{equation}
    \inp{\psi,\phi(g)\psi}_{\Fock}
    =
    2\operatorname{Re}
    \inp{\psi,a(g)\psi}_{\Fock}.
\end{equation}
Thus the annihilation estimate and Young's
inequality imply, for every $\delta>0$,
\begin{equation}
\begin{aligned}
    \left|
        \inp{\psi,\phi(g)\psi}_{\Fock}
    \right|
    &\leq
    2\|\psi\|_{\Fock}\|a(g)\psi\|_{\Fock}
    \\
    &\leq
    2\left\|
        \frac{\widehat g}{\sqrt{\omg}}
    \right\|_{L^2}
    \|\psi\|_{\Fock}
    \|\Hf^{1/2}\psi\|_{\Fock}
    \\
    &\leq
    \delta\|\Hf^{1/2}\psi\|_{\Fock}^2
    +
    \frac{1}{\delta}
    \left\|
        \frac{\widehat g}{\sqrt{\omg}}
    \right\|_{L^2}^2
    \|\psi\|_{\Fock}^2.
\end{aligned}
\end{equation}
This completes the proof.
\end{proof}

\begin{lemma}[Symmetry and consistency of reduced marginals]
\label{lem:reduced-marginals-are-density-matrices}
Suppose that \hyporef{hyp:standing-hypotheses} holds. Let $\Psi_{N,t} \in \HN$ be a normalized wave function, and let $\gamma_{N,t}^{(k)}$ denote its $k$-particle reduced density matrix defined for $1 \leq k \leq N$ by
\begin{equation}
    \gamma_{N,t}^{(k)} := \operatorname{Tr}_{k+1,\ldots,N,\Fock} \left( |\Psi_{N,t}\rangle\langle\Psi_{N,t}| \right).
\end{equation}
Then, for every $1 \leq k \leq N$, $\gamma_{N,t}^{(k)}$ is a density operator on the $k$-particle Hilbert space satisfying
\begin{equation}
    \gamma_{N,t}^{(k)} \geq 0 \quad \text{and} \quad \operatorname{Tr}_{1,\ldots,k}\left( \gamma_{N,t}^{(k)} \right) = 1,
\end{equation}
which implies $\gamma_{N,t}^{(k)} \in \mathcal{D}_k$. Moreover, the marginals are permutation-symmetric and satisfy the compatibility relation
\begin{equation}
\label{eq:finite-N-marginal-consistency}
    \operatorname{Tr}_{k+1} \left( \gamma_{N,t}^{(k+1)} \right) = \gamma_{N,t}^{(k)}, \qquad 1 \leq k < N.
\end{equation}
\end{lemma}

\begin{proof}
By construction, the pure-state density operator $\Gamma_{N,t}:=|\Psi_{N,t}\rangle\langle\Psi_{N,t}|$ is a positive, rank-one projection on $\HN$ with unit trace. Because the partial trace is a positive, trace-preserving map, it preserves both operator positivity and the trace constraint. This directly implies $\gamma_{N,t}^{(k)} \geq 0$ and $\operatorname{Tr}_{1,\ldots,k}(\gamma_{N,t}^{(k)}) = 1$, confirming that $\gamma_{N,t}^{(k)} \in \mathcal{D}_k$. The permutation symmetry of the marginals is inherited from the bosonic nature of the state vector. Since $\Psi_{N,t}$ resides in the symmetric tensor product space, the density matrix $|\Psi_{N,t}\rangle\langle\Psi_{N,t}|$ commutes with the canonical action of the symmetric group $S_N$. Consequently, tracing out the field and the remaining $N-k$ particles yields a $k$-particle operator $\gamma_{N,t}^{(k)}$ that is invariant under any permutation within $S_k$. Finally, the consistency relation is a direct consequence of the composition property of successive partial traces on tensor product spaces. Expanding the partial trace yields
\begin{align*}
    \operatorname{Tr}_{k+1} \left( \gamma_{N,t}^{(k+1)} \right) 
    &= \operatorname{Tr}_{k+1} \left( \operatorname{Tr}_{k+2,\ldots,N, \Fock} \left( |\Psi_{N,t}\rangle\langle\Psi_{N,t}| \right) \right) \\
    &= \operatorname{Tr}_{k+1,\ldots,N, \Fock} \left( |\Psi_{N,t}\rangle\langle\Psi_{N,t}| \right) \\
    &= \gamma_{N,t}^{(k)}.
\end{align*}
This completes the proof.
\end{proof}

\begin{proposition}[Uniform a priori energy estimates]
\label{prop:uniform-prior-bounds-on-BBGKY-hierarchy}
Assume the Hamiltonian, trapping, and symmetry assumptions of
\hyporef{hyp:standing-hypotheses}. In particular, the one-particle
Hamiltonian $-\Delta+V$ is the self-adjoint operator associated with
its nonnegative closed quadratic form, $V\geq0$, and
$\omg(\xi)=|\xi|$. Assume that $v$ is real-valued and even,
$v\in W^{1,\infty}(\R^3)$, and that the real-valued, even form factor
$w$ satisfies
\begin{equation}
\label{eq:energy-estimate-form-factor-constants}
\begin{aligned}
    C_w
    &:= \norm{w}_{L^2}
       + \norm{\widehat w/\sqrt{\omg}}_{L^2}
       <\infty,\\
    C_w'
    &:= \norm{\nabla w}_{L^2}
       + \norm{\widehat{\nabla w}/\sqrt{\omg}}_{L^2}
       <\infty.
\end{aligned}
\end{equation}
Let $\lam,\eta\geq0$ satisfy
\begin{equation}
\label{eq:energy-estimate-smallness}
    2\lam\norm{v}_{L^\infty}+4\eta C_w<1.
\end{equation}
Let $\HMN$ be the Hamiltonian defined in
\eqref{eq:defination-of-Hamiltonian-of-particle-field-system}.
Throughout this proposition, set
\begin{equation}
    \OT_j:=(1-\Delta_{x_j}+V(x_j))^{1/2},
    \qquad j=1,\ldots,N.
\end{equation}
There exists a constant $C>0$, independent of both $N$ and $k$,
such that for every integer $k\geq1$ there is an integer
$N_0(k)\geq k$ with the following property:
for every $N\geq N_0(k)$ and every bosonic vector
$\Psi\in\cD((\HMN+N)^{k/2})$, the product
$\OT_1\cdots\OT_k\Psi$ is well-defined and
\begin{equation}
\label{eq:uniform-prior-bound}
    \norm{(\HMN+N)^{k/2}\Psi}^{2}
    \geq
    C^kN^k\norm{\OT_1\cdots\OT_k\Psi}^{2}.
\end{equation}
Equivalently, this is the quadratic-form inequality
\begin{equation}
    (\HMN+N)^k
    \geq C^kN^k\prod_{j=1}^k\OT_j^2
\end{equation}
on the bosonic subspace.

One may take $C=1/4$. The threshold may depend on $k$ and on
the fixed model parameters; in particular, it may depend on
$\norm{v}_{W^{1,\infty}}$ and $C_w'$.
\end{proposition}

\noindent
The proof is given in
\appref{sec:prf-prop:uniform-prior-bounds-on-BBGKY-hierarchy}.

\begin{corollary}[Uniform microscopic first-energy bound]
\label{cor:one-particle-energy-bound}
Assume \hyporef{hyp:standing-hypotheses}. Set
\begin{equation}
    \OT=(\dch_0+\kappa)^{1/2},
    \qquad
    \mathcal Q_N
    :=
    \sum_{j=1}^N\OT_j^2+\Hf+N.
\end{equation}
Define
\begin{equation}
\begin{aligned}
    c_*
    &:=
    \kappa+1+\frac{\lambda}{2}\|v\|_\infty
    +2\eta^2\|\widehat w/\sqrt{\omega}\|_2^2,\\
    c_+
    &:=
    \max\left\{
       \frac32,\,
       1+
       \frac{
          \lambda\|v\|_\infty
          +4\eta^2\|\widehat w/\sqrt{\omega}\|_2^2
       }{\kappa+1}
    \right\}.
\end{aligned}
\end{equation}
Then, for every sufficiently large $N$,
\begin{equation}
\label{eq:uniform-total-first-energy}
This    \sup_{t\in\mathbb R}
    \frac1N
    \|\mathcal Q_N^{1/2}\Psi_{N,t}\|^2
    \leq 2(C_0+c_*-1),
\end{equation}
where $C_0$ is the constant in
\hypref{hyp:bounded-energy-per-particle}.
In particular,
\begin{equation}
\label{eq:one-particle-energy-bound}
    \sup_{t\in\mathbb R}
    \operatorname{Tr}\!\left(
       \OT\gamma_{N,t}^{(1)}\OT
    \right)
    \leq 2(C_0+c_*-1).
\end{equation}
All unbounded expectations are understood as quadratic
forms, and the weighted density matrix is understood
through its positive trace-class realization.
\end{corollary}
\begin{proof}
For a vector $\Psi$ in the free energy form domain,
the bounded pair interaction satisfies
\begin{equation}
    \left|
        \frac{\lambda}{N}
        \sum_{i<j}\langle\Psi,v_{ij}\Psi\rangle
    \right|
    \leq
    \frac{\lambda}{2}\|v\|_\infty N\|\Psi\|^2.
\end{equation}
Adjunction and the annihilation estimate give for every $\delta>0$,
\begin{equation}
\begin{aligned}
 &\left|
    \frac{\eta}{\sqrt N}
    \sum_{j=1}^N
        \langle\Psi,\phi(w(x_j-\cdot))\Psi\rangle
 \right|\\
 &\qquad\leq
 2\eta\sqrt N\|\widehat w/\sqrt{\omega}\|_2
       \|\Psi\|\|\Hf^{1/2}\Psi\|\\
 &\qquad\leq
 \delta\|\Hf^{1/2}\Psi\|^2
 +\frac{\eta^2}{\delta}
     \|\widehat w/\sqrt{\omega}\|_2^2
     N\|\Psi\|^2.
\end{aligned}
\end{equation}
The field expectation in this estimate is understood
as a quadratic form. The estimate extends from the
finite-photon core to the free energy form domain.

Taking $\delta=1/2$ and using $\dch_0\geq0$, we obtain
the following inequalities of quadratic forms:
\begin{equation}
\begin{aligned}
 \HMN+c_*N
 &\geq
 \sum_{j=1}^N\dch_0^{(j)}
 +\frac12\Hf+(\kappa+1)N
 \geq\frac12\mathcal Q_N,\\
 \HMN+c_*N
 &\leq
 \sum_{j=1}^N\dch_0^{(j)}
 +\frac32\Hf\\
 &\quad+
 \left(
    \kappa+1+\lambda\|v\|_\infty
    +4\eta^2\|\widehat w/\sqrt{\omega}\|_2^2
 \right)N
 \leq c_+\mathcal Q_N.
\end{aligned}
\end{equation}
Thus
\begin{equation}
\label{eq:first-energy-form-comparison}
    \frac12\mathcal Q_N
    \leq\HMN+c_*N
    \leq c_+\mathcal Q_N.
\end{equation}
The pair interaction is bounded, and the field
interaction is infinitesimally form bounded with
respect to the free Hamiltonian. The resulting
closed form has the same form domain as
$\mathcal Q_N$ and represents the self-adjoint
Hamiltonian constructed in
\propref{prop:appendix-self-adjointness}.
Consequently,
\begin{equation}
    \cD\bigl((\HMN+c_*N)^{1/2}\bigr)
    =
    \cD(\mathcal Q_N^{1/2}).
\end{equation}

The positivity of $\HMN+N$ follows from the
$k=1$ estimate in
\propref{prop:uniform-prior-bounds-on-BBGKY-hierarchy}.
Since $c_*>1$, spectral calculus also gives
\begin{equation}
    \cD\bigl((\HMN+c_*N)^{1/2}\bigr)
    =
    \cD\bigl((\HMN+N)^{1/2}\bigr).
\end{equation}
For every sufficiently large $N$, the case $k=1$
of \hypref{hyp:bounded-energy-per-particle} therefore
implies that $\Psi_{N,0}$ belongs to this common
form domain and satisfies
\begin{equation}
    \|(\HMN+N)^{1/2}\Psi_{N,0}\|^2<C_0N.
\end{equation}
Using $\|\Psi_{N,0}\|=1$, we obtain
\begin{equation}
\begin{aligned}
 \|(\HMN+c_*N)^{1/2}\Psi_{N,0}\|^2
 &=
 \|(\HMN+N)^{1/2}\Psi_{N,0}\|^2
 +(c_*-1)N\\
 &\leq (C_0+c_*-1)N.
\end{aligned}
\end{equation}

The spectral theorem implies that $e^{-it\HMN}$
preserves this form domain and that
\begin{equation}
    \|(\HMN+c_*N)^{1/2}\Psi_{N,t}\|^2
    =
    \|(\HMN+c_*N)^{1/2}\Psi_{N,0}\|^2.
\end{equation}
Combining this identity with
\eqref{eq:first-energy-form-comparison} gives
\begin{equation}
\begin{aligned}
 \frac1N\|\mathcal Q_N^{1/2}\Psi_{N,t}\|^2
 &\leq
 \frac2N\|(\HMN+c_*N)^{1/2}\Psi_{N,t}\|^2\\
 &=
 \frac2N\|(\HMN+c_*N)^{1/2}\Psi_{N,0}\|^2\\
 &\leq 2(C_0+c_*-1).
\end{aligned}
\end{equation}
This proves \eqref{eq:uniform-total-first-energy}
uniformly for $t\in\mathbb R$.

Finally, bosonic symmetry and the definition of
the reduced density matrix give
\begin{equation}
\begin{aligned}
    \operatorname{Tr}
        (\OT\gamma_{N,t}^{(1)}\OT)
    &=
    \frac1N
    \sum_{j=1}^N\|\OT_j\Psi_{N,t}\|^2\\
    &\leq
    \frac1N\|\mathcal Q_N^{1/2}\Psi_{N,t}\|^2.
\end{aligned}
\end{equation}
The trace identity follows first for bounded spectral
truncations of $\OT^2$ and then by monotone convergence.
Its finiteness provides the positive trace-class
realization of $\OT\gamma_{N,t}^{(1)}\OT$.
This proves \eqref{eq:one-particle-energy-bound}.
\end{proof}

\begin{corollary}[Uniform $k$-particle multiplicative energy bound]
\label{cor:k-particle-energy-bound}
Assume the hypotheses of
\propref{prop:uniform-prior-bounds-on-BBGKY-hierarchy}
and the initial moment condition
\hypref{hyp:bounded-energy-per-particle}.
Let $C_0>0$ and $K(N)$ be as in that initial moment
condition, and let
\begin{equation}
    \OT=(\dch_0+\kappa)^{1/2},
    \qquad
    \OT_k=\prod_{j=1}^k\OT^{(j)},
    \qquad \kappa>0.
\end{equation}
Then there exists a constant $C_E>0$, independent of
$N$, $k$, and $t$, such that, for every sufficiently
large $N$ and every integer $k$ satisfying
\begin{equation}
    1\leq k\leq K(N),
    \qquad N\geq N_0(k),
\end{equation}
the operator
$\OT_k\gamma_{N,t}^{(k)}\OT_k$
has a positive trace-class realization and
\begin{equation}
\label{eq:k-particle-energy-bound}
    \sup_{t\in\R}
    \Tr_{\Hkb}\!\left(
        \OT_k\gamma_{N,t}^{(k)}\OT_k
    \right)
    \leq C_E^k.
\end{equation}
One may take
\begin{equation}
    C_E=4\max\{1,\kappa\}C_0.
\end{equation}
In particular, for every fixed $k\geq1$,
\eqref{eq:k-particle-energy-bound} holds for all
sufficiently large $N$.
\end{corollary}

\begin{proof}
Fix a sufficiently large $N$ for which
\hypref{hyp:bounded-energy-per-particle} holds,
and fix an integer $k$ satisfying
\begin{equation}
    1\leq k\leq K(N),
    \qquad N\geq N_0(k).
\end{equation}
All operators acting on the first $k$ particle
coordinates below are understood to be tensored
with the identity on the remaining coordinates
and on the field space.

By \hypref{hyp:bounded-energy-per-particle},
with Hamiltonian moments understood in the
quadratic-form sense,
\begin{equation}
    \Psi_{N,0}\in
    \cD\bigl((\HMN+N)^{k/2}\bigr),
    \qquad
    \norm{(\HMN+N)^{k/2}\Psi_{N,0}}^2
    \leq C_0^kN^k.
\end{equation}
The evolution
$\Psi_{N,t}=e^{-it\HMN}\Psi_{N,0}$
preserves this domain. Indeed, by the spectral
calculus for the self-adjoint operator $\HMN$,
\begin{equation}
    (\HMN+N)^{k/2}\Psi_{N,t}
    =
    e^{-it\HMN}
    (\HMN+N)^{k/2}\Psi_{N,0}.
\end{equation}
Consequently, unitarity gives
\begin{equation}
\begin{aligned}
    \norm{(\HMN+N)^{k/2}\Psi_{N,t}}^2
    &=
    \norm{(\HMN+N)^{k/2}\Psi_{N,0}}^2\\
    &\leq C_0^kN^k,
    \qquad t\in\R.
\end{aligned}
\end{equation}

Apply
\propref{prop:uniform-prior-bounds-on-BBGKY-hierarchy}
with its permitted constant $C=1/4$.
Since $N\geq N_0(k)$, we obtain
\begin{equation}
\begin{aligned}
    \left\|
        \prod_{j=1}^k
        (1+\dch_0^{(j)})^{1/2}\Psi_{N,t}
    \right\|^2
    &\leq
    \frac{4^k}{N^k}
    \norm{(\HMN+N)^{k/2}\Psi_{N,t}}^2\\
    &\leq (4C_0)^k.
\end{aligned}
\end{equation}

We next compare the shift used in that proposition
with the shift in $\OT$.
Since $\dch_0\geq0$,
\begin{equation}
    \dch_0+\kappa
    \leq
    \max\{1,\kappa\}(1+\dch_0).
\end{equation}
The operators $\dch_0^{(j)}$ acting on different
particle coordinates strongly commute.
Their joint spectral calculus therefore gives
\begin{equation}
\begin{aligned}
    \left\|
        \prod_{j=1}^k\OT^{(j)}\Psi_{N,t}
    \right\|^2
    &\leq
    \max\{1,\kappa\}^k
    \left\|
        \prod_{j=1}^k
        (1+\dch_0^{(j)})^{1/2}\Psi_{N,t}
    \right\|^2\\
    &\leq
    \bigl[4\max\{1,\kappa\}C_0\bigr]^k.
\end{aligned}
\end{equation}
In particular, the product on the left is
well-defined on $\Psi_{N,t}$.

It remains to express this bound in terms of
the reduced density matrix.
For $L>0$, let $\min\{\OT_k,L\}$ denote the bounded
operator obtained from $\OT_k$ by spectral calculus.
The defining property of the partial trace yields
\begin{equation}
    \Tr_{\Hkb}\!\left(
        \gamma_{N,t}^{(k)}
        \min\{\OT_k,L\}^{\,2}
    \right)
    =
    \left\|
        \bigl(\min\{\OT_k,L\}\otimes\1\bigr)
        \Psi_{N,t}
    \right\|^2.
\end{equation}
As $L\to\infty$, monotone convergence for the
corresponding scalar spectral integrals gives
\begin{equation}
    \lim_{L\to\infty}
    \Tr_{\Hkb}\!\left(
        \gamma_{N,t}^{(k)}
        \min\{\OT_k,L\}^{\,2}
    \right)
    =
    \left\|
        \prod_{j=1}^k\OT^{(j)}\Psi_{N,t}
    \right\|^2
    <\infty.
\end{equation}
This finiteness implies that
$\OT_k(\gamma_{N,t}^{(k)})^{1/2}$
is Hilbert--Schmidt. Hence
$\OT_k\gamma_{N,t}^{(k)}\OT_k$
has the canonical positive trace-class realization
\begin{equation}
    \bigl(\OT_k(\gamma_{N,t}^{(k)})^{1/2}\bigr)
    \bigl(\OT_k(\gamma_{N,t}^{(k)})^{1/2}\bigr)^*,
\end{equation}
and its trace satisfies
\begin{equation}
\begin{aligned}
    \Tr_{\Hkb}\!\left(
        \OT_k\gamma_{N,t}^{(k)}\OT_k
    \right)
    &=
    \left\|
        \prod_{j=1}^k\OT^{(j)}\Psi_{N,t}
    \right\|^2\\
    &\leq
    \bigl[4\max\{1,\kappa\}C_0\bigr]^k.
\end{aligned}
\end{equation}
The right-hand side is independent of $t$.
Taking the supremum over $t\in\R$ proves
\eqref{eq:k-particle-energy-bound}.

Finally, fix $k\geq1$.
Because $K(N)\to\infty$, the condition
$k\leq K(N)$ holds for all sufficiently large $N$.
Increasing this lower threshold on $N$ also ensures
$N\geq N_0(k)$ and applicability of the initial
moment hypothesis.
This proves the final assertion.
\end{proof}

In the following two lemmas, expressions involving
individual particle labels are interpreted on the
full Hilbert space
\begin{equation}
    L^2(\R^{3N})\otimes\Fock.
\end{equation}
On this space, $\HMN$ denotes the self-adjoint
Hamiltonian defined by the same expression as above,
and $\UN(t,s)=e^{-i(t-s)\HMN}$ denotes its unitary
propagator. The Hamiltonian and propagator preserve
particle permutation symmetry, and their restrictions
to $\HN$ agree with the bosonic dynamics. For each $\ell$, multiplication by
$e^{-i\xi\cdot x_\ell}$ is unitary on the full
space, but need not preserve its bosonic subspace.
Accordingly, individual terms such as
\begin{equation}
    \UN(t,s)e^{-i\xi\cdot x_\ell}\Psi_{N,s}
\end{equation}
are evaluated in the full space. Their sum over
$\ell=1,\ldots,N$ is permutation symmetric.

A bounded particle test $\OJ^{(k)}$ on $\Hkb$
is extended by zero to $\h^{\otimes k}$ and then
tensored with the identity on the remaining
particle coordinates and the field.
Intermediate labeled source components and their
partial traces are interpreted on these full
tensor products. After summing over the first
$k$ particle labels, the reduced source acts
on $\Hkb$.

\begin{lemma}[Dynamics of the annihilation field]
\label{lem:affecrt-of-a(xi)-on-PsiNt}
Assume \hyporef{hyp:standing-hypotheses}, and let
\begin{equation}
    \Psi_{N,t}=\UN(t,0)\Psi_{N,0},
    \qquad
    \Psi_{N,0}
    =\phi_{N,0}\otimes\wyl[\sqrt N\,u_0]\vacf,
\end{equation}
where $\|\phi_{N,0}\|=1$ and
$\UN(t,s)=e^{-i(t-s)\HMN}$.

For every $t\geq0$,
$\Psi_{N,t}\in\cD(\cN_f^{1/2})$, and
\begin{equation}
\label{eq:PsiNt-axi-PsiNt}
\begin{aligned}
    a_\xi\Psi_{N,t}
    &=
    \sqrt N\,e^{-it\omg(\xi)}
       \widehat u_0(\xi)\Psi_{N,t}
    \\
    &\quad
    -i\frac{\eta}{\sqrt N}
      \int_0^t
      \widehat w(\xi)e^{-i(t-s)\omg(\xi)}
      \UN(t,s)
      \left(\sum_{\ell=1}^N e^{-i\xi\cdot x_\ell}\right)
      \Psi_{N,s}\,ds.
\end{aligned}
\end{equation}
Here $a_\xi\Psi_{N,t}$ denotes the sectorwise annihilation
field, and the identity holds in $L^2(\R^3_\xi;\HN)$. The time integral is a Bochner integral in this space. Consequently, for each fixed $t$, the identity holds for
almost every $\xi$ and may be smeared against any
$f\in L^2(\R^3)$.
\end{lemma}

\begin{proof}
We work in the momentum representation of the photon Fock
space. Define the closed annihilation map
\begin{equation}
    \mathcal A:\cD(\cN_f^{1/2})
       \subset\HN\longrightarrow L^2(\R^3_\xi;\HN)
\end{equation}
sectorwise by
\begin{equation}
    (\mathcal A\Psi)^{(n)}
       (\xi;\xi_1,\ldots,\xi_n)
    =
    \sqrt{n+1}\,
    \Psi^{(n+1)}(\xi,\xi_1,\ldots,\xi_n),
\end{equation}
with the particle variables suppressed.

Thus, $(\mathcal A\Psi)(\xi)=a_\xi\Psi$ in the
$L^2$ sense, and
\begin{equation}
    \|\mathcal A\Psi\|_{L^2(\R^3_\xi;\HN)}^2
    =
    \|\cN_f^{1/2}\Psi\|_{\HN}^2.
\end{equation}
The independent photon-number estimate
\eqref{eq:joint-source-uniform-number}, proved using
bounded photon-number cutoffs, implies that, for every
$T<\infty$,
\begin{equation}
    \sup_{0\leq t\leq T}
       \|\mathcal A\Psi_{N,t}\|_{L^2(\R^3_\xi;\HN)}
    <\infty.
\end{equation}
That cutoff argument establishes the required
membership in $\cD(\cN_f^{1/2})$ and does not use
the identity being proved here.

On $L^2(\R^3_\xi;\HN)$, let $K$ be the self-adjoint
direct-integral operator
\begin{equation}
    (KZ)(\xi)
    =
    \bigl(\HMN+\omg(\xi)\bigr)Z(\xi).
\end{equation}
Also define
$B:\HN\to L^2(\R^3_\xi;\HN)$ by
\begin{equation}
    (B\Psi)(\xi)
    =
    \frac{\eta}{\sqrt N}\widehat w(\xi)
    \left(\sum_{\ell=1}^N e^{-i\xi\cdot x_\ell}\right)
    \Psi.
\end{equation}
The sum over particle labels preserves bosonic symmetry,
and
\begin{equation}
    \|B\Psi\|_{L^2(\R^3_\xi;\HN)}
    \leq
    \eta\sqrt N\,\|\widehat w\|_{L^2}\|\Psi\|_{\HN}.
\end{equation}
In particular, $B$ is bounded.
We justify the commutator calculation using smeared
operators on an operator core. Let $\mathscr C_N$ be
the linear span of vectors
\begin{equation}
    \psi\otimes
    \ad(f_1)\cdots\ad(f_m)\vacf, 
\end{equation}
for $   \psi\in\cD(\HMNbos),\quad f_j\in\cD(\omg),\quad m\geq0.$
This is an operator core for $\HMNbos+\Hf$.
Since the particle--field interaction is infinitesimally
operator bounded relative to this free Hamiltonian,
$\mathscr C_N$ is also an operator core for $\HMN$.
Consequently,
\begin{equation}
    \mathscr D
    :=
    \operatorname{span}
    \{f\otimes\chi:
      f\in\cD(\omg),\ \chi\in\mathscr C_N\}
\end{equation}
is an operator core for $K$.

For a simple tensor in this core,
\begin{equation}
    \mathcal A^\dagger (f\otimes\chi)=\ad(f)\chi.
\end{equation}
For every $\Phi\in\mathscr D$, we have
$\mathcal A^\dagger \Phi\in\cD(\HMN)$ and
$K\Phi\in\cD(\mathcal A^\dagger )$.
The latter assertion follows because $K\Phi$ has only
finitely many photon sectors.
The smeared canonical commutation relations give the
vector identity
\begin{equation}
    \HMN\mathcal A^\dagger \Phi
    =
    \mathcal A^\dagger K\Phi+B^\dagger \Phi,
    \qquad \Phi\in\mathscr D.
\end{equation}
Using the convention that inner products are linear
in the second argument, the weak Schr\"odinger equation
therefore yields
\begin{equation}
\begin{aligned}
    i\frac{d}{dt}
      \langle\Phi,\mathcal A\Psi_{N,t}\rangle
    &=
    \langle\HMN\mathcal A^\dagger \Phi,\Psi_{N,t}\rangle
    \\
    &=
    \langle K\Phi,\mathcal A\Psi_{N,t}\rangle
    +
    \langle\Phi,B\Psi_{N,t}\rangle,
    \qquad \Phi\in\mathscr D.
\end{aligned}
\end{equation}
The local bound on $\mathcal A\Psi_{N,t}$ and the identity
\begin{equation}
    \langle\Phi,\mathcal A\Psi_{N,t}\rangle
    =
    \langle\mathcal A^\dagger \Phi,\Psi_{N,t}\rangle
\end{equation}
imply weak continuity of $t\mapsto\mathcal A\Psi_{N,t}$. By approximation in the graph norm of $K$, the
time-integrated weak equation extends from
$\mathscr D$ to every $\Phi\in\cD(K)$.  To recover the corresponding mild equation, set
\begin{equation}
    R_\varepsilon=(\1+i\varepsilon K)^{-1},
    \qquad \varepsilon>0.
\end{equation}
The weak equation implies that
$t\mapsto R_\varepsilon\mathcal A\Psi_{N,t}$
is locally absolutely continuous and satisfies
\begin{equation}
    i\partial_t
       \bigl(R_\varepsilon\mathcal A\Psi_{N,t}\bigr)
    =
    K R_\varepsilon\mathcal A\Psi_{N,t}
    +
    R_\varepsilon B\Psi_{N,t}.
\end{equation}
Here $K R_\varepsilon$ is bounded.
Duhamel formula gives
\begin{equation}
\begin{aligned}
    R_\varepsilon\mathcal A\Psi_{N,t}
    &=
    e^{-itK}R_\varepsilon\mathcal A\Psi_{N,0}
    -i\int_0^t
       e^{-i(t-s)K}R_\varepsilon B\Psi_{N,s}\,ds.
\end{aligned}
\end{equation}
Since $\|R_\varepsilon\|\leq1$ and
$R_\varepsilon\to \1$ strongly, boundedness of $B$
allows us to pass to the limit and obtain
\begin{equation}
    \mathcal A\Psi_{N,t}
    =
    e^{-itK}\mathcal A\Psi_{N,0}
    -
    i\int_0^t
      e^{-i(t-s)K}B\Psi_{N,s}\,ds.
\end{equation}
The coherent-state expansion from
\lemref{lem:basic-properties-of-Weyl-operators}
gives, sectorwise,
\begin{equation}
    (\mathcal A\Psi_{N,0})(\xi)
    =
    \sqrt N\,\widehat u_0(\xi)\Psi_{N,0}.
\end{equation}
Moreover, the direct-integral functional calculus gives
\begin{equation}
    (e^{-irK}Z)(\xi)
    =
    e^{-ir\omg(\xi)}\UN(r,0)Z(\xi).
\end{equation}
Substitution into the mild equation proves
\eqref{eq:PsiNt-axi-PsiNt} in
$L^2(\R^3_\xi;\HN)$. Finally, pairing this $L^2$ identity with
$\overline{f(\xi)}$ is justified by Cauchy--Schwarz for
every $f\in L^2(\R^3)$, proving the smearing assertion.
\end{proof}

\begin{lemma}[Uniform bound of the coherent photon source]
\label{lem:uniform-field-source-bound}
Suppose that \hyporef{hyp:standing-hypotheses} holds. For each $\xi\in\R^3$, $0\leq s\leq t\leq T$, and $1\leq\ell\leq N$, define the unitary operator
on $L^2(\R^{3N})\otimes\Fock$ by
\begin{equation}
\label{eq:transported-observable-for-source-bound}
    \cT_{x_\ell}(t,s;\xi)
    :=
    \UN(t,s)e^{-i\xi\cdot x_\ell}\UNd(t,s).
\end{equation}
Here the propagator is understood on the full space according to the preceding convention. Then, for every fixed $k\geq1$ and every $  \OJ^{(k)} \in\cB(\Hkb)$, there exists a constant $C_{k,J,T}<\infty$, independent of $N\geq k$, such that
\begin{equation}
\label{eq:uniform-field-source-bound}
    \sup_{N\geq k} \sup_{t\in[0,T]} \left| \Tr \prc{   \OJ^{(k)} \cS_{N,t}^{(k)} } \right| \leq C_{k,J,T},
\end{equation}
where one may explicitly choose
\begin{equation}
\label{eq:explicit-uniform-field-source-constant}
    C_{k,J,T} := 4k \norm{   \OJ^{(k)} }_{\op{op}} \norm{ \cc{\wh{w}}\,\wh{u}_0 }_{L^1} + 4k \eta T \norm{   \OJ^{(k)} }_{\op{op}} \norm{ |\wh{w}|^2 }_{L^1}.
\end{equation}
Consequently, the family of scalar functions $\brs{t \mapsto \Tr(   \OJ^{(k)} \cS_{N,t}^{(k)})}_{N\geq k}$ is uniformly integrable on $[0,T]$.
\end{lemma}

\begin{proof}
Fix $N\geq k$ and $t\in[0,T]$.
Individual labeled operators and partial traces are
interpreted on the full tensor-product spaces, as
specified above.

By \lemref{lem:affecrt-of-a(xi)-on-PsiNt},
unitarity of $\UN(t,s)$, and normalization of
$\Psi_{N,s}$, we have, for almost every $\xi$,
\begin{equation}
\begin{aligned}
    \frac1{\sqrt N}\|a_\xi\Psi_{N,t}\|
    &\leq
    |\widehat u_0(\xi)|
    +
    \frac{\eta}{N}|\widehat w(\xi)|
    \int_0^t
    \left\|
       \left(\sum_{\ell=1}^N
                   e^{-i\xi\cdot x_\ell}\right)
       \Psi_{N,s}
    \right\|\,ds
    \\
    &\leq
    |\widehat u_0(\xi)|
    +\eta t|\widehat w(\xi)|.
\end{aligned}
\end{equation}
The exceptional null set may depend on $t$.

For $1\leq j\leq k$, the Fourier representation of
the annihilation component gives
\begin{equation}
\begin{aligned}
    \mathcal Q_{N,t}^{(k),a,j}
    &=
    \frac1{\sqrt N}
    \int_{\R^3}\overline{\widehat w(\xi)}
    \operatorname{Tr}_{k+1,\ldots,N;\Fock}
    \left[
       e^{i\xi\cdot x_j},
       |a_\xi\Psi_{N,t}\rangle
       \langle\Psi_{N,t}|
    \right]\,d\xi.
\end{aligned}
\end{equation}
Indeed, the commutator produces the kernel multiplier
$e^{i\xi\cdot x_j}-e^{i\xi\cdot x_j'}$, which is precisely
the factor in
\begin{equation}
    \overline{\widehat f_j(\xi)}
    =
    \left(e^{i\xi\cdot x_j}
          -e^{i\xi\cdot x_j'}\right)
      \overline{\widehat w(\xi)}.
\end{equation}
This representation first holds for finite-photon vectors and extends to
$\cD(\cN_f^{1/2})$ by approximation in its graph norm.

We now justify the integral in trace norm.
The partial trace is contractive on trace-class
operators, including non-self-adjoint ones.
Moreover, multiplication by $e^{i\xi\cdot x_j}$ is
unitary on the full labeled particle space, and
\begin{equation}
    \bigl\|
       |a_\xi\Psi_{N,t}\rangle
       \langle\Psi_{N,t}|
    \bigr\|_{\mathfrak S_1}
    =
    \|a_\xi\Psi_{N,t}\|,
\end{equation}
because $\|\Psi_{N,t}\|=1$.
Consequently,
\begin{equation}
\begin{aligned}
    \|\mathcal Q_{N,t}^{(k),a,j}\|_{\mathfrak S_1}
    &\leq
    \frac2{\sqrt N}
    \int_{\R^3}
       |\widehat w(\xi)|
       \|a_\xi\Psi_{N,t}\|\,d\xi
    \\
    &\leq
    2\|\overline{\widehat w}\,\widehat u_0\|_{L^1}
    +2\eta t\||\widehat w|^2\|_{L^1}.
\end{aligned}
\end{equation}
Both terms are finite by Cauchy--Schwarz and
$w,u_0\in L^2$.
The integrand is strongly measurable in trace class,
so the displayed representation is a
trace-class-valued Bochner integral.

Extend $\OJ^{(k)}\in\cB(\Hkb)$ by zero on the
orthogonal complement of $\Hkb$ in $\h^{\otimes k}$.
This extension preserves its operator norm.
Trace duality therefore gives
\begin{equation}
\label{eq:annihilation-source-bound-one-j}
\begin{aligned}
    \left|
       \operatorname{Tr}
       \bigl(\OJ^{(k)}
             \mathcal Q_{N,t}^{(k),a,j}\bigr)
    \right|
    &\leq
    2\|\OJ^{(k)}\|_{\mathrm{op}}
      \|\overline{\widehat w}\,\widehat u_0\|_{L^1}
    \\
    &\quad+
    2\eta T\|\OJ^{(k)}\|_{\mathrm{op}}
      \||\widehat w|^2\|_{L^1}.
\end{aligned}
\end{equation}

Since $f_j$ changes sign when the primed and
unprimed particle variables are exchanged,
adjunction gives
\begin{equation}
    \mathcal Q_{N,t}^{(k),\ad,j}
    =
    -\bigl(\mathcal Q_{N,t}^{(k),a,j}\bigr)^*.
\end{equation}
Thus the creation component has the same
trace-norm bound as the annihilation component.

Although individual labeled components need not
preserve the bosonic subspace, their sums have
bosonic support.
Indeed, both $\Psi_{N,t}$ and $a_\xi\Psi_{N,t}$
are symmetric in the particle variables, while
$\sum_{j=1}^k e^{i\xi\cdot x_j}$ commutes with
particle permutations.
The summed commutator therefore acts on $\Hkb$.
It follows that
\begin{equation}
\begin{aligned}
    \|\mathcal S_{N,t}^{(k)}\|_{\mathfrak S_1(\Hkb)}
    &\leq
    4k\|\overline{\widehat w}\,\widehat u_0\|_{L^1}
    +4k\eta t\||\widehat w|^2\|_{L^1}.
\end{aligned}
\end{equation}
Applying trace duality on $\Hkb$ proves
\eqref{eq:uniform-field-source-bound}, with exactly
the constant in
\eqref{eq:explicit-uniform-field-source-constant}.

For completeness, the mild equation established in
\lemref{lem:affecrt-of-a(xi)-on-PsiNt} shows that
\begin{equation}
    t\longmapsto
    \bigl(\xi\longmapsto a_\xi\Psi_{N,t}\bigr)
\end{equation}
is continuous in $L^2(\R^3_\xi;\HN)$.
Together with strong continuity of $\Psi_{N,t}$,
the rank-one representation above and
Cauchy--Schwarz imply trace-norm continuity of
$t\mapsto\mathcal S_{N,t}^{(k)}$.
In particular, all tested scalar functions are
measurable.

Finally, for every measurable $E\subset[0,T]$,
\begin{equation}
    \sup_{N\geq k}
    \int_E
       \left|
          \operatorname{Tr}
          \bigl(\OJ^{(k)}\mathcal S_{N,t}^{(k)}\bigr)
       \right|\,dt
    \leq C_{k,J,T}|E|.
\end{equation}
This proves uniform integrability.
\end{proof}

\section{Existence and uniqueness of the Hartree hierarchy}
\label{sec:existence-and-uniqueness-of-GP-hierarchy}

\subsection{Weak star topology and compactness}
Throughout the following discussion, operators on
$\Hkb$ are extended by zero to $\h^{\otimes k}$
whenever labeled tensor-product operators are used.
We retain the same notation for these extensions.
The same convention applies to density matrices. Individual interaction multipliers, tensor products, and their commutators are first formed on the full
tensor-product spaces. Their pairings with bosonic
density matrices are evaluated there using the zero
extensions. Equivalently, the complete test operator
is orthogonally compressed to the bosonic subspace
before taking the trace.

\begin{definition}[State spaces and operators on the $k$-particle Hilbert space]
 Let $\Hkb$ denote the bosonic $k$-particle Hilbert space defined in \eqref{eq:definition-of-HNb}. Suppose that \hyporef{hyp:standing-hypotheses} holds. We define the following objects
\begin{enumerate}
    \item \textbf{Spectral cutoffs:} For any spectral cutoff parameter $M > 0$, we define the spectral projection operators onto the low-energy modes on $\h$ as
    \begin{equation}
        P_M := \1_{\{ \OT \leq M\}}, \qquad Q_M := \1 - P_M,
    \end{equation}
    where $\OT$ is the one-particle operator defined in \eqref{eq:definition-of-T}. These projections are extended to the $k$-particle Hilbert space $\Hkb$ via tensor products
    \begin{equation}
        P_M^{(k)} := P_M^{\otimes k} \quad \text{ and } \quad Q_M^{(k)} := Q_M^{\otimes k}.
    \end{equation}

    \item \textbf{Trace spaces and duality:} We denote the spaces of trace-class and compact operators on $\Hkb$ by $\mathfrak{S}_1^{(k)}$ and $\mathcal{K}^{(k)}$, respectively. The space $\mathfrak{S}_1^{(k)}$ is equipped with the trace norm $\|A\|_{\mathfrak{S}_1^{(k)}} := \Tr_{1,\ldots,k} |A|$. We utilize the isometric duality $\mathfrak{S}_1^{(k)} = \big(\mathcal{K}^{(k)}\big)^*$, where the duality pairing is given by
    \begin{equation}
        \langle \OJ, \gamma \rangle := \Tr_{1,\ldots,k} [\OJ \gamma], \qquad \text{for } \OJ \in \mathcal{K}^{(k)}, \, \gamma \in \mathfrak{S}_1^{(k)}.
    \end{equation}
    Consequently, the weak-$*$ topology on $\mathfrak{S}_1^{(k)}$ corresponds to the topology of convergence against compact observables.

    \item \textbf{Partial trace:} For a trace-class operator $\op{A}^{(k+1)} \in \mathfrak{S}_1^{(k+1)}$, its partial trace over the $(k+1)$-th particle, denoted $\Tr_{k+1}[ \op{A}^{(k+1)} ] \in \mathfrak{S}_1^{(k)}$, is defined as the unique operator satisfying
    \begin{equation}
        \Tr_{1,\ldots,k} \Big[ \OB^{(k)} \Tr_{k+1} [ \op{A}^{(k+1)} ] \Big] = \Tr_{1,\ldots,k+1} \Big[ (\OB^{(k)} \otimes \1) \op{A}^{(k+1)} \Big]
    \end{equation}
    for all bounded operators $\OB^{(k)}$ on $\Hkb$. When $\op{A}^{(k+1)}$ admits a sufficiently regular integral kernel, this partial trace acts by integrating out the final spatial coordinate
    \begin{equation}
        \big( \Tr_{k+1} \op{A}^{(k+1)} \big)(\ux_k; \ux_k') = \int_{\R^3} \op{A}^{(k+1)} (\ux_k, z; \ux_k', z) \, dz.
    \end{equation}

    \item \textbf{Density matrices:} The space of physical $k$-particle states is formalized as the convex set of normalized density matrices
    \begin{equation}
        \label{eq:definition-of-Dk}
        \mathcal{D}_k := \Big\{ \gamma \in \mathfrak{S}_1(\Hkb) \;\Big|\; \gamma \geq 0, \quad \Tr_{1,\ldots,k} [\gamma] = 1 \Big\}.
    \end{equation}
\end{enumerate}
\end{definition}

\begin{lemma}[Regular compact test family]
\label{lem:regular-compact-test-family}
Suppose that \hyporef{hyp:standing-hypotheses} holds. For every $k\geq1$, set the free $k$-particle Hamiltonian as
\begin{equation}
    \dch_{k} := \sum_{j=1}^{k}\dch_0^{(j)}
\end{equation}
on $\Hkb$. There exists a countable family $\left\{ \OJ_m^{(k)} \right\}_{m=1}^{\infty} \subset \mathcal{K}^{(k)}$ which is dense in the closed unit ball of $\mathcal{K}^{(k)}$, satisfies $\norm{\OJ_m^{(k)}}_{\mathrm{op}} \leq 1$, and possesses the following properties
\begin{enumerate}[label=(\roman*),itemsep=0.5em]
\item $\OJ_m^{(k)}$ is finite rank for every $m\geq1$.
\item $\operatorname{Ran}(\OJ_m^{(k)}) \mcup \operatorname{Ran} \prc{\OJ_m^{(k)}}^* \subset \mathcal{D}\left(\dch_k\right)$. 
\item
The commutator
$[\OJ_m^{(k)},\dch_k]$, initially defined on
$\cD(\dch_k)$, extends to a bounded finite-rank
operator on $\Hkb$.

\item
With the zero-extension convention above,
\begin{equation}
    [\OJ_m^{(k)},v(x_i-x_j)]
    \in\mathcal K(\h^{\otimes k}),
    \qquad 1\leq i<j\leq k.
\end{equation}

\item
For every $M>0$,
\begin{equation}
    [v(x_i-x_{k+1}),
       \OJ_m^{(k)}\otimes P_M]
    \in\mathcal K(\h^{\otimes(k+1)}),
    \qquad 1\leq i\leq k.
\end{equation}
\end{enumerate}
The commutators in \textup{(iv)} and \textup{(v)}
are finite rank. Their orthogonal compressions
to the corresponding bosonic subspaces are
therefore also compact.
\end{lemma}

\begin{proof}
Since $\dch_0$ has compact resolvent, there exists an orthonormal basis $\left\{ e_n \right\}_{n\geq1} \subset \h$ consisting of eigenvectors of $\dch_0$. The symmetrized tensor products of these $e_n$ form an orthonormal eigenbasis $\left\{ E_\alpha^{(k)} \right\}_{\alpha\geq1} \subset \Hkb$ for $\dch_k$. That is, for every $\alpha\geq1$, there exists $\lambda_\alpha^{(k)}\in\R$ such that $\dch_k E_\alpha^{(k)} = \lambda_\alpha^{(k)}E_\alpha^{(k)}$. For $\alpha,\beta\geq1$, define the rank-one operators $R_{\alpha,\beta}^{(k)} := \ket{E_\alpha^{(k)}} \bra{E_\beta^{(k)}}$. Let $\mathcal{R}_k$ be the set of finite linear combinations of the $R_{\alpha,\beta}^{(k)}$ with coefficients in $\mathbb{Q}+i\mathbb{Q}$. Because $\mathcal{R}_k$ is countable, the normalized set
\begin{equation}
    \left\{ \frac{R}{\max\{1,\norm{R}_{\mathrm{op}}\}} : R\in\mathcal{R}_k \right\}
\end{equation}
is countably dense in the closed unit ball of $\mathcal{K}^{(k)}$. Enumerate this family as $\{\OJ_m^{(k)}\}_{m\geq1}$. By definition, every $\OJ_m^{(k)}$ is finite rank, which proves (i). Moreover, both its range and the range of its adjoint lie entirely within a finite-dimensional span of eigenvectors of $\dch_k$, which proves (ii).  Evaluating the free energy commutator gives
\begin{equation}
    \left[ R_{\alpha,\beta}^{(k)}, \dch_k \right] = \left( \lambda_\beta^{(k)} - \lambda_\alpha^{(k)} \right) R_{\alpha,\beta}^{(k)},
\end{equation}
which means the commutator is simply a scalar multiple of $R_{\alpha,\beta}^{(k)}$. Consequently, $[\OJ_m^{(k)},\dch_k]$, initially
defined on $\cD(\dch_k)$, extends to a bounded
finite-rank operator. This proves (iii).

The zero extension of $\OJ_m^{(k)}$ to
$\h^{\otimes k}$ remains finite rank and has
the same operator norm.
Multiplication by $v(x_i-x_j)$ is bounded on
this full tensor-product space, with norm at
most $\|v\|_{L^\infty}$.
Both products in the commutator in (iv) are
therefore finite rank, proving (iv).

Since $\dch_0$ has compact resolvent, $P_M$
has finite rank for every $M>0$.
Thus $\OJ_m^{(k)}\otimes P_M$ has finite rank
on $\h^{\otimes(k+1)}$.
Multiplication by $v(x_i-x_{k+1})$ is bounded
on that space, so the commutator in (v) is
finite rank. Finally, orthogonal compression preserves finite
rank and hence compactness.
\end{proof}

\begin{definition}[Metric on reduced density matrices]
\label{def:metric-on-reduced-density-matrices}
For each $k \geq 1$, let $\{\OJ_m^{(k)}\}_{m \geq 1}$ be the family of regular compact test operators fixed via \lemref{lem:regular-compact-test-family}. For any pair of density matrices $\gamma, \widetilde{\gamma} \in \mathcal{D}_k$, we define the metric $d_k(\gamma, \widetilde{\gamma})$ by the absolutely convergent series
\begin{equation}
    d_k(\gamma, \widetilde{\gamma}) := \sum_{m=1}^{\infty} 2^{-m} \left| \operatorname{Tr}_{1,\ldots,k} \left( \OJ_m^{(k)} \left( \gamma - \widetilde{\gamma} \right) \right) \right|. 
\end{equation}
Absolute convergence is guaranteed for all $\gamma, \widetilde{\gamma} \in \mathcal{D}_k$ by the uniform operator bound
\begin{equation}
    \left| \operatorname{Tr} \left( \OJ_m^{(k)} \left( \gamma - \widetilde{\gamma} \right) \right) \right| \leq 2.
\end{equation}
\end{definition}

\begin{lemma}[Metric for the weak-$*$ topology]
\label{lem:metric-weak-star-topology}
Let $\{\gamma_n\}_{n\geq1} \subset \mathcal{D}_k$ and $\gamma \in \mathcal{D}_k$. Then 
\begin{equation}
    d_k(\gamma_n, \gamma) \longrightarrow 0 \quad \text{as } n \to \infty
\end{equation}
if and only if 
\begin{equation}
    \operatorname{Tr}_{1,\ldots,k}\left( \OJ\gamma_n \right) \longrightarrow \operatorname{Tr}_{1,\ldots,k}\left( \OJ\gamma \right) \quad \text{as } n \to \infty
\end{equation}
for every $\OJ \in \mathcal{K}^{(k)}$. Hence, the metric $d_k$ metrizes the weak-$*$ topology of the trace-class space $\mathfrak{S}_1^{(k)} = (\mathcal{K}^{(k)})^*$ when restricted to the subset of density matrices $\mathcal{D}_k$.
\end{lemma}

\begin{proof}
$(\Rightarrow)$ Suppose that $d_k(\gamma_n, \gamma) \to 0$ as $n \to \infty$. By the definition of the metric, this implies coordinate-wise convergence against every element of the dense family of test operators; that is, \begin{equation}
    \operatorname{Tr}_{1,\ldots,k}\left(\OJ_m^{(k)}(\gamma_n - \gamma)\right) \to 0
\end{equation} 
for each fixed $m \geq 1$. Let $\OJ \in \mathcal{K}^{(k)}$ be an arbitrary compact operator, and set $R := \max\{1, \|\OJ\|_{\mathrm{op}}\}$. Since the family $\{\OJ_m^{(k)}\}_{m \geq 1}$ is dense in the unit ball of $\mathcal{K}^{(k)}$, for any given $\varepsilon > 0$, we can find an index $m$ such that 
\begin{equation}
    \| R^{-1}\OJ - \OJ_m^{(k)} \|_{\mathrm{op}} < \varepsilon.
\end{equation}
Using the linearity of the trace and H\"{o}lder's inequality for operators ($|\operatorname{Tr}(AB)| \leq \|A\|_{\mathrm{op}}\|B\|_{\mathfrak{S}_1}$), alongside the triangle inequality, we find
\begin{align*}
    \left| \operatorname{Tr}_{1,\ldots,k} \left( \OJ(\gamma_n - \gamma) \right) \right| 
    &= R \left| \operatorname{Tr}_{1,\ldots,k} \left( R^{-1}\OJ(\gamma_n - \gamma) \right) \right| \\*
    &\leq R \left| \operatorname{Tr}_{1,\ldots,k} \left( \OJ_m^{(k)}(\gamma_n - \gamma) \right) \right|\\*
    &\quad + R \left| \operatorname{Tr}_{1,\ldots,k} \left( (R^{-1}\OJ - \OJ_m^{(k)})(\gamma_n - \gamma) \right) \right| \\*
    &\leq R \left| \operatorname{Tr}_{1,\ldots,k} \left( \OJ_m^{(k)}(\gamma_n - \gamma) \right) \right|\\*
    &\quad + R \| R^{-1}\OJ - \OJ_m^{(k)} \|_{\mathrm{op}} \|\gamma_n - \gamma\|_{\mathfrak{S}_1}.
\end{align*}
Since $\gamma_n, \gamma \in \mathcal{D}_k$ are positive density operators, their trace norms are bounded by $1$, yielding $\|\gamma_n - \gamma\|_{\mathfrak{S}_1} \leq 2$. Substituting this bound gives
\begin{equation}
    \left| \operatorname{Tr}_{1,\ldots,k} \left( \OJ(\gamma_n - \gamma) \right) \right| \leq R \left| \operatorname{Tr}_{1,\ldots,k} \left( \OJ_m^{(k)}(\gamma_n - \gamma) \right) \right| + 2R\varepsilon.
\end{equation}
Taking the limit superior as $n \to \infty$, the first term on the right-hand side vanishes, leaving
\begin{equation}
     \limsup_{n\to\infty} \left| \operatorname{Tr} \left( \OJ(\gamma_n - \gamma) \right) \right| \leq 2R\varepsilon
\end{equation}
Since $\varepsilon > 0$ can be chosen arbitrarily small, the limit superior must be zero, proving the weak-$*$ convergence against $\OJ$.

$(\Leftarrow)$ Conversely, assume $\operatorname{Tr}_{1,\ldots,k}(\OJ\gamma_n) \to \operatorname{Tr}_{1,\ldots,k}(\OJ\gamma)$ for every $\OJ \in \mathcal{K}^{(k)}$. It follows directly that each term in the definition of the metric converges to zero
\begin{equation}
    \lim_{n\to\infty} \left| \operatorname{Tr}_{1,\ldots,k} \left( \OJ_m^{(k)}(\gamma_n - \gamma) \right) \right| = 0
\end{equation}
for each $m \geq 1$. Furthermore, each summand in the metric series is uniformly bounded by $2^{-m} \cdot 2 = 2^{1-m}$, which forms a convergent geometric series $\sum_{m=1}^{\infty} 2^{1-m} = 2$. Applying the Dominated Convergence Theorem for series allows us to exchange the limit and the infinite summation, yielding
\begin{equation}
     \lim_{n\to\infty} d_k(\gamma_n, \gamma) = 0.
\end{equation}
\end{proof}

\begin{definition}[Hartree hierarchy path space and product metric]
\label{def:infinite-hierarchy-path-space}
Let $T > 0$, and let $C([0,T], \mathcal{D}_k; d_k)$ denote the space of $d_k$-continuous curves mapping the time interval $[0,T]$ into the metric space of $k$-particle density matrices $(\mathcal{D}_k, d_k)$. This space is equipped with the uniform metric
\begin{equation}
\label{eq:path-metric-dk-infty}
    d_{k,T}\left( \gamma^{(k)}(\cdot), \widetilde\gamma^{(k)}(\cdot) \right) := \sup_{t\in[0,T]} d_k\left( \gamma_t^{(k)}, \widetilde\gamma_t^{(k)} \right).
\end{equation}
The \emph{infinite-hierarchy path space} $\mathfrak{H}_T$ is the metric space defined by the Cartesian product
\begin{equation}
    \mathfrak{H}_T := \prod_{k=1}^{\infty} C([0,T], \mathcal{D}_k; d_k),
\end{equation}
and equipped with the product metric
\begin{equation}
\label{eq:hierarchy-product-metric}
    d_T(\Gamma, \widetilde\Gamma) := \sum_{k=1}^{\infty} 2^{-k} \frac{d_{k,T}\left( \gamma^{(k)}(\cdot), \widetilde\gamma^{(k)}(\cdot) \right)}{1 + d_{k,T}\left( \gamma^{(k)}(\cdot), \widetilde\gamma^{(k)}(\cdot) \right)},
\end{equation}
where $\Gamma = \left\{\gamma^{(k)}(\cdot)\right\}_{k=1}^\infty$ and $\widetilde\Gamma = \left\{\widetilde\gamma^{(k)}(\cdot)\right\}_{k=1}^\infty$ are elements of $\mathfrak{H}_T$. The metric $d_T$ naturally induces the product topology on $\mathfrak{H}_T$.
\end{definition} 

\begin{lemma}[Equicontinuity against compact observables]
\label{lem:equicontinuity-compact-observables}
Fix $k\geq1$ and $m\geq1$, and define
\begin{equation}
    f_{N,m}^{(k)}(t) := \Tr_{1,\ldots,k} \prc{ \OJ_m^{(k)} \gamma_{N,t}^{(k)} }, \qquad N\geq k.
\end{equation}
Then the family $\{f_{N,m}^{(k)}\}_{N\geq k}$ is uniformly bounded and uniformly Lipschitz on the closed interval $[0,T]$.
\end{lemma}

\begin{proof}
By \lemref{lem:reduced-marginals-are-density-matrices}, we have $\norm{\gamma_{N,t}^{(k)}}_{\mathfrak{S}_1^{(k)}}=1$. Since $\norm{\OJ_m^{(k)}}_{\mathrm{op}}\leq1$, it immediately follows that $\left| f_{N,m}^{(k)}(t) \right| \leq 1$. Let $0\leq s<t\leq T$. Testing the BBGKY hierarchy against the compact operator $\OJ_m^{(k)}$ gives
\begin{equation}
\begin{aligned}
    &f_{N,m}^{(k)}(t) - f_{N,m}^{(k)}(s) 
    = -i \int_s^t \Tr_{1,\ldots,k} \prc{ \left[ \OJ_m^{(k)}, \dch_k \right] \gamma_{N,\tau}^{(k)} } d\tau \\
    &\quad - i\frac{\lam}{N} \sum_{1\leq i<j\leq k} \int_s^t \Tr_{1,\ldots,k} \prc{ \left[ \OJ_m^{(k)}, v(x_i-x_j) \right] \gamma_{N,\tau}^{(k)} } d\tau \\
    &\quad - i\lam\frac{N-k}{N} \sum_{i=1}^{k} \int_s^t \Tr_{1,\ldots,k+1} \prc{ \left[ \OJ_m^{(k)}\otimes\1, v(x_i-x_{k+1}) \right] \gamma_{N,\tau}^{(k+1)} } d\tau \\
    &\quad - i \eta \int_s^t \Tr_{1,\ldots,k} \prc{ \OJ_m^{(k)} \mathcal{S}_{N,\tau}^{(k)} } d\tau.
\end{aligned}
\end{equation}
When $N=k$, the third term is absent. The identity for the third term follows directly from the definition of the partial trace and the cyclicity of the trace
\begin{equation}
\begin{aligned}
    &\Tr_{\h^{\otimes k}}
    \left(
       \OJ_m^{(k)}
       \Tr_{k+1}
       [v(x_i-x_{k+1}),\gamma_{N,\tau}^{(k+1)}]
    \right)
    \\
    &\qquad=
    \Tr_{\h^{\otimes(k+1)}}
    \left(
       [\OJ_m^{(k)}\otimes\1,
          v(x_i-x_{k+1})]
       \gamma_{N,\tau}^{(k+1)}
    \right).
\end{aligned}
\end{equation}
Define the finite operator-norm constants
\begin{equation}
    A_{k,m}^{(0)} := \norm{ [\OJ_m^{(k)},\dch_k] }_{\mathrm{op}}, \quad \text{and} \quad A_{k,m}^{(v)} := \sum_{1\leq i<j\leq k} \norm{ [\OJ_m^{(k)},v(x_i-x_j)] }_{\mathrm{op}}.
\end{equation}
By \lemref{lem:regular-compact-test-family}, these constants are bounded. Moreover, we have the bound $\norm{ [\OJ_m^{(k)}\otimes\1,v(x_i-x_{k+1})] }_{\mathrm{op}} \leq 2\norm{v}_{L^\infty}$. Utilizing the trace-one normalization of the marginals and the uniform source bounds of \lemref{lem:uniform-field-source-bound}, we obtain the Lipschitz bound
\begin{equation}
    \left| f_{N,m}^{(k)}(t) - f_{N,m}^{(k)}(s) \right| \leq L_{k,m,T}|t-s|,
\end{equation}
where the Lipschitz constant is independent of $N$,
\begin{equation}
     L_{k,m,T} := A_{k,m}^{(0)} + \lam A_{k,m}^{(v)} + 2 \lam k\norm{v}_{L^\infty} + \eta\, C_{k,\OJ_m^{(k)},T}. 
\end{equation}
\end{proof}

\begin{theorem}[Compactness of the BBGKY hierarchy]
\label{thm:compactness-bbgky-sequence}
Assume \hyporef{hyp:standing-hypotheses}.
Use the source bound of
\lemref{lem:uniform-field-source-bound}
and the independent first-energy estimate of
\corref{cor:one-particle-energy-bound}. Let the marginals be defined from normalized bosonic solutions as in \lemref{lem:reduced-marginals-are-density-matrices}. Then, there exist a subsequence $N_\ell\to\infty$ such that for every $k\geq1$ the sequence $\gamma_{N_\ell,\cdot}^{(k)}$ converges weakly-$*$ to
\begin{equation}
    \gamma_{\infty,\cdot}^{(k)} \in C([0,T],\mathcal{D}_k;d_k)
\end{equation}
that is
\begin{equation}
\label{eq:uniform-weak-star-compactness-conclusion}
    \lim_{\ell\to\infty} \sup_{t\in[0,T]} \left| \Tr_{\Hkb} \prc{ \OJ^{(k)} \prc{ \gamma_{N_\ell,t}^{(k)} - \gamma_{\infty,t}^{(k)} } } \right| = 0
\end{equation}
for every $\OJ^{(k)} \in\mathcal{K}(\Hkb)$. For fixed $k$, the sequence $\gamma_{N_\ell,t}^{(k)}$ is evaluated once $N_\ell\geq k$.
\end{theorem}

\begin{proof}
Let $\mathbb{D} := \mathbb{Q} \mcap [0,T]$. For every fixed $k\geq1$ and every $t\in\mathbb{D}$, the family of operators $\{ \gamma_{N,t}^{(k)} : N\geq k \}$ is contained entirely within the closed unit ball of $\mathfrak{S}_1(\Hkb) = \mathcal{K}(\Hkb)^*$, owing to the normalization $\norm{ \gamma_{N,t}^{(k)} }_{\mathfrak{S}_1^{(k)}} = \Tr_{\Hkb}( \gamma_{N,t}^{(k)} ) = 1$. By the Banach--Alaoglu theorem, this closed unit ball is compact in the weak-$*$ topology $\sigma(\mathfrak{S}_1(\Hkb), \mathcal{K}(\Hkb))$. Furthermore, since $\Hkb$ is separable, its compact operator predual $\mathcal{K}(\Hkb)$ is also separable, since the weak-$*$ topology is metrizable on the closed unit ball. Thus, this ball is weak-$*$ sequentially compact. Because the set $\{ (k,t): k\geq1, \, t\in\mathbb{D} \}$ is countable, a standard diagonal extraction argument yields a subsequence $N_\ell\to\infty$ and, for every $k\geq1$ and $t\in\mathbb{D}$, an operator $\widetilde{\gamma}_{\infty,t}^{(k)} \in \mathfrak{S}_1(\Hkb)$ such that
\begin{equation}
\label{eq:weak-star-convergence-rational-times}
    \lim_{\ell\to\infty} \Tr_{\Hkb} \prc{ \OJ^{(k)} \gamma_{N_\ell,t}^{(k)} } = \Tr_{\Hkb} \prc{ \OJ^{(k)} \widetilde{\gamma}_{\infty,t}^{(k)} }
\end{equation}
for every test observable $\OJ^{(k)} \in\mathcal{K}(\Hkb)$. For every fixed pair $(k,m)$, we define the sequence evaluated against the dense family as
\begin{equation}
     h_m^{(k)}(t) := \Tr_{\Hkb} \prc{ \OJ_m^{(k)} \widetilde{\gamma}_{\infty,t}^{(k)} }\quad \text{ for} \quad  t\in\mathbb{D}.
\end{equation}
By \eqref{eq:weak-star-convergence-rational-times}, we have the pointwise convergence
\begin{equation}
\label{eq:pointwise-convergence-rational-times}
    f_{N_\ell,m}^{(k)}(t) \longrightarrow h_m^{(k)}(t) \quad \text{for every } t\in\mathbb{D}.
\end{equation}
By \lemref{lem:equicontinuity-compact-observables}, the sequence of functions $\{ f_{N,m}^{(k)} : N\geq k \}$ is uniformly equicontinuous on $[0,T]$. Thus, for every $\varepsilon>0$, there exists a $\delta>0$, independent of $N$, such that $| f_{N,m}^{(k)}(t) - f_{N,m}^{(k)}(s) | < \varepsilon$ whenever $|t-s|<\delta$. Restricting to $s,t\in\mathbb{D}$ and passing to the limit $\ell\to\infty$ yields $| h_m^{(k)}(t) - h_m^{(k)}(s) | \leq \varepsilon$ whenever $|t-s|<\delta$. This guarantees that $h_m^{(k)}$ is uniformly continuous on the dense subset $\mathbb{D} \subset [0,T]$. It therefore admits a unique continuous extension $g_m^{(k)} \in C([0,T])$. We now prove that the pointwise convergence $f_{N_\ell,m}^{(k)} \to g_m^{(k)}$ is, in fact, uniform on $[0,T]$. Fix $\varepsilon>0$. By the uniform equicontinuity of $\{f_{N,m}^{(k)}\}_{N\geq k}$ and the uniform continuity of $g_m^{(k)}$, we can select $\delta>0$ such that $| f_{N,m}^{(k)}(t) - f_{N,m}^{(k)}(s) | < \frac{\varepsilon}{3}$ and $| g_m^{(k)}(t) - g_m^{(k)}(s) | < \frac{\varepsilon}{3}$ whenever $|t-s|<\delta$. Because $[0,T]$ is compact and $\mathbb{D}$ is dense, we can select a finite net $t_1,\ldots,t_R \in \mathbb{D}$ such that for every $t\in[0,T]$, there exists an index $r$ satisfying $|t-t_r|<\delta$. Due to the pointwise convergence on $\mathbb{D}$, there exists an $\ell_0\geq1$ such that $\max_{1\leq r\leq R} | f_{N_\ell,m}^{(k)}(t_r) - g_m^{(k)}(t_r) | < \frac{\varepsilon}{3}$ for all $\ell\geq\ell_0$. Consequently, for any $t\in[0,T]$, we select its nearest net point $t_r$ to find
\begin{align*}
    \left| f_{N_\ell,m}^{(k)}(t) - g_m^{(k)}(t) \right| 
    &\leq \left| f_{N_\ell,m}^{(k)}(t) - f_{N_\ell,m}^{(k)}(t_r) \right| + \left| f_{N_\ell,m}^{(k)}(t_r) - g_m^{(k)}(t_r) \right|\\
    &\quad + \left| g_m^{(k)}(t_r) - g_m^{(k)}(t) \right| < \varepsilon.
\end{align*}
Therefore, $\lim_{\ell\to\infty} \sup_{t\in[0,T]} \left| f_{N_\ell,m}^{(k)}(t) - g_m^{(k)}(t) \right| = 0$.

Let
\begin{equation}
    \mathscr{F}_k := \operatorname{span}_{\C} \{ \OJ_m^{(k)} : m\geq1 \},
\end{equation}
which forms a norm-dense complex linear subspace of $\mathcal{K}(\Hkb)$. For fixed $t\in[0,T]$, define the functional
\begin{equation}
    \ell_{k,t}(\OJ) := \lim_{\ell\to\infty} \Tr_{\Hkb} \prc{ \OJ\gamma_{N_\ell,t}^{(k)} }, \qquad \OJ \in\mathscr{F}_k.
\end{equation}
This limit is well-defined because $\OJ$ is a finite linear combination of the $\OJ_m^{(k)}$. Moreover, $|\ell_{k,t}(\OJ)| \leq \norm{\OJ}_{\mathrm{op}}$. Thus, $\ell_{k,t}$ extends uniquely to a bounded linear functional on $\mathcal{K}(\Hkb)$ with norm at most one. By trace-class duality, there exists a unique operator $\gamma_{\infty,t}^{(k)} \in \mathfrak{S}_1(\Hkb)$ satisfying
\begin{equation}
\label{eq:definition-of-compactness-limit}
    \ell_{k,t}(\OJ) = \Tr_{\Hkb} \prc{ \OJ\gamma_{\infty,t}^{(k)} }
\end{equation}
for every compact $\OJ$, with $\norm{\gamma_{\infty,t}^{(k)}}_{\mathfrak{S}_1^{(k)}} \leq 1$. For every non-negative finite-rank test operator $K^{(k)}$, we have $\Tr_{\Hkb} ( K^{(k)}\gamma_{\infty,t}^{(k)} ) = \lim_{\ell\to\infty} \Tr_{\Hkb} ( K^{(k)}\gamma_{N_\ell,t}^{(k)} ) \geq 0$, which confirms that $\gamma_{\infty,t}^{(k)} \geq 0$. 

Next, we establish that no trace is lost in the limit.
Fix $k\geq1$ and $t\in[0,T]$. For $M>0$, set
\begin{equation}
    P_M:=\mathbf1_{\{\OT\leq M\}},
    \qquad
    Q_M:=I-P_M.
\end{equation}
Compact resolvent of $\dch_0$ implies that $P_M$
has finite rank. Moreover, $P_M^{\otimes k}$
preserves the bosonic subspace and is a finite-rank
projection on $\Hkb$.

On the full tensor product, the commuting coordinate
projections satisfy
\begin{equation}
    I-P_M^{\otimes k}
    \leq
    \sum_{j=1}^k
       I^{\otimes(j-1)}\otimes Q_M
       \otimes I^{\otimes(k-j)}.
\end{equation}
Both sides preserve the bosonic subspace, so the
inequality also holds there. Furthermore,
\begin{equation}
    Q_M\leq M^{-2}\OT^2
\end{equation}
as nonnegative quadratic forms.

Consequently, positivity, symmetry, and consistency
of the finite-$N$ marginals give
\begin{equation}
\begin{aligned}
    1-\Tr_{\Hkb}\!\left(
       P_M^{\otimes k}\gamma_{N,t}^{(k)}
    \right)
    &\leq
    k\operatorname{Tr}\!\left(
       Q_M\gamma_{N,t}^{(1)}
    \right)\\
    &\leq
    \frac{k}{M^2}
    \operatorname{Tr}\!\left(
       \OT\gamma_{N,t}^{(1)}\OT
    \right)\\
    &\leq
    \frac{2k(C_0+c_*-1)}{M^2},
\end{aligned}
\end{equation}
for every sufficiently large $N\geq k$, where the
last inequality follows from
\eqref{eq:one-particle-energy-bound}.
The inequality involving the unbounded operator
$\OT^2$ is justified by bounded spectral truncation
and monotone convergence.

Since $N_\ell\to\infty$, this estimate applies
along the extracted subsequence for all sufficiently
large $\ell$. Testing the weak-$*$ convergence
against the compact operator $P_M^{\otimes k}$
therefore yields
\begin{equation}
    \Tr_{\Hkb}\!\left(
       P_M^{\otimes k}\gamma_{\infty,t}^{(k)}
    \right)
    \geq
    1-\frac{2k(C_0+c_*-1)}{M^2}.
\end{equation}
As $M\to\infty$, the projections
$P_M^{\otimes k}$ increase strongly to the identity.
Because $\gamma_{\infty,t}^{(k)}$ is positive
and trace class, monotone convergence gives
\begin{equation}
    \Tr_{\Hkb}\!\left(
       \gamma_{\infty,t}^{(k)}
    \right)\geq1.
\end{equation}
On the other hand, positivity and the previously
established trace-norm bound imply
\begin{equation}
    \Tr_{\Hkb}\!\left(
       \gamma_{\infty,t}^{(k)}
    \right)
    =
    \bigl\|\gamma_{\infty,t}^{(k)}\bigr\|_{\mathfrak S_1}
    \leq1.
\end{equation}
Hence
\begin{equation}
    \Tr_{\Hkb}\!\left(
       \gamma_{\infty,t}^{(k)}
    \right)=1.
\end{equation}
Since $k$ and $t$ were arbitrary,
$\gamma_{\infty,t}^{(k)}\in\mathcal D_k$
for every $k\geq1$ and $t\in[0,T]$.

For any $\OJ\in\mathscr{F}_k$, the scalar function $t \mapsto \Tr_{\Hkb} ( \OJ\gamma_{\infty,t}^{(k)} )$ is continuous. For an arbitrary compact $\OJ\in\mathcal{K}(\Hkb)$, we can approximate it in operator norm by elements of $\mathscr{F}_k$. Since
\begin{equation}
    \sup_{t\in[0,T]} \norm{ \gamma_{\infty,t}^{(k)} }_{\mathfrak{S}_1^{(k)}} = 1,
\end{equation} the corresponding sequence of scalar functions converges uniformly. Therefore, the map $t\mapsto\gamma_{\infty,t}^{(k)}$ is weak-$*$ continuous and subsequently $d_k$-continuous by \lemref{lem:metric-weak-star-topology}.

Finally, let $\OJ\in\mathcal{K}(\Hkb)$ and choose $\OJ_r\in\mathscr{F}_k$ such that 
\begin{equation}
    \norm{\OJ-\OJ_r}_{\mathrm{op}}\to0.
\end{equation}
Then, by adding and subtracting the approximation,
\begin{equation}
\begin{aligned}
    \sup_{t\in[0,T]} \left| \Tr_{\Hkb} \prc{ \OJ \prc{ \gamma_{N_\ell,t}^{(k)} - \gamma_{\infty,t}^{(k)} } } \right| 
    &\leq 2\norm{\OJ-\OJ_r}_{\mathrm{op}}\\
    &\quad + \sup_{t\in[0,T]} \left| \Tr_{\Hkb} \prc{ \OJ_r \prc{ \gamma_{N_\ell,t}^{(k)} - \gamma_{\infty,t}^{(k)} } } \right|.
\end{aligned}
\end{equation}
For any fixed $r$, the final supremum term tends to zero as $\ell\to\infty$. Sending $r\to\infty$ drives the approximation error to zero, which proves \eqref{eq:uniform-weak-star-compactness-conclusion}.
\end{proof}

\subsection{Normalization and admissibility}

\begin{lemma}[Upgrade from weak-$*$ to strong $\fkS_1$-convergence]
\label{lem:weak-str-conv-and-traceclass-imply-trace-norm-conv}
Let $k\geq 1$, and let $\gamma_n^{(k)},\gamma^{(k)} \in \fkS_1(\Hkb)$ be a sequence of positive trace-class operators satisfying the normalization condition
\begin{equation}
\label{eq:trace-one-assumption-for-trace-norm-upgrade}
    \Tr \left( \gamma_n^{(k)} \right) = \Tr \left( \gamma^{(k)} \right) = 1
\end{equation}
for all $n\geq 1$. Assume that $\gamma_n^{(k)}$ converges weak-$*$ to $\gamma^{(k)}$ in $\fkS_1(\Hkb)$, that is,
\begin{equation}
\label{eq:weak-star-tested-against-compacts}
    \Tr \left[ \OJ^{(k)} \gamma_n^{(k)} \right] \longrightarrow \Tr \left[ \OJ^{(k)} \gamma^{(k)} \right]
\end{equation}
for every compact operator $\OJ^{(k)}$. Then, $\gamma_n^{(k)}$ converges to $\gamma^{(k)}$ strongly in the trace-norm topology
\begin{equation}
\label{eq:trace-norm-convergence-from-weak-star-and-trace}
    \norm{ \gamma_n^{(k)} - \gamma^{(k)} }_{\fkS_1(\Hkb)} \longrightarrow 0 \quad \text{as } n\to\infty.
\end{equation}
Consequently, convergence holds against all bounded operators; that is, for every $B^{(k)}\in\cB(\Hkb)$,
\begin{equation}
\label{eq:convergence-against-all-bounded-operators}
    \Tr \left[ \OB^{(k)}\gamma_n^{(k)} \right] \longrightarrow \Tr \left[ \OB^{(k)}\gamma^{(k)} \right] \quad \text{as } n\to\infty.
\end{equation}
\end{lemma}

\begin{proof}
Fix $\varepsilon>0$. Since $\gamma^{(k)}$ is a positive trace-class operator, the spectral theorem yields the trace-norm convergent representation
\begin{equation}
\label{eq:spectral-decomposition-of-limit-density}
    \gamma^{(k)} = \sum_{j\geq 1} \lambda_j |e_j\rangle\langle e_j|,
\end{equation}
where $\{e_j\}_{j\geq 1}\subset\Hkb$ is an orthonormal family and $\lambda_j \geq 0$ satisfy $\sum_{j\geq 1}\lambda_j = 1$. Choose $M\geq 1$ large enough so that the spectral tail satisfies $\sum_{j>M}\lambda_j < \varepsilon$. Define the finite rank orthogonal projection $P := \sum_{j=1}^{M} |e_j\rangle\langle e_j|$ onto the span of the first $M$ eigenvectors, and set $Q := \1-P$. By construction, we have
\begin{equation}
\label{eq:small-tail-of-limit-density-rigorous}
    \Tr \left( Q\gamma^{(k)}Q \right) = \sum_{j>M}\lambda_j < \varepsilon.
\end{equation}
Since $\dim(P\Hkb) = M < \infty$, the projection $P$ is compact. The weak-$*$ convergence assumption \eqref{eq:weak-star-tested-against-compacts} evaluated at $\OJ^{(k)} = P$ yields $\Tr ( P\gamma_n^{(k)} ) \to \Tr ( P\gamma^{(k)} )$. By the cyclicity of the trace and the orthogonality of the decomposition $\1 = P + Q$, we have
\begin{equation}
    \Tr(A) = \Tr(PAP) + \Tr(QAQ)
\end{equation}
for any trace-class operator $A$. Together with the normalization \eqref{eq:trace-one-assumption-for-trace-norm-upgrade}, this identity implies
\begin{equation}
\label{eq:convergence-of-Q-masses}
    \Tr \left( Q\gamma_n^{(k)}Q \right) = 1 - \Tr \left( P\gamma_n^{(k)}P \right) \longrightarrow 1 - \Tr \left( P\gamma^{(k)}P \right) = \Tr \left( Q\gamma^{(k)}Q \right).
\end{equation}
Hence, there exists an index $n_\varepsilon\geq 1$ such that 
\begin{equation}
\label{eq:small-tail-of-approximating-density-rigorous}
    \Tr \left( Q\gamma_n^{(k)}Q \right) < 2\varepsilon
\end{equation}
for all $n\geq n_\varepsilon$. We now utilize the algebraic block decomposition
\begin{equation}
\label{eq:four-block-decomposition}
\begin{aligned}
    \gamma_n^{(k)}-\gamma^{(k)} &= P \left( \gamma_n^{(k)}-\gamma^{(k)} \right) P + P \left( \gamma_n^{(k)}-\gamma^{(k)} \right) Q \\
    &\quad + Q \left( \gamma_n^{(k)}-\gamma^{(k)} \right) P + Q \left( \gamma_n^{(k)}-\gamma^{(k)} \right) Q.
\end{aligned}
\end{equation}
We estimate each block individually in the trace norm,
\begin{enumerate}
    \item \textbf{The $P$-$P$ block:} Let $E_{ji} := |e_j\rangle\langle e_i|$ for $1\leq i,j\leq M$. Since each $E_{ji}$ is compact, weak-$*$ convergence ensures that the matrix elements satisfy
    \begin{equation}
        \langle e_i, (\gamma_n^{(k)}-\gamma^{(k)})e_j \rangle \to 0 \quad \text{ as } \quad n\to\infty.
    \end{equation}
    Because $P\Hkb$ is finite-dimensional, component-wise convergence implies trace-norm convergence
    \begin{equation}
    \label{eq:PP-block-converges-in-trace-norm}
        \norm{ P \left( \gamma_n^{(k)}-\gamma^{(k)} \right) P }_{\fkS_1(\Hkb)} \longrightarrow 0 \quad \text{ as } \quad n\to\infty.
    \end{equation}
    
    \item \textbf{The off-diagonal blocks:} For any positive operator $A \geq 0$, factoring $PAQ = (PA^{1/2})(A^{1/2}Q)$ and applying the Schatten--Hölder inequality yields the standard estimate
    \begin{equation}
    \label{eq:off-diagonal-positive-estimate}
        \norm{PAQ}_{\fkS_1(\Hkb)} \leq \left[ \Tr(PAP) \right]^{1/2} \left[ \Tr(QAQ) \right]^{1/2}.
    \end{equation}
    Applying \eqref{eq:off-diagonal-positive-estimate} to $\gamma_n^{(k)}$ and $\gamma^{(k)}$ via bounds \eqref{eq:small-tail-of-limit-density-rigorous} and \eqref{eq:small-tail-of-approximating-density-rigorous} shows that for all $n \geq n_\varepsilon$,
    \begin{equation}
        \norm{P\gamma_n^{(k)}Q}_{\fkS_1} < \sqrt{2\varepsilon} \quad \text{and} \quad \norm{P\gamma^{(k)}Q}_{\fkS_1} < \sqrt{\varepsilon}.
    \end{equation}
    Triangle inequality and isometric invariance under adjoints yield
    \begin{equation}
    \label{eq:off-diagonal-combined}
        \norm{ P \left( \gamma_n^{(k)}-\gamma^{(k)} \right) Q }_{\fkS_1} = \norm{ Q \left( \gamma_n^{(k)}-\gamma^{(k)} \right) P }_{\fkS_1} < \sqrt{2\varepsilon} + \sqrt{\varepsilon}.
    \end{equation}
    
    \item \textbf{The $Q$-$Q$ block:} Since $\gamma_n^{(k)}$ and $\gamma^{(k)}$ are positive, the operators $Q\gamma_n^{(k)}Q$ and $Q\gamma^{(k)}Q$ are positive semidefinite. For positive operators, the trace norm coincides with the trace. The triangle inequality along with estimates \eqref{eq:small-tail-of-limit-density-rigorous} and \eqref{eq:small-tail-of-approximating-density-rigorous} yields
    \begin{equation}
    \label{eq:QQ-block-estimate}
        \norm{ Q \left( \gamma_n^{(k)}-\gamma^{(k)} \right) Q }_{\fkS_1} \leq \Tr\left( Q\gamma_n^{(k)}Q \right) + \Tr\left( Q\gamma^{(k)}Q \right) < 3\varepsilon.
    \end{equation}
\end{enumerate}
Applying the triangle inequality to the block decomposition \eqref{eq:four-block-decomposition} and taking the upper limit as $n\to\infty$ gives
\begin{equation}
    \limsup_{n\to\infty} \norm{ \gamma_n^{(k)} - \gamma^{(k)} }_{\fkS_1(\Hkb)} \leq 2\sqrt{2\varepsilon} + 2\sqrt{\varepsilon} + 3\varepsilon.
\end{equation}
Since $\varepsilon > 0$ was arbitrary, sending $\varepsilon \downarrow 0$ establishes \eqref{eq:trace-norm-convergence-from-weak-star-and-trace}. Finally, assertion \eqref{eq:convergence-against-all-bounded-operators} follows directly from the canonical duality pairing between $\fkS_1(\Hkb)$ and $\cB(\Hkb)$, satisfying $|\Tr(BA)| \leq \norm{B}_{\op{op}} \norm{A}_{\fkS_1}$. Setting $A = \gamma_n^{(k)} - \gamma^{(k)}$ delivers the required limit.
\end{proof}

\begin{proposition}[Admissibility of the Hartree hierarchy]
\label{prop:admissibility-limit}
Let $\set{N_\ell}_{\ell\geq1}$ be a sequence of positive integers such that $N_\ell\to\infty$ and let $\{\gamma_{\infty,t}^{(k)}\}_{k\geq1}$ be a limit hierarchy obtained from \thmref{thm:compactness-bbgky-sequence}. Then, for every $k\geq1$ and every $t\in[0,T]$,
\begin{equation}
\label{eq:limit-hierarchy-admissibility}
    \Tr_{k+1} \prc{ \gamma_{\infty,t}^{(k+1)} } = \gamma_{\infty,t}^{(k)}.
\end{equation}
\end{proposition}

\begin{proof}
Fix $k\geq1$ and $t\in[0,T]$. By \thmref{thm:compactness-bbgky-sequence},
\begin{equation}
    \gamma_{N_\ell,t}^{(k)} \overset{*}{\rightharpoonup} \gamma_{\infty,t}^{(k)} \qquad \text{in} \qquad \mathfrak S_1(\Hkb) = \mathcal K(\Hkb)^*,
\end{equation}
and
\begin{equation}
    \gamma_{N_\ell,t}^{(k+1)} \overset{*}{\rightharpoonup} \gamma_{\infty,t}^{(k+1)} \qquad \text{in } \qquad  \mathfrak S_1(\h_{k+1}) = \mathcal K(\h_{k+1})^*.
\end{equation}
Moreover, for every $k\geq1$ and $t\in[0,T]$,
\begin{equation}
    \gamma_{N_\ell,t}^{(k)} \geq 0, \qquad \gamma_{\infty,t}^{(k)} \geq 0,
\end{equation}
with
\begin{equation}
    \Tr_{\Hkb} \prc{ \gamma_{N_\ell,t}^{(k)} } = \Tr_{\Hkb} \prc{ \gamma_{\infty,t}^{(k)} } = 1.
\end{equation}
Therefore, \lemref{lem:weak-str-conv-and-traceclass-imply-trace-norm-conv} gives that for every $k\geq1$ and $t\in[0,T]$,
\begin{equation}
\label{eq:trace-norm-convergence-k-for-admissibility}
    \norm{ \gamma_{N_\ell,t}^{(k)} - \gamma_{\infty,t}^{(k)} }_{\mathfrak S_1(\Hkb)} \longrightarrow 0
\end{equation}
as $\ell\to\infty$. We recall that the partial trace is contractive in the trace norm. More precisely, for every $\op{A}^{(k+1)} \in \mathfrak S_1(\h_{k+1})$, one has
\begin{equation}
\label{eq:partial-trace-contractive-trace-norm}
    \norm{ \Tr_{k+1} \prc{ \op{A}^{(k+1)} } }_{\mathfrak S_1(\Hkb)} \leq \norm{ \op{A}^{(k+1)} }_{\mathfrak S_1(\h_{k+1})}.
\end{equation}
Indeed, by trace-class duality and the defining property of the partial trace,
\begin{equation}
\begin{aligned}
    \norm{ \Tr_{k+1} \prc{ \op{A}^{(k+1)} } }_{\mathfrak S_1(\Hkb)} 
    &= \sup_{\substack{ B^{(k)}\in\mathcal B(\Hkb) \\ \norm{B^{(k)}}_{\mathrm{op}}\leq1 }} \left| \Tr_{\Hkb} \prc{ B^{(k)} \Tr_{k+1} \prc{ \op{A}^{(k+1)} } } \right| \\
    &= \sup_{\substack{ B^{(k)}\in\mathcal B(\Hkb) \\ \norm{B^{(k)}}_{\mathrm{op}}\leq1 }} \left| \Tr_{\fkH_{k+1}} \prc{ \left( B^{(k)}\otimes\1 \right) \op{A}^{(k+1)} } \right| \\
    &\leq \sup_{\substack{ B^{(k)}\in\mathcal B(\Hkb) \\ \norm{B^{(k)}}_{\mathrm{op}}\leq1 }} \norm{ B^{(k)}\otimes\1 }_{\mathrm{op}} \norm{ \op{A}^{(k+1)} }_{\mathfrak S_1(\h_{k+1})} \\
    &\leq \norm{ \op{A}^{(k+1)} }_{\mathfrak S_1(\h_{k+1})}.
\end{aligned}
\end{equation}
Thus,
\begin{equation}
    \norm{ \Tr_{k+1} \prc{ \op{A}^{(k+1)} } }_{\mathfrak S_1(\Hkb)}   \leq \norm{ \op{A}^{(k+1)} }_{\mathfrak S_1(\h_{k+1})}.
\end{equation}
Since $N_\ell\to\infty$, then $N_\ell\geq k+1$ for all sufficiently large $\ell$. Hence, finite-$N_\ell$ consistency gives
\begin{equation}
\label{eq:finite-N-consistency-used-for-limit-admissibility}
    \Tr_{k+1} \prc{ \gamma_{N_\ell,t}^{(k+1)} } = \gamma_{N_\ell,t}^{(k)}.
\end{equation}
Using \eqref{eq:finite-N-consistency-used-for-limit-admissibility}, the triangle inequality, and \eqref{eq:partial-trace-contractive-trace-norm}, we obtain
\begin{equation}
\begin{aligned}
    \norm{ \Tr_{k+1} \prc{ \gamma_{\infty,t}^{(k+1)} } - \gamma_{\infty,t}^{(k)} }_{\mathfrak S_1} 
    &\leq \norm{ \Tr_{k+1} \prc{ \gamma_{\infty,t}^{(k+1)} - \gamma_{N_\ell,t}^{(k+1)} } }_{\mathfrak S_1}\\
    &\quad + \norm{ \Tr_{k+1} \prc{ \gamma_{N_\ell,t}^{(k+1)} } - \gamma_{N_\ell,t}^{(k)} }_{\mathfrak S_1} \\
    &\quad + \norm{ \gamma_{N_\ell,t}^{(k)} - \gamma_{\infty,t}^{(k)} }_{\mathfrak S_1} \\
    &= \norm{ \Tr_{k+1} \prc{ \gamma_{\infty,t}^{(k+1)} - \gamma_{N_\ell,t}^{(k+1)} } }_{\mathfrak S_1(\Hkb)}\\
    &\quad + \norm{ \gamma_{N_\ell,t}^{(k)} - \gamma_{\infty,t}^{(k)} }_{\mathfrak S_1(\Hkb)} \\
    &\leq \norm{ \gamma_{\infty,t}^{(k+1)} - \gamma_{N_\ell,t}^{(k+1)} }_{\mathfrak S_1(\h_{k+1})}\\
    &\quad + \norm{ \gamma_{N_\ell,t}^{(k)} - \gamma_{\infty,t}^{(k)} }_{\mathfrak S_1(\Hkb)}.
\end{aligned}
\end{equation}
By \eqref{eq:trace-norm-convergence-k-for-admissibility}, the right-hand side converges to zero as $\ell\to\infty$. Consequently,
\begin{equation}
    \norm{ \Tr_{k+1} \prc{ \gamma_{\infty,t}^{(k+1)} } - \gamma_{\infty,t}^{(k)} }_{\mathfrak S_1(\Hkb)} = 0.
\end{equation}
Therefore,
\begin{equation}
    \Tr_{k+1} \prc{ \gamma_{\infty,t}^{(k+1)} } = \gamma_{\infty,t}^{(k)}.
\end{equation}
Since $k\geq1$ and $t\in[0,T]$ are arbitrary, the limiting Hartree hierarchy is admissible.
\end{proof}

\begin{corollary}[Limit to Hartree term]
\label{cor:bounded-BBGKY-interaction-limit}
Suppose that the assumptions of \hyporef{hyp:standing-hypotheses} hold. Let $N_\ell\to\infty$ be the subsequence extracted in \thmref{thm:compactness-bbgky-sequence}, and let $\Gamma_{\infty,t}=\set{\gamma_{\infty,t}^{(m)}}_{m\geq1}$ be the corresponding normalized limiting hierarchy. Fix $k\geq1$ and let $\OJ^{(k)} \in \mathcal{K}(\Hkb)$ be a compact observable, and let 
\begin{equation}
    \iota_{k+1}: \Hkp \longrightarrow \Hkb\otimes\h
\end{equation} be the canonical isometric embedding. Define the sum of external two-body interactions as
\begin{equation}
\label{eq:summed-k-k-plus-one-interaction-direct-limit}
    \mathcal{V}_{k,k+1} := \sum_{j=1}^{k} v(x_j-x_{k+1})
\end{equation}
on $\Hkb\otimes\h$. Define the bounded operator $B_J^{(k+1)} \in \mathcal{B}(\Hkp)$ by
\begin{equation}
\label{eq:compressed-BBGKY-interaction-test}
    B_J^{(k+1)} := \iota_{k+1}^\dagger \left[ \OJ^{(k)} \otimes\1_{\h}, \mathcal{V}_{k,k+1} \right] \iota_{k+1}.
\end{equation}
Then, $B_J^{(k+1)} \in\mathcal{B}(\Hkp)$ satisfies the uniform operator bound
\begin{equation}
\label{eq:compressed-BBGKY-interaction-test-bound}
    \norm{ B_J^{(k+1)} }_{\mathrm{op}} \leq 2k \norm{v}_{L^\infty} \norm{\OJ^{(k)} }_{\mathrm{op}}.
\end{equation}
Moreover, for every positive trace-class operator $\gamma^{(k+1)}\in\mathfrak{S}_1(\Hkp)$, one has the trace identity
\begin{equation}
\label{eq:interaction-test-trace-identity}
\begin{aligned}
    \Tr_{\Hkb} \left[ \OJ^{(k)} \sum_{j=1}^{k} \Tr_{k+1} \left[ v(x_j-x_{k+1}), \gamma^{(k+1)} \right] \right] = \Tr_{\Hkp} \left[ B_J^{(k+1)} \gamma^{(k+1)} \right].
\end{aligned}
\end{equation}
Consequently, for every $t\in[0,T]$,
\begin{equation}
\label{eq:direct-limit-surviving-BBGKY-interaction}
\begin{aligned}
    \lim_{\ell\to\infty} \lambda \frac{N_\ell-k}{N_\ell} \Tr_{\Hkp} \left[ B_J^{(k+1)} \gamma_{N_\ell,t}^{(k+1)} \right] = \lambda \Tr_{\Hkp} \left[ B_J^{(k+1)} \gamma_{\infty,t}^{(k+1)} \right].
\end{aligned}
\end{equation}
This convergence also holds after time integration, that is, for every $\tau\in[0,T]$,
\begin{equation}
\label{eq:time-integrated-direct-limit-surviving-BBGKY-interaction}
\begin{aligned}
    \lim_{\ell\to\infty} \lambda \frac{N_\ell-k}{N_\ell} \int_0^\tau \Tr_{\Hkp} \left[ B_J^{(k+1)} \gamma_{N_\ell,s}^{(k+1)} \right] ds = \lambda \int_0^\tau \Tr_{\Hkp} \left[ B_J^{(k+1)} \gamma_{\infty,s}^{(k+1)} \right] ds.
\end{aligned}
\end{equation}
Finally, the internal interaction among the first $k$ particles vanishes uniformly in time
\begin{equation}
\label{eq:vanishing-internal-finite-k-interaction}
\begin{aligned}
    \lim_{\ell\to\infty} \sup_{t\in[0,T]} \left| \frac{\lambda}{N_\ell} \sum_{1\leq i<j\leq k} \Tr_{\Hkb} \left[ [\OJ^{(k)}, v(x_i-x_j)] \gamma_{N_\ell,t}^{(k)} \right] \right| = 0.
\end{aligned}
\end{equation}
\end{corollary}

\begin{proof}
Since $v\in L^\infty(\R^3,\R)$, the multiplication operator $\mathcal{V}_{k,k+1}$ is bounded and satisfies
\begin{equation}
\label{eq:summed-interaction-operator-bound}
    \norm{ \mathcal{V}_{k,k+1} }_{\mathrm{op}} \leq k\norm{v}_{L^\infty}.
\end{equation}
Since $\iota_{k+1}$ is an isometry, we have $\norm{\iota_{k+1}}_{\mathrm{op}} = \norm{\iota_{k+1}^\dagger}_{\mathrm{op}} = 1$. Therefore,
\begin{align*}
    \norm{ B_J^{(k+1)} }_{\mathrm{op}} 
    &\leq \norm{ \left[ \OJ^{(k)} \otimes\1_{\h}, \mathcal{V}_{k,k+1} \right] }_{\mathrm{op}} \\
    &\leq 2 \norm{\OJ^{(k)} \otimes\1_{\h}}_{\mathrm{op}} \norm{\mathcal{V}_{k,k+1}}_{\mathrm{op}} \\
    &\leq 2k \norm{v}_{L^\infty} \norm{\OJ^{(k)} }_{\mathrm{op}}.
\end{align*}
This proves \eqref{eq:compressed-BBGKY-interaction-test-bound}. By the definition of the partial trace, we can lift the trace from the $k$-particle to the $(k+1)$-particle space
\begin{align*}
    &\Tr_{\Hkb} \left[ \OJ^{(k)} \sum_{j=1}^{k} \Tr_{k+1} \left[ v(x_j-x_{k+1}), \gamma^{(k+1)} \right] \right] \\
    &\qquad \qquad = \Tr_{\Hkb\otimes\h} \left[ (\OJ^{(k)} \otimes\1_{\h}) \left[ \mathcal{V}_{k,k+1}, \iota_{k+1} \gamma^{(k+1)} \iota_{k+1}^\dagger \right] \right].
\end{align*}
Since the product of a bounded operator and a trace-class operator is trace-class, the cyclicity of the trace gives
\begin{align*}
    &\Tr_{\Hkb\otimes\h} \left[ (\OJ^{(k)} \otimes\1_{\h}) \left[ \mathcal{V}_{k,k+1}, \iota_{k+1} \gamma^{(k+1)} \iota_{k+1}^\dagger \right] \right] \\*
    &\qquad \qquad= \Tr_{\Hkb\otimes\h} \left[ \left[ \OJ^{(k)} \otimes\1_{\h}, \mathcal{V}_{k,k+1} \right] \iota_{k+1} \gamma^{(k+1)} \iota_{k+1}^\dagger \right] \\*
    &\qquad \qquad = \Tr_{\Hkp} \left[ \iota_{k+1}^\dagger \left[ \OJ^{(k)} \otimes\1_{\h}, \mathcal{V}_{k,k+1} \right] \iota_{k+1} \gamma^{(k+1)} \right],
\end{align*}
which establishes \eqref{eq:interaction-test-trace-identity}.  Fix $s\in[0,T]$. By \thmref{thm:compactness-bbgky-sequence}, we have the weak-$*$ convergence $\gamma_{N_\ell,s}^{(k+1)} \overset{*}{\rightharpoonup} \gamma_{\infty,s}^{(k+1)}$ in $\mathfrak{S}_1(\Hkp)$. Since both operators are non-negative and possess unit trace, \lemref{lem:weak-str-conv-and-traceclass-imply-trace-norm-conv} implies convergence in trace norm
\begin{equation}
\label{eq:trace-norm-convergence-used-for-BBGKY-interaction}
    \norm{ \gamma_{N_\ell,s}^{(k+1)} - \gamma_{\infty,s}^{(k+1)} }_{\mathfrak{S}_1(\Hkp)} \longrightarrow 0 \quad \text{ as } \quad  \ell\to\infty.
\end{equation}
Since $B_J^{(k+1)}$ is bounded, we can apply the triangle inequality to evaluate the convergence of the test interaction
\begin{align*}
    &\left| \frac{N_\ell-k}{N_\ell} \Tr_{\Hkp} \left[ B_J^{(k+1)} \gamma_{N_\ell,s}^{(k+1)} \right] - \Tr_{\Hkp} \left[ B_J^{(k+1)} \gamma_{\infty,s}^{(k+1)} \right] \right| \\
    &\qquad \qquad \qquad \qquad \qquad \leq \norm{ B_J^{(k+1)} }_{\mathrm{op}} \left( \norm{ \gamma_{N_\ell,s}^{(k+1)} - \gamma_{\infty,s}^{(k+1)} }_{\mathfrak{S}_1} + \frac{k}{N_\ell} \right).
\end{align*}
Because the sequence of trace norms converges to zero and $\frac{k}{N_\ell} \to 0$, the entire right-hand side tends to zero, which proves \eqref{eq:direct-limit-surviving-BBGKY-interaction}. For every $\ell$ and $s\in[0,T]$, we have the uniform bound
\begin{equation}
    \left| \lambda \frac{N_\ell-k}{N_\ell} \Tr_{\Hkp} \left[ B_J^{(k+1)} \gamma_{N_\ell,s}^{(k+1)} \right] \right| \leq |\lambda| \norm{ B_J^{(k+1)} }_{\mathrm{op}},
\end{equation}
where we used the fact that $\norm{\gamma_{N_\ell,s}^{(k+1)}}_{\mathfrak{S}_1}=1$. Therefore, by coupling the established pointwise convergence with the Dominated Convergence Theorem, we arrive at \eqref{eq:time-integrated-direct-limit-surviving-BBGKY-interaction}. Finally, for the internal interaction terms, we compute the supremum over time
\begin{align*}
    \sup_{t\in[0,T]} \Bigg| \frac{\lambda}{N_\ell} \sum_{1\leq i<j\leq k} \Tr_{\Hkb} &\left[ [\OJ^{(k)}, v(x_i-x_j)] \gamma_{N_\ell,t}^{(k)} \right] \Bigg| \\
    &\leq \frac{|\lambda|}{N_\ell} \sum_{1\leq i<j\leq k} \norm{ [\OJ^{(k)}, v(x_i-x_j)] }_{\mathrm{op}} \\
    &\leq \frac{|\lambda|k(k-1)}{N_\ell} \norm{v}_{L^\infty} \norm{\OJ^{(k)} }_{\mathrm{op}}.
\end{align*}
As $N_\ell \to \infty$, this supremum vanishes, which proves \eqref{eq:vanishing-internal-finite-k-interaction}.
\end{proof}

\subsection{Quantum de Finetti theorem}

\begin{theorem}[Strong quantum de Finetti \cite{ammari2008mean, chen2015unconditional, hudson1976locally, lewin2014derivation, stormer1969symmetric}]
\label{thm:strong-quantum-de-Finetti}
Let $\Hkb$ be any bosonic $k$-particle Hilbert space over a separable Hilbert space $\h$. For any fixed $0<t<\infty$, let $\Gamma_t$ be a family of admissible reduced density matrices
\begin{equation}
\label{eq:admissible-hierarchy}
    \Gamma_t = (\gamma^{(1)}_t, \gamma^{(2)}_t, \ldots),
\end{equation}
where each $\gamma^{(k)}_t$ is a non-negative trace-class operator on $\Hkb$ satisfying the partial trace consistency relation $\gamma^{(k)}_t = \Tr_{k+1}\gamma^{(k+1)}_t$ for all $k\geq 1$. Then, there exists a unique $S^1$-invariant Borel probability measure $\mu_t$ on the unit sphere $\mathbb{S}\h := \{\varphi\in\h : \norm{\varphi}_{\h}=1\}$ such that
\begin{equation}
\label{eq:strong-quantum-de-Finetti-representation}
    \gamma^{(k)}_t = \int_{\mathbb{S}\h} \left|\varphi^{\otimes k}\right\rangle \left\langle\varphi^{\otimes k}\right| \,d\mu_t(\varphi), \qquad k\geq1.
\end{equation}
Moreover, if $A \geq 0$ is a self-adjoint operator on $\h$ such that $\Tr_{\h}(A\gamma^{(1)}_t) < \infty$, then the measure is supported on its quadratic-form domain, $\mu_t(\cQ(A))=1$, where $\cQ(A) := \mathcal{D}(A^{1/2})$. Furthermore, the energy expectation is given by
\begin{equation}
\label{eq:strong-de-Finetti-energy-identity}
    \int_{\mathbb{S}\h} \norm{A^{1/2}\varphi}_{\h}^{2} \,d\mu_t(\varphi) = \Tr_{\h}(A\gamma^{(1)}_t),
\end{equation}
where the integrand is understood to be $+\infty$ outside $\cQ(A)$.
\end{theorem}

\begin{theorem}[Weak quantum de Finetti \cite{ammari2008mean, ammari2011mean, chen2015unconditional, lewin2014derivation}]
\label{thm:weak-quantum-de-Finetti}
Let $\Hkb$ be any bosonic $k$-particle Hilbert space over a separable Hilbert space $\h$. For a fixed $0<t<\infty$, assume that $\Gamma_{t,N}^{(N)}$ is a sequence of mixed states on $\h^{\otimes_s N}$, for $N \in \N$, satisfying $\Gamma_{t,N}^{(N)} \geq 0$ and $\Tr_{\h^{\otimes_s N}} \prc{\Gamma_{t,N}^{(N)}} = 1$. Suppose that for every $k\geq1$, its $k$-particle marginals have weak-$*$ limits in $\fkS_1(\Hkb)$,
\begin{equation}
    \gamma_{t,N}^{(k)} := \Tr_{k+1,\ldots,N} \prc{\Gamma_{t,N}^{(N)}} \rightharpoonup^* \gamma^{(k)} \quad \text{ as }  \quad N\to\infty,
\end{equation} 
where $\Tr_{k+1,\ldots,N}$ denotes the partial trace over the last $N-k$ variables. Then, there exists a unique $S^1$-invariant Borel probability measure $\mu_t$ on the unit ball $B\h := \{\varphi\in\h : \norm{\varphi}_{\h}\leq1\}$ such that
\begin{equation}
\label{eq:weak-quantum-de-Finetti-representation}
    \gamma^{(k)}_t = \int_{B\h} \left|\varphi^{\otimes k}\right\rangle \left\langle\varphi^{\otimes k}\right| \,d\mu_t(\varphi), \qquad k\geq1.
\end{equation}
\end{theorem}

\begin{corollary}[De Finetti representation and bounded energy support]
\label{cor:bounded-energy-support-de-finetti}
Assume that \hyporef{hyp:standing-hypotheses} holds, and for any $t\in[0,T]$, $  \Gamma_{\infty,t}
=\{\gamma_{\infty,t}^{(k)}\}_{k\geq1}$ be a limiting hierarchy obtained from
\thmref{thm:compactness-bbgky-sequence}. Recall the definitions
\begin{equation}
    \OT=(\dch_0+\kappa)^{1/2},
    \qquad
    \OT_k=\prod_{j=1}^k\OT_j,
    \qquad
    Y^1=\cD(\OT),
\end{equation}
and let $C_E>0$ be the constant in
\corref{cor:k-particle-energy-bound}.
In particular, one may take
\begin{equation}
    C_E=4\max\{1,\kappa\}C_0,
\end{equation}
where $C_0$ is the constant in
\hypref{hyp:bounded-energy-per-particle}.

For every $t\in[0,T]$, there exists a unique
$S^1$-invariant Borel probability measure $\mu_t$
on $\mathbb S\h$ such that
\begin{equation}
\label{eq:quantum-de-Finetti-representation-limit-point}
    \gamma_{\infty,t}^{(k)}
    =
    \int_{\mathbb S\h}
        |\varphi^{\otimes k}\rangle
        \langle\varphi^{\otimes k}|
        \,d\mu_t(\varphi),
    \qquad k\geq1.
\end{equation}
Moreover, for every integer $k\geq1$,
\begin{equation}
\label{eq:definetti-uniform-energy-moments}
    \Tr_{\Hkb}\!\left(
        \OT_k\gamma_{\infty,t}^{(k)}\OT_k
    \right)
    =
    \int_{\mathbb S\h}
        \|\OT\varphi\|^{2k}\,d\mu_t(\varphi)
    \leq C_E^k.
\end{equation}
The integrand is assigned the value $+\infty$
when $\varphi\notin\cD(\OT)$, and the weighted
density matrix is understood through its
positive trace-class realization.

Consequently, for any $t\in[0,T]$
\begin{equation}
\label{eq:definetti-bounded-energy-support}
    \mu_t\!\left(
        \left\{
            \varphi\in Y^1\mcap\mathbb S\h:
            \|\OT\varphi\|^2\leq C_E
        \right\}
    \right)=1.
\end{equation}
The set appearing in
\eqref{eq:definetti-bounded-energy-support}
is compact in $\h$.
In particular, the initial measure $\mu_0$
has the same bounded energy support.
\end{corollary}

\begin{proof}
By \propref{prop:admissibility-limit},
the limiting hierarchy is nonnegative,
normalized, and admissible.
The strong quantum de Finetti theorem therefore
gives the unique $S^1$-invariant probability
measure $\mu_t$ satisfying
\eqref{eq:quantum-de-Finetti-representation-limit-point}. 

We next prove the energy-moment bound. Fix $t\in[0,T]$ and an integer $k\geq1$. Let
\begin{equation}
    P_d
    =
    \sum_{j=0}^{d-1}
        |\chi_j\rangle\langle\chi_j|
\end{equation}
be the spectral projection onto the first $d$
trap eigenfunctions.

Because $P_d$ commutes with $\OT$, the operator
\begin{equation}
    (P_d\OT^2P_d)^{\otimes k}
\end{equation}
is positive and finite rank on $\Hkb$.
Furthermore, as nonnegative quadratic forms,
\begin{equation}
    (P_d\OT^2P_d)^{\otimes k}
    \leq \OT_k^2.
\end{equation}

Let $\{N_\ell\}_{\ell\geq1}$ be the subsequence
defining the limiting hierarchy.
Since $k$ is fixed and $K(N)\to\infty$, for all
sufficiently large $\ell$ we have
\begin{equation}
    k\leq K(N_\ell),
    \qquad
    N_\ell\geq N_0(k),
\end{equation}
and the initial moment hypothesis also applies.
Thus \corref{cor:k-particle-energy-bound} yields
\begin{equation}
\begin{aligned}
    \Tr_{\Hkb}\!\left(
        (P_d\OT^2P_d)^{\otimes k}
        \gamma_{N_\ell,t}^{(k)}
    \right)
    &\leq
    \Tr_{\Hkb}\!\left(
        \OT_k\gamma_{N_\ell,t}^{(k)}\OT_k
    \right)
    \leq C_E^k.
\end{aligned}
\end{equation}

We first let $\ell\to\infty$, keeping $k$ and
$d$ fixed.
Weak-$*$ convergence against the finite-rank
test operator gives
\begin{equation}
\begin{aligned}
    &\Tr_{\Hkb}\!\left(
        (P_d\OT^2P_d)^{\otimes k}
        \gamma_{\infty,t}^{(k)}
    \right)=
    \lim_{\ell\to\infty}
    \Tr_{\Hkb}\!\left(
        (P_d\OT^2P_d)^{\otimes k}
        \gamma_{N_\ell,t}^{(k)}
    \right)
    \leq C_E^k.
\end{aligned}
\end{equation}
On the other hand, testing the de Finetti
representation against this operator yields
\begin{equation}
\begin{aligned}
    \Tr_{\Hkb}\!\left(
        (P_d\OT^2P_d)^{\otimes k}
        \gamma_{\infty,t}^{(k)}
    \right)&=
    \int_{\mathbb S\h}
        \left(
            \langle\varphi,
                P_d\OT^2P_d\varphi\rangle
        \right)^k
        \,d\mu_t(\varphi)\\
    &=
    \int_{\mathbb S\h}
        \|\OT P_d\varphi\|^{2k}
        \,d\mu_t(\varphi).
\end{aligned}
\end{equation}
For every $\varphi\in\h$,
\begin{equation}
    \|\OT P_d\varphi\|^2
    =
    \sum_{j=0}^{d-1}
        (E_j+\kappa)
        |\langle\chi_j,\varphi\rangle|^2.
\end{equation}
These quantities increase to
$\|\OT\varphi\|^2$ as $d\to\infty$,
with value $+\infty$ outside $\cD(\OT)$.
Monotone convergence therefore gives
\begin{equation}
    \int_{\mathbb S\h}
        \|\OT\varphi\|^{2k}\,d\mu_t(\varphi)
    \leq C_E^k.
\end{equation}

The same spectral truncations increase, in
the sense of nonnegative quadratic forms,
to $\OT_k^2$.
Consequently, monotone convergence of the
scalar traces gives
\begin{equation}
\begin{aligned}
    \Tr_{\Hkb}\!\left(
        \OT_k\gamma_{\infty,t}^{(k)}\OT_k
    \right)
    &=
    \lim_{d\to\infty}
    \Tr_{\Hkb}\!\left(
        (P_d\OT^2P_d)^{\otimes k}
        \gamma_{\infty,t}^{(k)}
    \right)\\
    &=
    \int_{\mathbb S\h}
        \|\OT\varphi\|^{2k}\,d\mu_t(\varphi)
    \leq C_E^k.
\end{aligned}
\end{equation}
Finiteness of this extended trace implies that
$\OT_k(\gamma_{\infty,t}^{(k)})^{1/2}$
is Hilbert--Schmidt and hence supplies the
claimed positive trace-class realization.
This proves \eqref{eq:definetti-uniform-energy-moments}. In particular, taking $k=1$ shows that
\begin{equation}
    \mu_t(Y^1\mcap\mathbb S\h)=1.
\end{equation}

To obtain bounded energy support, fix $\varepsilon>0$. For every integer $k\geq1$, Markov's inequality gives
\begin{equation}
\begin{aligned}
    \mu_t\!\left(
        \left\{
            \varphi\in\mathbb S\h:
            \|\OT\varphi\|^2>C_E+\varepsilon
        \right\}
    \right)&\leq
    \frac{
        \displaystyle
        \int_{\mathbb S\h}
            \|\OT\varphi\|^{2k}\,d\mu_t(\varphi)
    }{(C_E+\varepsilon)^k}\\
    &\leq
    \left(
        \frac{C_E}{C_E+\varepsilon}
    \right)^k.
\end{aligned}
\end{equation}
Letting $k\to\infty$ shows that this probability
is zero.
Taking the countable union over
$\varepsilon=1/m$, $m\geq1$, proves
\eqref{eq:definetti-bounded-energy-support}.

Finally, we prove compactness in $\h$.
The extended quadratic form
\begin{equation}
    \varphi\longmapsto\|\OT\varphi\|^2
\end{equation}
is lower semicontinuous in $\h$.
Its sublevel set at $C_E$, intersected with
the unit sphere, is therefore closed in $\h$.

For every $\varphi$ in this set,
\begin{equation}
\begin{aligned}
    \|(I-P_d)\varphi\|^2
    &=
    \sum_{j=d}^{\infty}
        |\langle\chi_j,\varphi\rangle|^2\\
    &\leq
    \frac{1}{E_d+\kappa}
    \sum_{j=d}^{\infty}
        (E_j+\kappa)
        |\langle\chi_j,\varphi\rangle|^2\\
    &\leq
    \frac{C_E}{E_d+\kappa}.
\end{aligned}
\end{equation}
Since $\dch_0$ has compact resolvent,
$E_d\to\infty$.
Thus these spectral tails vanish uniformly
on the set.
Its finite-dimensional projections are bounded,
so the set is totally bounded in $\h$.
Being closed in the complete space $\h$,
it is compact.

The argument applies to every $t\in[0,T]$,
including $t=0$, with the same constant $C_E$.
\end{proof}

\subsection{Joint compactness and the GP radiation-source}
\label{sec:joint-radiation-source}
To pass the radiation source to the limit, we retain the correlations between the particle and field variables in a probability measure on the joint phase space, the phase-space viewpoint used for the Nelson model in
\cite{ammari2014wigner}. We identify the resulting measure dynamically in the subsequent sections. 

Define the function $ g_x(\xi)=e^{-i\xi\cdot x}\widehat w(\xi)$ and
\begin{equation}
\label{eq:joint-source-scalings}
 \phi_N(f)=\frac{a(f)+\ad(f)}{\sqrt N},\quad
 \wyl_N(f)=\exp(i\phi_N(f)),
\end{equation}
where $\wyl_N(f)$ denotes the displacement operator with parameter
$if/\sqrt N$.

For $\ux_k=(x_1,\ldots,x_k)$, define
\begin{equation}
\label{eq:def-Gxk-and-VNk}
    G_{\ux_k}:=\displaystyle\sum_{j=1}^k g_{x_j},
    \qquad
    \dcV_N^{(k)}:=\phi_N(G_{\ux_k}),
\end{equation}
and
\begin{equation}
\label{eq:def-Fu-and-Fuk}
    F_u(x):=2\re\langle g_x,u\rangle,
    \qquad
    F_u^{(k)}(\ux_k):=\displaystyle\sum_{j=1}^k F_u(x_j).
\end{equation}
On the field fiber over $\ux_k$, it acts as $\phi_N(G_{\ux_k})$.  
Since $w$ is a real and even function,
\begin{equation}
\label{eq:joint-source-classical-field}
\begin{aligned}
 F_u(x)&=2\re \int_{\R^3}
       e^{i\xi\cdot x}\widehat w(\xi)u(\xi)\,d\xi,\\
\|F_u\|_\infty &\leq 2\|w\|_{L^2} \|u\|_{L^2} .
 \end{aligned}
\end{equation}
Set $\Gamma_{N,t}=|\Psi_{N,t}\rangle\langle\Psi_{N,t}|$, and let
$\gamma_{N,t}^{(k)}$ denote its particle marginal of trace one.
Expectations of nonnegative unbounded operators are understood in
the quadratic-form sense, equivalently as the monotone limits of
expectations of bounded spectral truncations.
The source is normalized as in the BBGKY equation \eqref{eq:finite-N-BBGKY-operator-form-clean}. In particular, the parameter $\eta$ is not included in $\cS_{N,t}^{(k)}$.

For the following computations, regard $\Psi_{N,t}$ as a vector in the ambient vector space $\h^{\otimes N}\otimes\Fock$ and extend $\Gamma_{N,t}$ by zero to this ambient space. The operator $\dcV_N^{(k)}$ acts on the first $k$
particle variables and the field, with the identity
on the remaining particle variables. The photon-number bound \eqref{eq:joint-source-uniform-number} implies $\Psi_{N,t}\in\cD(\cN_f^{1/2})$. Given that $\|G_{\ux_k}\|_{L^2}\leq k\|w\|_{L^2}$, the creation and annihilation estimates give
\begin{equation}
    \|\dcV_N^{(k)}\Psi_{N,t}\|
        _{\h^{\otimes N}\otimes\Fock}
    \leq
    2k\|w\|_{L^2}
    \left\|
        \left(\frac{\cN_f+1}{N}\right)^{1/2}
        \Psi_{N,t}
    \right\|_{\HN}.
\end{equation}
Thus, the commutator has the trace-class realization
\begin{equation}
\begin{aligned}
    [\dcV_N^{(k)},\Gamma_{N,t}]
    &:=
    |\dcV_N^{(k)}\Psi_{N,t}\rangle
        \langle\Psi_{N,t}|
    \\
    &\qquad \qquad -
    |\Psi_{N,t}\rangle
        \langle\dcV_N^{(k)}\Psi_{N,t}|
    \in
    \mathfrak S_1(\h^{\otimes N}\otimes\Fock).
\end{aligned}
\end{equation}
We define
\begin{equation}
\label{eq:joint-source-finite-definition}
    \cS_{N,t}^{(k)}
    :=
    \Tr_{k+1,\ldots,N;\Fock}
        [\dcV_N^{(k)},\Gamma_{N,t}].
\end{equation}
Both $\Psi_{N,t}$ and
$\dcV_N^{(k)}\Psi_{N,t}$ are symmetric in the
first $k$ particle variables.
Consequently, the partial trace is supported on
$\Hkb$, and
\begin{equation}
    \cS_{N,t}^{(k)}\in\mathfrak S_1(\Hkb),
    \qquad
    \Tr_{\Hkb} \prc{\cS_{N,t}^{(k)}}=0.
\end{equation}
For a  $\OJ^{(k)}\in\mathcal K(\Hkb)$, extend
$\OJ^{(k)}$ by zero to $\h^{\otimes k}$ and then
tensor it with the identity on the remaining
particle variables and the field.
This bounded operator commutes with $\cN_f$,
and therefore preserves $\cD(\cN_f^{1/2})$,
which is contained in $\cD(\dcV_N^{(k)})$.
Using the same notation for this lift, we obtain
\begin{equation}
\begin{aligned}
    \Tr_{\Hkb} \bigl(\OJ^{(k)}\cS_{N,t}^{(k)}\bigr)
    &=\left\langle (\OJ^{(k)})^*\Psi_{N,t},
        \dcV_N^{(k)}\Psi_{N,t}
    \right\rangle\\
    &\qquad \qquad -
    \left\langle
        \dcV_N^{(k)}\Psi_{N,t},
        \OJ^{(k)}\Psi_{N,t}
    \right\rangle
    \\
    &=
    \left\langle
        \Psi_{N,t},
        [\OJ^{(k)},\dcV_N^{(k)}]\Psi_{N,t}
    \right\rangle.
\end{aligned}
\end{equation}
The inner products in this identity are taken
in $\h^{\otimes N}\otimes\Fock$. Finally, contractivity of the partial trace,
the rank-one trace-norm identity, and
$\|\Psi_{N,t}\|=1$ give
\begin{equation}
\begin{aligned}
    \|\cS_{N,t}^{(k)}\|_{\mathfrak S_1(\Hkb)}
    &\leq
    \|[\dcV_N^{(k)},\Gamma_{N,t}]\|
        _{\mathfrak S_1(\h^{\otimes N}\otimes\Fock)}
    \\
    &\leq
    2\|\dcV_N^{(k)}\Psi_{N,t}\|
        _{\h^{\otimes N}\otimes\Fock}
    \\
    &\leq
    4k\|w\|_{L^2}
    \left\|
        \left(\frac{\cN_f+1}{N}\right)^{1/2}
        \Psi_{N,t}
    \right\|_{\HN}.
\end{aligned}
\end{equation}

The following compactness lemma will be applied to the
finite-dimensional joint Husimi measures constructed in
\appref{sec:joint-radiation-source-limit-proof}.

\begin{lemma}[Finite-dimensional tightness and consistent limits]
\label{lem:joint-husimi-tightness}
Let $(E_d)_{d\geq1}$ and $(L_m)_{m\geq1}$ be increasing
families of finite-dimensional subspaces of $\h$ and
$\h_f$, respectively, with
\begin{equation}
    \dim E_d=d,
    \qquad
    \dim L_m=m.
\end{equation}
Denote the corresponding orthogonal projections by
$P_d$ and $P_m^{\mathrm f}$.

For every $N,d,m\geq1$, let $\mathfrak m_N^{d,m}$ be a
Borel probability measure on $E_d\times L_m$. Suppose
that these measures are consistent under coordinate
projections:
\begin{equation}
    (P_d\times P_m^{\mathrm f})_{\#}
        \mathfrak m_N^{d',m'}
    =
    \mathfrak m_N^{d,m},
    \qquad d'\geq d,\quad m'\geq m,
\end{equation}
where the projections are restricted to
$E_{d'}\times L_{m'}$.
Assume that there exists a constant $\cM_T<\infty$,
independent of $N,d,m$, such that
\begin{equation}
    \begin{aligned}
    \int_{E_d\times L_m}\|z\|^2\,
        \dcd\mathfrak m_N^{d,m}(z,u)
    &\leq 1+\frac dN,\\
    \qquad
    \int_{E_d\times L_m}\|u\|^2\,
        \dcd\mathfrak m_N^{d,m}(z,u)
    &\leq \cM_T+\frac mN.
    \end{aligned}
\end{equation}
Then, for every fixed $d,m$, the family
$\{\mathfrak m_N^{d,m}\}_{N\geq1}$ is tight on
$E_d\times L_m$.
Moreover, from every sequence of particle numbers
tending to infinity, one can extract a subsequence
$(N_\ell)_{\ell\geq1}$ and find Borel probability
measures $\nu^{d,m}$ on $E_d\times L_m$ such that
\begin{equation}
    \mathfrak m_{N_\ell}^{d,m}
    \rightharpoonup \nu^{d,m}
    \qquad\text{as }\ell\to\infty
\end{equation}
weakly, simultaneously for every $d,m\geq1$.
The limiting measures are consistent:
\begin{equation}
    (P_d\times P_m^{\mathrm f})_{\#}\nu^{d',m'}
    =
    \nu^{d,m},
    \qquad d'\geq d,\quad m'\geq m.
\end{equation}
\end{lemma}

\begin{proof}
Fix $d,m\geq1$. For $R>0$, set
\begin{equation}
    K_R^{d,m}
    :=
    \left\{
        (z,u)\in E_d\times L_m:
        \|z\|^2+\|u\|^2\leq R^2
    \right\}.
\end{equation}
This set is compact because $E_d\times L_m$ is
finite-dimensional. By Markov's inequality and the
assumed moment bounds,
\begin{align*}
    \mathfrak m_N^{d,m}
        \bigl((E_d\times L_m)\setminus K_R^{d,m}\bigr)
    &\leq
    \frac1{R^2}
    \int_{E_d\times L_m}
        \bigl(\|z\|^2+\|u\|^2\bigr)\,
        \dcd\mathfrak m_N^{d,m}(z,u)
    \\
    &\leq
    \frac{1+\cM_T+(d+m)/N}{R^2}
    \\
    &\leq
    \frac{1+\cM_T+d+m}{R^2}.
\end{align*}
The last expression tends to zero as $R \to \infty$,
uniformly in $N$. This proves tightness.

Starting from any prescribed sequence of particle
numbers tending to infinity, Prokhorov's theorem
provides a weakly convergent subsequence for each
fixed pair $(d,m)$. Enumerate the countable set
$\mathbb N^2$, extract successively nested
subsequences, and take a diagonal subsequence.
The resulting sequence $(N_\ell)_{\ell\geq1}$ yields
weak convergence for every fixed pair $(d,m)$.

It remains to verify consistency of the limits.
Fix $d'\geq d$, $m'\geq m$, and
$F\in C_b(E_d\times L_m)$.
The function
\begin{equation}
    (z,u)\longmapsto F(P_dz,P_m^{\mathrm f}u)
\end{equation}
is bounded and continuous on $E_{d'}\times L_{m'}$.
Finite-$N$ consistency and weak convergence give
\begin{align*}
    \int_{E_d\times L_m}F\,\dcd\nu^{d,m}
    &=
    \lim_{\ell\to\infty}
    \int_{E_d\times L_m}
        F\,\dcd\mathfrak m_{N_\ell}^{d,m}
    \\
    &=
    \lim_{\ell\to\infty}
    \int_{E_{d'}\times L_{m'}}
        F(P_dz,P_m^{\mathrm f}u)\,
        \dcd\mathfrak m_{N_\ell}^{d',m'}(z,u)
    \\
    &=
    \int_{E_{d'}\times L_{m'}}
        F(P_dz,P_m^{\mathrm f}u)\,
        \dcd\nu^{d',m'}(z,u).
\end{align*}
Since bounded continuous functions determine Borel
probability measures on finite-dimensional spaces,
this proves the asserted consistency.
\end{proof}

\begin{proposition}[Joint source limit]
    \label{prop:joint-radiation-source-limit}
Assume that \hyporef{hyp:standing-hypotheses} holds.  Fix $T<\infty$, and start with any subsequence on which all
particle marginals converge to $\gamma_{\infty,t}^{(k)}$ as in the
compactness result of \thmref{thm:compactness-bbgky-sequence}.  There is a further subsequence, independent
of $t\in[0,T]$ and $k$, and a Borel probability kernel
$t\mapsto\nu_t$ on
\begin{equation}
 \mathbb S\h\times\h_f
 =\{\vfi:\|\vfi\|_{L^2} =1\}\times\h_f,
\end{equation}
invariant under $(\vfi,u)\mapsto(e^{i\theta}\vfi,u)$ for $\theta\in[0,2\pi)$, with the
following properties.  Its particle marginal is the de Finetti measure
$\mu_t$ representing $(\gamma_{\infty,t}^{(k)})_{k\geq1}$, and
\begin{equation}\label{eq:joint-source-initial-law}
 \nu_0=\mu_0\otimes\delta_{\widehat u_0}.
\end{equation}
For every $k\geq1$ and $\OJ^{(k)}\in\mathcal K(\Hkb)$,
\begin{equation}\label{eq:joint-source-uniform-limit}
 \lim_{N\to\infty}\sup_{0\leq t\leq T}
 \left|\Tr\bigl[\OJ^{(k)}(\cS_{N,t}^{(k)}
                    -\cS_{\infty,t}^{(k)})\bigr]\right|=0,
\end{equation}
where $N$ runs through this further subsequence and
\begin{equation}\label{eq:joint-source-limit-operator}
 \boxed{\displaystyle
 \cS_{\infty,t}^{(k)}
 =\int_{\mathbb S\h\times\h_f}
    [F_u^{(k)},|\vfi^{\otimes k}\rangle
                    \langle\vfi^{\otimes k}|]\,d\nu_t(\vfi,u).}
\end{equation}
The integral is a Bochner integral in the trace class.  Equivalently, its
kernel, understood as an $L^2$-kernel, is
\begin{equation}\label{eq:joint-source-limit-kernel}
 \cS_{\infty,t}^{(k)}(\ux_k;\ux_k')
 =\sum_{j=1}^k\int
  \bigl(F_u(x_j)-F_u(x'_j)\bigr)
  \prod_{\ell=1}^k\vfi(x_\ell)\overline{\vfi(x'_\ell)}
  \,d\nu_t(\vfi,u).
\end{equation}
In particular, the source converges weak-$*$ in the trace class at every
time, and also after integration against any $L^1$ function of time.
\end{proposition}
\begin{proof}
The proof is given in \appref{sec:joint-radiation-source-limit-proof}. 
\end{proof}

\subsection{Derivation of the joint Liouville equation}
\label{sec:identification-joint-dynamics}
The source representation in
\propref{prop:joint-radiation-source-limit} leaves two dynamical
identifications to establish: the joint measure must satisfy the
Liouville equation of the coupled system, and that equation must
transport the initial joint measure along its unique characteristics.
The precise remaining proof obligations at this stage are:
\begin{enumerate}
 \item Extend mixed characteristic convergence to bounded particle
 tests and to a field insertion in the presence of a field character
 (\lemref{lem:mixed-observable-extensions} and
 \lemref{lem:mixed-radiation-commutator});
 \item Derive the joint Liouville equation from the microscopic
 commutator identity, including its extension to every smooth
 cylindrical test
 (\propref{prop:joint-compactness-and-Liouville-equation});
 \item Construct the coupled characteristics and identify the joint
 measure as their pushforward
 (\lemref{lem:global-joint-Hilbert-flow} and
 \thmref{thm:joint-measure-transport});
 \item Substitute that transport law into the radiation source and
 remove the subsequence when the initial hierarchy converges
 (\lemref{lem:annihilation-as-transported-source-operator} and
 \thmref{thm:uniqueness-of-gp-hierarchy}).
\end{enumerate}
The first two steps are proved in this subsection. The construction
and transport of the coupled flow, persistence of the energy space,
and the consequences for the particle hierarchy are treated in the
subsequent results listed above. Keep the momentum-space convention for the field coordinate and set
\begin{equation}\label{eq:joint-phase-space-and-source}
 \mathscr Z:=\h\oplus\h_f,\qquad
 R(\vfi):=\int_{\R^3}g_x|\vfi(x)|^2\,dx
         =\widehat w\,\widehat{|\vfi|^2}.
\end{equation}
The integral defining $R$ is a Bochner integral in $\h_f$:
$x\mapsto g_x$ is continuous in $\h_f$, and
$\int\|g_x\||\vfi(x)|^2\,dx=\|w\|_{L^2}\|\vfi\|_{L^2}^2$.
The Fourier identity uses the transform and momentum measure
fixed in the microscopic formulation.
The complex space $\mathscr Z$ is regarded as a real Hilbert space
when derivatives of scalar functions are taken.  Define the propagator
\begin{equation}
    \cU_{fr}(t,s):=e^{-i(t-s)\dch_0}\oplus e^{-i(t-s)\omg}.
\end{equation}
Also define
\begin{equation}
\label{eq:joint-free-group-and-nonlinearity}
 \mathcal R(\vfi,u):=
 \begin{pmatrix}
  \bigl(\lam(v*|\vfi|^2)+\eta F_u\bigr)\vfi\vspace{1em}\\
  \eta R(\vfi)
 \end{pmatrix},\\
\end{equation}
For a joint phase-space point $z=(\vfi,u)\in\mathscr Z=\h\oplus\h_f$, define the interaction-picture vector field by
\begin{equation}
 \cVec(t,z):=-i\cU_{fr}(0,t)
              \mathcal R(\cU_{fr}(t,0)z).
\end{equation}
Here $e^{it\omg}$ acts by multiplication in momentum space.
Both free factors in the definition of $\cVec$ are needed:
$\cU_{fr}(t,0)$ returns the interaction-picture point to physical
coordinates, while $\cU_{fr}(0,t)$ maps the nonlinear velocity
back to interaction coordinates.  On the particle unit ball,
\begin{equation}\label{eq:joint-velocity-growth}
 \|\cVec(t,(\vfi,u))\|_{\mathscr Z}
 \leq \lam\|v\|_{L^\infty} +\eta\|w\|_{L^2}
                      +2\eta\|w\|_{L^2} \|u\|_{L^2} .
\end{equation}
Indeed,
\begin{equation}
    \begin{aligned}
        \|R(\vfi)\|_{L^2} &\leq\|w\|_{L^2} \|\vfi\|_{L^2}^2\\
        \|F_u\|_{L^\infty} &\leq2\|w\|_{L^2} \|u\|_{L^2}
    \end{aligned}
\end{equation}
and $\cU_{fr}(t,s)$ preserves both component norms.
The same estimates, together with
$\||\vfi|^2-|\psi|^2\|_{L^1}
\leq(\|\vfi\|+\|\psi\|)\|\vfi-\psi\|$,
show that $\mathcal R$ is continuous on $\mathscr Z$.
Thus, $\cVec$ is jointly continuous, hence Borel, and bounded on
bounded subsets of $[0,T]\times\mathscr Z$.

Recall the free particle evolution $\dch_k=\sum_{j=1}^k\dch_0^{(j)}$. For a bounded particle operator $A$ on $\h^{\otimes r}$, define the classical observable
\begin{equation}
 \begin{aligned}
 b_A(\vfi)&:=\langle\vfi^{\otimes r},A\vfi^{\otimes r}\rangle,
 \\
 G_{N,t}^{(r)}(A,f)&:=\langle\Psi_{N,t},A\wyl_N(f)\Psi_{N,t}\rangle,
 \\
 \rchi_f(u)&:=e^{2i\re \langle f,u\rangle}.
 \end{aligned}
\end{equation}
These expressions agree with the previous notation on $\Hkb$.  Let $S_r$ denote the group of permutations of $\{1,\ldots,r\}$. The orthogonal projection $P_{\mathrm{sym}}^{(r)}$ from $\h^{\otimes r}$ onto its symmetric subspace is defined by
\begin{equation}
 \bigl(P_{\mathrm{sym}}^{(r)}\psi\bigr)(x_1,\ldots,x_r)
 :=
 \frac{1}{r!}\sum_{\sigma\in S_r}
 \psi(x_{\sigma(1)},\ldots,x_{\sigma(r)}),
 \qquad \psi\in\h^{\otimes r}.
\end{equation}
If $A$ does not preserve this subspace, we use its compression $P_{\mathrm{sym}}^{(r)}AP_{\mathrm{sym}}^{(r)}$, restricted to that subspace. Indeed,
\begin{equation}
\begin{aligned}
 b_A(\vfi)
 &=b_{P_{\mathrm{sym}}^{(r)}AP_{\mathrm{sym}}^{(r)}}(\vfi),\\
 G_{N,t}^{(r)}(A,f)
 &=G_{N,t}^{(r)}
   \bigl(P_{\mathrm{sym}}^{(r)}AP_{\mathrm{sym}}^{(r)},f\bigr).
\end{aligned}
\end{equation}
The first equality follows from the symmetry of $\vfi^{\otimes r}$. The second follows from the symmetry of $\Psi_{N,t}$ in its first $r$ particle variables and the fact that $\wyl_N(f)$ acts only on the field. We use this convention in particular for operators such as $[v_{j,k+1},\OJ^{(k)}\otimes\1]$, which distinguish the $(k+1)$-st particle.
\begin{lemma}[Two extensions of mixed observable convergence]
\label{lem:mixed-observable-extensions}
Along the subsequence of
\propref{prop:joint-radiation-source-limit}, the following limits hold
uniformly for any $t\in[0,T]$
\begin{equation}\label{eq:bounded-particle-mixed-limit}
 G_{N,t}^{(r)}(A,f)
 \longrightarrow\int b_A(\vfi)\rchi_f(u)\,d\nu_t(\vfi,u),
\end{equation}
where $A\in\mathcal B(\h^{\otimes r})$ and $f\in\h_f$, and
\begin{equation}\label{eq:weyl-weighted-field-insertion}
\begin{aligned}
 &\langle\Psi_{N,t},K\wyl_N(f)\phi_N(h)\Psi_{N,t}\rangle\longrightarrow
 \int b_K(\vfi)\rchi_f(u)
       2\re\langle h,u\rangle\,d\nu_t(\vfi,u),
\end{aligned}
\end{equation}
where $K\in\mathcal K(\Hkb)$ and $f,h\in\h_f$. The second limit also holds for $k=0$, $K=\1$.
\end{lemma}
\begin{proof}
Let $P_d$ be the orthogonal projection onto the first $d$ spectral
modes of $\OT^2$, and let
$\Pi_d=P_d^{\otimes r}$.  If $\lambda_{d+1}$ is the next eigenvalue,
particle symmetry and \eqref{eq:joint-source-uniform-energy} give
\begin{equation}
 \langle\Psi_{N,t},(1-\Pi_d)\Psi_{N,t}\rangle
 \leq\frac{r \, C_E}{\lambda_{d+1}}.
\end{equation}
Decompose  $A-\Pi_dA\Pi_d=(1-\Pi_d)A+\Pi_dA(1-\Pi_d)$. Then, 
Cauchy--Schwarz, unitarity of $\wyl_N(f)$, and its commutation
with particle projections imply
\begin{equation}\label{eq:bounded-particle-cutoff-error}
 \left|\langle\Psi_{N,t},
       (A-\Pi_dA\Pi_d)\wyl_N(f)\Psi_{N,t}\rangle\right|
 \leq2\|A\|\sqrt{\frac{rC_E}{\lambda_{d+1}}}.
\end{equation}
The identical estimate for the classical integral follows from
\eqref{eq:joint-source-joint-moment-bounds}, using
$1-P_d^{\otimes r}\leq\sum_{j=1}^r(1-P_d)^{(j)}$ on tensor
powers. 

More explicitly, on $\|\vfi\|=1$, we have
\begin{equation}
 \|(1-\Pi_d)\vfi^{\otimes r}\|^2
 \leq r\|(1-P_d)\vfi\|^2
 \leq\frac{r}{\lambda_{d+1}}\|\OT\vfi\|^2.
\end{equation}
Similarly, decomposing the classical pairing and applying Cauchy--Schwarz in $\nu_t$ gives
\begin{equation}
 \left|\int
 (b_A-b_{\Pi_dA\Pi_d})(\vfi)\rchi_f(u)\,d\nu_t\right|
 \leq2\|A\|\sqrt{\frac{rC_E}{\lambda_{d+1}}}.
\end{equation}
The compression $\Pi_dA\Pi_d$ is finite rank. Apply
\eqref{eq:joint-source-uniform-characteristics} to this compression,
then let $d\to\infty$ proves \eqref{eq:bounded-particle-mixed-limit}.

For the second assertion put $\epsilon=N^{-1}$ and
$\alpha=\im \langle f,h\rangle$.  The Weyl relation becomes
\begin{equation}\label{eq:mixed-weyl-directional-identity}
 \wyl_N(f+rh)
 =e^{i\epsilon r\alpha}\wyl_N(f)e^{ir\phi_N(h)},
 \qquad r\in\R.
\end{equation}
In particular, the first derivative at $r=0$ is
\begin{equation}
 \left.\frac{d}{dr}G_{N,t}^{(k)}(K,f+rh)\right|_{r=0}
 =i\langle K\wyl_N(f)\phi_N(h)\rangle_{N,t}
          +i\epsilon\alpha G_{N,t}^{(k)}(K,f),
\end{equation}
where $\langle\cdot\rangle_{N,t}$ denotes expectation in $\Psi_{N,t}$.

We justify the convergence of this derivative.
Let $B$ be self-adjoint, $r\in\R$, and $a,b\in\cD(B)$. Taylor's formula gives
\begin{equation}
 e^{irx}-1-irx=x^2q_r(x),
 \qquad
 q_r(x):=-r^2\int_0^1(1-s)e^{irsx}\,ds,
 \qquad x\in\R.
\end{equation}
Since $|e^{irsx}|=1$, we have $\|q_r\|_{L^\infty}\leq r^2/2$. The spectral theorem therefore yields the quadratic-form identity
\begin{equation}
 \langle a,(e^{irB}-\1-irB)b\rangle
 =\langle Ba,q_r(B)Bb\rangle.
\end{equation}
Consequently, the functional calculus and Cauchy--Schwarz imply
\begin{equation}\label{eq:mixed-spectral-remainder}
\begin{aligned}
 \left|\langle a,(e^{irB}-\1-irB)b\rangle\right|
 &\leq \|q_r(B)\|\,\|Ba\|\,\|Bb\|\leq \frac{r^2}{2}\|Ba\|\,\|Bb\|.
\end{aligned}
\end{equation}
Only $a,b\in\cD(B)$ is required, since one factor of $B$ acts on each vector and $q_r(B)$ is bounded.
Apply this with
$B=\phi_N(h)$,
$a=\wyld_N(f) K^*\Psi_{N,t}$, and $b=\Psi_{N,t}$.
These vectors belong to $\cD(B)$: $K$ commutes with the photon
number operator, and Weyl displacement preserves
$\cD(\cN_f^{1/2})\subset\cD(\phi_N(h))$.
The number and displacement bounds in
\eqref{eq:joint-source-quadrature-bound} and
\eqref{eq:joint-source-displaced-number} imply
\begin{equation}
 \|Bb\|\leq2\|h\|\sqrt{\cM_T},\qquad
 \|Ba\|\leq2\|K\|\|h\|
                      \bigl(\sqrt{\cM_T}+\epsilon\|f\|\bigr).
\end{equation}
The scalar phase in \eqref{eq:mixed-weyl-directional-identity}
has a uniformly bounded second derivative for fixed $f,h$.
Consequently, for $|r|\leq1$, we have 
\begin{align}
 \big|G_{N,t}^{(k)}&(K,f+rh)-G_{N,t}^{(k)}(K,f)\\
 &-ir\langle K\wyl_N(f)\phi_N(h)\rangle_{N,t}-i\epsilon r\alpha G_{N,t}^{(k)}(K,f)\big|
 \leq C_{K,f,h,T}r^2,
 \label{eq:mixed-uniform-taylor-remainder}
\end{align}
uniformly in $N,t$.  

To estimate the classical remainder, define
\begin{equation}
 G_t^{(k)}(K,g)
 :=\int b_K(\vfi)\rchi_g(u)\,d\nu_t(\vfi,u),
 \qquad g\in\h_f.
\end{equation}
The scalar inequality $|e^{ix}-1-ix|\leq x^2/2$, together with $|b_K(\vfi)|\leq\|K\|$ and the field second-moment bound, gives
\begin{equation}
\begin{aligned}
 \Bigg|
 G_t^{(k)}(K,f+rh)-G_t^{(k)}(K,f)
 -&ir\int b_K(\vfi)\rchi_f(u)
              2\re\langle h,u\rangle\,d\nu_t
 \Bigg|\\
 &\qquad\leq
 2r^2\|K\|\int|\re\langle h,u\rangle|^2\,d\nu_t\\
 &\qquad\leq
 2r^2\|K\|\|h\|^2\cM_T.
\end{aligned}
\end{equation}
Thus, both the quantum and classical Taylor remainders are bounded by a constant times $r^2$, uniformly in $N$ and $t\in[0,T]$.

For the fixed test $K$, write
\begin{equation}
 \delta_N(g)
 :=\sup_{0\leq t\leq T}
 \left|G_{N,t}^{(k)}(K,g)-G_t^{(k)}(K,g)\right|.
\end{equation}
By \eqref{eq:joint-source-uniform-characteristics}, we have $\delta_N(g)\to0$ as $N\to\infty$ for every fixed $g\in\h_f$. Subtracting the two Taylor expansions, dividing by $|r|$, and using $|G_{N,t}^{(k)}(K,f)|\leq\|K\|$, we obtain, for $0<|r|\leq1$,
\begin{equation}
\begin{aligned}
 \sup_{0\leq t\leq T}
 \Bigg|
 \langle K\wyl_N(f)&\phi_N(h)\rangle_{N,t}
 -\int b_K(\vfi)\rchi_f(u)
              2\re\langle h,u\rangle\,d\nu_t
 \Bigg|\\
 &\qquad\leq
 \frac{\delta_N(f+rh)+\delta_N(f)}{|r|}
 +\frac{|\alpha|\|K\|}{N}
 +C_{K,f,h,T}|r|.
\end{aligned}
\end{equation}
Here $C_{K,f,h,T}$ includes both remainder bounds and is independent of $N$, $t$, and $r\in[-1,1]$. For each fixed nonzero $r$, the first two terms on the right tend to zero as $N\to\infty$. Consequently,
\begin{equation}
\begin{aligned}
 \limsup_{N\to\infty}\sup_{0\leq t\leq T}
 \Bigg| \langle K\, \wyl_N(f)\phi_N(h)&\rangle_{N,t}
 -\int b_K(\vfi)\rchi_f(u)
              2\re\langle h,u\rangle\,d\nu_t
 \Bigg|\\
 &\leq C_{K,f,h,T}|r|, \quad \text{as } N\to\infty.
\end{aligned}
\end{equation}
Letting $r\to0$ proves \eqref{eq:weyl-weighted-field-insertion}. The quantum remainder requires only the first photon-number moment, while the classical remainder uses $\int\|u\|^2\,d\nu_t\leq\cM_T$; no second moment of $\cN_f$ is needed.
\end{proof}

\begin{lemma}[The radiation commutator with a field character]
\label{lem:mixed-radiation-commutator}
Along the subsequence of \propref{prop:joint-radiation-source-limit}, for every $\OJ^{(k)}\in\mathcal K(\Hkb)$ and $f\in\h_f$, we have
\begin{equation}
\label{eq:mixed-radiation-commutator-limit}
\begin{aligned}
 &\langle\Psi_{N,t},
      [\dcV_N^{(k)},\OJ^{(k)}]\wyl_N(f)\Psi_{N,t}\rangle\\
 &\qquad \qquad \qquad \longrightarrow
 \int\rchi_f(u)
   \langle\vfi^{\otimes k},
      [F_u^{(k)},\OJ^{(k)}]\vfi^{\otimes k}\rangle
 \,d\nu_t(\vfi,u),
\end{aligned}
\end{equation}
uniformly for $t\in[0,T]$.
\end{lemma}

\begin{proof}
Bounded particle operators and $\wyl_N(f)$ preserve $\cD(\cN_f^{1/2})$. Since $\|G_{\ux_k}\|_{\h_f}\leq k\|w\|_{L^2}$, the fiberwise quadrature bound places this space inside $\cD(\dcV_N^{(k)})$. Thus, the commutator products in the statement are well-defined on $\Psi_{N,t}$. Fix $m$. With $c_\ell(\ux_k)=\langle e_\ell,G_{\ux_k}\rangle$, the finite-mode expansion \eqref{eq:joint-source-real-mode-expansion} reads
\begin{equation}
 \dcV_{N,m}^{(k)}
 =\sum_{\ell=1}^m
 \left(
   M_{\re c_\ell}\phi_N(e_\ell)
   +M_{\operatorname{Im}c_\ell}\phi_N(ie_\ell)
 \right).
\end{equation}
Each multiplier $M_{\re c_\ell}$ and $M_{\operatorname{Im}c_\ell}$ is bounded, so its commutator with $\OJ^{(k)}$ is compact. For a typical term $M_c\phi_N(h)$ in this expansion, the particle operators commute with the field operators, and the Weyl commutator identity gives
\begin{equation}
\begin{aligned}
 &\langle
   [M_c\phi_N(h),\OJ^{(k)}]\wyl_N(f)
  \rangle_{N,t}\\
 &\quad=
 \langle
   [M_c,\OJ^{(k)}]\wyl_N(f)\phi_N(h)
 \rangle_{N,t}-\frac{2}{N}\operatorname{Im}\langle h,f\rangle\,
 G_{N,t}^{(k)}([M_c,\OJ^{(k)}],f).
\end{aligned}
\end{equation}
Here $\langle\cdot\rangle_{N,t}$ denotes expectation in $\Psi_{N,t}$, and $h$ is either $e_\ell$ or $ie_\ell$. The second term is bounded in absolute value by
\begin{equation}
 \frac{2}{N}\|h\|\|f\|\,
 \|[M_c,\OJ^{(k)}]\|,
\end{equation}
uniformly in $t$. Applying \eqref{eq:weyl-weighted-field-insertion} to the first term and summing over the finitely many modes therefore yields
\begin{equation}
\begin{aligned}
 &\langle\Psi_{N,t},
   [\dcV_{N,m}^{(k)},\OJ^{(k)}]
   \wyl_N(f)\Psi_{N,t}\rangle\\
 &\qquad \qquad \qquad \longrightarrow
 \int\rchi_f(u)
   \langle\vfi^{\otimes k},
      [F_{P_m^{\mathrm f}u}^{(k)},\OJ^{(k)}]
      \vfi^{\otimes k}\rangle
 \,d\nu_t,
\end{aligned}
\end{equation}
uniformly for $t\in[0,T]$. The classical multiplier is identified by the corresponding real-mode expansion of $F_{P_m^{\mathrm f}u}^{(k)}$. We next remove the field projection. As in Step~6 of the source proof, set
\begin{equation}
 R_m(\ux_k)=(\1-P_m^{\mathrm f})G_{\ux_k},
 \qquad
 d_m(p)^2
 =\int |p(\ux_k)|^2
       \|R_m(\ux_k)\|_{\h_f}^2\,d\ux_k.
\end{equation}
Strong convergence of $P_m^{\mathrm f}$ gives $R_m(\ux_k)\to0$ for every $\ux_k$, while $\|R_m(\ux_k)\|_{\h_f}\leq k\|w\|_{L^2}$. Dominated convergence consequently gives $d_m(p)\to0$ for every $p\in\Hkb$. First take $\OJ^{(k)}=|p\rangle\langle q|$. The partial-inner-product estimate from Step~6 gives, for every $\zeta\in\cD(\cN_f^{1/2})$,
\begin{equation}
 \|\phi_N(R_m)\OJ^{(k)}\zeta\|
 \leq
 2\|q\|d_m(p)
 \left\|
   \left(\frac{\cN_f+1}{N}\right)^{1/2}\zeta
 \right\|.
\end{equation}
Apply this estimate to $\zeta=\wyl_N(f)\Psi_{N,t}$ and, after interchanging $p$ and $q$, to $\zeta=\Psi_{N,t}$. In the second commutator term, move the self-adjoint operator $\phi_N(R_m)$ onto $(\OJ^{(k)})^*\Psi_{N,t}$. Unitarity of $\wyl_N(f)$ and \eqref{eq:joint-source-displaced-number} then imply
\begin{equation}
\label{eq:mixed-radiation-tail}
\begin{aligned}
 \Bigg|\langle\Psi_{N,t},
   [\phi_N(R_m)&,\OJ^{(k)}]
   \wyl_N(f)\Psi_{N,t}\rangle\Bigg|\\
 &\quad\leq
 \|\phi_N(R_m)\OJ^{(k)}
          \wyl_N(f)\Psi_{N,t}\|
 +\|\phi_N(R_m)(\OJ^{(k)})^*\Psi_{N,t}\|\\
 &\quad\leq
 2\|q\|d_m(p)
       \left(\sqrt{\cM_T}+\frac{\|f\|}{N}\right)
 +2\|p\|d_m(q)\sqrt{\cM_T}.
\end{aligned}
\end{equation}
On the classical side,
\begin{equation}
 \|(F_u^{(k)}-F_{P_m^{\mathrm f}u}^{(k)})p\|
 \leq2\|u\|d_m(p).
\end{equation}
The rank-one commutator identity, $|\rchi_f(u)|=1$, and $\int\|u\|\,d\nu_t\leq\sqrt{\cM_T}$ therefore give
\begin{equation}
\begin{aligned}
 &\left|
 \int\rchi_f(u)
 \langle\vfi^{\otimes k},
   [F_u^{(k)}-F_{P_m^{\mathrm f}u}^{(k)},\OJ^{(k)}]
   \vfi^{\otimes k}\rangle\,d\nu_t
 \right|\\
 &\qquad \qquad \qquad  \leq
 2\sqrt{\cM_T}
 \bigl(\|q\|d_m(p)+\|p\|d_m(q)\bigr).
\end{aligned}
\end{equation}
Both tail estimates are uniform in $t$ and extend by the triangle inequality to finite sums of rank-one operators. Combining them with the finite-mode convergence, first letting $N\to\infty$ along the extracted subsequence with $m$ fixed and then letting $m\to\infty$, proves \eqref{eq:mixed-radiation-commutator-limit} for every finite-rank test, uniformly on $[0,T]$. Finally, for any bounded particle operator $\OJ^{(k)}$, self-adjointness of $\dcV_N^{(k)}$, unitarity of $\wyl_N(f)$, and the quadrature and displaced-number bounds yield
\begin{equation}
\begin{aligned}
 \Bigg|\langle\Psi_{N,t},&
   [\dcV_N^{(k)},\OJ^{(k)}]
   \wyl_N(f)\Psi_{N,t}\Bigg|\\
 &\quad\leq
 \|\OJ^{(k)}\|
 \left(
   \|\dcV_N^{(k)}\Psi_{N,t}\|
   +\|\dcV_N^{(k)}\wyl_N(f)\Psi_{N,t}\|
 \right)\\
 &\quad\leq
 2k\|w\|_{L^2}\|\OJ^{(k)}\|
 \left(2\sqrt{\cM_T}+\frac{\|f\|}{N}\right).
\end{aligned}
\end{equation}
Similarly, $\|F_u^{(k)}\|\leq2k\|w\|_{L^2}\|u\|$ gives
\begin{equation}
\begin{aligned}
 \Bigg|
 \int\rchi_f(u)
 \langle\vfi^{\otimes k},&
   [F_u^{(k)},\OJ^{(k)}]\vfi^{\otimes k}\rangle
 \,d\nu_t
 \Bigg|\\
 &\quad\leq
 4k\|w\|_{L^2}\|\OJ^{(k)}\|
 \int\|u\|\,d\nu_t\\
 &\quad\leq
 4k\|w\|_{L^2}\|\OJ^{(k)}\|\sqrt{\cM_T}.
\end{aligned}
\end{equation}
Given a compact $\OJ^{(k)}$, choose finite-rank operators $K_n$ with $\|K_n-\OJ^{(k)}\|\to0$. Apply these two bounds to $\OJ^{(k)}-K_n$, use the established convergence for each fixed $K_n$, and then let $n\to\infty$. This proves \eqref{eq:mixed-radiation-commutator-limit} for every compact test, uniformly on $[0,T]$.
\end{proof}

\begin{proposition}[The joint Liouville equation]
\label{prop:joint-compactness-and-Liouville-equation}
Let $\nu_t$ be the measure obtained in \propref{prop:joint-radiation-source-limit}, and define
\begin{equation}\label{eq:joint-interaction-picture-measure}
 \widetilde\nu_t
 :=\bigl(\cU_{fr}(0,t)\bigr)_{\#}\nu_t.
\end{equation}
The curve $t\mapsto\widetilde\nu_t$ is narrowly continuous for the weak topology of $\mathscr Z$. It satisfies
\begin{equation}\label{eq:joint-Liouville-integral}
 \int_{\mathscr Z}b(z)\,d\widetilde\nu_t(z)
 -\int_{\mathscr Z}b(z)\,d\widetilde\nu_s(z)
 =
 \int_s^t\int_{\mathscr Z}
       Db(z)[\cVec(r,z)]\,d\widetilde\nu_r(z)\,dr
\end{equation}
for $0\leq s,t\leq T$ and every smooth cylindrical test $b$ with a compactly supported finite-dimensional representative, where derivatives are understood in the real sense. Furthermore,
\begin{equation}\label{eq:joint-scalar-velocity-bound}
 \int_0^T\int_{\mathscr Z}
       \|\cVec(t,z)\|_{\mathscr Z}\,d\widetilde\nu_t(z)\,dt
 \leq CT(1+\sqrt{\cM_T})<\infty.
\end{equation}
\end{proposition}

\begin{proof}
The free propagators preserve the norms in both factors of $\mathscr Z$. Thus the particle sphere support and field second-moment bound pass from $\nu_t$ to $\widetilde\nu_t$, and \eqref{eq:joint-velocity-growth} gives \eqref{eq:joint-scalar-velocity-bound}. We derive the Liouville equation from the microscopic dynamics.

\medskip\noindent\textbf{The finite-$N$ mixed identity.}
Take a regular finite-rank particle test $\OJ^{(k)}$, with range and adjoint range contained in a finite span of eigenvectors of $\dch_k$, and take $f\in\cD(\omg)$. Set
\begin{equation}
 m_f(x):=\im\langle g_x,f\rangle,
 \qquad
 \|m_f\|_{L^\infty}\leq\|w\|_{L^2}\|f\|.
\end{equation}
As in Step~3 of the source proof, take the permutation average of the labeled particle observable when acting on the bosonic $N$-particle space. Its expectation is unchanged by symmetry. The regularity of the particle test and the condition $f\in\cD(\omg)$ ensure that the observable and its adjoint preserve the energy form domain, while the photon-number bound makes the field insertions well defined. The integrated Heisenberg identity from that step therefore applies. In differential form, the following equality holds almost everywhere in time; all subsequent uses are in its integrated form:
\begin{align}
 \partial_tG_{N,t}^{(k)}(\OJ^{(k)},f)
 &=iG_{N,t}^{(k)}([\dch_k,\OJ^{(k)}],f)
 \nonumber\\
 &\quad+\frac{i\lam}{N}
 \sum_{1\leq i<j\leq k}
 G_{N,t}^{(k)}([v_{ij},\OJ^{(k)}],f)
 \nonumber\\
 &\quad+i\lam\frac{N-k}{N}
 \sum_{j=1}^k
 G_{N,t}^{(k+1)}
   ([v_{j,k+1},\OJ^{(k)}\otimes\1],f)
 \nonumber\\
 &\quad+i\eta
 \langle[\dcV_N^{(k)},\OJ^{(k)}]\wyl_N(f)\rangle_{N,t}
 \nonumber\\
 &\quad+i
 \langle\OJ^{(k)}\wyl_N(f)\phi_N(i\omg f)\rangle_{N,t}
 +\frac{i}{N}\langle f,\omg f\rangle
 G_{N,t}^{(k)}(\OJ^{(k)},f)
 \nonumber\\
 &\quad-\frac{2i\eta}{N}
 \sum_{j=1}^k
 G_{N,t}^{(k)}(\OJ^{(k)}m_f(x_j),f)
 \nonumber\\
 &\quad-2i\eta\frac{N-k}{N}
 G_{N,t}^{(k+1)}(\OJ^{(k)}\otimes m_f,f).
 \label{eq:finite-N-joint-characteristic-hierarchy}
\end{align}
It suffices that $N\geq k+1$, as holds eventually along the extracted subsequence. For $k=0$, take $\OJ^{(0)}=\1$; the particle terms vanish, and the last term is $-2i\eta G_{N,t}^{(1)}(m_f,f)$.

To identify the contributions in \eqref{eq:finite-N-joint-characteristic-hierarchy}, write
\begin{equation}
 \HMN
 =\sum_{j=1}^N\dch_0^{(j)}+\Hf
 +\frac{\lam}{N}\sum_{1\leq i<j\leq N}v_{ij}
 +\eta\sum_{j=1}^N\phi_N(g_{x_j}).
\end{equation}
Pairs whose two labels exceed $k$ commute with the observable. Each pair with one retained label gives the same $(k+1)$-particle expectation after permutation of the other labels, producing the factor $(N-k)/N$. For the radiation interaction, the product rule gives
\begin{equation}
\begin{aligned}
 &\left[
   \sum_{j=1}^N\phi_N(g_{x_j}),
   \OJ^{(k)}\wyl_N(f)
 \right]\\
 &\qquad=
 [\dcV_N^{(k)},\OJ^{(k)}]\wyl_N(f)
 +\OJ^{(k)}
   \sum_{j=1}^N[\phi_N(g_{x_j}),\wyl_N(f)].
\end{aligned}
\end{equation}
The two final lines of \eqref{eq:finite-N-joint-characteristic-hierarchy} follow from
\begin{equation}
 [\phi_N(g_x),\wyl_N(f)]
 =-\frac{2}{N}m_f(x)\wyl_N(f),
\end{equation}
together with symmetry among the labels exceeding $k$. The free-field contributions follow from \eqref{eq:joint-source-free-field-commutator}. Thus every term in \eqref{eq:finite-N-joint-characteristic-hierarchy} is a microscopic commutator contribution.

\medskip\noindent\textbf{Cancellation of the free evolution.}
For fixed $\OJ^{(k)}$ and $f$ as above, use the time-dependent tests
\begin{equation}\label{eq:joint-free-evolved-tests}
 \OJ_t^{(k)}
 =e^{-it\dch_k}\OJ^{(k)}e^{it\dch_k},
 \qquad
 f_t=e^{-it\omg}f.
\end{equation}
Then
\begin{equation}
 b_{\OJ_t^{(k)}}(\vfi)\rchi_{f_t}(u)
 =
 b_{\OJ^{(k)}}(e^{it\dch_0}\vfi)
 \rchi_f(e^{it\omg}u).
\end{equation}
Consequently, mixed characteristic convergence and the pushforward definition give, for each $t\in[0,T]$,
\begin{equation}
 G_{N,t}^{(k)}(\OJ_t^{(k)},f_t)
 \longrightarrow
 \int b_{\OJ^{(k)}}(\vfi)\rchi_f(u)
 \,d\widetilde\nu_t(\vfi,u).
\end{equation}
The additional test derivatives are
\begin{align*}
 \partial_t\OJ_t^{(k)}
 &=-i[\dch_k,\OJ_t^{(k)}],\\
 \partial_t\wyl_N(f_t)
 &=-i\wyl_N(f_t)\phi_N(i\omg f_t)
 -\frac{i}{N}\langle f_t,\omg f_t\rangle\wyl_N(f_t).
\end{align*}
The Weyl derivative is understood strongly on $\cD(\cN_f^{1/2})$. Indeed, \eqref{eq:mixed-weyl-directional-identity} gives
\begin{equation}
 \partial_t\wyl_N(f_t)
 =
 i\wyl_N(f_t)\phi_N(\partial_t f_t)
 +\frac{i}{N}
   \im\langle f_t,\partial_t f_t\rangle\wyl_N(f_t).
\end{equation}
Since $\partial_t f_t=-i\omg f_t$, real linearity of $\phi_N$ and the scalar-product convention imply
\begin{equation}
 \phi_N(\partial_t f_t)=-\phi_N(i\omg f_t),
 \qquad
 \im\langle f_t,\partial_t f_t\rangle
 =-\langle f_t,\omg f_t\rangle.
\end{equation}
This proves the displayed derivative, including its scalar term. In the integrated Heisenberg identity, the derivative is interpreted in the energy-form sense, consistently with the free-field covariance
\begin{equation}
 \wyl_N(f_t)
 =e^{-it\Hf}\wyl_N(f)e^{it\Hf}.
\end{equation}
More explicitly, the field-energy estimates give, for $\zeta\in\cD(\Hf^{1/2})$,
\begin{equation}
 \|\phi_N(i\omg f_t)\zeta\|
 \leq
 \frac{2}{\sqrt N}\|\sqrt{\omg}f\|
      \|\Hf^{1/2}\zeta\|
 +\frac{\|\omg f\|}{\sqrt N}\|\zeta\|.
\end{equation}
To justify this extension directly, let $\zeta_M=\mathbf1_{\{\cN_f\leq M\}}\zeta$ for $\zeta\in\cD(\Hf^{1/2})$. Since $\Hf$ and $\cN_f$ strongly commute, $\zeta_M\to\zeta$ in the field-energy form norm. Integrate the Weyl derivative identity for $\zeta_M$ and use the preceding bound to pass to the limit $M\to\infty$. The resulting derivative expression is strongly continuous in time, since $\omg f_t$ and $\sqrt{\omg}f_t$ are norm continuous. Hence it gives the strong derivative on $\cD(\Hf^{1/2})$. The expression for the Weyl derivative therefore extends to this domain and satisfies
\begin{equation}
\begin{aligned}
 \|(\partial_t\wyl_N(f_t))\zeta\|
 \leq
 \frac{2}{\sqrt N}\|\sqrt{\omg}f\|
      \|\Hf^{1/2}\zeta\|+
 \left(
   \frac{\|\omg f\|}{\sqrt N}
   +\frac{\|\sqrt{\omg}f\|^2}{N}
 \right)\|\zeta\|.
\end{aligned}
\end{equation}
All coefficients are finite because $f\in\cD(\omg)$, and the bounds are uniform in $t$. The same energy-form regularization used above therefore applies to these time-dependent tests. Their derivatives cancel the free particle and field terms in \eqref{eq:finite-N-joint-characteristic-hierarchy}, including the scalar term of order $N^{-1}$.

The remaining internal-particle and retained-label terms vanish as $N\to\infty$: their absolute values are bounded, respectively, by
\begin{equation}
 \frac{\lam k(k-1)}{N}
 \|v\|_{L^\infty}\|\OJ^{(k)}\|,
 \qquad
 \frac{2\eta k}{N}
 \|w\|_{L^2}\|f\|\|\OJ^{(k)}\|.
\end{equation}
The other particle operators are bounded and are therefore covered by \eqref{eq:bounded-particle-mixed-limit}. The radiation commutator is covered by \eqref{eq:mixed-radiation-commutator-limit}. These limits apply at each fixed time to the tests in \eqref{eq:joint-free-evolved-tests}. The preceding lemmas bound all the remaining integrands uniformly on $[0,T]$. Dominated convergence therefore permits passage to the limit in the integrated identity with these time-dependent tests.

We identify the resulting generator explicitly. At a physical phase point $(\vfi,u)$, put $V_\vfi=v*|\vfi|^2$. Partial contraction with the last particle in a tensor power yields
\begin{equation}\label{eq:joint-Hartree-symbol-contraction}
 \sum_{j=1}^k
 b_{[v_{j,k+1},\OJ^{(k)}\otimes\1]}(\vfi)
 =
 \left\langle\vfi^{\otimes k},
 \left[\sum_{j=1}^kV_\vfi(x_j),\OJ^{(k)}\right]
 \vfi^{\otimes k}\right\rangle.
\end{equation}
For any bounded real multiplier $V$,
\begin{equation}\label{eq:joint-particle-symbol-derivative}
 Db_{\OJ^{(k)}}(\vfi)[-iV\vfi]
 =
 i\left\langle\vfi^{\otimes k},
 \left[\sum_{j=1}^kV(x_j),\OJ^{(k)}\right]
 \vfi^{\otimes k}\right\rangle.
\end{equation}
The field contribution satisfies
\begin{equation}\label{eq:joint-field-symbol-derivative}
 -2i\eta\rchi_f(u)
 \int m_f(x)|\vfi(x)|^2\,dx
 =
 D\rchi_f(u)[-i\eta R(\vfi)].
\end{equation}
Indeed, both sides equal $2i\eta\rchi_f(u)\im\langle f,R(\vfi)\rangle$. For the generating test
\begin{equation}\label{eq:joint-generating-symbols}
 b(\vfi,u)=b_{\OJ^{(k)}}(\vfi)\rchi_f(u),
\end{equation}
identities \eqref{eq:joint-Hartree-symbol-contraction}--\eqref{eq:joint-field-symbol-derivative}, with the free factors in \eqref{eq:joint-free-evolved-tests}, identify the nonlinear generator in physical coordinates. The chain rule then gives
\begin{equation}
\begin{aligned}
 D\bigl(b\circ\cU_{fr}(0,t)\bigr)&
   \bigl(\cU_{fr}(t,0)z\bigr)
   \bigl[-i\mathcal R(\cU_{fr}(t,0)z)\bigr]\\
 &\quad=
 Db(z)\bigl[-i\cU_{fr}(0,t)
                \mathcal R(\cU_{fr}(t,0)z)\bigr]\\
 &\quad=Db(z)[\cVec(t,z)].
\end{aligned}
\end{equation}
Thus the limiting integrated identity is \eqref{eq:joint-Liouville-integral} for the generating tests \eqref{eq:joint-generating-symbols}. Although these generating symbols are not compactly supported, their derivatives are integrable on the support of the measures. In particular, on $\|\vfi\|\leq1$,
\begin{equation}
 |Db_{\OJ^{(k)}}(\vfi)[a]|
 \leq2k\|\OJ^{(k)}\|\|a\|.
\end{equation}
The identity extends to every finite-rank particle test and every $f\in\h_f$. To see this, approximate the particle test by its spectral compressions. If $\Pi_d=P_d^{\otimes k}$, then on the particle unit ball,
\begin{equation}
\begin{aligned}
 |b_{\OJ^{(k)}}-b_{\Pi_d\OJ^{(k)}\Pi_d}|
 &\leq
 \|\OJ^{(k)}-\Pi_d\OJ^{(k)}\Pi_d\|,\\
 \|Db_{\OJ^{(k)}}-Db_{\Pi_d\OJ^{(k)}\Pi_d}\|
 &\leq
 2k\|\OJ^{(k)}-\Pi_d\OJ^{(k)}\Pi_d\|.
\end{aligned}
\end{equation}
Both right-hand sides tend to zero by compactness. Next approximate $f$ in $\h_f$ by elements of $\cD(\omg)$. The corresponding field characters and their first derivatives converge pointwise, with uniformly bounded first derivatives. The velocity bound \eqref{eq:joint-scalar-velocity-bound} justifies passage to the limit in the derivative terms.

\medskip\noindent\textbf{Every smooth cylindrical test.}
First let $b(\vfi,u)=B(P_{\mathrm p}\vfi,P_{\mathrm f}u)$, where $P_{\mathrm p}:\h\to\h$ and $P_{\mathrm f}:\h_f\to\h_f$ are arbitrary complex-linear orthogonal projections of finite rank and
\begin{equation}
 B\in C_c^\infty
 \bigl(P_{\mathrm p}\h\times P_{\mathrm f}\h_f\bigr).
\end{equation}
All smoothness and derivatives are understood in real coordinates. The measures $\widetilde\nu_t$ are invariant under
\begin{equation}
 \mathcal G_\theta(\vfi,u)=(e^{i\theta}\vfi,u),
\end{equation}
and the vector field satisfies
\begin{equation}
 \cVec(t,\mathcal G_\theta z)
 =\mathcal G_\theta\cVec(t,z).
\end{equation}
Therefore, replacing $b$ by its particle-phase average
\begin{equation}
 \frac{1}{2\pi}\int_0^{2\pi}
 b(\mathcal G_\theta z)\,d\theta
\end{equation}
leaves both sides of \eqref{eq:joint-Liouville-integral} unchanged. For the derivative term, the chain rule and equivariance give
\begin{equation}
 D(b\circ\mathcal G_\theta)(z)[\cVec(t,z)]
 =
 Db(\mathcal G_\theta z)
   [\cVec(t,\mathcal G_\theta z)].
\end{equation}
Its integral equals that of $Db(z)[\cVec(t,z)]$ by invariance of $\widetilde\nu_t$. Fubini's theorem applies because $Db$ is bounded and \eqref{eq:joint-scalar-velocity-bound} holds. It therefore suffices to treat $B$ as invariant under the particle phase.

On a real cube in $P_{\mathrm p}\h$ containing its unit ball, approximate $B$ in its particle coordinates by multivariate Bernstein polynomials. Their coefficients are smooth functions of $P_{\mathrm f}u$ supported in the same fixed compact set in the field coordinates. The approximations and their first derivatives in both groups of coordinates converge uniformly on the cube, uniformly in $P_{\mathrm f}u$. Convergence of particle derivatives follows from the finite-difference formula for derivatives of Bernstein polynomials and uniform continuity of the first derivatives of $B$; field derivatives commute with the finite polynomial approximation.

Average each approximating polynomial over the particle phase. Since rotations preserve the particle unit ball, this averaging preserves the stated $C^1$ convergence on that ball. The averaged polynomial is a finite sum of balanced monomials in the particle coordinates and their conjugates. Accordingly, it has the form
\begin{equation}\label{eq:joint-cylinder-polynomial-approximation}
 b_n(\vfi,u)
 =
 \sum_{\ell=1}^{L_n}
 b_{\OJ_\ell^{(k_\ell)}}(\vfi)c_\ell(P_{\mathrm f}u),
 \qquad
 c_\ell\in C_c^\infty(P_{\mathrm f}\h_f),
\end{equation}
where the $\OJ_\ell^{(k_\ell)}$ have finite rank on suitable symmetric tensor powers of $P_{\mathrm p}\h$. The constant monomial corresponds to $k=0$. Finite-dimensional Fourier inversion expresses each $c_\ell(P_{\mathrm f}u)$ as an integral of $\rchi_f(u)$, with $f\in P_{\mathrm f}\h_f$, whose Fourier weight has integrable product with $1+\|f\|$. Apply the Liouville identity already established for the generating symbols \eqref{eq:joint-generating-symbols}, and use Fubini's theorem to obtain \eqref{eq:joint-Liouville-integral} for each $b_n$. Indeed, $\|D\rchi_f(u)\|=2\|f\|$, and on $\|\vfi\|\leq1$,
\begin{equation}
 \left|
 D\bigl(b_{\OJ^{(k)}}\rchi_f\bigr)(z)
       [\cVec(t,z)]
 \right|
 \leq
 2\|\OJ^{(k)}\|(k+\|f\|)
 \|\cVec(t,z)\|.
\end{equation}
The Fourier transform of each $c_\ell$ is rapidly decreasing on the real finite-dimensional space $P_{\mathrm f}\h_f$, so its absolute value times $1+\|f\|$ is integrable. Together with \eqref{eq:joint-scalar-velocity-bound}, this bounds the absolute integral over the Fourier parameter, time, and phase space. Uniform $C^1$ convergence on the particle unit ball and all field coordinates gives numbers $\delta_n\to0$ such that
\begin{equation}
 |b_n-b|\leq\delta_n,
 \qquad
 |D(b_n-b)(z)[\cVec(t,z)]|
 \leq\delta_n\|\cVec(t,z)\|
\end{equation}
on the support of every $\widetilde\nu_t$. The error in passing from the integrated identity for $b_n$ to that for $b$ is bounded by
\begin{equation*}
\begin{aligned}
 &\left|\int(b_n-b)\,d\widetilde\nu_t\right|
 +\left|\int(b_n-b)\,d\widetilde\nu_s\right|\\
 &\qquad \qquad +
 \int_{\min(s,t)}^{\max(s,t)}
 \int
 \left|D(b_n-b)(z)[\cVec(r,z)]\right|
 \,d\widetilde\nu_r(z)\,dr\\
 &\qquad\leq
 \delta_n
 \left(
  2+\int_0^T\int
  \|\cVec(r,z)\|_{\mathscr Z}
  \,d\widetilde\nu_r(z)\,dr
 \right)
 \longrightarrow0.
\end{aligned}
\end{equation*}
Thus, \eqref{eq:joint-Liouville-integral} holds for $b$.

To cover a cylinder defined by an arbitrary finite family of continuous real linear functionals, take the complex spans of their representing vectors in each factor and let $P_{\mathrm p}$ and $P_{\mathrm f}$ be the orthogonal projections onto these spans. The test then factors through $(P_{\mathrm p}\vfi,P_{\mathrm f}u)$, although its representative may not have compact support in the added coordinates. That representative and its first derivative remain bounded. Multiply it by smooth cutoffs equal to one on balls of radius $R$ and with gradient bounded by $C/R$. For $R\geq1$, the resulting functions and first derivatives are uniformly bounded and converge pointwise to those of the original test. Dominated convergence, using \eqref{eq:joint-scalar-velocity-bound} for the derivative term, permits $R\to\infty$ and proves the identity for every cylindrical test in the statement. Multiplication by smooth scalar time tests and integration by parts gives the equivalent distributional Liouville equation.

\smallskip

\medskip\noindent\textbf{Narrow continuity for the weak topology.}
The measure construction in the source proof gives continuity against bounded continuous finite-coordinate tests. For $R>0$, set
\begin{equation}
 K_R
 :=
 \{\vfi\in\h:\|\vfi\|\leq1\}
 \times
 \{u\in\h_f:\|u\|\leq R\}.
\end{equation}
This set is compact for the weak topology, and \eqref{eq:joint-source-joint-moment-bounds} gives
\begin{equation}
 \sup_{0\leq t\leq T}\nu_t(K_R^c)
 \leq\frac{\cM_T}{R^2}.
\end{equation}
On $K_R$, the weak and coordinate topologies agree. The restrictions of bounded continuous finite-coordinate functions form a point-separating algebra containing the constants. By the Stone--Weierstrass theorem, they approximate every continuous function on $K_R$ uniformly.

For a bounded weakly continuous real-valued test on $\mathscr Z$, these approximations can be truncated to its bounded range, retaining their finite-coordinate dependence and giving uniform global bounds. Continuity against the approximations, followed by removal of the approximation error and then $R\to\infty$, proves continuity against every bounded weakly continuous test. Thus $t\mapsto\nu_t$ is narrowly continuous for the weak topology of $\mathscr Z$. To obtain the same conclusion for $\widetilde\nu_t$, first observe that $(t,z)\mapsto\cU_{fr}(0,t)z$ is jointly continuous for the weak topology on bounded sets. Indeed, for $t_n\to t$, bounded $z_n\rightharpoonup z$, and $a\in\mathscr Z$,
\begin{equation}
\begin{aligned}
 &\langle a,\cU_{fr}(0,t_n)z_n\rangle
 -\langle a,\cU_{fr}(0,t)z\rangle\\
 &\quad=
 \langle\cU_{fr}(t,0)a,z_n-z\rangle
 +\langle
   (\cU_{fr}(t_n,0)-\cU_{fr}(t,0))a,z_n
  \rangle.
\end{aligned}
\end{equation}
The first term tends to zero by weak convergence, and the second by strong continuity of the free group and boundedness of $z_n$.

The sets $K_R$ are weakly compact and metrizable because $\mathscr Z$ is a separable Hilbert space. Hence, on the compact product $[0,T]\times K_R$, composition of the free propagator with a bounded weakly continuous test is uniformly continuous. Its variation with $t$ is therefore uniform on $K_R$. Combining this fact with narrow continuity of $\nu_t$ for the weak topology, and then using the uniform tail estimate to let $R\to\infty$, proves the asserted continuity of $\widetilde\nu_t$.
\end{proof}

\subsection{Coupled flow, transport, and uniqueness}
\label{sec:uniqueness-gp-hierarchy}

We first construct the coupled flow in the underlying Hilbert space.
This also supplies uniqueness in the class of trajectories produced by
the superposition principle, without an additional regularity
assumption on those trajectories.
Throughout this subsection, the field coordinate is in momentum space.
We use the free propagator and interaction-picture vector field defined
in \eqref{eq:joint-free-group-and-nonlinearity}; in particular,
\begin{equation}
 \cVec(t,z)=-i\cU_{fr}(0,t)
                   \mathcal R(\cU_{fr}(t,0)z).
\end{equation}
The existing shift $\kappa$ in the definition of $Y^1$ is chosen
sufficiently large that $V+\kappa\geq1$. Such a choice is possible
because $V$ is bounded below. Increasing this shift does not change
the form domain $Y^1$ and changes its norm only to an equivalent norm.

\begin{lemma}[Global coupled flow]
\label{lem:global-joint-Hilbert-flow}
Under the standing assumptions, the coupled system
\begin{equation}\label{eq:joint-classical-coupled-system}
 \begin{cases}
 i\partial_t\vfi_t
   =\bigl(\dch_0+\lam(v*|\vfi_t|^2)+\eta F_{u_t}\bigr)\vfi_t,\\
 i\partial_tu_t=\omg u_t+\eta R(\vfi_t)
 \end{cases}
\end{equation}
has a unique global mild solution in $C(\R;\mathscr Z)$ for each
initial datum in $\mathscr Z$.  Write its flow as
$\mathcal U_t(\vfi_0,u_0)=(\vfi_t,u_t)$, with $u_0$ in this lemma
denoting an arbitrary momentum-space datum. For the microscopic
initial state of the paper, this argument of the flow is
$\widehat u_0$. The flow is jointly continuous and satisfies
\begin{equation}\label{eq:joint-flow-mass-and-field-bound}
 \|\vfi_t\|_{L^2} =\|\vfi_0\|_{L^2} ,\qquad
 \|u_t\|_{L^2} \leq\|u_0\|_{L^2} +
                    \eta|t|\|w\|_{L^2} \|\vfi_0\|_{L^2} ^2.
\end{equation}
\end{lemma}

\begin{proof}
The estimates
\begin{align}
 \|R(\vfi)-R(\psi)\|_{L^2} 
 &\leq\|w\|_{L^2} (\|\vfi\|_{L^2} +\|\psi\|_{L^2} )\|\vfi-\psi\|_{L^2} ,
 \label{eq:joint-source-Lipschitz}\\
 \|F_u-F_{\widetilde u}\|_{L^\infty} 
 &\leq2\|w\|_{L^2} \|u-\widetilde u\|_{L^2} ,
 \label{eq:joint-field-Lipschitz}\\
 \|(v*|\vfi|^2)\vfi-(v*|\psi|^2)\psi\|_{L^2} 
 &\leq\|v\|_{L^\infty} 
       (\|\vfi\|_{L^2} ^2+\|\vfi\|_{L^2} \|\psi\|_{L^2} +\|\psi\|_{L^2} ^2)
                                    \|\vfi-\psi\|_{L^2} 
 \label{eq:joint-Hartree-Lipschitz}
\end{align}
follows from Cauchy--Schwarz and
$\||\vfi|^2-|\psi|^2\|_1
\leq(\|\vfi\|_{L^2} +\|\psi\|_{L^2} )\|\vfi-\psi\|_{L^2} $.
Also,
\begin{equation}
 \|F_u\vfi-F_{\widetilde u}\psi\|_{L^2} 
 \leq2\|w\|_{L^2} 
  \bigl(\|u\|_{L^2} \|\vfi-\psi\|_{L^2} 
                   +\|\psi\|_{L^2} \|u-\widetilde u\|_{L^2} \bigr).
\end{equation}
Hence $\mathcal R:\mathscr Z\to\mathscr Z$ is locally Lipschitz,
and $\cVec(t,z)$ is jointly continuous and locally Lipschitz in $z$,
with Lipschitz constants on each ball independent of $t$.
Picard iteration for
\begin{equation}\label{eq:joint-interaction-ODE}
 \dot z(t)=\cVec(t,z(t)),\qquad z(0)=z_0,
\end{equation}
gives a unique maximal $C^1$ solution in the real Hilbert space
underlying $\mathscr Z$. Setting
$(\vfi_t,u_t)=\cU_{fr}(t,0)z(t)$ in its integral equation gives the mild equations
\begin{align}
 \vfi_t&=e^{-it\dch_0}\vfi_0
    -i\int_0^t e^{-i(t-s)\dch_0}
       \bigl(\lam(v*|\vfi_s|^2)+\eta F_{u_s}\bigr)\vfi_s\,ds,
 \label{eq:duhamel-particle}\\
 u_t&=e^{-it\omg}u_0-i\eta\int_0^t
                        e^{-i(t-s)\omg}R(\vfi_s)\,ds.
 \label{eq:duhamel-field}
\end{align}
All integrals are Bochner integrals in the corresponding Hilbert
spaces, and for $t<0$ they have their usual oriented meaning.
Conversely, applying $\cU_{fr}(0,t)$ to a continuous mild solution
gives the integral equation for \eqref{eq:joint-interaction-ODE}.
Thus, uniqueness of that ordinary differential equation proves
uniqueness among all continuous mild solutions in $\mathscr Z$.

For the interaction-picture particle component $p(t)=e^{it\dch_0}\vfi_t$,
its strong derivative gives
\begin{equation}
 \frac{d}{dt}\|p(t)\|_{L^2} ^2
 =2\re 
   \langle\vfi_t,-i(\lam(v*|\vfi_t|^2)+\eta F_{u_t})\vfi_t\rangle
 =0.
\end{equation}
The multiplier is real.  This proves particle mass conservation
without differentiating the unbounded free Hamiltonian.
The second mild equation and $\|R(\vfi)\|_{L^2} \leq\|w\|_{L^2} \|\vfi\|_{L^2} ^2$
give, for either sign of $t$,
\begin{equation}
 \|u_t\|_{L^2}
 \leq\|u_0\|_{L^2}
   +\eta\int_{\min(0,t)}^{\max(0,t)}\|R(\vfi_s)\|_{L^2}\,ds
 \leq\|u_0\|_{L^2}
   +\eta|t|\|w\|_{L^2}\|\vfi_0\|_{L^2}^2.
\end{equation}
This proves \eqref{eq:joint-flow-mass-and-field-bound}.
Thus, the solution norm cannot blow up on a finite interval.
The continuation criterion for \eqref{eq:joint-interaction-ODE}
proves global existence in both time directions.
These bounds are uniform for bounded sets of initial data.
Gronwall's inequality on such bounded solution sets gives continuous
dependence, locally uniformly in time, and hence joint continuity.
Autonomy of \eqref{eq:joint-classical-coupled-system} and uniqueness
also give $\mathcal U_{t+s}=\mathcal U_t\circ\mathcal U_s$.
Finally, $R(e^{i\theta}\vfi)=R(\vfi)$ and the particle equation
is equivariant under multiplication by $e^{i\theta}$.
Uniqueness therefore gives
$\mathcal U_t\circ\mathcal G_\theta
 =\mathcal G_\theta\circ\mathcal U_t$.
\end{proof}

\begin{theorem}[Persistence of the energy space]
\label{thm:gwp-strong-final}
Under the standing assumptions, if
$(\vfi_0,u_0)\in Y^1\times\cD(\omg^{1/2})$, then the solution in
\lemref{lem:global-joint-Hilbert-flow} belongs to
\begin{equation}
 C(\R;Y^1\times\cD(\omg^{1/2})),
\end{equation}
where the field-form domain carries its graph norm.  Under the Fourier
identification, this is the space $Y^1\times H^{1/2}$ used for the
coupled effective system.
\end{theorem}

\begin{proof}
Write $\mathscr E_f=\cD(\omg^{1/2})$, equipped with the graph norm
\begin{equation}
 \|u\|_{\mathscr E_f}^2
 =\|u\|_{L^2}^2+\|\omg^{1/2}u\|_{L^2}^2
 =\|(1+\omg)^{1/2}u\|_{L^2}^2.
\end{equation}
The standing condition $w\in H^{1/2}$ and
$\|\widehat{|\vfi|^2}\|_{L^\infty}\leq\|\vfi\|_{L^2}^2$ give
\begin{equation}
 \|R(\vfi)\|_{\mathscr E_f}
 \leq\|(1+\omg)^{1/2}\widehat w\|_{L^2}\|\vfi\|_{L^2}^2.
\end{equation}
The field free group is strongly continuous and isometric in this
graph norm. The source is continuous in this norm along every
continuous particle trajectory, by the difference estimate below.
Consequently, the field mild equation and conservation of particle
mass give, on any interval of existence,
\begin{equation}\label{eq:joint-field-energy-growth}
 \|u_t\|_{\mathscr E_f}
 \leq\|u_0\|_{\mathscr E_f}
   +\eta|t|\|(1+\omg)^{1/2}\widehat w\|_{L^2} \|\vfi_0\|_{L^2} ^2.
\end{equation}
Moreover, since $\omg(\xi)=|\xi|$,
\begin{equation}
\begin{aligned}
 \|\nabla F_u\|_{L^\infty}
 &\leq2\int_{\R^3}|\xi|\,|\widehat w(\xi)|\,|u(\xi)|\,d\xi\\
 &\leq2\|\omg^{1/2}\widehat w\|_{L^2}
                  \|\omg^{1/2}u\|_{L^2}.
\end{aligned}
\end{equation}
The integrable majorant in the first line justifies differentiating
the Fourier integral defining $F_u$.
Together with \eqref{eq:joint-source-classical-field}, this bounds
$\|F_u\|_{W^{1,\infty}}$ by $C_{w,1}\|u\|_{\mathscr E_f}$,
where $C_{w,1}$ depends only on $\|w\|_{H^{1/2}}$.
The same density estimate as in
\eqref{eq:joint-source-Lipschitz} gives
\begin{equation}
 \|R(\vfi)-R(\psi)\|_{\mathscr E_f}
 \leq\|(1+\omg)^{1/2}\widehat w\|_{L^2} 
       (\|\vfi\|_{L^2} +\|\psi\|_{L^2} )\|\vfi-\psi\|_{L^2} .
\end{equation}
For $a\in W^{1,\infty}$ the quadratic-form definition of $Y^1$
and the choice $V+\kappa\geq1$ give
\begin{equation}
 \|\vfi\|_{Y^1}^2
 =\|\nabla\vfi\|_{L^2}^2
    +\int_{\R^3}(V+\kappa)|\vfi|^2\,dx.
\end{equation}
The weak product rule therefore gives $a\vfi\in Y^1$ and
\begin{equation}
\begin{aligned}
 \|a\vfi\|_{Y^1}^2
 &\leq2\|a\|_{L^\infty}^2\|\nabla\vfi\|_{L^2}^2
       +2\|\nabla a\|_{L^\infty}^2\|\vfi\|_{L^2}^2\\
 &\quad+\|a\|_{L^\infty}^2
             \int_{\R^3}(V+\kappa)|\vfi|^2\,dx.
\end{aligned}
\end{equation}
In particular,
\begin{equation}
 \|a\vfi\|_{Y^1}\leq C\|a\|_{W^{1,\infty}}\|\vfi\|_{Y^1}.
\end{equation}
Here the assumptions include $v\in W^{1,\infty}$. Thus
\begin{equation}
\begin{aligned}
 \|v*|\vfi|^2\|_{W^{1,\infty}}
 &\leq\|v\|_{W^{1,\infty}}\|\vfi\|_{L^2}^2,\\
 \|v*(|\vfi|^2-|\psi|^2)\|_{W^{1,\infty}}
 &\leq\|v\|_{W^{1,\infty}}
  (\|\vfi\|_{L^2}+\|\psi\|_{L^2})\|\vfi-\psi\|_{L^2},\\
 \|F_u-F_{\widetilde u}\|_{W^{1,\infty}}
 &\leq C_{w,1}\|u-\widetilde u\|_{\mathscr E_f}.
\end{aligned}
\end{equation}
For example, split the Hartree difference as
\begin{equation}
 (v*|\vfi|^2)(\vfi-\psi)
       +\bigl(v*(|\vfi|^2-|\psi|^2)\bigr)\psi.
\end{equation}
The multiplier estimate applied to both terms proves local
Lipschitz continuity in $Y^1$; splitting
$F_u\vfi-F_{\widetilde u}\psi$
in the same way proves the corresponding mixed estimate.
Together with the weighted estimate for $R$, these inequalities
show that the full nonlinearity is locally Lipschitz on
$Y^1\times\mathscr E_f$. Both free groups preserve this space
and are strongly continuous isometries in its norm. Picard
iteration in interaction coordinates therefore gives local
solutions in that space, which coincide with the Hilbert-space
solution by uniqueness.

For every finite $T$, \eqref{eq:joint-field-energy-growth} and mass
conservation bound
$\|\lam(v*|\vfi_t|^2)+\eta F_{u_t}\|_{W^{1,\infty}}$ on $[-T,T]$.
Since $e^{-it\dch_0}$ preserves $Y^1$, the particle mild equation
and the multiplier estimate give, for $0\leq t\leq T$ in the
local interval of existence,
\begin{equation}
 \|\vfi_t\|_{Y^1}
 \leq\|\vfi_0\|_{Y^1}
   +C\int_0^t
      \|\lam(v*|\vfi_s|^2)+\eta F_{u_s}\|_{W^{1,\infty}}
      \|\vfi_s\|_{Y^1}\,ds.
\end{equation}
Gronwall gives a finite bound depending only on $T$ and the initial
norms. The same argument with reversed time applies on $[-T,0]$.
Together with \eqref{eq:joint-field-energy-growth}, this excludes
finite-time blowup in the energy norm. Continuation proves the
assertion.  This argument also accounts for the low-frequency
$L^2$ part of the massless field, which energy alone does not
control.
\end{proof}

\begin{theorem}[Transport and uniqueness of the joint measure]
\label{thm:joint-measure-transport}
Suppose that $t\mapsto\nu_t$ is a curve of probability measures
on $\mathscr Z$ for which
\begin{equation}
 \widetilde\nu_t=\bigl(\cU_{fr}(0,t)\bigr)_{\#}\nu_t
\end{equation}
is narrowly continuous for the weak topology of $\mathscr Z$ and satisfies
\eqref{eq:joint-Liouville-integral} for every smooth real cylindrical
test with compactly supported finite-dimensional representative.
Assume
\begin{equation}
 \nu_t(\mathbb S\h\times\h_f)=1,\qquad
 \sup_{0\leq t\leq T}\int\|u\|_{L^2}^2\,d\nu_t(\vfi,u)<\infty.
\end{equation}
Then
\begin{equation}\label{eq:joint-measure-transport}
 \boxed{\nu_t=(\mathcal U_t)_{\#}\nu_0,\qquad 0\leq t\leq T.}
\end{equation}
In particular, for the measures extracted from the microscopic states,
\begin{equation}\label{eq:joint-source-transport-law}
 \nu_t=(\mathcal U_t)_{\#}
                   (\mu_0\otimes\delta_{\widehat u_0}).
\end{equation}
\end{theorem}

\begin{proof}
The vector field $\cVec$ is Borel and bounded on bounded sets by
\lemref{lem:global-joint-Hilbert-flow}. Estimate
\eqref{eq:joint-velocity-growth} and the assumed field moment give
\begin{equation}
 \int_0^T\int\|\cVec(t,z)\|_{\mathscr Z}\,d\widetilde\nu_t(z)\,dt
 <\infty.
\end{equation}
We apply the infinite-dimensional superposition principle
\cite[Proposition~4.1]{ammari2016uniqueness}. In its Hilbert rigging,
take all three spaces equal to the complex Hilbert space
$\mathscr Z$, identifying its dual canonically, and use the real
differential structure of this space. We verify the topological
and endpoint details in this application.

The narrowly weak topology used in that reference is induced by
a weighted coordinate norm, which agrees with the weak Hilbert
topology on every norm-bounded ball and has the same Borel sets.
Our moment assumption gives the uniform tail estimate
\begin{equation}
 \sup_{0\leq t\leq T}
 \widetilde\nu_t\bigl(\{z:\|z\|_{\mathscr Z}>R\}\bigr)
 \leq\frac{1+\sup_{0\leq t\leq T}
                   \displaystyle\int\|u\|_{L^2}^2\,d\nu_t}{R^2}.
\end{equation}
On a bounded weakly compact ball, a bounded continuous function
for the weighted coordinate norm is uniformly approximated by
bounded continuous functions of finitely many coordinates.
The approximations can be truncated so that their global bounds
are uniform. Their integrals are continuous by the assumed weak
weak continuity. Removing the uniform approximation error and
then the tail proves continuity in the precise topology of the
cited proposition.

To apply its open-interval formulation with initial time zero,
extend $\widetilde\nu_t$ by $\nu_0$ on $(-1,0)$ and by
$\widetilde\nu_T$ on $(T,T+1)$, and set the vector field equal
to zero outside $[0,T]$. For every cylindrical test $b$ in the statement,
\eqref{eq:joint-Liouville-integral} says that
$t\mapsto\int b\,d\widetilde\nu_t$ is absolutely continuous on
$[0,T]$. Its constant extension is absolutely continuous on
$(-1,T+1)$, with derivative zero outside $[0,T]$.
Consequently, the distributional Liouville equation holds on
this open interval, with no endpoint contribution. The extended
vector field is Borel and bounded on bounded sets, and its scalar
velocity integral remains finite. Thus all hypotheses of the
superposition principle are satisfied.

Restrict the resulting paths to $[0,T]$ and project from the
initial-point/path pairs used in the cited proposition to the
path coordinate. This gives a probability measure $\mathbb P$
on continuous $\mathscr Z$-valued paths, concentrated on paths in
$W^{1,1}([0,T];\mathscr Z)$ satisfying
\begin{equation}\label{eq:joint-superposition-paths}
 \gamma(t)=\gamma(0)+\int_0^t\cVec(s,\gamma(s))\,ds,
 \qquad (e_t)_{\#}\mathbb P=\widetilde\nu_t.
\end{equation}
Here $e_t(\gamma)=\gamma(t)$. In the coincident Hilbert spaces used
above, the $W^{1,1}$ regularity means absolute continuity in the
norm of $\mathscr Z$, and the path identity holds for all
$t\in[0,T]$ on a set of full $\mathbb P$-measure. In particular,
$(e_0)_{\#}\mathbb P=\widetilde\nu_0=\nu_0$.

For $\mathbb P$-almost every path compare $\gamma(t)$ with
$\cU_{fr}(0,t)\mathcal U_t(\gamma(0))$.
By \lemref{lem:global-joint-Hilbert-flow}, the latter is a
$C^1$ solution of the same interaction-picture integral equation.
Both paths are norm-continuous on $[0,T]$, so their ranges lie
in a common ball of finite radius. On that ball $\cVec$ has a
Lipschitz constant independent of time. Subtracting their
integral equations gives
\begin{equation}
\begin{aligned}
 &\|\gamma(t)-\cU_{fr}(0,t)\mathcal U_t(\gamma(0))\|_{\mathscr Z}\\
 &\quad\leq C\int_0^t
    \|\gamma(s)-\cU_{fr}(0,s)\mathcal U_s(\gamma(0))\|_{\mathscr Z}\,ds.
\end{aligned}
\end{equation}
The constant may depend on these two bounded path ranges but
is independent of $t$. Gronwall therefore gives
\begin{equation}
 \gamma(t)=\cU_{fr}(0,t)\mathcal U_t(\gamma(0)),
 \qquad 0\leq t\leq T.
\end{equation}
Thus, for every bounded Borel function $B$,
\begin{equation}
\begin{aligned}
 \int B\,d\widetilde\nu_t
 &=\int B(\gamma(t))\,d\mathbb P(\gamma)\\
 &=\int B(\cU_{fr}(0,t)\mathcal U_t(z))\,d\nu_0(z).
\end{aligned}
\end{equation}
Applying $\bigl(\cU_{fr}(t,0)\bigr)_{\#}$ proves
\eqref{eq:joint-measure-transport}.
The initial law in \eqref{eq:joint-source-initial-law} gives
\eqref{eq:joint-source-transport-law}.

Conversely, start with a probability measure $\nu_0$ supported on
$\mathbb S\h\times\h_f$ and with finite second moment.
Particle mass conservation preserves this support, and
\eqref{eq:joint-flow-mass-and-field-bound} gives for its pushforward
\begin{equation}
 \sup_{0\leq t\leq T}\int\|u\|_{L^2}^2\,d\nu_t
 \leq2\int\|u\|_{L^2}^2\,d\nu_0
       +2\eta^2T^2\|w\|_{L^2}^2.
\end{equation}
For every fixed initial point, both the physical trajectory and
its interaction-picture trajectory are norm-continuous.
Integrating a bounded weakly continuous test along these
trajectories and applying dominated convergence proves narrow
weak continuity of both measure curves. This uses continuity
in time along each trajectory. The chain rule along the $C^1$
interaction-picture trajectories, followed by Fubini using the
scalar velocity bound and boundedness of the derivative of a
cylindrical test, proves \eqref{eq:joint-Liouville-integral}.
Hence, the pushforward is the unique solution in the stated class
for its prescribed initial joint measure.
\end{proof}

\begin{theorem}[Uniqueness of the joint hierarchy and convergence of the particle marginals]
\label{thm:uniqueness-of-gp-hierarchy}
Fix the initial joint measure
$\nu_0=\mu_0\otimes\delta_{\widehat u_0}$.
The joint characteristic hierarchy obtained by taking the limit in
\eqref{eq:finite-N-joint-characteristic-hierarchy} has at most one
particle-phase-invariant probability representation with this
initial law, particle norm one, a uniformly bounded field second
moment on $[0,T]$, and narrowly weak continuity of
$\bigl(\cU_{fr}(0,t)\bigr)_{\#}\nu_t$.
Here the limiting characteristic identities are required in
integrated form for the regular particle and field tests used in
\propref{prop:joint-compactness-and-Liouville-equation}, including
$k=0$. Their meaning is to drop the terms of order $N^{-1}$,
replace $(N-k)/N$ by one, and express the remaining terms by the
joint-measure integrals in
\eqref{eq:bounded-particle-mixed-limit},
\eqref{eq:weyl-weighted-field-insertion}, and
\eqref{eq:mixed-radiation-commutator-limit}.
For the microscopic limits constructed above, this representation
exists and equals $(\mathcal U_t)_{\#}\nu_0$.
Its particle marginals and radiation sources are consequently
determined uniquely by the initial joint measure.

If, for every fixed $k\geq1$, the initial particle marginals
converge along the full sequence weak-$*$ in $\fkS_1(\Hkb)$ to
the hierarchy represented by $\mu_0$, that is,
\begin{equation}
 \Tr\bigl(\OJ^{(k)}\gamma_{N,0}^{(k)}\bigr)
 \longrightarrow\int b_{\OJ^{(k)}}(\vfi)\,d\mu_0(\vfi)
 \quad\text{for every }\OJ^{(k)}\in\mathcal K(\Hkb),
\end{equation}
then all the
mixed characteristic functions and radiation-source pairings above
converge along the full sequence, uniformly on $[0,T]$ for each
fixed compact test and field parameter.  Moreover,
\begin{equation}\label{eq:joint-full-sequence-particle-convergence}
 \sup_{0\leq t\leq T}
 \|\gamma_{N,t}^{(k)}-\gamma_{\infty,t}^{(k)}\|_{\fkS_1}
 \longrightarrow0,\qquad k\geq1.
\end{equation}
\end{theorem}

\begin{proof}
Apply the free-evolved tests from
\eqref{eq:joint-free-evolved-tests} to a competing solution of the
integrated characteristic identities. Their use is justified by
time approximation in the graph norm for the field parameters
and in the operator norm, together with their commutators with
$\dch_k$, for the regular particle tests. 

For the field characters, we have
\begin{equation*}
 |\rchi_f(u)-\rchi_{f'}(u)|
 \leq 2\|f-f'\|\,\|u\|.
\end{equation*}
To include one field insertion, set $\ell_h(u):=2\re\langle h,u\rangle$, then splitting the difference gives
\begin{equation*}
\begin{aligned}
 &|\rchi_f(u)\ell_h(u)-\rchi_{f'}(u)\ell_{h'}(u)|\\
 &\quad\leq
 |\ell_h(u)-\ell_{h'}(u)|
 +|\rchi_f(u)-\rchi_{f'}(u)|\,|\ell_{h'}(u)|\\
 &\quad\leq
 2\|h-h'\|\,\|u\|
 +4\|h'\|\,\|f-f'\|\,\|u\|^2.
\end{aligned}
\end{equation*}
Writing
\begin{equation*}
 M:=\sup_{0\leq t\leq T}\int\|u\|^2\,d\nu_t<\infty,
\end{equation*}
we therefore obtain
\begin{equation*}
\begin{aligned}
 &\sup_{0\leq t\leq T}
 \int|\rchi_f(u)\ell_h(u)-\rchi_{f'}(u)\ell_{h'}(u)|\,d\nu_t\\
 &\qquad\leq
 2\sqrt M\,\|h-h'\|
 +4M\|h'\|\,\|f-f'\|.
\end{aligned}
\end{equation*}
For $f_t=e^{-it\omg}f$ with $f\in\cD(\omg)$, the map $t\mapsto f_t$ is continuous in the graph norm of $\omg$ and continuously differentiable in $\h_f$, with derivative $-i\omg f_t$. The regular particle tests $\OJ_t^{(k)}$ are continuously differentiable in operator norm, and their commutators with $\dch_k$ are norm continuous. These estimates justify piecewise time approximation in the integrated characteristic identities and the cancellation of the free terms for every competing measure solution.

The polynomial and Fourier approximation argument in
\propref{prop:joint-compactness-and-Liouville-equation} then yields
the Liouville identity for every smooth cylindrical test.
The resulting measure is unique by
\thmref{thm:joint-measure-transport}.  Its mixed moments and source
are then unique by their integral representations.

Under full-sequence initial convergence, every subsequence of the
microscopic family admits a further subsequence with a joint
limit, by the compactness proof above.  The initial measure of
each such limit is $\mu_0\otimes\delta_{\widehat u_0}$, and the
transport theorem identifies every limit with the same curve.
A subsequence contradiction proves the claimed full-sequence
convergence of each characteristic and source pairing, including
the uniformity in time supplied by the compactness and insertion
estimates: failure of uniform convergence for one fixed test would
give a subsequence with an error bounded below by a positive
constant, whereas its further convergent subsequence has the
unique limit just identified. The fixed coherent initial field
in \hyporef{hyp:microscopic-Hamiltonian-and-initial-data} is used
here to identify the same joint initial law from convergence of
the initial particle hierarchy alone.

To prove \eqref{eq:joint-full-sequence-particle-convergence}, set
$\Pi_d=P_d^{\otimes k}$.  For a positive trace-one operator
$\gamma$, write
\begin{equation}
 \gamma-\Pi_d\gamma\Pi_d
   =(1-\Pi_d)\gamma+\Pi_d\gamma(1-\Pi_d).
\end{equation}
Factoring each term through $\gamma^{1/2}$ and using the
Schatten H\"older inequality gives
\begin{equation}
\begin{aligned}
 \|\gamma-\Pi_d\gamma\Pi_d\|_{\fkS_1}
 &\leq\|(1-\Pi_d)\gamma^{1/2}\|_{\fkS_2}
                 \|\gamma^{1/2}\|_{\fkS_2}\\
 &\quad+\|\Pi_d\gamma^{1/2}\|_{\fkS_2}
                 \|\gamma^{1/2}(1-\Pi_d)\|_{\fkS_2}.
\end{aligned}
\end{equation}
Since $\Tr\gamma=1$, this implies
\begin{equation}
 \|\gamma-\Pi_d\gamma\Pi_d\|_{\fkS_1}
 \leq2\sqrt{\Tr((1-\Pi_d)\gamma)}.
\end{equation}
For both the finite-$N$ and limiting marginals, the right-hand
side is at most $2\sqrt{kC_E/\lambda_{d+1}}$.
On the finite-dimensional symmetric range of $\Pi_d$, convergence
of all matrix entries uniformly in time follows by testing with
the finitely many corresponding rank-one particle operators.
Equivalence of norms in finite dimension therefore gives
\begin{equation}
 \sup_{0\leq t\leq T}
 \|\Pi_d(\gamma_{N,t}^{(k)}-\gamma_{\infty,t}^{(k)})\Pi_d\|_{\fkS_1}
 \longrightarrow0
 \quad\text{for each fixed }d.
\end{equation}
The triangle inequality and the two cutoff errors now give
\begin{equation}
 \limsup_{N\to\infty}\sup_{0\leq t\leq T}
 \|\gamma_{N,t}^{(k)}-\gamma_{\infty,t}^{(k)}\|_{\fkS_1}
 \leq4\sqrt{\frac{kC_E}{\lambda_{d+1}}}.
\end{equation}
Compact resolvent implies $\lambda_{d+1}\to\infty$. Letting
$d\to\infty$ proves \eqref{eq:joint-full-sequence-particle-convergence}.
\end{proof}

\subsection{The limiting Gross--Pitaevskii hierarchy and its radiation source}
\label{sec:limit-of-radiation-source}

For each initial particle state $\vfi$, write
\begin{equation}\label{eq:joint-trajectory-with-fixed-field}
 (\vfi_t,u_t):=\mathcal U_t(\vfi,\widehat u_0),
 \qquad S_t(\vfi):=\vfi_t.
\end{equation}
Here $S_t$ denotes the particle projection of the flow with
the field fixed only at the initial time.  Restarts at an
intermediate time are performed with the pair $(\vfi_s,u_s)$:
$\mathcal U_{t-s}(\vfi_s,u_s)=(\vfi_t,u_t)$.
No particle-only restart identity is used.

\begin{lemma}[The radiation source expressed through the initial measure]
\label{lem:annihilation-as-transported-source-operator}
Under the standing assumptions, along the subsequence under
consideration the weak-$*$ limit of the radiation source is
\begin{equation}\label{eq:def-of-source-term-Sk-infty}
 \cS_{\infty,t}^{(k)}
 =
 \int_{\mathbb S\h}
 \left[
    \sum_{j=1}^kF_{\vfi_t}(x_j),
    |\vfi_t^{\otimes k}\rangle\langle\vfi_t^{\otimes k}|
 \right]\,d\mu_0(\vfi),
\end{equation}
where $F_{\vfi_t}:=F_{u_t}$ is evaluated on the coupled trajectory
\eqref{eq:joint-trajectory-with-fixed-field}.  Explicitly,
\begin{align}
 F_{\vfi_t}(x)&=2\operatorname{Re}\left(
      \cA_t^a(x)-i\eta\int_0^t
                 \cB_{t-s}^a(|\vfi_s|^2)(x)\,ds\right),
 \label{eq:retarded-radiation-field-def}\\
 \cA_t^a(x)&=\int_{\R^3}e^{i\xi\cdot x}
       e^{-it\omg(\xi)}\widehat w(\xi)\widehat u_0(\xi)\,d\xi,
 \label{eq:joint-retarded-free-field}\\
 \cB_\tau^a(\rho)(x)&=\int_{\R^3}e^{i\xi\cdot x}
        e^{-i\tau\omg(\xi)}|\widehat w(\xi)|^2\widehat\rho(\xi)\,d\xi.
 \label{eq:joint-retarded-sourced-field}
\end{align}
The convergence is uniform on $[0,T]$ against each compact particle
test.  It holds along the full sequence when the initial particle
hierarchy has a full-sequence limit.
\end{lemma}

\begin{proof}
Insert \eqref{eq:joint-source-transport-law} into the source
representation \eqref{eq:joint-source-limit-operator}.
The resulting integral is a trace-class Bochner integral:
its integrand has trace norm at most
\begin{equation}
 4k\|w\|_{L^2} \|u_t\|_{L^2} 
 \leq4k\|w\|_{L^2} (\|u_0\|_{L^2} +\eta T\|w\|_{L^2} )
\end{equation}
for $\|\vfi\|_{L^2} =1$, and depends continuously on the initial
pair.  The field mild equation reads
\begin{equation}\label{eq:alpha-t-def}
 u_t(\xi)=e^{-it\omg(\xi)}\widehat u_0(\xi)
      -i\eta\int_0^t e^{-i(t-s)\omg(\xi)}
                  \widehat w(\xi)\widehat{|\vfi_s|^2}(\xi)\,ds.
\end{equation}
Substitution into $F_{u_t}$ gives
\eqref{eq:retarded-radiation-field-def}.
All Fourier/time interchanges are justified by
$\|\widehat{|\vfi_s|^2}\|_\infty\leq1$,
$\widehat w\,\widehat u_0\in L^1(d\xi)$, and $|\widehat w|^2\in L^1(d\xi)$. Indeed, conservation of particle mass gives
\begin{equation*}
\begin{aligned}
 \int_{\R^3}|\widehat w(\xi)\widehat u_0(\xi)|\,d\xi
 &\leq \|w\|_{L^2}\|\widehat u_0\|_{L^2},\\
 \int_0^t\int_{\R^3}
 |\widehat w(\xi)|^2
 |\widehat{|\vfi_s|^2}(\xi)|\,d\xi\,ds
 &\leq T\|w\|_{L^2}^2,
 \qquad 0\leq t\leq T.
\end{aligned}
\end{equation*}
Since $w$ is real and even, $\widehat w^2=|\widehat w|^2$. The convergence assertions follow from
\eqref{eq:joint-source-uniform-limit} and
\thmref{thm:uniqueness-of-gp-hierarchy}.
\end{proof}

\begin{corollary}[The limiting Gross--Pitaevskii hierarchy]
\label{cor:explicit-gp-hierarchy}
The particle marginals determined by the joint dynamics satisfy
\begin{equation}\label{eq:GP-hierarchy-operator-form}
 i\partial_t\gamma_{\infty,t}^{(k)}
 =
 [\dch_k,\gamma_{\infty,t}^{(k)}]
 +\lam\sum_{j=1}^k\Tr_{k+1}
          [v_{j,k+1},\gamma_{\infty,t}^{(k+1)}]
 +\eta\cS_{\infty,t}^{(k)},
\end{equation}
with the source in \eqref{eq:def-of-source-term-Sk-infty}.
The equation is understood after pairing with regular finite-rank
tests and integrating in time, or equivalently in its mild form.
In the following formula, the diagonal notation
$\gamma_{\infty,t}^{(k+1)}(\ux_k,y;\ux_k',y)$
is understood through the de Finetti expression in
\eqref{eq:de-Finetti-representation-gamma-k-plus-one}.

The de Finetti expression defines an element of $L^1(\R^3_y;L^2(d\ux_k\,d\ux_k'))$. Indeed, Minkowski's inequality and conservation of particle mass give
\begin{equation*}
\begin{aligned}
 \int_{\R^3}
 \left\|
 \gamma_{\infty,t}^{(k+1)}(\cdot,y;\cdot,y)
 \right\|_{L^2(d\ux_k\,d\ux_k')}\,dy\leq
 \int_{\mathbb S\h}\int_{\R^3}
 |\vfi_t(y)|^2
 \|\vfi_t^{\otimes k}\|_{L^2}^2
 \,dy\,d\mu_0(\vfi)
 =1.
\end{aligned}
\end{equation*}
Since $|v(x_j-y)-v(x_j'-y)|\leq2\|v\|_{L^\infty}$, each collision integral is consequently well defined in $L^2(d\ux_k\,d\ux_k')$. With this convention, the kernel form is interpreted
distributionally and reads
\begin{align}
 i\partial_t\gamma_{\infty,t}^{(k)}(\ux_k;\ux_k')
 &=\sum_{j=1}^k(\dch_0(x_j)-\dch_0(x'_j))
                         \gamma_{\infty,t}^{(k)}(\ux_k;\ux_k')
 \nonumber\\
 &\quad+\lam\sum_{j=1}^k\int_{\R^3}
      (v(x_j-y)-v(x'_j-y))
      \gamma_{\infty,t}^{(k+1)}(\ux_k,y;\ux_k',y)\,dy
 \nonumber\\
 &\quad+\eta\sum_{j=1}^k\int_{\mathbb S\h}
     (F_{\vfi_t}(x_j)-F_{\vfi_t}(x'_j))
     \prod_{\ell=1}^k\vfi_t(x_\ell)\overline{\vfi_t(x'_\ell)}
                                               \,d\mu_0(\vfi).
 \label{eq:gp-hierarchy-kernel-form}
\end{align}
\end{corollary}

\begin{proof}
Set $f=0$ in the integrated limiting mixed identity.  Then
$\wyl_N(0)=1$, $m_0=0$, and the field-character derivative
terms vanish.  The remaining terms give
\eqref{eq:GP-hierarchy-operator-form}, with the source sign fixed
by $\Tr(\OJ^{(k)}[F,P])=\Tr([\OJ^{(k)},F]P)$.
Equivalently, conjugate by the free $k$-particle group and
integrate to obtain the mild identity.
The interaction terms are trace class: $v$ is bounded and the
source has the uniform trace bound already proved.
Writing the kernels of the commutators gives
\eqref{eq:gp-hierarchy-kernel-form}.
\end{proof}

\begin{corollary}[De Finetti representation along the coupled flow]
\label{cor:de-Finetti-form-of-limiting-GP-hierarchy}
The particle de Finetti measure and marginals satisfy
\begin{align}
 \mu_t&=(S_t)_{\#}\mu_0,
 \label{eq:transported-de-Finetti-measures-GP}\\
 \gamma_{\infty,t}^{(k)}
 &=\int_{\mathbb S\h}
       |\vfi_t^{\otimes k}\rangle\langle\vfi_t^{\otimes k}|
                                                \,d\mu_0(\vfi),
 \label{eq:de-Finetti-representation-gamma-k}\\
 \gamma_{\infty,t}^{(k+1)}(\ux_k,y;\ux_k',y)
 &=\int_{\mathbb S\h}|\vfi_t(y)|^2
       \prod_{\ell=1}^k\vfi_t(x_\ell)\overline{\vfi_t(x'_\ell)}
                                                \,d\mu_0(\vfi).
 \label{eq:de-Finetti-representation-gamma-k-plus-one}
\end{align}
In particular, the particle component on each trajectory satisfies
\begin{equation}\label{eq:nonlinear-retarded-Hartree-equation-x}
 i\partial_t\vfi_t
   =\bigl(\dch_0+\lam(v*|\vfi_t|^2)+\eta F_{\vfi_t}\bigr)\vfi_t,
\end{equation}
where $F_{\vfi_t}$ has the history-dependent meaning in \eqref{eq:retarded-radiation-field-def}. This equation is understood in the mild sense, through \eqref{eq:duhamel-particle} and the field representation \eqref{eq:alpha-t-def}.
\end{corollary}

\begin{proof}
Project \eqref{eq:joint-source-transport-law} onto the particle
coordinate and use the particle de Finetti representation.
The trajectory equation is the particle component of
\eqref{eq:joint-classical-coupled-system}; its retarded form follows
from \eqref{eq:alpha-t-def}.  Thus
\eqref{eq:transported-de-Finetti-measures-GP} is a conclusion of
joint transport, rather than an assumption used to prove it.
\end{proof}

\begin{corollary}[The coupled mean-field system for pure condensates]
\label{cor:coupled-mean-field-system}
If
$\gamma_{\infty,0}^{(k)}
=|\vfi_0^{\otimes k}\rangle\langle\vfi_0^{\otimes k}|$ for all $k$,
then
\begin{equation}
 \gamma_{\infty,t}^{(k)}
 =|\vfi_t^{\otimes k}\rangle\langle\vfi_t^{\otimes k}|,
 \qquad
 (\vfi_t,u_t)=\mathcal U_t(\vfi_0,\widehat u_0).
\end{equation}
In particular, the effective equations are
\begin{equation}\label{eq:coupled-system}
 \begin{cases}
 i\partial_t\vfi_t
   =\bigl(\dch_0+\lam(v*|\vfi_t|^2)+\eta F_{u_t}\bigr)\vfi_t,\\
 i\partial_tu_t(\xi)
   =\omg(\xi)u_t(\xi)
          +\eta\widehat w(\xi)\widehat{|\vfi_t|^2}(\xi),
 \end{cases}
\end{equation}
with
\begin{equation}\label{eq:total-classical-radiation-field}
 F_{u_t}(x)=2\operatorname{Re}\int_{\R^3}
            e^{i\xi\cdot x}\widehat w(\xi)u_t(\xi)\,d\xi.
\end{equation}
The inverse Fourier transform of $u_t$ is the position-space
field used in the introductory coupled system. The system is understood in the mild sense established in \lemref{lem:global-joint-Hilbert-flow}.
\end{corollary}

\begin{proof}
In the phase-invariant de Finetti convention, the initial measure
is the uniform measure on
$\{e^{i\theta}\vfi_0:0\leq\theta<2\pi\}$.
The coupled flow satisfies
$\mathcal U_t(e^{i\theta}\vfi_0,\widehat u_0)
=(e^{i\theta}\vfi_t,u_t)$ by its equations and uniqueness.
The projector and the field are therefore constant on this
initial phase orbit.  The preceding integral representations
reduce to the displayed factorized solution.
A Dirac mass at $\vfi_0$ gives the same integrals for phase-invariant observables, but the phase-invariant representing measure on the sphere is the uniform phase-orbit measure above.
\end{proof}

\appendix
\section{Self-adjointness of the Hamiltonian}
\label{app:self-adjointness}

In this appendix, we rigorously establish that the Hamiltonian $\HMN$ is self-adjoint and generates a well-defined unitary propagator, ensuring that the quantum dynamics are mass-preserving and globally well-posed. 

\begin{proposition}[Self-adjointness and unitary dynamics]
\label{prop:appendix-self-adjointness}
Assume \hyporef{hyp:standing-hypotheses}.
For every fixed $N\geq1$, the Hamiltonian $\HMN$
is self-adjoint on
\begin{equation}
    \mathcal D_0
    :=
    \cD\left(\sum_{j=1}^N\dch_0^{(j)}\right)
    \mcap\cD(\Hf)
\end{equation}
in $\HN$, with the operators acting on their
corresponding tensor factors.
The same Hamiltonian expression defines a
self-adjoint operator on
$L^2(\R^{3N})\otimes\Fock$, whose restriction to
the reducing bosonic subspace is the operator above.

Consequently,
\begin{equation}
    \UN(t,s)=e^{-i(t-s)\HMN}
\end{equation}
is a strongly continuous unitary propagator.
For every $\Psi_{N,0}\in\HN$, the evolution
$\Psi_{N,t}=\UN(t,0)\Psi_{N,0}$ is norm preserving.
If $\Psi_{N,0}\in\mathcal D_0$, it satisfies
$i\partial_t\Psi_{N,t}=\HMN\Psi_{N,t}$
as a strong equation in $\HN$.
\end{proposition}

\begin{proof}
We first work on the full labeled space
$L^2(\R^{3N})\otimes\Fock$.
Suppressing tensor-product identities, write
\begin{equation}
    \HMN
    =
    \mathcal H_{\mathrm{free}}
    +\mathcal V_{\mathrm{int}}
    +\mathcal W_{\mathrm{field}},
\end{equation}
where
\begin{equation}
\begin{aligned}
    \mathcal H_{\mathrm{free}}
    &:=
    \sum_{j=1}^N\dch_0^{(j)}+\Hf,
    \\
    \mathcal V_{\mathrm{int}}
    &:=
    \frac{\lam}{N}
       \sum_{1\leq i<j\leq N}v(x_i-x_j),
    \\
    \mathcal W_{\mathrm{field}}
    &:=
    \frac{\eta}{\sqrt N}
       \sum_{j=1}^N\phi(w(x_j-\cdot)).
\end{aligned}
\end{equation}

The particle and field free Hamiltonians are
nonnegative, self-adjoint, and strongly commute.
Their joint spectral calculus therefore shows that
$\mathcal H_{\mathrm{free}}$ is self-adjoint on
\begin{equation}
    \cD(\mathcal H_{\mathrm{free}})
    =
    \cD\left(\sum_{j=1}^N\dch_0^{(j)}\right)
    \mcap\cD(\Hf).
\end{equation}
The particle interaction is bounded and symmetric,
with
\begin{equation}
    \|\mathcal V_{\mathrm{int}}\|_{\mathrm{op}}
    \leq
    \frac{\lam(N-1)}2\|v\|_{L^\infty}.
\end{equation}

For the field interaction, the creation and
annihilation estimates in
\lemref{lem:field-energy-operator-bound},
applied fiberwise in the particle coordinates,
give
\begin{equation}
\begin{aligned}
    \|\mathcal W_{\mathrm{field}}\Psi\|
    &\leq
    2\eta\sqrt N
       \left\|\frac{\widehat w}{\sqrt{\omg}}\right\|_{L^2}
       \|\Hf^{1/2}\Psi\|
    \\
    &\quad+
    \eta\sqrt N\,\|w\|_{L^2}\|\Psi\|,
    \qquad
    \Psi\in\cD(\Hf^{1/2}).
\end{aligned}
\end{equation}
The form-factor norms are uniform in the particle
positions, since translation changes the Fourier
transform only by a phase, using the evenness of $w$.

Now take $\Psi\in\cD(\mathcal H_{\mathrm{free}})$.
Nonnegativity gives
\begin{equation}
    \|\Hf^{1/2}\Psi\|^2
    \leq
    \langle\Psi,\mathcal H_{\mathrm{free}}\Psi\rangle
    \leq
    \|\Psi\|\,
    \|\mathcal H_{\mathrm{free}}\Psi\|.
\end{equation}
Young's inequality therefore yields, for every
$\varepsilon>0$,
\begin{equation}
\begin{aligned}
    \|(\mathcal V_{\mathrm{int}}
         +\mathcal W_{\mathrm{field}})\Psi\|
    &\leq
    \varepsilon\|\mathcal H_{\mathrm{free}}\Psi\|
    \\
    &\quad+
    \left(
       \frac{\eta^2N}{\varepsilon}
       \left\|\frac{\widehat w}{\sqrt{\omg}}\right\|_{L^2}^2
       +\eta\sqrt N\,\|w\|_{L^2}
       +\frac{\lam(N-1)}2\|v\|_{L^\infty}
    \right)\|\Psi\|.
\end{aligned}
\end{equation}
Thus, for each fixed $N$, the symmetric perturbation
is infinitesimally operator bounded relative to
$\mathcal H_{\mathrm{free}}$.

The Kato--Rellich theorem
\cite{kato1966perturbation}
implies that $\HMN$ is self-adjoint on
$\cD(\mathcal H_{\mathrm{free}})$.
The relative bound also gives equivalence of the
graph norms of $\HMN$ and
$\mathcal H_{\mathrm{free}}$.
Consequently, every operator core of the free
Hamiltonian remains an operator core of $\HMN$.

Particle permutations preserve this domain and
commute with the complete Hamiltonian.
Here evenness of $v$ ensures invariance of the
pair interaction, and the field interaction is
summed over all particle labels.
Hence the bosonic subspace reduces $\HMN$.
Its restriction is self-adjoint on the corresponding
bosonic domain $\mathcal D_0$.

Stone's theorem now gives the strongly continuous
unitary propagator and norm conservation.
For initial data in $\mathcal D_0$, invariance
of the operator domain under the unitary group
gives the asserted strong Schr\"odinger equation.
\end{proof}

\section{Proof of \texorpdfstring{\lemref{lem:coherent-oscillator-moments}}{Lemma}}
\label{sec:coherent-oscillator-moments}

\begin{proof}
Constants may depend on the matrix $A$, the parameters $E_*$ and $a_*$,
the fixed couplings, and the norms in the hypotheses, but not on $N$
or the moment order $k$. Identity operators on tensor factors will be
suppressed. We first establish the estimates with Fourier cutoffs; the
removal of the cutoffs and all domain statements are justified at the
end of the proof.

\smallskip
\noindent\emph{Positivity and the Weyl transform.}
Put $c_w=\norm{\widehat w/\sqrt\omega}_{L^2}$. The standard
annihilation-operator estimate, applied at fixed particle coordinates,
gives
\begin{equation}
\label{eq:moment-field-form-estimate}
\begin{aligned}
    \left|\left\langle\psi,\frac\eta{\sqrt N}W_N\psi
         \right\rangle\right|
    &\leq2\eta\sqrt N\,c_w\norm{\Hf^{1/2}\psi}\norm\psi\\
    &\leq\tfrac12\norm{\Hf^{1/2}\psi}^2
           +2\eta^2Nc_w^2\norm\psi^2.
\end{aligned}
\end{equation}
For completeness, the annihilation estimate follows on each photon
sector from Cauchy--Schwarz with weight $\omega(\xi)$; integration over
the remaining variables yields
$\norm{a(g)\psi}\leq\norm{g/\sqrt\omega}_{L^2}
\norm{\Hf^{1/2}\psi}$.
The same proof applies to the multiplication-valued kernel
$\sum_j e^{-i\xi\cdot x_j}\widehat w(\xi)$, whose absolute value is
at most $N|\widehat w(\xi)|$.
Since $V\geq0$ and
$\lam N^{-1}\sum_{i<j}v(x_i-x_j)\geq-\lam N\norm v_{L^\infty}/2$,
\begin{equation}
    A_N\geq\sum_{j=1}^N\dch_0^{(j)}+\tfrac12\Hf
       +N\left(1-\tfrac\lam2\norm v_{L^\infty}
                   -2\eta^2c_w^2\right)\geq0.
\end{equation}
In the notation of the standing weak-coupling condition, $c_w\leq C_w$;
thus $2\lam\norm v_{L^\infty}+4\eta C_w<1$ implies
\eqref{eq:moment-positivity-condition}. The Weyl identities give
\begin{equation}
\label{eq:moment-G-decomposition}
    G_N=D_N+R_N+\mathcal L_N,
    \qquad D_N:=\mathbb K_N+\Hf+Nb_0,
\end{equation}
where
\begin{equation}
\label{eq:moment-bounded-particle-part}
\begin{aligned}
    b_0&:=1+\int_{\R^3}\omega(\xi)|\widehat u_0(\xi)|^2\,d\xi,\\
    c_{u_0}(x)&:=2\eta\re (w*u_0)(x),
      \qquad B(x):=U(x)+c_{u_0}(x),\\
    R_N&:=\sum_{j=1}^N B(x_j)
            +\frac\lam N\sum_{1\leq i<j\leq N}v(x_i-x_j),
\end{aligned}
\end{equation}
and
\begin{equation}
\label{eq:moment-linear-field-part}
\begin{aligned}
    \mathcal L_N
      &=\sqrt N\int_{\R^3}
            \bigl(F_N(\xi)\ad_\xi+F_N(\xi)^*a_\xi\bigr)\,d\xi,\\
    F_N(\xi)
      &:=\omega(\xi)\widehat u_0(\xi)
            +\frac\eta N\widehat w(\xi)
                  \sum_{j=1}^N e^{-i\xi\cdot x_j}.
\end{aligned}
\end{equation}
Here $F_N(\xi)$ is a multiplication operator on the particle space.
In particular, the original interaction $W_N$ is unscaled in
\eqref{eq:moment-microscopic-Hamiltonian}; its factor $\eta/\sqrt N$
has been included in \eqref{eq:moment-linear-field-part}.

Define the nonnegative function and finite measures
\begin{equation}
\label{eq:moment-majorant-measures}
\begin{aligned}
    f(p)&:=\omega(p)|\widehat u_0(p)|+\eta|\widehat w(p)|,\\
    d\nu(p)&:=f(p)^2\,dp,\\
    d\mu(q)&:=\left(|\widehat B(q)|
                   +\tfrac\lam2|\widehat v(q)|\right)dq.
\end{aligned}
\end{equation}
The Fourier transform of the real part in $c_{u_0}$ involves a function
and its reflected complex conjugate. Cauchy--Schwarz therefore gives
\begin{equation}
\begin{aligned}
    \int e^{a_*|q|}|\widehat c_{u_0}(q)|\,dq
    \leq2\eta
       \norm{e^{a_*|\cdot|/2}\widehat w}_{L^2}
       \norm{e^{a_*|\cdot|/2}\widehat u_0}_{L^2}.
\end{aligned}
\end{equation}
Consequently,
\begin{equation}
\label{eq:moment-majorant-exponential}
    \int e^{a_*|q|}\,d\mu(q)
       +\int e^{a_*|p|}\,d\nu(p)<\infty.
\end{equation}
Also $f,f/\sqrt\omega\in L^2$. Indeed, the only infrared term
requiring a separate assumption is
$\widehat w/\sqrt\omega\in L^2$, supplied by the hypotheses on $w$.

\smallskip
\noindent\emph{Particle estimates for fixed momentum transfers.}
After an orthogonal change of coordinates, let the eigenvalues of $A$
be $\gamma_1^2,\gamma_2^2,\gamma_3^2$ with $\gamma_\ell>0$.
In the Hermite basis of the reference oscillator,
\begin{equation}
    \mathbb K_N=\sum_{j=1}^N\sum_{\ell=1}^3
             \gamma_\ell(2b_{j,\ell}^*b_{j,\ell}+1).
\end{equation}
Let $\mathcal E_M$ be the span of the $N$-particle Hermite functions
with total excitation number at most $M$. These spaces are considered
in the full labeled particle space, since individual Fourier factors
need not preserve symmetry. For $\chi\in\mathcal E_M$,
\begin{equation}
\label{eq:moment-Hermite-estimates}
    \norm{\mathbb K_N\chi}\leq C_A(N+M)\norm\chi,
    \qquad
    \norm{\nabla_X\chi}\leq C_A(N+M)^{1/2}\norm\chi,
\end{equation}
where $X=(x_1,\ldots,x_N)$.
The first estimate follows from the displayed diagonal form; the
second follows from $-\Delta_X\leq\mathbb K_N$ as quadratic forms.
Each derivative maps $\mathcal E_M$ into $\mathcal E_{M+1}$.
The cutoff \eqref{eq:moment-initial-reference-cutoff} implies
$\phi_{N,0}\in\mathcal E_{c_*N}$ for a constant $c_*$ depending only
on $A$ and $E_*$, with integer parts understood.

For $P\in\R^{3N}$, set $T_P\chi(X)=e^{-iP\cdot X}\chi(X)$.
Then
\begin{equation}
\label{eq:moment-translated-oscillator}
    T_P^*\mathbb K_NT_P
       =\mathbb K_N+2iP\cdot\nabla_X+|P|^2.
\end{equation}
If $M\leq c_*N+k$, $s\geq0$, and $\chi\in\mathcal E_M$, it follows
that
\begin{equation}
\label{eq:moment-fixed-transfer-bound}
\begin{aligned}
    \norm{(T_P^*\mathbb K_NT_P+Nb_0+s)\chi}
      &\leq C(N+k+|P|^2+s)\norm\chi,\\
    (T_P^*\mathbb K_NT_P+Nb_0+s)\chi
      &\in\mathcal E_{M+1}.
\end{aligned}
\end{equation}
This estimate will be applied separately to each fixed sequence of
Fourier momenta. No operator norm bound on a span of differently
translated Hermite spaces is used.

\smallskip
\noindent\emph{Vacuum contractions and field energies.}
Expand $G_N^k$ using \eqref{eq:moment-G-decomposition}, retaining the
order of every particle operator. A nonzero vacuum expectation contains
an even number $2r$ of factors $\mathcal L_N$, with
$0\leq r\leq\lfloor k/2\rfloor$. Choose their positions and contract
their field operators. There are at most
\begin{equation}
\label{eq:moment-pairing-count}
    \binom{k}{2r}(2r-1)!!
       =\frac{k!}{2^r r!(k-2r)!}
\end{equation}
choices, with the value $1$ when $r=0$.
For each pair, its left endpoint is an annihilation operator and its
right endpoint is a creation operator. These are the only orientations
that contribute to a vacuum contraction. One can obtain the expansion
by repeatedly using the canonical commutation relations, or directly
by applying the word to the vacuum from right to left and recording
which existing photon each annihilation removes.

Leave each factor $\Hf$ at its original position. In a fixed complete
contraction, its value is
\begin{equation}
\label{eq:moment-field-shifts}
    s=\sum_{i:\,\text{pair }i\text{ crosses that position}}\omega(p_i),
    \qquad 0\leq s\leq\sum_{i=1}^r|p_i|,
\end{equation}
where $p_i$ is the momentum of the $i$-th contracted pair.
To verify this, start with the vacuum state. A creation operator inserts a photon
with energy $\omega(p_i)$, the corresponding annihilation operator removes it,
and $\Hf$ multiplies by the sum of the energies of the photons currently
present. This also proves the pull-through rule in the presence of
several factors $\Hf$. Particle operators retain their original order
throughout this procedure.

Suppose that the remaining $k-2r$ positions contain $b$ factors $R_N$
and $d$ factors $D_N$, so that $2r+b+d=k$.
Expand each coefficient $F_N(p)$ as the sum of its scalar term and its
$N$ plane-wave terms. The sum of the absolute values of these
coefficients is at most $f(p)$. After pairing, the integration for each
pair is therefore bounded by $d\nu(p)$, and the $2r$ factors $\sqrt N$
give $N^r$.
For a factor $R_N$, the sum of the absolute Fourier coefficients over
all particle labels is bounded by
\begin{equation}
    N|\widehat B(q)|\,dq
      +\frac{\lam}{N}\frac{N(N-1)}2|\widehat v(q)|\,dq
    \leq N\,d\mu(q).
\end{equation}
Thus the $b$ particle-potential factors give $N^b$ and $b$ integrations
against $\mu$.

For fixed momenta $p_1,\ldots,p_r,q_1,\ldots,q_b$, each field vertex
transfers either zero momentum or momentum of norm $|p_i|$ to one
particle. Each pair uses $p_i$ at two vertices. A one-particle potential
transfers momentum of norm $|q_j|$, and a pair potential transfers
$(q_j,-q_j)$ to two particles, of norm $\sqrt2|q_j|$ in $\R^{3N}$.
Every cumulative particle momentum $P$ consequently satisfies
\begin{equation}
\label{eq:moment-total-transfer}
    |P|\leq S,
    \qquad
    S:=2\sum_{i=1}^r|p_i|+\sqrt2\sum_{j=1}^b|q_j|.
\end{equation}
The field energy in \eqref{eq:moment-field-shifts} is also at most $S$.
Apply the particle word from right to left. After each plane wave, the
vector has the form $T_P\chi$; after each free factor, use
\eqref{eq:moment-translated-oscillator}. At most $d\leq k$ free factors
occur, so the Hermite excitation number never exceeds $c_*N+k$.
Equations \eqref{eq:moment-fixed-transfer-bound} and
\eqref{eq:moment-total-transfer} bound the absolute value of the
resulting particle matrix element by
\begin{equation}
\label{eq:moment-word-integrand}
    C^d(N+k+S^2)^d.
\end{equation}
Here $S\leq1+S^2$ absorbs the field-energy term. The plane-wave factors
are unitary, and the initial and final particle vectors have norm one.

\smallskip
\noindent\emph{Momentum integrations and summation of the words.}
Choose $\varepsilon=a_*/2$. By
\eqref{eq:moment-majorant-exponential}, there is $M_*\geq1$ such that
\begin{equation}
\label{eq:moment-exponential-transfer}
    \int e^{\varepsilon S}
       \prod_{i=1}^r d\nu(p_i)\prod_{j=1}^b d\mu(q_j)
       \leq M_*^{r+b}.
\end{equation}
The integral is interpreted as evaluation at $S=0$ if $r=b=0$.
For $d\geq1$, use
\begin{equation}
    (N+k+S^2)^d\leq2^{d-1}\bigl((N+k)^d+S^{2d}\bigr),
    \qquad
    S^{2d}\leq (2d)!\,\varepsilon^{-2d}e^{\varepsilon S}.
\end{equation}
Since $d\leq k$, $(2d)!\leq(2k)^{2d}$ and $N+k\leq N+k^2$.
The case $d=0$ follows directly from
\eqref{eq:moment-exponential-transfer}. In both cases,
\begin{equation}
\label{eq:moment-integrated-word}
    \int (N+k+S^2)^d
       \prod_{i=1}^r d\nu(p_i)\prod_{j=1}^b d\mu(q_j)
       \leq C^k(N+k^2)^d.
\end{equation}
Combining the preceding estimates, each fixed placement and pairing
of the field vertices, with $b$ particle-potential positions, contributes
in absolute value at most
\begin{equation}
    C^kN^{r+b}(N+k^2)^{k-2r-b}.
\end{equation}
Write $L_k=N+k^2$. Summing over the potential positions gives
\begin{equation}
    \sum_{b=0}^{k-2r}\binom{k-2r}{b}N^bL_k^{k-2r-b}
       =(N+L_k)^{k-2r}\leq(2L_k)^{k-2r}.
\end{equation}
Using \eqref{eq:moment-pairing-count}, and increasing $C$ to absorb
the powers of $2$, we obtain
\begin{equation}
\label{eq:moment-final-pairing-sum}
\begin{aligned}
    |\langle\Phi_N,G_N^k\Phi_N\rangle|
    &\leq C^kL_k^k
       \sum_{r=0}^{\lfloor k/2\rfloor}
         \frac{k!}{2^r r!(k-2r)!}
               \left(\frac{N}{L_k^2}\right)^r\\
    &\leq C^kL_k^k
       \exp\left(\frac{k^2N}{2(N+k^2)^2}\right)
     \leq e^{1/8}C^kL_k^k.
\end{aligned}
\end{equation}
For the second line, use $k!/(k-2r)!\leq k^{2r}$ and the exponential
series. For the last inequality, use
$Nk^2/(N+k^2)^2\leq1/4$. The numerical factor $e^{1/8}$ can be
absorbed into $C^k$.

\smallskip
\noindent\emph{Removal of cutoffs and domains of powers.}
We give the details to justify the expansions without assuming the
desired energy moments. Let $\Lambda>0$. In
\eqref{eq:moment-G-decomposition} replace
$\widehat B(q),\widehat v(q)$ and $F_N(p)$ by their restrictions to
$|q|\leq\Lambda$ and $|p|\leq\Lambda$, respectively, leaving $D_N$
unchanged. Denote the resulting operator by $G_{N,\Lambda}$.
The cutoffs are symmetric, so the particle potentials remain real
and the field interaction remains symmetric.
All estimates above hold with constants independent of $\Lambda$,
because their nonnegative majorant measures are simply restricted.

For fixed $N$, both $G_N$ and $G_{N,\Lambda}$ are self-adjoint on
\begin{equation}
\label{eq:moment-common-domain}
    \cD(T_N),\qquad T_N:=\mathbb K_N+\Hf+1.
\end{equation}
Indeed, $R_N$ and its cutoffs are bounded. The field estimates give for either the original or a cutoff field interaction,
\begin{equation}
    \norm{\mathcal L_N\psi}
       \leq C\sqrt N\bigl(
            \norm{f/\sqrt\omega}_{L^2}\norm{\Hf^{1/2}\psi}
             +\norm f_{L^2}\norm\psi\bigr).
\end{equation}
The creation estimate used here follows from the annihilation estimate
and the canonical commutation relations, at each fixed $X$.
Since $\Hf^{1/2}$ is infinitesimally bounded relative to $T_N$,
self-adjointness on \eqref{eq:moment-common-domain} follows from the
Kato--Rellich theorem. The Weyl identities above identify this $G_N$
with $\mathcal W_N^*A_N\mathcal W_N$: the profile condition implies
$\omega\widehat u_0\in L^2$, and the induced particle multiplication
term is the bounded function $c_{u_0}$.

For each $\Lambda$, $\Phi_N$ lies in the domain of every power of
$G_{N,\Lambda}$. In fact, the reference cutoff makes $\phi_{N,0}$ a
finite linear combination of Hermite functions. Repeated applications
of $G_{N,\Lambda}$ produce only finitely many photon sectors; after
$m$ applications there are at most $m$ photons, all with momenta in
$\{|p|\leq\Lambda\}$. In the particle variables the resulting vectors
are Schwartz functions with values in these photon sectors. To check
this last assertion, differentiation of a Fourier coefficient costs
only powers of the bounded momenta, while multiplication by a compactly
supported Fourier integral has bounded derivatives of every order.
Polynomial weights in $X$ are preserved by these multipliers, by
derivatives, and by the quadratic reference potential. The same
argument in $L^2$ for the photon variables applies to the field
coefficients. These vectors lie in $\cD(T_N)$ at every stage.
Thus all the finite-word expansions and contractions used above are
legitimate for $G_{N,\Lambda}$.

Dominated convergence in the Fourier $L^1$ norms and in the field
norms $L^2\mcap L^2(\omega^{-1}dp)$ yields, for each fixed $N$,
\begin{equation}
    \norm{(G_{N,\Lambda}-G_N)T_N^{-1}}\longrightarrow0.
\end{equation}
The common-domain assertion implies that
$T_N(G_N-i)^{-1}$ is bounded, by the closed graph theorem.
The resolvent identity, together with the self-adjoint resolvent bound,
therefore shows norm-resolvent convergence to $G_N$.
Let $\sigma_{N,\Lambda}$ and $\sigma_N$ denote the scalar spectral
measures of $G_{N,\Lambda}$ and $G_N$ in $\Phi_N$. They have total
mass one. Resolvent convergence implies convergence of their
integrals against continuous functions of compact support.
For every integer $m\geq1$, the even-order estimate already proved
gives
\begin{equation}
\label{eq:moment-uniform-even-cutoff-bound}
    \int |t|^{2m}\,d\sigma_{N,\Lambda}(t)
       \leq C^{2m}\bigl(N+4m^2\bigr)^{2m}.
\end{equation}
Testing first with continuous compactly supported truncations of
$|t|^{2m}$ and then increasing the truncation shows the same bound
for $\sigma_N$. Hence $\Phi_N\in\cD(G_N^m)$ for every $m$.

For a fixed $k$, choose $2m>k$. The estimate
\begin{equation}
    \int_{|t|>R}|t|^k\,d\sigma_{N,\Lambda}(t)
       \leq R^{k-2m}\int |t|^{2m}\,d\sigma_{N,\Lambda}(t)
\end{equation}
and its counterpart for $\sigma_N$ prove uniform integrability.
It follows that
\begin{equation}
    \langle\Phi_N,G_{N,\Lambda}^k\Phi_N\rangle
       \longrightarrow\langle\Phi_N,G_N^k\Phi_N\rangle.
\end{equation}
Thus \eqref{eq:moment-final-pairing-sum} passes to $G_N$ for every
$k$, with the original constants and moment order. Positivity of the
cutoff operators was not needed. Positivity of $G_N$, proved above,
allows the absolute value to be removed. Unitary equivalence gives
the asserted domain statements for $\Psi_{N,0}=\mathcal W_N\Phi_N$
and the equality in \eqref{eq:coherent-oscillator-moment-bound}.

\smallskip
\noindent\emph{Factorials and the logarithmic range.}
Finally,
\begin{equation}
    (N+k^2)^k\leq2^{k-1}\bigl(N^k+k^{2k}\bigr),
    \qquad k^{2k}\leq e^{2k}(k!)^2.
\end{equation}
The second inequality follows from $k!\geq(k/e)^k$.
This proves the last inequality in
\eqref{eq:coherent-oscillator-moment-bound} after increasing $D$.
For $k\geq2$, $N^{k-2}\geq1$, which gives
\eqref{eq:coherent-oscillator-N-factorial-bound}.
Dividing the first bound by $N^k$ proves
\eqref{eq:coherent-oscillator-normalized-bound}.
\end{proof}

\section{Proof of \texorpdfstring{\propref{prop:uniform-prior-bounds-on-BBGKY-hierarchy}}{Proposal}}
\label{sec:prf-prop:uniform-prior-bounds-on-BBGKY-hierarchy}

\begin{proof}[Proof of  \propref{prop:uniform-prior-bounds-on-BBGKY-hierarchy}{:}]
We proceed by the two-step induction over $k \geq 0$. The base case $k=0$ holds trivially. For $k=1$, let $\Psi$ be a bosonic vector in the common quadratic-form domain of $\sum_{j=1}^N \OT_j^2$ and $\Hf$. All identities and inequalities below are understood in the quadratic-form sense. Recalling that $\OT_j^2 := 1-\Delta_{x_j}+V(x_j)$ for $j=1,\ldots,N$, and observing that
\begin{equation}
    \sum_{j=1}^N \bigl(-\Delta_{x_j}+V(x_j)\bigr)+N = \sum_{j=1}^N \OT_j^2,
\end{equation}
we have
\begin{equation}
\begin{aligned}
    \inp{\Psi,(\HMN+N)\Psi}
    &=
    \sum_{j=1}^N \norm{\OT_j\Psi}^2
    +
    \norm{\Hf^{1/2}\Psi}^2\\
    &\quad+
    \frac{\lam}{N}
    \sum_{1\leq i<j\leq N}
    \inp{\Psi,v(x_i-x_j)\Psi}
    \\
    &\quad
    +
    \frac{\eta}{\sqrt N}
    \sum_{j=1}^N
    \inp{\Psi,\phi\bigl(w(x_j-\cdot)\bigr)\Psi}.
\end{aligned}
\end{equation}
Since $v\in L^\infty(\R^3)$, we obtain
\begin{equation}
\begin{aligned}
    \frac{\lam}{N}
    \sum_{1\leq i<j\leq N}
    \inp{\Psi,v(x_i-x_j)\Psi}
    &\geq
    -
    \frac{ \lam }{N}
    \binom{N}{2}
    \norm{v}_{L^\infty}
    \norm{\Psi}^2
    \\
    &=
    -
    \frac{ \lam }{2}
    (N-1)
    \norm{v}_{L^\infty}
    \norm{\Psi}^2
    \\
    &\geq
    -
    \frac{ \lam }{2}
    \norm{v}_{L^\infty}
    N\norm{\Psi}^2.
\end{aligned}
\end{equation}
By \lemref{lem:field-energy-operator-bound}, for every $\delta\in(0,1)$, Young's inequality yields
\begin{equation}
\begin{aligned}
    \left|
        \frac{\eta}{\sqrt N}
        \sum_{j=1}^N
        \inp{\Psi,
        \phi\bigl(w(x_j-\cdot)\bigr)\Psi}
    \right|
    &\leq
    \delta\norm{\Hf^{1/2}\Psi}^2
    +
    \frac{ \eta ^2}{\delta}
    C_w^2 N\norm{\Psi}^2.
\end{aligned}
\end{equation}
Consequently,
\begin{equation}
\begin{aligned}
    \inp{\Psi,(\HMN+N)\Psi}
    &\geq
    \sum_{j=1}^N \norm{\OT_j\Psi}^2
    +
    (1-\delta)
    \norm{\Hf^{1/2}\Psi}^2
    -
    C_\delta(\lam,\eta)
    N\norm{\Psi}^2,
\end{aligned}
\end{equation}
where
\begin{equation}
    C_\delta(\lam,\eta)
    :=
    \frac{ \lam }{2}\norm{v}_{L^\infty}
    +
    \frac{ \eta ^2}{\delta}C_w^2.
\end{equation}
Since $V\geq0$, then $\OT_j^2\geq1$ in the quadratic-form sense. Therefore,
\begin{equation}
\sum_{j=1}^N \norm{\OT_j\Psi}^2 \geq N\norm{\Psi}^2.
\end{equation}
It follows that
\begin{equation}
\begin{aligned}
    \inp{\Psi,(\HMN+N)\Psi}
    &\geq
    \bigl(1-C_\delta(\lam,\eta)\bigr)
    \sum_{j=1}^N \norm{\OT_j\Psi}^2
    +
    (1-\delta)
    \norm{\Hf^{1/2}\Psi}^2.
\end{aligned}
\end{equation}
Take $\delta=1/2$. By \eqref{eq:energy-estimate-smallness},
\begin{equation}
    \lam\norm{v}_{L^\infty}<\frac12,
    \qquad
    \eta C_w<\frac14.
\end{equation}
Consequently,
\begin{equation}
\begin{aligned}
    C_{1/2}(\lam,\eta)
    &=
    \frac{\lam}{2}\norm{v}_{L^\infty}
    +2\eta^2C_w^2<
    \frac14+\frac18
    =\frac38.
\end{aligned}
\end{equation}
It follows that
\begin{equation}
\label{eq:energy-base-coercivity}
\begin{aligned}
    \norm{(\HMN+N)^{1/2}\Psi}^2
    &\geq
    \frac58\sum_{j=1}^N\norm{\OT_j\Psi}^2
    +\frac12\norm{\Hf^{1/2}\Psi}^2.
\end{aligned}
\end{equation}
In particular, $\HMN+N$ is strictly positive. By bosonic symmetry,
\begin{equation}
    \sum_{j=1}^N\norm{\OT_j\Psi}^2
    =N\norm{\OT_1\Psi}^2.
\end{equation}
We therefore fix
\begin{equation}
    C:=\frac14
\end{equation}
once and for all. Then
\begin{equation}
    \norm{(\HMN+N)^{1/2}\Psi}^2
    \geq C N\norm{\OT_1\Psi}^2.
\end{equation}
Thus, the assertion holds at orders $k=0$ and $k=1$, with
$N_0(0)=N_0(1)=1$. 

For the operator manipulations below, we use the natural extension
of $\HMN$, given by the same expression, to the full,
unsymmetrized particle--field tensor product.
The individual operators $\OT_j$ and their products act on this
full ambient space. The total Hamiltonian preserves the bosonic subspace,
and the induction estimates are asserted for bosonic vectors. We next justify the induction by spectral truncation.
For each integer $m\geq1$, let
\begin{equation}
    P_m:=
    \left(\prod_{j=1}^N
    \mathbf{1}_{[1,m]}(\OT_j^2)\right)
    \otimes
    \mathbf{1}_{[0,m]}(\Hf).
\end{equation}
These projections preserve bosonic symmetry, commute with all
$\OT_j$ and with $\Hf$, and converge strongly to the identity.
On $\operatorname{Ran}P_m$, define the compressed Hamiltonian by
\begin{equation}
    \mathcal H_{N,m}+N
    :=
    \left.P_m(\HMN+N)P_m\right|_{\operatorname{Ran}P_m},
\end{equation}
where the compression is understood through the quadratic form.
This is a bounded self-adjoint operator on
$\operatorname{Ran}P_m$. By \eqref{eq:energy-base-coercivity},
\begin{equation}
    \mathcal H_{N,m}+N\geq \frac58 N.
\end{equation}
The base cases therefore hold for these compressed Hamiltonians,
with the same constant $C=1/4$, uniformly in $m$.

We perform the following induction on $\operatorname{Ran}P_m$.
In the displayed computations, the cutoff notation is suppressed:
$\HMN+N$ denotes $\mathcal H_{N,m}+N$, and the interaction
operators in $\H_2$ are compressed as
\begin{equation}
    P_m v(x_i-x_j)P_m,
    \qquad
    P_m\phi(w_j)P_m.
\end{equation}
All algebraic expansions of the compressed Hamiltonian take
place on $\operatorname{Ran}P_m$.
In terms linear in an interaction, we first move the remaining
free factors to the first argument and then remove the
projections, since those factors commute with $P_m$.

In the resulting weak pairings and in the multiplier and field
estimates below, $v(x_i-x_j)$, $\phi(w_j)$, $\OT_j$, and $\Hf$
denote the original, uncompressed operators.
In particular, the weak product rule is applied to the original
field operator after the projections have been removed.
No second derivative of $v$ is used.

Fix $n\geq0$, and assume that the desired estimate holds at order
$k=n$, uniformly in $m$, for every $N\geq N_0(n)$.
We prove it at order $k=n+2$. Empty products are the identity, and terms with zero combinatorial
coefficient are omitted. In particular, the internal-pair term of
$B_2$ is absent when $n=0$, and that of $B_5$ is absent when
$n\leq1$. We express
\begin{equation}
    \HMN + N := \H_1 + \H_2,
\end{equation}
where 
\begin{equation}
    \begin{aligned}
        \H_1 &:= \sum_{j=n+1}^{N} \OT^2_j + \Hf, \qquad \text{where} \quad \OT^2_j:= 1-\Delta_j +V(x_j), \\
        \H_2 &:= \sum_{j=1}^{n} \OT^2_j + \frac{\lambda}{N} \sum_{1\leq i<j\leq N} v(x_i-x_j) + \frac{\eta}{\sqrt N} \sum_{j=1}^{N} \phi\bigl(w(x_j-\cdot)\bigr).
    \end{aligned}
\end{equation} 
By the inductive hypothesis for $k=n$, we obtain
\begin{equation}
    \begin{aligned}
    \inp{\Psi, (\HMN +N )^{n+2} \Psi }  &\geq C^{n} N^{n} \inp{\Psi, (\H_1 + \H_2)(\OT^2_1 \dots \OT^2_n) (\H_1 + \H_2) \Psi}.
    \end{aligned}
\end{equation}
Expanding this gives
\begin{equation}
    \label{eq:uniform-prior-bound-on-BBGKY-hierarchy-n+2}
       \inp{\Psi, (\HMN + N)^{n+2} \Psi }   \geq C^{n} N^{n} (A + B + \cc{B} + D),
\end{equation}
where
\begin{align}
    A &:= \inp{\Psi, \H_1 \OT^2_1 \dots \OT^2_n \H_1 \Psi},\\
    B &:= \inp{\Psi, \H_1 \OT^2_1 \dots \OT^2_n \H_2 \Psi},\\
    \cc{B} &:= \inp{\Psi, \H_2 \OT^2_1 \dots \OT^2_n \H_1 \Psi},\\
    D &:= \inp{\Psi, \H_2 \OT^2_1 \dots \OT^2_n \H_2 \Psi}.
\end{align}
Since the fourth term $D$ is positive, it can be dropped. The first inner product $A$ can be decomposed as
\begin{equation}
    A = A_1+ A_2 +  A_3 + A_4,
\end{equation}
where
\begin{align}
    A_1 &= \sum_{i=n+1}^{N} \sum_{j=n+1}^{N} \inp{\Psi, \OT^2_i \OT^2_1 \dots \OT^2_n \OT^2_j \Psi}, \\
    A_2 &= \sum_{i=n+1}^{N}  \inp{\Psi, \OT^2_i \OT^2_1 \dots \OT^2_n \Hf \Psi}, \\
    A_3 &=  \sum_{j=n+1}^{N} \inp{\Psi, \Hf \OT^2_1 \dots \OT^2_n \OT^2_j \Psi}, \\
    A_4 &=  \inp{\Psi, \Hf \OT^2_1 \dots \OT^2_n \Hf \Psi}. \label{eq:A4}
\end{align}
Exploiting bosonic symmetry, we simplify the positive terms as follows. Decomposing the sum $A_1$ into cases $i=j$ and $i\neq j$ gives
\begin{equation}
    A_1 = A_{11} + A_{12},
\end{equation}
where
\begin{align}
    A_{11} &= (N-n) \inp{ \Psi, \OT^4_1 \OT^2_2 \dots \OT^2_{n+1} \Psi}, \label{eq:A11}\\
    A_{12} &= (N-n)(N-n-1) \inp{ \Psi, \OT^2_1 \dots \OT^2_{n+2} \Psi}. \label{eq:A12}
\end{align}
The terms $A_2$ and $A_3$ can be combined into a single term:
\begin{equation}
    A_2 + A_3 = 2 (N-n) \inp{\Psi, \OT^2_1 \dots \OT^2_{n+1} \Hf \Psi}. \label{eq:A2+A3}
\end{equation}
Similarly, we decompose 
\begin{equation}
    B = B_1 + B_2 + B_3 + B_4 + B_5 + B_6,
\end{equation}
where
\begin{align}
    B_1 &= \sum_{\ell=n+1}^{N} \sum_{m=1}^{n} \inp{\Psi, \OT^2_\ell \OT^2_1 \dots \OT^2_n \OT^2_m \Psi}, \label{eq:I-21}\\
    B_2 &= \frac{\lam}{N} \sum_{\ell=n+1}^{N} \sum_{1 \leq i <j \leq N} \inp{\Psi, \OT^2_\ell \OT^2_1 \dots \OT^2_n v(x_i-x_j) \Psi}, \label{eq:I-22}\\ 
    B_3 &= \frac{\eta}{\sqrt N} \sum_{\ell=n+1}^{N} \sum_{j=1}^{N} \inp{\Psi, \OT^2_\ell \OT^2_1 \dots \OT^2_n \phi(w(x_j-\cdot)) \Psi}, \label{eq:I-23}\\
    B_4 &= \sum_{m=1}^{n} \inp{\Psi, \Hf \OT_1^2 \dots \OT^2_n \OT^2_m \Psi}, \label{eq:I-24}\\
    B_5 &= \frac{\lam}{N} \sum_{1 \leq i <j \leq N} \inp{\Psi, \Hf \OT^2_1 \dots \OT^2_n v(x_i-x_j) \Psi}, \label{eq:I-25}\\
    B_6 &= \frac{\eta}{\sqrt N} \sum_{j=1}^{N} \inp{\Psi, \Hf \OT^2_1 \dots \OT^2_n \phi(w(x_j-\cdot)) \Psi}. \label{eq:I-26}
\end{align}
Using bosonic symmetry, we simplify the positive terms $B_1$ and $B_4$ as follows:
\begin{align}
    B_1 &= n (N-n) \inp{\Psi, \OT^4_1 \OT^2_2 \dots \OT^2_{n+1} \Psi}, \label{eq:B1}\\
    B_4 &= n \inp{\Psi, \Hf \OT^4_1 \OT^2_2 \dots \OT^2_n \Psi}. \label{eq:B4}
\end{align}

Next, we compute lower bounds for $B_2$, $B_3$, $B_5$, and $B_6$. Consider the term $B_2$ in \eqref{eq:I-22}. The multiplier $v(x_i-x_j)$ commutes with $\OT_\ell$
for $\ell\notin\{i,j\}$, whereas commutation can fail for
$\ell=i,j$. We therefore first establish a bound for
multiplication by $v(x_i-x_j)$ in the one-particle form norm.
We first establish a useful operator bound. Fix a $j >i$. Denote by $v_j:= v(x_i-x_j)$, we have
\begin{equation}
    \begin{aligned}
    \norm{\OT_j v(x_i-x_j) \Psi}^2 &= \inp{\nabla_j v_j \Psi + v_j \nabla_j \Psi, \nabla_j v_j \Psi + v_j \nabla_j \Psi} \\
        &\quad + \inp{v_j \Psi, (1+V(x_j)) v_j \Psi} \\
        &= \inp{\nabla_j v_j \Psi, \nabla_j v_j \Psi} + 2 \re \inp{\nabla_j v_j \Psi, v_j \nabla_j \Psi}\\
         &\quad + \inp{v_j \nabla_j \Psi, v_j \nabla_j \Psi} + \inp{v_j \Psi, (1+V(x_j)) v_j \Psi} \\
         &\leq \norm{\nabla_j v_j}_{L^\infty}^2 \norm{\Psi}^2 +2 \norm{\nabla_j v_j}_{L^\infty} \norm{v_j}_{L^\infty} \norm{\nabla_j \Psi} \norm{\Psi}\\
         &\quad + \norm{v_j}_{L^\infty}^2 \norm{\nabla_j \Psi}^2 + \norm{v_j}_{L^\infty}^2 \norm{(1+V(x_j))^{1/2} \Psi}^2 \\
         &\leq \norm{\nabla_j v_j}_{L^\infty}^2 \norm{\Psi}^2 +2 \norm{\nabla_j v_j}_{L^\infty} \norm{v_j}_{L^\infty} \norm{\nabla_j \Psi} \norm{\Psi}\\
         &\quad + \norm{v_j}_{L^\infty}^2 \norm{\OT_j \Psi}^2.
    \end{aligned}
\end{equation}
Since $\OT_j \geq 1$, and $\norm{\nabla_j \Psi}^2 \leq \norm{\OT_j \Psi}^2$, we have
\begin{equation}
    \begin{aligned}
         \norm{\OT_j v_j \Psi}^2 &\leq \prc{\norm{\nabla v}_{L^\infty} + \norm{v}_{L^\infty}}^2 \norm{\OT_j \Psi}^2.
    \end{aligned}
\end{equation}
Taking the square root and using $\norm{v}_{W^{1,\infty}} := \norm{v}_{L^\infty} + \norm{\nabla v}_{L^\infty}$, we obtain
\begin{equation}
    \label{eq:T-v-bound}
    \norm{\OT_j v(x_i-x_j) \Psi} \leq \norm{v}_{W^{1,\infty}} \norm{\OT_j \Psi}.
\end{equation}
Decomposing the sum in $B_2$ into three cases, and using bosonic symmetry,
\begin{equation}
    \label{eq:B2}
    \begin{aligned}
        B_2 &= \frac{\lam (N-n)}{N} \Big( \sum_{n+1 < i < j \leq N} \inp{\Psi, \OT^2_1 \dots \OT^2_{n+1} v(x_i-x_j) \Psi}\\
         &\quad + \sum_{1 \leq i \leq n+1 < j \leq N} \inp{\Psi, \OT^2_1 \dots \OT^2_{n+1} v(x_i-x_j) \Psi}\\
         &\quad + \sum_{1 \leq i < j \leq n+1} \inp{\Psi, \OT^2_1 \dots \OT^2_{n+1} v(x_i-x_j) \Psi} \Big).
    \end{aligned}
\end{equation}
That is,
\begin{equation}
    \begin{aligned}
        B_2 &= \frac{\lam (N-n)(N-n-1)(N-n-2)}{2N} \inp{\Psi, v(x_{n+2}-x_{n+3}) \OT^2_1 \dots \OT^2_{n+1} \Psi} \\
        &+ \frac{\lam (N-n)(n+1)(N-n-1)}{N} \inp{\Psi,\OT_1^2 v(x_1-x_{n+2}) \OT_2^2 \dots \OT^2_{n+1} \Psi} \\
        &+ \frac{\lam (N-n)(n+1)n}{2N} \inp{\Psi, \OT_1^2 \OT_2^2 v(x_1-x_2) \OT_3^2 \dots \OT^2_{n+1} \Psi}.
    \end{aligned}
\end{equation}
Using the bound \eqref{eq:T-v-bound} and denoting $v_1 := v(x_1-x_{n+2})$, we have
\begin{equation}
    \begin{aligned}
        \Big| \inp{\Psi, \OT_1^2 v_1 \OT_2^2 \dots \OT^2_{n+1} \Psi} \Big| &= \Big| \inp{\OT_1 \dots \OT_{n+1} \Psi, \OT_1 v_1 \OT_2 \dots \OT_{n+1} \Psi} \Big|\\
        &\leq \norm{\OT_1 \dots \OT_{n+1} \Psi} \norm{\OT_1 v_1 \OT_2 \dots \OT_{n+1} \Psi} \\
        &\leq \norm{v}_{W^{1,\infty}} \norm{\OT_1 \dots \OT_{n+1} \Psi}^2.
    \end{aligned}
\end{equation} 
For the internal-pair term, assume $n\geq1$ and set
$v_{12}:=v(x_1-x_2)$. In the weak-pairing convention above,
\begin{equation}
\begin{aligned}
    &\left|
    \inp{\Psi,\OT_1^2\OT_2^2v_{12}
                   \OT_3^2\cdots\OT_{n+1}^2\Psi}
    \right|\\
    &\quad=
    \left|
    \inp{\OT_1^2\OT_2^2\OT_3\cdots\OT_{n+1}\Psi,
         v_{12}\OT_3\cdots\OT_{n+1}\Psi}
    \right|\\
    &\quad=
    \left|
    \inp{\OT_1^2\OT_2\cdots\OT_{n+1}\Psi,
         \OT_2v_{12}\OT_3\cdots\OT_{n+1}\Psi}
    \right|\\
    &\quad\leq
    \norm{v}_{W^{1,\infty}}
    \norm{\OT_1^2\OT_2\cdots\OT_{n+1}\Psi}
    \norm{\OT_2\cdots\OT_{n+1}\Psi}\\
    &\quad\leq
    \norm{v}_{W^{1,\infty}}
    \norm{\OT_1^2\OT_2\cdots\OT_{n+1}\Psi}^{2}.
\end{aligned}
\end{equation}
Therefore, we have
\begin{equation}
    \abs{B_2} \leq {B_{2}^{N^2} } + {B_{2}^{N}} + {B_{2}^{1}},
\end{equation}
where
\begin{align}
    B_{2}^{N^2} &= \frac{\lam (N-n)(N-n-1)(N-n-2)}{2N} \norm{v}_{L^\infty} \inp{\Psi, \OT^2_1 \dots \OT^2_{n+1} \Psi}, \label{eq:B2N2}\\*
    B_{2}^{N} &= \frac{\lam (N-n)(n+1)(N-n-1)}{N} \norm{v}_{W^{1,\infty}} \inp{\Psi, \OT^2_1 \dots \OT^2_{n+1} \Psi}, \label{eq:B2N}\\*
    B_{2}^{1} &= \frac{\lam (N-n)(n+1)n}{2N} \norm{v}_{W^{1,\infty}} \inp{\Psi, \OT^4_1 \dots \OT^2_{n+1} \Psi}.\label{eq:B21}
\end{align}

Next, we bound $B_5$. Since the interaction commutes with $\Hf$,
we split the field energy equally between the two arguments.
Using bosonic symmetry and the weak-pairing convention above,
\begin{equation}
\begin{aligned}
    B_5
    &=
    \frac{\lam(N-n)(N-n-1)}{2N}
    \inp{\OT_1\cdots\OT_n\Hf^{1/2}\Psi,
         v(x_{n+1}-x_{n+2})
         \OT_1\cdots\OT_n\Hf^{1/2}\Psi}\\
    &\quad+
    \frac{\lam n(N-n)}{N}
    \inp{\OT_1\cdots\OT_n\Hf^{1/2}\Psi,
         \OT_1v(x_1-x_{n+1})
         \OT_2\cdots\OT_n\Hf^{1/2}\Psi}\\
    &\quad+
    \frac{\lam n(n-1)}{2N}
    \inp{\OT_1^2\OT_2\cdots\OT_n\Hf^{1/2}\Psi,
         \OT_2v(x_1-x_2)
         \OT_3\cdots\OT_n\Hf^{1/2}\Psi}.
\end{aligned}
\end{equation}
Then, using the Cauchy-Schwarz inequality then \eqref{eq:T-v-bound}, gives
\begin{equation}
    \begin{aligned}
       | <\Psi, \OT_1^2& \dots \OT_n^2  v(x_1-x_{n+1}) \Hf\Psi>| \\
       &=\abs{\inp{\OT_1 \dots \OT_n \Hf^{1/2} \Psi,\OT_1 v(x_1-x_{n+1}) \OT_2 \dots \OT_n \Hf^{1/2} \Psi}} \\
        &\leq \norm{\OT_1 \dots \OT_n \Hf^{1/2} \Psi} \norm{\OT_1 v(x_1-x_{n+1}) \OT_2 \dots \OT_n \Hf^{1/2}\Psi} \\
        &\leq \norm{v}_{W^{1,\infty}} \norm{\OT_1 \dots \OT_n \Hf^{1/2} \Psi}^2.
    \end{aligned}
\end{equation}
Similarly, \eqref{eq:T-v-bound}, and $\OT_1 \geq 1$ gives
\begin{equation}
    \begin{aligned}
        \Big| < \OT_1^2 \OT_2 \dots &\OT_n \Hf^{1/2} \Psi, \OT_2 v(x_1-x_2)  \OT_3 \dots \OT_n \Hf^{1/2} \Psi> \Big|\\
        &\leq \norm{\OT_1^2 \dots \OT_n \Hf^{1/2} \Psi} \norm{\OT_2 v(x_1-x_2) \OT_3 \dots \OT_n \Hf^{1/2} \Psi}\\
        &\leq \norm{v}_{W^{1,\infty}} \norm{\OT_1^2 \OT_2 \dots \OT_n \Hf^{1/2} \Psi} \norm{\OT_2 \dots \OT_n \Hf^{1/2} \Psi}\\
        &\leq \norm{v}_{W^{1,\infty}} \norm{\OT_1^2 \OT_2 \dots \OT_n \Hf^{1/2} \Psi}^2.
    \end{aligned}
\end{equation}

\begin{equation}
    \abs{B_5} \leq {B_{5}^{2} } + {B_{5}^{1} } + {B_{5}^{0} },
\end{equation}
where 
\begin{align}
    B_{5}^{2} &= \frac{\lam (N-n)(N-n-1)}{2N} \norm{v}_{L^\infty} \inp{\Psi, \OT^2_1 \dots \OT^2_n \Hf \Psi}, \label{eq:B52}\\
    B_{5}^{1} &= \frac{\lam n (N-n)}{N} \norm{v}_{W^{1,\infty}} \inp{\Psi, \OT^2_1 \dots \OT^2_n \Hf \Psi},\label{eq:B51}\\
    B_{5}^{0} &= \frac{\lam n(n-1)}{2N} \norm{v}_{W^{1,\infty}} \inp{\Psi, \OT^4_1 \dots \OT^2_n \Hf \Psi}.\label{eq:B50}
\end{align}

Now, we turn to bounding the terms that involve the field operator $\phi(w(x_j))$, where $w_j(y) = w(x_j-y)$.
\begin{equation}
    \begin{aligned}
        \abs{B_3} =& \frac{\eta}{\sqrt N} \abs{\sum_{\ell=n+1}^{N} \sum_{j=1}^{N} \inp{\Psi, \OT^2_1 \dots \OT^2_n \OT^2_\ell \phi(w_j) \Psi}} \\
        &\leq \frac{\eta (N-n)}{\sqrt N} \sum_{j=1}^{N} \abs{\inp{\Psi, \OT^2_1 \dots \OT^2_n \OT^2_{n+1} \phi(w_j) \Psi}}.
    \end{aligned}
\end{equation}
For $j\geq n+2$, the field operator $\phi(w_j)$ commutes with
$\OT_1^2\cdots\OT_{n+1}^2$.
We therefore separate the sum into $j\geq n+2$ and
$j\leq n+1$, obtaining
\begin{equation}
    |B_3| \leq B_3^{non} + B_3^{com},
\end{equation}
where
\begin{align}
    B_3^{com} &:= \frac{\eta (N-n)}{\sqrt N} \sum_{j=n+2}^{N} \abs{\inp{\Psi, \OT^2_1 \dots \OT^2_n \OT^2_{n+1} \phi(w_j) \Psi}}\\
    B_3^{non} &:= \frac{\eta (N-n)}{\sqrt N} \sum_{j=1}^{n+1} \abs{\inp{\Psi, \OT^2_1 \dots \OT^2_n \OT^2_{n+1} \phi(w_j) \Psi}}.
\end{align}
Bosonic symmetry, Cauchy-Schwarz inequality, and \lemref{lem:field-energy-operator-bound} gives
\begin{equation}
    \begin{aligned}
        B_3^{com} &= \frac{\eta (N-n) (N-n-1)}{\sqrt N} \abs{ \inp{\OT_1 \dots \OT_{n+1} \Psi, \OT_1 \dots \OT_{n+1} \phi(w_N) \Psi}} \\
        &\leq \frac{\eta (N-n) (N-n-1)}{\sqrt N} \norm{\OT_1 \dots \OT_{n+1} \Psi} \norm{\phi(w_N) \OT_1 \dots \OT_{n+1}  \Psi} \\
        &\leq \frac{2\eta (N-n) (N-n-1)}{\sqrt N} C_w \norm{\OT_1 \dots \OT_{n+1} \Psi} \norm{ \OT_1 \dots \OT_{n+1} (1+\Hf)^{1/2}  \Psi},
    \end{aligned}
\end{equation}
where
\begin{equation}
    C_w := \norm{w}_{L^2} + \norm{\frac{\wh{w}}{\sqrt{\omg}}}_{L^2}.
\end{equation}
Young's inequality gives
\begin{equation}
    \begin{aligned}
        B_3^{com} &\leq {\eta (N-n) (N-n-1)} C_w  \norm{\OT_1 \dots \OT_{n+1} \Psi}^2\\
        &\quad  +\frac{\eta (N-n) (N-n-1)}{ N} C_w \norm{  \OT_1 \dots \OT_{n+1}  (1+\Hf)^{1/2} \Psi}^2.
    \end{aligned}
\end{equation}
Therefore, using the fact
\begin{equation}
    \label{eq:OT-Hf-bound}
    \begin{aligned}
        \|(1+\Hf)^{1/2} \OT_1 \dots \OT_{n+1}  \Psi\|^2 &= \norm{\OT_1 ...\, \OT_{n+1} \Psi}^2 + \norm{\OT_1 \dots \OT_{n+1} \Hf^{1/2} \Psi}^2,
    \end{aligned}
\end{equation}
we obtain
\begin{equation}
    \begin{aligned}
        B_3^{com} &\leq {\eta C_w (N-n) (N-n-1)} (1 + \frac{1}{N}) \norm{\OT_1 \dots \OT_{n+1} \Psi}^2\\
        &\quad + \frac{\eta C_w (N-n) (N-n-1)}{ N} \norm{\OT_1 \dots \OT_{n+1} \Hf^{1/2} \Psi}^2.
    \end{aligned}
\end{equation}

\begin{equation}
    B_3^{com} \leq B_{3}^{com,N^2} + B_{3}^{com,N},
\end{equation}
where
\begin{align}
    B_{3}^{com,N^2} &:= \eta C_w {(N-n)(N-n-1)} (1 + \frac{1}{N}) \inp{\Psi, \OT^2_1 \dots \OT^2_{n+1} \Psi},\label{eq:B3com1}\\
    B_{3}^{com,N} &:= \eta C_w \frac{(N-n)(N-n-1)}{  N} \inp{\Psi, \OT^2_1 \dots \OT^2_{n+1} \Hf \Psi}.\label{eq:B3com2}
\end{align}
Next, we analyze the term $B_3^{non}$. Using bosonic symmetry, the Cauchy-Schwarz inequality, and \lemref{lem:field-energy-operator-bound}, we obtain
\begin{equation}
    \begin{aligned}
        B_3^{non} &= \frac{\eta (N-n) (n+1)}{\sqrt N} \abs{ \inp{\OT^2_1 \OT_2 \dots \OT_{n+1} \Psi,\OT_2 \dots \OT_{n+1} \phi(w_1) \Psi}} \\
        &\leq \frac{\eta (N-n) (n+1)}{\sqrt N} \norm{\OT^2_1 \OT_2 \dots \OT_{n+1} \Psi} \norm{\phi(w_1) \OT_2 \dots \OT_{n+1}  \Psi} \\
        &\leq 2 \frac{\eta (N-n) (n+1)}{\sqrt N} C_w \norm{\OT^2_1 \OT_2 \dots \OT_{n+1} \Psi} \norm{(1+\Hf)^{1/2}  \OT_2 \dots \OT_{n+1}  \Psi}.
    \end{aligned}
\end{equation}
Young's inequality gives 
\begin{equation}
    \begin{aligned}
        B_3^{non} &\leq  \eta C_w \frac{(N-n)(n+1)}{\sqrt N}  \norm{\OT^2_1 \OT_2 \dots \OT_{n+1} \Psi}^2\\
        &\quad + \eta C_w \frac{(N-n)(n+1)}{\sqrt N} \norm{(1+\Hf)^{1/2}  \OT_2 \dots \OT_{n+1}  \Psi}^2.
    \end{aligned}
\end{equation}
By the fact $\OT_1 \geq 1$, and \eqref{eq:OT-Hf-bound}, we have
\begin{equation}
    B_3^{non} \leq B_{3}^{non,1} + B_{3}^{non,2} + B_{3}^{non,3}, 
\end{equation}
where
\begin{align}
    B_{3}^{non,1} &:= \eta C_w \frac{(N-n)(n+1)}{\sqrt N} \inp{\Psi, \OT^4_1 \OT^2_2 \dots \OT^2_{n+1} \Psi},\label{eq:B3non1}\\*
    B_{3}^{non,2} &:= \eta C_w \frac{(N-n)(n+1)}{ \sqrt N} \inp{\Psi, \OT^2_1 \dots \OT^2_{n+1} \Psi},\label{eq:B3non2}\\*
    B_{3}^{non,3} &:= \eta C_w \frac{(N-n)(n+1)}{ \sqrt N} \inp{\Psi, \OT^2_1 \dots \OT^2_{n+1} \Hf \Psi}.\label{eq:B3non3}
\end{align}

To bound $B_6$ in \eqref{eq:I-26}, we decompose the sum into $j \leq n$ and $j > n$:
\begin{equation}
    \begin{aligned}
        B_6 &= \frac{\eta}{\sqrt N} \sum_{j=1}^{N} \inp{\Psi, \Hf \OT^2_1 \dots \OT^2_n \phi(w_j) \Psi}\\
        &= \frac{\eta}{\sqrt N} \sum_{j=1}^{n} \inp{\Psi, \Hf \OT^2_1 \dots \OT^2_n \phi(w_j) \Psi}\\
        & + \frac{\eta}{\sqrt N} \sum_{j=n+1}^{N} \inp{\Psi, \Hf \OT^2_1 \dots \OT^2_n \phi(w_j) \Psi}.
    \end{aligned}
\end{equation}
Thus, we have 
\begin{equation}
    |B_6| \leq B_6^{com} + B_6^{non} ,
\end{equation}
where
\begin{align}
        B_6^{com} &:= \frac{\eta (N-n) }{\sqrt N} \abs{\inp{\Psi, \Hf \OT^2_1 \dots \OT^2_n \phi(w_N) \Psi}}\\
         B_6^{non} &:= \frac{\eta n }{\sqrt N} \abs{\inp{\Psi, \Hf \OT^2_1 \dots \OT^2_n \phi(w_1) \Psi}}.
\end{align}
Bosonic symmetry, Cauchy-Schwarz inequality, and \lemref{lem:field-energy-operator-bound}, we obtain
\begin{equation}
    \begin{aligned}
        B_6^{com} &\leq  \frac{2\eta (N-n)}{\sqrt N} C_w \norm{\OT_1 \dots \OT_n \Hf \Psi} \norm{(1+\Hf)^{1/2}   \OT_1 \dots \OT_n  \Psi}.
    \end{aligned}
\end{equation}
Young's inequality gives
\begin{equation}
    \begin{aligned}
        B_6^{com} &\leq \eta C_w \frac{(N-n)}{ N} \norm{\OT_1 \dots \OT_n \Hf \Psi}^2\\
        &\quad + \eta C_w {(N-n)} \norm{(1+\Hf)^{1/2}   \OT_1 \dots \OT_n  \Psi}^2.
    \end{aligned}
\end{equation}
Using \eqref{eq:OT-Hf-bound}, we obtain
\begin{equation}
    \begin{aligned}
        B_6^{com} &\leq \eta C_w \frac{(N-n)}{ N} \norm{\OT_1 \dots \OT_n \Hf \Psi}^2 + \eta C_w {(N-n)} \norm{\OT_1 \dots \OT_n  \Psi}^2\\
        &\qquad  + \eta C_w {(N-n)} \norm{\OT_1 \dots \OT_n  \Hf^{1/2} \Psi}^2.
    \end{aligned}
\end{equation}

\begin{equation}
    B_6^{com} \leq B_{6}^{com,1} + B_{6}^{com,2} + B_{6}^{com,3},
\end{equation}
where
\begin{align}
    B_{6}^{com,1} &:= \eta C_w {(N-n)} \inp{\Psi, \OT^2_1 \dots \OT^2_n \Psi},\label{eq:B6com1}\\*
    B_{6}^{com,2} &:= \eta C_w {(N-n)} \inp{\Psi, \OT^2_1 \dots \OT^2_n \Hf \Psi} \label{eq:B6com2}\\
    B_{6}^{com,3} &:= \eta C_w \frac{(N-n)}{ N} \inp{\Psi, \OT^2_1 \dots \OT^2_n \Hf^2 \Psi}.\label{eq:B6com3}
\end{align}
The term $B_6^{non}$ requires a one-particle quadratic-form
calculation. Using the form defining $\OT_1^2$ and the weak
product rule
\begin{equation}
    \nabla_{x_1}\bigl(\phi(w_1)\Phi\bigr)
    =
    \phi(w_1)\nabla_{x_1}\Phi
    +\phi(\nabla_{x_1}w_1)\Phi,
\end{equation}
we obtain
\begin{equation}
\begin{aligned}
    &\inp{\Psi,\Hf\OT_1^2\cdots\OT_n^2\phi(w_1)\Psi}\\
    &=
    \inp{(1+V(x_1))^{1/2}\Hf\OT_2\cdots\OT_n\Psi,
         \phi(w_1)(1+V(x_1))^{1/2}
         \OT_2\cdots\OT_n\Psi}\\
    &\quad+
    \inp{\nabla_{x_1}\Hf\OT_2\cdots\OT_n\Psi,
         \phi(w_1)\nabla_{x_1}\OT_2\cdots\OT_n\Psi}\\
    &\quad+
    \inp{\nabla_{x_1}\Hf\OT_2\cdots\OT_n\Psi,
         \phi(\nabla_{x_1}w_1)\OT_2\cdots\OT_n\Psi}.
\end{aligned}
\end{equation}
The inner products involving gradients are summed over their
three components. Applying the field estimate componentwise and then
Cauchy--Schwarz to the direct sum of the potential and gradient
components gives
\begin{equation}
    \begin{aligned}
        | \langle \Psi&, \Hf \OT^2_1 \dots \OT^2_n \phi(w_1) \Psi \rangle|\\
        & \leq 2 C_w \norm{\OT_1 \dots \OT_n \Hf \Psi} \norm{(1+\Hf)^{1/2} \OT_1 \dots \OT_n \Psi} \\
        &\quad + 2 C_w' \norm{\OT_1 \dots \OT_n \Hf \Psi} \norm{(1+\Hf)^{1/2} \OT_2 \dots \OT_n \Psi}\\
        &\leq 2(C_w + C_w') \norm{\OT_1 \dots \OT_n \Hf \Psi} \norm{(1+\Hf)^{1/2} \OT_1 \dots \OT_n \Psi}.
    \end{aligned}
\end{equation}
where
\begin{equation}
    C_w := \norm{w}_{L^2} + \norm{\frac{\wh w}{\sqrt{\omg}}}_{L^2}, \quad 
    \text{ and } \quad
    C_w' :=\norm{\nabla w}_{L^2} +  \norm{\frac{\widehat{\nabla w}}{\sqrt{\omg}}}_{L^2}
\end{equation}
Young's inequality, $\OT_1 \geq 1$, and \eqref{eq:OT-Hf-bound} gives
\begin{equation}
    \begin{aligned}
        | \langle \Psi, \Hf \OT^2_1 &\dots \OT^2_n \phi(w_1) \Psi \rangle| \leq (C_w + C_w') \Big[ \norm{\OT_1 \dots \OT_n \Psi}^2\\
        &\quad + \norm{\OT_1 \dots \OT_n \Hf^{1/2} \Psi}^2 + \norm{\OT_1 \dots \OT_n \Hf \Psi}^2\Big].
    \end{aligned}
\end{equation}
Rearranging the terms and translating the squared norms back into inner products, we obtain
\begin{equation}
    |B_6^{non}| \leq |B_6^{non,1}| + |B_6^{non,2}| + |B_6^{non,3}|,
\end{equation}
where
\begin{align}
    B_6^{non,1} &:= \frac{\eta n}{\sqrt N} (C_w + C_w') \inp{\Psi, \OT^2_1 \dots \OT^2_n \Psi}, \label{eq:B6non1}\\*
    B_6^{non,2} &:= \frac{\eta n}{ \sqrt N} (C_w + C'_w) \inp{\Psi, \OT^2_1 \dots \OT^2_n \Hf \Psi}, \label{eq:B6non2}\\*
    B_6^{non,3} &:=  \frac{\eta n}{\sqrt N} (C_w + C_w') \inp{\Psi, \OT^2_1 \dots \OT^2_n \Hf^2 \Psi}. \label{eq:B6non3}
\end{align}

\noindent
\medskip
\noindent
\textbf{Grouping the error terms:}
We collect the preceding estimates in
\eqref{eq:uniform-prior-bound-on-BBGKY-hierarchy-n+2}.
Fix $n\geq0$ and take $N\geq n+3$.
We continue on the common domain of the preceding computations.
Empty products are understood as the identity. When $n=0$, terms
carrying a factor $n$ are absent, and the fifth group below is omitted.
The constants $C_w$ and $C_w'$ have the definitions given above,
and we use the weak-coupling condition
\begin{equation}
\label{eq:grouping-weak-coupling}
    2\lam\norm{v}_{L^\infty}+4\eta C_w<1.
\end{equation}
In particular,
\begin{equation}
\label{eq:grouping-positive-margins}
\begin{aligned}
    1-\lam\norm{v}_{L^\infty}-2\eta C_w&>\frac12,\\
    2-\lam\norm{v}_{L^\infty}-4\eta C_w&>1,\\
    2\eta C_w&<\frac12.
\end{aligned}
\end{equation}

Since $D\geq0$ and $B_1,B_4$ are real and nonnegative,
\begin{equation}
\label{eq:grouping-cross-term-lower-bound}
\begin{aligned}
    A+B+\cc{B}+D
    &=A+2B_1+2B_4
      +2\re (B_2+B_3+B_5+B_6)+D\\
    &\geq A_{12}+A_{11}+2B_1+A_2+A_3+A_4+2B_4\\
    &\quad-2\abs{B_2}-2\abs{B_3}
          -2\abs{B_5}-2\abs{B_6}.
\end{aligned}
\end{equation}
Substituting the nonnegative bounds for the four error terms and
collecting each contribution exactly once gives
\begin{equation}
\label{eq:grouping-five-groups}
\begin{aligned}
    A+B+\cc{B}+D
    &\geq
    \Big[
        A_{12}-2B_2^{N^2}-2B_2^N-2B_3^{com,N^2}\\
    &\hspace{3em}
        -2B_3^{non,2}-2B_6^{com,1}-2B_6^{non,1}
    \Big]\\
    &\quad+
    \Big[
        A_{11}+2B_1-2B_2^1-2B_3^{non,1}
    \Big]\\
    &\quad+
    \Big[
        A_2+A_3-2B_5^2-2B_5^1-2B_3^{com,N}\\
    &\hspace{3em}
        -2B_3^{non,3}-2B_6^{com,2}-2B_6^{non,2}
    \Big]\\
    &\quad+
    \Big[
        A_4-2B_6^{com,3}-2B_6^{non,3}
    \Big]\\
    &\quad+
    \Big[2B_4-2B_5^0\Big].
\end{aligned}
\end{equation}
The factors $2$ in this expression come from
$B+\cc{B}=2\re B$; the individual error bounds
have not been rescaled.

The operators $\OT_j$ commute with one another and with $\Hf$,
and $\OT_j\geq1$. Hence,
\begin{equation}
\label{eq:grouping-particle-monotonicity}
\begin{aligned}
    0\leq\inp{\Psi,\OT_1^2\cdots\OT_n^2\Psi}
    &\leq\inp{\Psi,\OT_1^2\cdots\OT_{n+1}^2\Psi}\leq\inp{\Psi,\OT_1^2\cdots\OT_{n+2}^2\Psi},
\end{aligned}
\end{equation}
and
\begin{equation}
\label{eq:grouping-field-monotonicity}
    0\leq\inp{\Psi,\OT_1^2\cdots\OT_n^2\Hf\Psi}
    \leq\inp{\Psi,\OT_1^2\cdots\OT_{n+1}^2\Hf\Psi}.
\end{equation}

\medskip
\noindent
\textbf{The $n+2$-particle group.}
Using \eqref{eq:A12}, \eqref{eq:B2N2}, \eqref{eq:B2N},
\eqref{eq:B3com1}, \eqref{eq:B3non2}, \eqref{eq:B6com1},
and \eqref{eq:B6non1}, together with
\eqref{eq:grouping-particle-monotonicity}, we obtain
\begin{equation}
\label{eq:grouping-nplus2-expanded}
\begin{aligned}
    &A_{12}-2B_2^{N^2}-2B_2^N-2B_3^{com,N^2}\\
    &\qquad-2B_3^{non,2}-2B_6^{com,1}-2B_6^{non,1}\\
    &\geq
    \Bigg[
        (N-n)(N-n-1)\\
    &\qquad
        -\frac{\lam(N-n)(N-n-1)(N-n-2)}{N}
         \norm{v}_{L^\infty}\\
    &\qquad
        -\frac{2\lam(N-n)(n+1)(N-n-1)}{N}
         \norm{v}_{W^{1,\infty}}\\
    &\qquad
        -2\eta C_w(N-n)(N-n-1)\left(1+\frac1N\right)\\
    &\qquad
        -\frac{2\eta C_w(N-n)(n+1)}{\sqrt N}
        -2\eta C_w(N-n)
        -\frac{2\eta n(C_w+C_w')}{\sqrt N}
    \Bigg]\\
    &\qquad\times
    \inp{\Psi,\OT_1^2\cdots\OT_{n+2}^2\Psi}.
\end{aligned}
\end{equation}
The coefficient in brackets can be written as
\begin{equation}
\label{eq:grouping-nplus2-coefficient}
\begin{aligned}
    &(N-n)(N-n-1)
      \Bigg[
        1-\lam\frac{N-n-2}{N}\norm{v}_{L^\infty}\\
    &\hspace{8em}
        -\frac{2\lam(n+1)}{N}\norm{v}_{W^{1,\infty}}
        -2\eta C_w\left(1+\frac1N\right)
      \Bigg]\\
    &\quad
        -\frac{2\eta C_w(N-n)(n+1)}{\sqrt N}
        -2\eta C_w(N-n)
        -\frac{2\eta n(C_w+C_w')}{\sqrt N}\\
    &\geq
        \frac12(N-n)(N-n-1)\\
    &\quad
        -\frac{(N-n)(N-n-1)}{N}
         \Big[2\lam(n+1)\norm{v}_{W^{1,\infty}}+2\eta C_w\Big]\\
    &\quad
        -\frac{2\eta C_w(N-n)(n+1)}{\sqrt N}
        -2\eta C_w(N-n)
        -\frac{2\eta n(C_w+C_w')}{\sqrt N}.
\end{aligned}
\end{equation}
Here we used $(N-n-2)/N\leq1$ and the first inequality
in \eqref{eq:grouping-positive-margins}. Moreover,
\begin{equation}
\begin{aligned}
    (N-n)(N-n-1)
    &=N^2-(2n+1)N+n(n+1)\\
    &\geq N^2-(2n+1)N,\\
    \frac{(N-n)(N-n-1)}{N}&\leq N,\\
    \frac{N-n}{\sqrt N}&\leq\sqrt N\leq N,\\
    N-n&\leq N,
    \qquad\frac1{\sqrt N}\leq N.
\end{aligned}
\end{equation}
Consequently, the coefficient in
\eqref{eq:grouping-nplus2-expanded} is bounded below by
\begin{equation}
\label{eq:grouping-nplus2-coefficient-lower-bound}
\begin{aligned}
    &\frac{N^2}{2}
    -N\Bigg[
        \frac{2n+1}{2}
        +2\lam(n+1)\norm{v}_{W^{1,\infty}}\\
    &\hspace{7em}
        +2\eta C_w(n+3)
        +2\eta n(C_w+C_w')
    \Bigg].
\end{aligned}
\end{equation}
The term $2\eta C_w(n+3)$ results from
$2\eta C_w+2\eta C_w(n+1)+2\eta C_w$.
It follows that the sufficient condition
\begin{equation}
\label{eq:grouping-nplus2-threshold}
\begin{aligned}
    N\geq{}&2(2n+1)
       +8\lam(n+1)\norm{v}_{W^{1,\infty}}\\
       &+8\eta C_w(n+3)+8\eta n(C_w+C_w')
\end{aligned}
\end{equation}
implies
\begin{equation}
\label{eq:grouping-nplus2-positive}
\begin{aligned}
    &A_{12}-2B_2^{N^2}-2B_2^N-2B_3^{com,N^2}\\
    &\qquad-2B_3^{non,2}-2B_6^{com,1}-2B_6^{non,1}\\
    &\geq
      \frac{N^2}{4}
      \inp{\Psi,\OT_1^2\cdots\OT_{n+2}^2\Psi}.
\end{aligned}
\end{equation}

\medskip
\noindent
\textbf{The $\OT_1^4$ group.}
All terms in this group contain the same quadratic form.
By \eqref{eq:A11}, \eqref{eq:B1}, \eqref{eq:B21}, and
\eqref{eq:B3non1},
\begin{equation}
\label{eq:grouping-T4-expanded}
\begin{aligned}
    &A_{11}+2B_1-2B_2^1-2B_3^{non,1}\\
    &=\Bigg[
        (N-n)+2n(N-n)\\
    &\qquad
        -\frac{\lam(N-n)n(n+1)}{N}\norm{v}_{W^{1,\infty}}
        -\frac{2\eta C_w(N-n)(n+1)}{\sqrt N}
      \Bigg]\\
    &\qquad\times
      \inp{\Psi,\OT_1^4\OT_2^2\cdots\OT_{n+1}^2\Psi}\\
    &=(N-n)
      \Bigg[
        2n+1-\frac{\lam n(n+1)}{N}\norm{v}_{W^{1,\infty}}
        -\frac{2\eta C_w(n+1)}{\sqrt N}
      \Bigg]\\
    &\qquad\times
      \inp{\Psi,\OT_1^4\OT_2^2\cdots\OT_{n+1}^2\Psi}.
\end{aligned}
\end{equation}
Since $n+1\leq2n+1$, the coefficient in the last bracket satisfies
\begin{equation}
\begin{aligned}
    &2n+1-\frac{\lam n(n+1)}{N}\norm{v}_{W^{1,\infty}}
      -\frac{2\eta C_w(n+1)}{\sqrt N}\\
    &\geq
      (2n+1)\left[
        1-\frac{\lam n}{N}\norm{v}_{W^{1,\infty}}
        -\frac{2\eta C_w}{\sqrt N}
      \right].
\end{aligned}
\end{equation}
Thus, if
\begin{equation}
\label{eq:grouping-T4-threshold}
    N\geq2\lam n\norm{v}_{W^{1,\infty}},
\end{equation}
then $N\geq1$ and \eqref{eq:grouping-positive-margins} give
\begin{equation}
    1-\frac{\lam n}{N}\norm{v}_{W^{1,\infty}}
      -\frac{2\eta C_w}{\sqrt N}
    \geq\frac12-2\eta C_w>0.
\end{equation}
Therefore,
\begin{equation}
\label{eq:grouping-T4-positive}
    A_{11}+2B_1-2B_2^1-2B_3^{non,1}\geq0.
\end{equation}

\medskip
\noindent
\textbf{The single-$\Hf$ group.}
Using \eqref{eq:A2+A3}, \eqref{eq:B52}, \eqref{eq:B51},
\eqref{eq:B3com2}, \eqref{eq:B3non3}, \eqref{eq:B6com2},
and \eqref{eq:B6non2}, we first obtain
\begin{equation}
\label{eq:grouping-Hf-before-comparison}
\begin{aligned}
    &A_2+A_3-2B_5^2-2B_5^1-2B_3^{com,N}\\
    &\qquad-2B_3^{non,3}-2B_6^{com,2}-2B_6^{non,2}\\
    &=\Bigg[
        2(N-n)
        -\frac{2\eta C_w(N-n)(N-n-1)}{N}\\
    &\hspace{4em}
        -\frac{2\eta C_w(N-n)(n+1)}{\sqrt N}
      \Bigg]
      \inp{\Psi,\OT_1^2\cdots\OT_{n+1}^2\Hf\Psi}\\
    &\quad-
      \Bigg[
        \frac{\lam(N-n)(N-n-1)}{N}\norm{v}_{L^\infty}
        +\frac{2\lam n(N-n)}{N}\norm{v}_{W^{1,\infty}}\\
    &\hspace{4em}
        +2\eta C_w(N-n)
        +\frac{2\eta n(C_w+C_w')}{\sqrt N}
      \Bigg]
      \inp{\Psi,\OT_1^2\cdots\OT_n^2\Hf\Psi}.
\end{aligned}
\end{equation}
The second bracket is nonnegative. We may therefore use
\eqref{eq:grouping-field-monotonicity} in the subtracted term to obtain
\begin{equation}
\label{eq:grouping-Hf-expanded}
\begin{aligned}
    &A_2+A_3-2B_5^2-2B_5^1-2B_3^{com,N}\\
    &\qquad-2B_3^{non,3}-2B_6^{com,2}-2B_6^{non,2}\\
    &\geq
    \Bigg[
        2(N-n)
        -\frac{\lam(N-n)(N-n-1)}{N}\norm{v}_{L^\infty}\\
    &\qquad
        -\frac{2\lam n(N-n)}{N}\norm{v}_{W^{1,\infty}}
        -\frac{2\eta C_w(N-n)(N-n-1)}{N}\\
    &\qquad
        -\frac{2\eta C_w(N-n)(n+1)}{\sqrt N}
        -2\eta C_w(N-n)
        -\frac{2\eta n(C_w+C_w')}{\sqrt N}
    \Bigg]\\
    &\qquad\times
    \inp{\Psi,\OT_1^2\cdots\OT_{n+1}^2\Hf\Psi}.
\end{aligned}
\end{equation}
To simplify its coefficient, observe that
\begin{equation}
\begin{aligned}
    &-\frac{\lam(N-n)(N-n-1)}{N}\norm{v}_{L^\infty}
      -\frac{2\eta C_w(N-n)(N-n-1)}{N}
      -2\eta C_w(N-n)\\
    &=(N-n)\left[
        -\lam\norm{v}_{L^\infty}-4\eta C_w
        +\frac{n+1}{N}
         \big(\lam\norm{v}_{L^\infty}+2\eta C_w\big)
      \right].
\end{aligned}
\end{equation}
Consequently, the coefficient in
\eqref{eq:grouping-Hf-expanded} equals
\begin{equation}
\label{eq:grouping-Hf-coefficient}
\begin{aligned}
    &(N-n)\Bigg[
        2-\lam\norm{v}_{L^\infty}-4\eta C_w\\
    &\qquad
        +\frac{n+1}{N}
         \big(\lam\norm{v}_{L^\infty}+2\eta C_w\big)
        -\frac{2\lam n}{N}\norm{v}_{W^{1,\infty}}
        -\frac{2\eta C_w(n+1)}{\sqrt N}
      \Bigg]\\
    &\quad-\frac{2\eta n(C_w+C_w')}{\sqrt N}.
\end{aligned}
\end{equation}
Dropping the nonnegative term proportional to $(n+1)/N$, using
$2-\lam\norm{v}_{L^\infty}-4\eta C_w>1$, and using $N-n\geq1$
in the last negative term, we bound this coefficient from below by
\begin{equation}
\label{eq:grouping-Hf-coefficient-lower-bound}
\begin{aligned}
    (N-n)\Bigg[
        1-\frac{2\lam n}{N}\norm{v}_{W^{1,\infty}}
        -\frac{2\eta}{\sqrt N}
         \big(C_w(n+1)+n(C_w+C_w')\big)
    \Bigg].
\end{aligned}
\end{equation}
The two negative terms in this bracket are each at most $1/2$ if
\begin{equation}
\label{eq:grouping-Hf-threshold}
\begin{aligned}
    N&\geq4\lam n\norm{v}_{W^{1,\infty}},\\
    N&\geq16\eta^2
       \big(C_w(n+1)+n(C_w+C_w')\big)^2.
\end{aligned}
\end{equation}
Under these conditions,
\begin{equation}
\label{eq:grouping-Hf-positive}
\begin{aligned}
    &A_2+A_3-2B_5^2-2B_5^1-2B_3^{com,N}\\
    &\qquad-2B_3^{non,3}-2B_6^{com,2}-2B_6^{non,2}
    \geq0.
\end{aligned}
\end{equation}
\medskip
\noindent
\textbf{The $\Hf^2$ group.}
By \eqref{eq:A4}, \eqref{eq:B6com3}, and \eqref{eq:B6non3},
\begin{equation}
\label{eq:grouping-Hf2-expanded}
\begin{aligned}
    &A_4-2B_6^{com,3}-2B_6^{non,3}\\
    &=\left[
        1-2\eta C_w\frac{N-n}{N}
        -\frac{2\eta n(C_w+C_w')}{\sqrt N}
      \right]
      \inp{\Psi,\OT_1^2\cdots\OT_n^2\Hf^2\Psi}\\
    &\geq
      \left[
        \frac12-\frac{2\eta n(C_w+C_w')}{\sqrt N}
      \right]
      \inp{\Psi,\OT_1^2\cdots\OT_n^2\Hf^2\Psi}.
\end{aligned}
\end{equation}
Here we used $(N-n)/N\leq1$ and $2\eta C_w<1/2$.
Thus,
\begin{equation}
\label{eq:grouping-Hf2-threshold}
    N\geq16\eta^2n^2(C_w+C_w')^2
\end{equation}
implies
\begin{equation}
\label{eq:grouping-Hf2-positive}
    A_4-2B_6^{com,3}-2B_6^{non,3}\geq0.
\end{equation}
Condition \eqref{eq:grouping-Hf2-threshold} is already implied by the
second condition in \eqref{eq:grouping-Hf-threshold}.

\medskip
\noindent
\textbf{The $\OT_1^4\Hf$ group.}
For $n\geq1$, \eqref{eq:B4} and \eqref{eq:B50} give
\begin{equation}
\label{eq:grouping-T4Hf-expanded}
\begin{aligned}
    2B_4-2B_5^0
    &=\left[
        2n-\frac{\lam n(n-1)}{N}\norm{v}_{W^{1,\infty}}
      \right]
      \inp{\Psi,\OT_1^4\OT_2^2\cdots\OT_n^2\Hf\Psi}\\
    &=n\left[
        2-\frac{\lam(n-1)}{N}\norm{v}_{W^{1,\infty}}
      \right]
      \inp{\Psi,\OT_1^4\OT_2^2\cdots\OT_n^2\Hf\Psi}.
\end{aligned}
\end{equation}
This group is nonnegative whenever
\begin{equation}
\label{eq:grouping-T4Hf-threshold}
    N\geq\frac{\lam(n-1)}{2}\norm{v}_{W^{1,\infty}}.
\end{equation}
For $n\geq1$, this is implied by the first condition in
\eqref{eq:grouping-Hf-threshold}. When $n=0$, the group is absent.

\medskip
\noindent
\textbf{Closing the induction step.}
Set $N_0(0)=N_0(1)=1$, in accordance with the base cases.
For fixed $n\geq0$, choose the integer threshold
\begin{equation}
\label{eq:grouping-common-threshold}
    N_0(n+2):=
    \left\lceil
    \max\left\{
    \begin{aligned}
        &N_0(n),\quad n+3,\quad
          4\lam n\norm{v}_{W^{1,\infty}},\\
        &2(2n+1)+8\lam(n+1)\norm{v}_{W^{1,\infty}}\\
        &\qquad+8\eta C_w(n+3)+8\eta n(C_w+C_w'),\\
        &16\eta^2\big(C_w(n+1)+n(C_w+C_w')\big)^2
    \end{aligned}
    \right\}
    \right\rceil.
\end{equation}
This threshold is finite because $v\in W^{1,\infty}$ and
$C_w,C_w'<\infty$. It guarantees
\eqref{eq:grouping-nplus2-threshold},
\eqref{eq:grouping-T4-threshold},
\eqref{eq:grouping-Hf-threshold},
\eqref{eq:grouping-Hf2-threshold}, and, for $n\geq1$,
\eqref{eq:grouping-T4Hf-threshold}.
It also guarantees $N\geq N_0(n)$, as required to apply the induction
hypothesis at order $n$. Combining the five group estimates in
\eqref{eq:grouping-five-groups}, we conclude that, for
$N\geq N_0(n+2)$,
\begin{equation}
\label{eq:grouping-final-lower-bound}
    A+B+\cc{B}+D
    \geq\frac{N^2}{4}
    \inp{\Psi,\OT_1^2\cdots\OT_{n+2}^2\Psi}.
\end{equation}
We keep the fixed choice $C=1/4$, for which $C^2\leq1/4$. Substitution into
\eqref{eq:uniform-prior-bound-on-BBGKY-hierarchy-n+2} yields
\begin{equation}
\label{eq:grouping-induction-closure}
\begin{aligned}
    \inp{\Psi,(\HMN+N)^{n+2}\Psi}
    &\geq C^nN^n\big(A+B+\cc{B}+D\big)\\
    &\geq\frac{C^nN^{n+2}}{4}
       \inp{\Psi,\OT_1^2\cdots\OT_{n+2}^2\Psi}\\
    &\geq C^{n+2}N^{n+2}
       \inp{\Psi,\OT_1^2\cdots\OT_{n+2}^2\Psi},
\end{aligned}
\end{equation}
since $C^2\leq1/4$. Thus the induction closes for every compressed Hamiltonian,
uniformly in $m$. In particular, for $N\geq N_0(k)$ and every
bosonic $\Psi\in\operatorname{Ran}P_m$,
\begin{equation}
\label{eq:energy-compressed-estimate}
    \norm{(\mathcal H_{N,m}+N)^{k/2}\Psi}^{2}
    \geq
    C^kN^k\norm{\OT_1\cdots\OT_k\Psi}^{2}.
\end{equation}

We now remove the spectral truncation. From this point onward, $\HMN$ again denotes the original, uncompressed Hamiltonian. The estimates used in the
base case also give the upper bound
\begin{equation}
    \norm{(\HMN+N)^{1/2}\Phi}^{2}
    \leq
    \frac32\left(
        \sum_{j=1}^N\norm{\OT_j\Phi}^{2}
        +\norm{\Hf^{1/2}\Phi}^{2}
    \right).
\end{equation}
Together with \eqref{eq:energy-base-coercivity}, this shows that
the quadratic-form norm of $\HMN+N$ is equivalent to the free
quadratic-form norm. Consequently,
\begin{equation}
    \norm{(\HMN+N)^{1/2}(P_m-I)\Phi}\longrightarrow0
\end{equation}
for every $\Phi$ in their common form domain.

For every vector $\Phi$, the variational equations for
$(\HMN+N)^{-1}\Phi$ and
$(\mathcal H_{N,m}+N)^{-1}P_m\Phi$ give Galerkin orthogonality.
Since the latter vector lies in $\operatorname{Ran}P_m$,
\begin{equation}
\begin{aligned}
    &\left\|
    (\HMN+N)^{1/2}
    \left[
        (\mathcal H_{N,m}+N)^{-1}P_m\Phi
        -(\HMN+N)^{-1}\Phi
    \right]
    \right\|\\
    &\qquad\leq
    \left\|
    (\HMN+N)^{1/2}(P_m-I)(\HMN+N)^{-1}\Phi
    \right\|
    \longrightarrow0.
\end{aligned}
\end{equation}
In particular,
\begin{equation}
    (\mathcal H_{N,m}+N)^{-1}P_m
    \longrightarrow(\HMN+N)^{-1}
\end{equation}
strongly. Here the compressed inverse is extended by zero on
$(\operatorname{Ran}P_m)^\perp$.

These extended inverses are positive and uniformly bounded by
$8/(5N)$. Continuous functional calculus therefore implies,
for every integer $k\geq1$,
\begin{equation}
    (\mathcal H_{N,m}+N)^{-k/2}P_m
    \longrightarrow(\HMN+N)^{-k/2}
\end{equation}
strongly.

Fix a bosonic
$\Psi\in\cD((\HMN+N)^{k/2})$ and set
\begin{equation}
    \Psi_m:=
    (\mathcal H_{N,m}+N)^{-k/2}
    P_m(\HMN+N)^{k/2}\Psi.
\end{equation}
Then $\Psi_m$ is bosonic, belongs to $\operatorname{Ran}P_m$,
and converges strongly to $\Psi$. Moreover,
\begin{equation}
    (\mathcal H_{N,m}+N)^{k/2}\Psi_m
    =
    P_m(\HMN+N)^{k/2}\Psi.
\end{equation}
Applying \eqref{eq:energy-compressed-estimate} yields
\begin{equation}
    C^kN^k\norm{\OT_1\cdots\OT_k\Psi_m}^{2}
    \leq
    \norm{P_m(\HMN+N)^{k/2}\Psi}^{2}.
\end{equation}

The product $\OT_1\cdots\OT_k$, defined by the joint spectral
calculus of the strongly commuting positive operators $\OT_j$,
is closed. The preceding uniform bound, weak compactness, and
strong convergence $\Psi_m\to\Psi$ imply
\begin{equation}
    \Psi\in\cD(\OT_1\cdots\OT_k),
    \qquad
    \norm{\OT_1\cdots\OT_k\Psi}^{2}
    \leq
    \liminf_{m\to\infty}
    \norm{\OT_1\cdots\OT_k\Psi_m}^{2}.
\end{equation}
Consequently,
\begin{equation}
    C^kN^k\norm{\OT_1\cdots\OT_k\Psi}^{2}
    \leq
    \norm{(\HMN+N)^{k/2}\Psi}^{2}.
\end{equation}
This proves the asserted domain inclusion and quadratic-form
inequality, with $C=1/4$ and the threshold
\eqref{eq:grouping-common-threshold}.
\end{proof}

\section{Proof of \texorpdfstring{\propref{prop:joint-radiation-source-limit}}{Proposition}}
\label{sec:joint-radiation-source-limit-proof}

\begin{proof}
We first establish the joint compactness and the convergence of a single
field insertion, including the uniformity in time needed in
\eqref{eq:joint-source-uniform-limit}.

\medskip\noindent\textbf{Step 1: Uniform bounds.}
For every $f\in\h_f$ and $\zeta\in\cD(\cN_f^{1/2})$, the
creation and annihilation estimates give
\begin{align*}
 \|a(f)\zeta\|&\leq\|f\|\,\|\cN_f^{1/2}\zeta\|,\\
 \|\ad(f)\zeta\|&\leq\|f\|\,\|(\cN_f+1)^{1/2}\zeta\|.
\end{align*}
Consequently,
\begin{equation}\label{eq:joint-source-quadrature-bound}
\begin{aligned}
 \|\phi_N(f)\zeta\|\leq2\|f\|
       \left\|\left(\frac{\cN_f+1}{N}\right)^{1/2}\zeta\right\|.
\end{aligned}
\end{equation}
This estimate also holds after tensoring with the identity on the
particle space. It applies fiberwise when $f$ depends measurably on
the particle positions, with a uniform bound on its $\h_f$ norm. Recall that
$\OT=(\dch_0+\kappa)^{1/2}\geq\sqrt{\kappa}$,
since $\dch_0\geq0$ and $\kappa>0$, and set
\begin{equation}
    \dcQ_N
    :=
    \sum_{j=1}^N \OT_j^2+\Hf+N.
\end{equation}
Its quadratic-form domain is
\begin{equation}
    \cD(\dcQ_N^{1/2})
    =
    \mcap_{j=1}^N\cD (\OT_j) 
    \mcap
    \cD(\Hf^{1/2}).
\end{equation}
We establish the comparison
\begin{equation}
    c_-\dcQ_N
    \leq
    \HMN+c_*N
    \leq
    c_+\dcQ_N
\end{equation}
as quadratic forms, with constants independent of $N$. The field interaction is infinitesimally form bounded with respect to the free Hamiltonian,
and the pair interaction is bounded. Consequently, their form sum is bounded from below and closed on $\cD(\dcQ_N^{1/2})$, meaning that this domain is complete with respect to the associated shifted form norm.
Its associated operator agrees with the
self-adjoint Hamiltonian constructed in
\propref{prop:appendix-self-adjointness}.
Thus
\begin{equation}
    \cD\bigl((\HMN+c_*N)^{1/2}\bigr)
    =
    \cD(\dcQ_N^{1/2}).
\end{equation}
We will use this identification to justify the
evolution identities for mixed particle–field observables.

Let $C_E$ be the constant in
\corref{cor:k-particle-energy-bound}.
Applying that corollary with $k=1$ gives for every sufficiently large $N$,
\begin{equation}
\label{eq:joint-source-uniform-energy}
    \sup_{0\leq t\leq T}
    \Tr_{\h}\!\left(
        \OT\gamma_{N,t}^{(1)}\OT
    \right)
    \leq C_E,
\end{equation}
where $ C_E=4\max\{1,\kappa\}C_0$.

The weighted density matrix is understood through
its positive trace-class realization.
All subsequent compactness and limit arguments
are carried out on this tail of the sequence.
Discarding finitely many initial terms does not
affect any limit as $N\to\infty$. 

To establish finiteness without assuming it in the differentiation argument, let $L\geq1$ be an integer and set
\begin{equation}
     \dcq_{N,L}(t)=\frac1N
        \langle\Psi_{N,t},\chi_L(\cN_f)\Psi_{N,t}\rangle,
\end{equation}
where
\begin{equation}
     \chi_L(n)=\min(n,L).
\end{equation}
The free Hamiltonians and the two-body interaction commute with
$\chi_L(\cN_f)$. On the finite-photon core,
\begin{equation}
 \begin{aligned}
 [a(f),\chi_L(\cN_f)]
     &=a(f)\1_{\{1\leq\cN_f\leq L\}},\\
 [\ad(f),\chi_L(\cN_f)]
     &=-\bigl(a(f)\1_{\{1\leq\cN_f\leq L\}}\bigr)^\dagger.
 \end{aligned}
\end{equation}
These commutators are bounded, with norms at most $\|f\|_{L^2} \sqrt L$.
The same conclusion holds fiberwise for $f=g_{x_j}$, uniformly in
$x_j$. Thus the Hamiltonian commutator is bounded for fixed $N,L$.
Moreover, $\chi_L(\cN_f)$ preserves the free energy form domain.
The integrated quadratic-form Heisenberg identity therefore applies,
and Cauchy--Schwarz gives, for almost every $t$,
\begin{align*}
 |\pr_t\dcq_{N,L}(t)|
 &\leq\frac{2\eta}{N^{3/2}}\sum_{j=1}^N
   \left|\left\langle\Psi_{N,t},
     a(g_{x_j})\1_{\{1\leq\cN_f\leq L\}}
        \Psi_{N,t}\right\rangle\right|\\
 &\leq2\eta\|w\|_{L^2} 
   \left(\frac1N\left\langle\Psi_{N,t},
      \cN_f\1_{\{\cN_f\leq L\}}\Psi_{N,t}
      \right\rangle\right)^{1/2}\\
 &\leq2\eta\|w\|_{L^2} \sqrt{\dcq_{N,L}(t)},
\end{align*}
where in the last inequality follows from $ n \1_{\{1\leq n\leq L\}}\leq \min(n,L)$. Since $\dcq_{N,L}(t)+N^{-1}>0$, the chain rule yields
\begin{equation}
 \left|\pr_t\sqrt{\dcq_{N,L}(t)+N^{-1}}\right|
 \leq\eta\|w\|_{L^2} .
\end{equation}
Integration and then monotone convergence as $L\to\infty$ give
\begin{equation}
 \sqrt{\dcq_N(t)+N^{-1}}
 \leq\sqrt{\dcq_N(0)+N^{-1}}+\eta\|w\|_2t.
\end{equation}
For the prescribed coherent initial field, the annihilation
eigenvector identity gives
\begin{equation}
 \dcq_N(0)
 =\frac1N\int_{\R^3}\|a_\xi\Psi_{N,0}\|^2\,d\xi
 =\int_{\R^3}|\widehat u_0(\xi)|^2\,d\xi
 =\|u_0\|_{L^2} ^2,
\end{equation}
where the final equality is Plancherel's identity with the Fourier
transform. Consequently,
\begin{equation}\label{eq:joint-source-uniform-number}
\begin{aligned}
 \cM_T
 :=&\sup_{N,\,0\leq t\leq T}
     \left\langle\Psi_{N,t},\frac{\cN_f+1}{N}
                                  \Psi_{N,t}\right\rangle\\
 &\leq\bigl(\sqrt{\|u_0\|_{L^2} ^2+1}+\eta\|w\|_2T\bigr)^2<\infty.
\end{aligned}
\end{equation}
The cutoff commutator identities can now also be passed to the
limit in their time-integrated form: the discarded annihilation
terms tend to zero in norm at each time and are dominated by the
number bound just proved. Hence $\dcq_N$ is absolutely continuous
and satisfies
\begin{equation}
 |\pr_t\dcq_N(t)|
 \leq2\eta\|w\|_{L^2} \sqrt{\dcq_N(t)+N^{-1}},
\end{equation}
for almost every $t$. Only the first moments in
\eqref{eq:joint-source-uniform-energy} and
\eqref{eq:joint-source-uniform-number} will be used below.

\medskip\noindent\textbf{Step 2: Construction of a joint measure at a fixed time.}
For now, suppress $t$. Recall \eqref{eq:definition-of-bosonic-Fock-space}, the particle Fock space over the $\h$, and use $\cF(\h_f)$ for the radiation Fock space. Embed $\Gamma_N$ into
$\fkH(\h)\otimes\cF(\h_f)$, with its particle factor supported on
particle number  $N$, and set $\epsilon=N^{-1}$. Let
$E_d=\operatorname{span}\{e_1,\ldots,e_d\}$, where $(e_i)_{i\geq1}$
is an orthonormal eigenbasis of $\OT^2$, and let $L_m$ be the span
of the first $m$ vectors of a fixed orthonormal basis of $\h_f$.
Let $P_d$ denote the orthogonal projection of $\h$ onto $E_d$,
and let $P_m^{\mathrm f}$ denote the orthogonal projection
of $\h_f$ onto $L_m$.
The Fock factorizations over a subspace and its orthogonal complement
define the reduced state $\Gamma_N^{d,m}$ on
$\fkH(E_d)\otimes\cF(L_m)$, that is,
\begin{equation}
    \Gamma_N^{d,m}
    =\Tr_{\fkH(E_d^\perp)\otimes\cF(L_m^\perp)}(\Gamma_N).
\end{equation}
It is positive with trace one. On a finite-dimensional complex Hilbert space, define the unitary coherent vector
\begin{equation}
 \dcE_\epsilon(z)
 =e^{-\|z\|_{L^2}^2/(2\epsilon)}
   \bigoplus_{n=0}^\infty
   \frac{z^{\otimes n}}{\sqrt{\epsilon^n \, n!}},
\end{equation}
or explicitly,
\begin{equation}
    \dcE_\eps (z) = e^{- \sum_{j=1}^d |z_j|^2/(2 \eps)} \sum_{n_1,\dots, n_d \geq 0} \prc{ \prod_{j=1}^d \frac{z_j^{n_j}}{\sqrt{\eps^{n_j} n_j !}}} \ket{n_1, \dots, n_d},
\end{equation}
where $z=z_1e_1+\cdots+z_de_d\in E_d$, and $\ket{n_1,\dots, n_j}$ is the normalized occupation vector with $n_j$ particles in $e_j$. The factorials follow from being divided by $\sqrt{\eps^n n!}$
\begin{equation}
    z^{\otimes n} = \sum_{n_1 + \dots + n_d=n} \sqrt{\frac{n!}{n_1 ! \dots n_d !}} z_1^{n_1} \dots z_d^{n_d} \ket{n_1, \dots , n_d} 
\end{equation}
The joint Husimi probability measure is
\begin{equation}\label{eq:joint-source-husimi}
 \dcd\mathfrak m_N^{d,m}(z,u)
 =\frac{\langle\dcE_\epsilon(z)\otimes\dcE_\epsilon(u),
             \Gamma_N^{d,m}
             \dcE_\epsilon(z)\otimes\dcE_\epsilon(u)\rangle}
             {(\pi\epsilon)^{d+m}}\,dz\,du.
\end{equation}
Here $dz$ and $du$ are Lebesgue measures in orthonormal complex
coordinates. Positivity follows from that of $\Gamma_N^{d,m}$.
For nonnegative integers $p,q$, the one-coordinate identity
\begin{equation}
 \frac1{\pi\epsilon}\int_{\mathbb C}
       \overline z^{\,p}z^q e^{-|z|^2/\epsilon}\,dz
 =\begin{cases}
    \epsilon^p p!,&p=q,\\
    0,&p\ne q
   \end{cases}
\end{equation}
which, by expansion in particle-number sectors, gives the weak resolution
of the identity
\begin{equation}
 \int_{E_d}|\dcE_\epsilon(z)\rangle
                \langle\dcE_\epsilon(z)|
                \frac{dz}{(\pi\epsilon)^d}=\1.
\end{equation}
The same identity holds on $u\in\cF(L_m)$. Coherent vectors factor over
orthogonal subspaces; integrating out any additional coordinates
and using this resolution gives the reduced Husimi probability measure \eqref{eq:joint-source-husimi} which is consequently consistent under the coordinate projections for each fixed $N$.

Let $\mathcal N_{E_d}$ and $\mathcal N_p$ denote the unscaled
particle-number operators on $\fkH(E_d)$ and $\fkH(\h)$,
respectively: both act as multiplication by $n$ on their
$n$-particle sectors. To compute the radial moments, consider
the weak coherent-projector integral with weight $\|z\|^{2r}$,
where $r\geq1$ is an integer.
Integration over the phase of each coordinate makes this operator
diagonal in the normalized occupation basis. For an occupation
with $n_j\geq0$ particles in the mode $e_j$, where
$\sum_{j=1}^d n_j=n$, the diagonal entry is
\begin{equation}
\begin{aligned}
 &\frac1{(\pi\epsilon)^d\epsilon^n\prod_{j=1}^d n_j!}
 \int_{\mathbb C^d}\|z\|^{2r}\prod_{j=1}^d|z_j|^{2n_j}
                     e^{-\|z\|^2/\epsilon}\,dz\\
 &\qquad=\epsilon^r\frac{(n+d+r-1)!}{(n+d-1)!}
 =\epsilon^r\prod_{j=0}^{r-1}(n+d+j).
\end{aligned}
\end{equation}
Indeed, the angular integral cancels in the ratio to the case
$r=0$, and the radial integral is the elementary Gamma integral.
This depends only on $n$. Taking the trace against $\Gamma_N^{d,m}$
therefore gives
\begin{equation}
\begin{aligned}
 \int\|z\|^2\,\dcd\mathfrak m_N^{d,m}
 &=\epsilon\Tr\bigl[(\mathcal N_{E_d}+d)
                          \Gamma_N^{d,m}\bigr]
 \leq1+\epsilon d,\\
 \int\|u\|^2\,\dcd\mathfrak m_N^{d,m}
 &\leq\epsilon\Tr\bigl[\Gamma_N\cN_f\bigr]+\epsilon m
 \leq\cM_T+\epsilon m.
 \label{eq:joint-source-husimi-moments}
\end{aligned}
\end{equation}
For the second inequality \eqref{eq:joint-source-husimi-moments}, the same calculation gives $\epsilon$
times the expected field number in $L_m$, plus $\epsilon \, m$;
the field number in $L_m$ is bounded above by $\cN_f$.
In particular,
\begin{align}
 \int\|z\|^{2r}\,\dcd\mathfrak m_N^{d,m}
 &=\epsilon^r\Tr\left[
      \prod_{j=0}^{r-1}(\mathcal N_{E_d}+d+j)
                    \Gamma_N^{d,m}\right]\nonumber\\
 &\leq\prod_{j=0}^{r-1}(1+\epsilon(d+j)),\qquad r\geq1.
 \label{eq:joint-source-husimi-high-moments}
\end{align}
The last inequality uses $\mathcal N_{E_d}\leq N$ on the support of $\Gamma_N^{d,m}$.

At the fixed time under consideration, the consistency
of the Husimi measures and the bounds
\eqref{eq:joint-source-husimi-moments}, with
$\epsilon=N^{-1}$, verify the hypotheses of
\lemref{lem:joint-husimi-tightness}.
Consequently, after passing to a further subsequence
of the particle numbers already selected, there exist
Borel probability measures $\nu^{d,m}$ such that for every $d,m\geq1$,
\begin{equation}
    \mathfrak m_{N_\ell}^{d,m}
    \rightharpoonup\nu^{d,m}
    \qquad\text{as}\quad  N_\ell\to\infty.
\end{equation}
The subsequence is common to all finite-dimensional
pairs, and the limiting measures are consistent under
coordinate projections. Identify $E_d$ and $L_m$ with $\mathbb C^d$ and $\mathbb C^m$
through the chosen bases. For $d,m\geq1$, define the coordinate projection $ \pi_{d,m}:\mathbb C^{\mathbb N}\times\mathbb C^{\mathbb N} \longrightarrow \mathbb C^d\times\mathbb C^m$ as
\begin{equation}
    \pi_{d,m}\bigl((z_j)_{j\geq1},(u_j)_{j\geq1}\bigr)
    :=
    \bigl((z_1,\ldots,z_d),(u_1,\ldots,u_m)\bigr).
\end{equation}
By consistency and Kolmogorov's extension theorem, there exists
a unique Borel probability measure $\widetilde\nu$ on
$\mathbb C^{\mathbb N}\times\mathbb C^{\mathbb N}$, equipped
with its product topology, such that
\begin{equation}
    (\pi_{d,m})_{\#}\widetilde\nu=\nu^{d,m},
    \qquad d,m\geq1.
\end{equation}
Equivalently, for every Borel set
$B\subset\mathbb C^d\times\mathbb C^m$,
\begin{equation}
    \widetilde\nu\bigl(\pi_{d,m}^{-1}(B)\bigr)=\nu^{d,m}(B).
\end{equation}
Thus, $\widetilde\nu$ is the joint probability law of the two
infinite coordinate sequences, with finite-dimensional
marginals $\nu^{d,m}$. Consistency passes to the limits because coordinate projections are continuous, that is, for $d'\geq d$ and $m'\geq m$,
\begin{equation}
    \begin{aligned}
        \int_{E_d\times L_m}
            F(z,u)\,\dcd\nu^{d,m}(z,u) &= \lim_{\ell\to\infty} \int_{E_d\times L_m} F(z,u)\,\dcd\mathfrak m_{N_\ell}^{d,m}(z,u) \\
            &=\lim_{\ell\to\infty} \int_{E_{d'}\times L_{m'}} F(P_d z, P_m^{\mathrm f} u)\,\dcd\mathfrak m_{N_\ell}^{d',m'}(z,u) \\
            &=\int_{E_{d'}\times L_{m'}} F(P_d z, P_m^{\mathrm f} u)\,\dcd\nu^{d',m'}(z,u).
    \end{aligned}
\end{equation}
Lower semicontinuity and \eqref{eq:joint-source-husimi-moments}, with $d,m$ fixed before
$\ell\to\infty$, imply
\begin{equation}
 \begin{aligned}
  \int\sum_{j=1}^d|z_j|^2\,\dcd\widetilde\nu &\leq 1,
  \qquad
  \int\sum_{j=1}^m|u_j|^2\,\dcd\widetilde\nu &\leq \cM_T.
 \end{aligned}
\end{equation}
Monotone convergence as $d,m\to\infty$ shows that both coordinate
sequences are square summable almost surely. Their finite partial
sums in the chosen Hilbert bases are measurable; on the Borel
set of square-summable sequences they converge in Hilbert norm.
Define $\nu$ as the pushforward of $\widetilde\nu$ under this
reconstruction, assigning the value $(0,0)$ on its null complement.
It is a Borel probability measure on $\h\times\h_f$, with
\begin{equation}
 (P_d,P_m^{\mathrm f} )_{\#}\nu=\nu^{d,m}\qquad(d,m\geq1),
\end{equation}
where the projected spaces are identified with their coordinate
spaces. In particular, Parseval's identity gives
\begin{equation}
    \label{eq:joint-source-measure-number}
 \int_{\h\times\h_f}\|u\|^2\,\dcd\nu(\vfi,u)\leq\cM_T.
\end{equation}
Likewise, \eqref{eq:joint-source-husimi-high-moments} and lower
semicontinuity give for every $d,r\geq1$
\begin{equation}
    \int\|P_d\vfi\|^{2r}\,\dcd\nu\leq1.
\end{equation}
Hence, for every $\delta>0$,
\begin{equation}
 \nu\{(\vfi,u):\|P_d\vfi\|>1+\delta\}
       \leq(1+\delta)^{-2r}.
\end{equation}
Letting $r\to\infty$, then taking a countable sequence
$\delta\downarrow0$ and intersecting over $d$, proves
\begin{equation}
 \nu\{(\vfi,u):\|\vfi\|\leq1\}=1.
\end{equation}
Since $\Gamma_N$ is supported on total particle number $N$,
it commutes with $e^{i\theta\mathcal N_p}\otimes\1$.
Factoring this unitary over $E_d$ and $E_d^\perp$ and using
invariance of the partial trace under unitary conjugation on
the traced factors gives
\begin{equation}
 (e^{i\theta\mathcal N_{E_d}}\otimes\1)\Gamma_N^{d,m}
 (e^{-i\theta\mathcal N_{E_d}}\otimes\1)=\Gamma_N^{d,m}.
\end{equation}
The identity
$e^{i\theta\mathcal N_{E_d}}\dcE_\epsilon(z)
 =\dcE_\epsilon(e^{i\theta}z)$ then shows that each Husimi
measure is invariant under $(z,u)\mapsto(e^{i\theta}z,u)$.
This invariance passes to the finite-dimensional limits and
their reconstruction. Thus, $\nu$ is invariant under
$(\vfi,u)\mapsto(e^{i\theta}\vfi,u)$. Henceforth, we relabel the extracted subsequence as $N$.

Next, we identify the observables represented by $\nu$: For
$k\geq1$ and $\OJ^{(k)}\in\mathcal K(\Hkb)$, set
\begin{equation}
\begin{aligned}
 b_{\OJ^{(k)}}(\vfi)
     &:=\langle\vfi^{\otimes k},\OJ^{(k)}\vfi^{\otimes k}\rangle,\\
 G_N^{(k)}(\OJ^{(k)},f)
     &:=\Tr\bigl[\Gamma_N(\OJ^{(k)}\otimes\wyl_N(f))\bigr].
\end{aligned}
\end{equation}
The operator $\OJ^{(k)}$ acts on the first $k$ particles. Also allow $k=0$, $\OJ^{(0)}=\1$ and $b_{\OJ^{(0)}}=1$. Initially, suppose $\OJ^{(k)}=P_d^{\otimes k}\OJ^{(k)}P_d^{\otimes k}$ and $f\in L_m$, and let $a_p(e_i),a_p^\dagger(e_i)$ be the usual
particle annihilation and creation operators on $\fkH(E_d)$,
and define the scaled creation and annihilation operators:
\begin{equation}
 b_i:=\sqrt\epsilon\,a_p(e_i),\qquad
 b_i^\dagger:=\sqrt\epsilon\,a_p^\dagger(e_i).
\end{equation}
They are distinct from the radiation creation and annihilation operators $a(f),\ad(f)$.
On the algebraic sum of particle-number sectors,
$[b_i,b_j^\dagger]=\epsilon\delta_{ij}\1$.
For $z=\sum_{i=1}^d z_i e_i$, Wick quantization is defined by
linearity and
\begin{equation}\label{eq:joint-source-wick-definition}
 \bigl(\overline z_{i_1}\cdots\overline z_{i_p}
                    z_{j_1}\cdots z_{j_q}\bigr)^{\mathrm{Wick}}
 :=b_{i_1}^\dagger\cdots b_{i_p}^\dagger
                  b_{j_1}\cdots b_{j_q},\qquad p,q\geq0.
\end{equation}
An empty product is the identity. Thus
$(|z_i|^2)^{\mathrm{Wick}}=\epsilon a_p^\dagger(e_i)a_p(e_i)$.
The operator $b_{\OJ^{(k)}}^{\mathrm{Wick}}$ is obtained by
applying \eqref{eq:joint-source-wick-definition} to the coordinate
expansion of $b_{\OJ^{(k)}}$. This convention incorporates the
scaling $\epsilon$ in the operators; see also \cite{ammari2011mean}.

Anti-Wick quantization of this polynomial is the weak integral
\begin{equation}
 b_{\OJ^{(k)}}^{\mathrm{AW}}
 :=\int_{E_d}b_{\OJ^{(k)}}(z)
       |\dcE_\epsilon(z)\rangle\langle\dcE_\epsilon(z)|
       \frac{dz}{(\pi\epsilon)^d},
\end{equation}
initially as a form on the algebraic sum of number sectors.
The same definition applies to each monomial. The Gaussian
identity above, applied to its matrix elements, yields
\begin{equation}
 \bigl(\overline z_{i_1}\cdots\overline z_{i_k}
                   z_{j_1}\cdots z_{j_k}\bigr)^{\mathrm{AW}}
 =b_{j_1}\cdots b_{j_k}b_{i_1}^\dagger\cdots b_{i_k}^\dagger.
\end{equation}
Put $\epsilon=N^{-1}$ and
$P_{\leq N}:=\mathbf1_{\{\mathcal N_p\leq N\}}$.
For any product $M$ of $\ell\leq2k$ scaled creation and
annihilation operators, the intermediate particle numbers
are at most $N+2k$. Hence
\begin{equation}
 \|MP_{\leq N}\|
 \leq [\epsilon(N+2k+1)]^{\ell/2}
 \leq (2k+2)^k.
\end{equation}
Write the coordinate expansion as
\begin{equation}
 b_{\OJ^{(k)}}(z)
 =\sum_{\mathbf i,\mathbf j\in\{1,\ldots,d\}^k}
 c_{\mathbf i,\mathbf j}\,
 \overline z_{i_1}\cdots\overline z_{i_k}
 z_{j_1}\cdots z_{j_k},
 \qquad
 |c_{\mathbf i,\mathbf j}|\leq\|\OJ^{(k)}\|.
\end{equation}
Normal ordering with
$b_jb_i^\dagger=b_i^\dagger b_j+\epsilon\delta_{ij}\1$
gives the Wick monomial and, for each $1\leq r\leq k$,
at most $\binom{k}{r}^2r!$ terms with $r$ contractions.
Each such term contains $\epsilon^r$ and a monomial
of length $2k-2r$. Therefore,
\begin{align}
 \bigl\|
 (b_{\OJ^{(k)}}^{\mathrm{AW}}
 -b_{\OJ^{(k)}}^{\mathrm{Wick}})P_{\leq N}
 \bigr\|
 &\leq C_{d,k,\OJ^{(k)}}\epsilon,
 \label{eq:joint-source-aw-wick-error}
\end{align}
where one may take
\begin{equation}
 C_{d,k,\OJ^{(k)}}
 :=
 d^{2k}(2k+2)^k\|\OJ^{(k)}\|
 \sum_{r=1}^k\binom{k}{r}^2r!.
\end{equation}
Here the operators on $\fkH(E_d)$ have been lifted to
$\fkH(\h)$ by the identity on $\fkH(E_d^\perp)$.
This also explains the use of the total particle-number
cutoff $\mathcal N_p$ in this estimate.

We choose the exponential symbol
$\widetilde u\mapsto e^{2i\re\langle f,\widetilde u\rangle}$
because its integral against the field Husimi measure is
the characteristic function that we wish to compare with
the quantum Weyl expectation. To make this comparison,
we apply the coherent-projector definition of anti-Wick
quantization to this bounded symbol and compute the resulting
operator. We will show that
\begin{equation}
 \bigl(e^{2i\re\langle f,\cdot\rangle}\bigr)^{\mathrm{AW}}
 =
 e^{-\epsilon\|f\|^2/2}\wyl_N(f).
\end{equation}
For fixed $f$, the scalar factor tends to one as
$\epsilon=N^{-1}\to0$, thereby identifying the limiting
Husimi characteristic function with the limiting Weyl
expectation. On the field factor, anti-Wick quantization is defined by
\begin{equation}
 \bigl(e^{2i\re\langle f,\cdot\rangle}\bigr)^{\mathrm{AW}}
 =
 \frac1{(\pi\epsilon)^m}
 \int_{L_m}e^{2i\re\langle f,\widetilde u\rangle}
 |\dcE_\epsilon(\widetilde u)\rangle
 \langle\dcE_\epsilon(\widetilde u)|\,d\widetilde u.
\end{equation}
The resolution of the identity and Cauchy--Schwarz imply
that this operator has norm at most one. For $u,\widetilde u\in L_m$, the coherent overlap formula gives
\begin{align}
 \langle\dcE_\epsilon(u),\dcE_\epsilon(\widetilde u)\rangle
 &=
 \exp\left(
 -\frac{\|u\|^2+\|\widetilde u\|^2}{2\epsilon}
 +\frac{\langle u,\widetilde u\rangle}{\epsilon}
 \right),\\
 \bigl|\langle\dcE_\epsilon(u),
                 \dcE_\epsilon(\widetilde u)\rangle\bigr|^2
 &=e^{-\|\widetilde u-u\|^2/\epsilon}.
\end{align}
Consequently, the change of variables
$\zeta=\widetilde u-u$ yields
\begin{equation}
\begin{aligned}
 \left\langle\dcE_\epsilon(u),
 \bigl(e^{2i\re\langle f,\cdot\rangle}\bigr)^{\mathrm{AW}}
 \dcE_\epsilon(u)\right\rangle
 &=
 \frac1{(\pi\epsilon)^m}
 \int_{L_m}
 e^{2i\re\langle f,\widetilde u\rangle}
 e^{-\|\widetilde u-u\|^2/\epsilon}\,d\widetilde u\\
 &=
 \frac{e^{2i\re\langle f,u\rangle}}{(\pi\epsilon)^m}
 \int_{L_m}
 e^{2i\re\langle f,\zeta\rangle}
 e^{-\|\zeta\|^2/\epsilon}\,d\zeta.
\end{aligned}
\end{equation}
Choose orthonormal complex coordinates on $L_m$ and write
$f_\ell=\alpha_\ell+i\beta_\ell$ and
$\zeta_\ell=x_\ell+iy_\ell$. Then
\begin{equation}
 \re\langle f,\zeta\rangle
 =\sum_{\ell=1}^m(\alpha_\ell x_\ell+\beta_\ell y_\ell),
 \qquad
 \|\zeta\|^2=\sum_{\ell=1}^m(x_\ell^2+y_\ell^2).
\end{equation}
For $s\in\mathbb R$, using the Gaussian Fourier integral
\begin{equation}
 \frac1{\sqrt{\pi\epsilon}}
 \int_{\mathbb R}e^{-r^2/\epsilon+2isr}\,dr
 =e^{-\epsilon s^2},
\end{equation}
we obtain
\begin{equation}
\begin{aligned}
 \frac1{(\pi\epsilon)^m}&\int_{L_m} e^{2i\re\langle f,\zeta\rangle}
                   e^{-\|\zeta\|^2/\epsilon}\,d\zeta\\*
 &=
 \prod_{\ell=1}^m
 \left(
 \frac1{\sqrt{\pi\epsilon}}
 \int_{\mathbb R}
 e^{-x_\ell^2/\epsilon+2i\alpha_\ell x_\ell}\,dx_\ell
 \right)
 \left(
 \frac1{\sqrt{\pi\epsilon}}
 \int_{\mathbb R}
 e^{-y_\ell^2/\epsilon+2i\beta_\ell y_\ell}\,dy_\ell
 \right)\\*
 &=
 \prod_{\ell=1}^m
 e^{-\epsilon\alpha_\ell^2}
 e^{-\epsilon\beta_\ell^2}
 =e^{-\epsilon\|f\|^2}.
\end{aligned}
\end{equation}
Thus,
\begin{equation}
 \left\langle\dcE_\epsilon(u),
 \bigl(e^{2i\re\langle f,\cdot\rangle}\bigr)^{\mathrm{AW}}
 \dcE_\epsilon(u)\right\rangle
 =
 e^{-\epsilon\|f\|^2}
 e^{2i\re\langle f,u\rangle}.
\end{equation}
For the Weyl expectation, recall that
$\epsilon=N^{-1}$ and
$\wyl_N(f)=e^{i\sqrt\epsilon(a(f)+\ad(f))}$.
Since
\begin{equation}
 [i\sqrt\epsilon\,\ad(f),i\sqrt\epsilon\,a(f)]
 =\epsilon\|f\|^2\1,
\end{equation}
Since this commutator is a scalar multiple of the identity,
all higher nested commutators vanish. The
Baker--Campbell--Hausdorff formula, applied on coherent
vectors where the relevant exponential series converge,
therefore gives
\begin{equation}
 \wyl_N(f)
 =
 e^{-\epsilon\|f\|^2/2}
 e^{i\sqrt\epsilon\,\ad(f)}
 e^{i\sqrt\epsilon\,a(f)}.
\end{equation}
The coherent eigenvector identity
\begin{equation}
 a(f)\dcE_\epsilon(u)
 =\epsilon^{-1/2}\langle f,u\rangle\dcE_\epsilon(u)
\end{equation}
implies, for $p,q\geq0$,
\begin{equation}
 \left\langle\dcE_\epsilon(u),
 \ad(f)^p a(f)^q\dcE_\epsilon(u)\right\rangle
 =
 \epsilon^{-(p+q)/2}
 \langle u,f\rangle^p\langle f,u\rangle^q.
\end{equation}
The exponential series converge in norm on coherent vectors.
Expanding them therefore yields
\begin{equation}
\begin{aligned}
 \langle\dcE_\epsilon(u),
             \wyl_N(f)\dcE_\epsilon(u)\rangle
 &=
 e^{-\epsilon\|f\|^2/2}
 \sum_{p,q\geq0}
 \frac{(i\sqrt\epsilon)^{p+q}}{p!\,q!}
 \epsilon^{-(p+q)/2}
 \langle u,f\rangle^p\langle f,u\rangle^q\\
 &=
 e^{-\epsilon\|f\|^2/2}
 e^{i\langle u,f\rangle}
 e^{i\langle f,u\rangle}
 =
 e^{-\epsilon\|f\|^2/2}
 e^{2i\re\langle f,u\rangle}.
\end{aligned}
\end{equation}
The second identity follows by normal ordering the exponential
of $i\phi_N(f)$ and using the coherent eigenvector identity for
annihilation operators. Coherent diagonal matrix elements determine a bounded operator on $\fkH(L_m)$. To see this, identify $L_m$ with $\mathbb C^m$ through its chosen orthonormal basis, and let $\{|\alpha\rangle\}_{\alpha\in\mathbb N_0^m}$ be the normalized occupation-number basis. We use the notation $|\alpha|=\sum_{j=1}^m\alpha_j$, $\alpha!=\prod_{j=1}^m\alpha_j!$, and $u^\alpha=\prod_{j=1}^m u_j^{\alpha_j}$. For a bounded operator $A$, the coherent-vector expansion gives
\begin{equation}
 e^{\|u\|^2/\epsilon}
 \langle\dcE_\epsilon(u),A\dcE_\epsilon(u)\rangle
 =
 \sum_{\alpha,\beta\in\mathbb N_0^m}
 \frac{\langle\alpha|A|\beta\rangle\,
       \overline u^{\,\alpha}u^\beta}
      {\epsilon^{(|\alpha|+|\beta|)/2}
       \sqrt{\alpha!\beta!}}.
\end{equation}
This series converges absolutely and locally uniformly. Indeed, for $|u_j|\leq R$, its absolute sum is bounded by
\begin{equation}
 \sum_{\alpha,\beta\in\mathbb N_0^m}
 \frac{|\langle\alpha|A|\beta\rangle|\,
       |\overline u^{\,\alpha}u^\beta|}
      {\epsilon^{(|\alpha|+|\beta|)/2}
       \sqrt{\alpha!\beta!}}
 \leq
 \|A\|
 \left(
   \sum_{n=0}^{\infty}
   \frac{(R/\sqrt\epsilon)^n}{\sqrt{n!}}
 \right)^{2m}
 <\infty,
\end{equation}
where the scalar series converges by the ratio test. Uniqueness of this real-analytic expansion in $u$ and $\overline u$ shows that, if $\langle\dcE_\epsilon(u),A\dcE_\epsilon(u)\rangle=0$ for every $u\in L_m$, then every matrix element $\langle\alpha|A|\beta\rangle$ vanishes, and hence $A=0$. Applying this observation to the difference of the two bounded operators whose coherent expectations were computed above yields
\begin{equation}
 \bigl(e^{2i\re\langle f,\cdot\rangle}\bigr)^{\mathrm{AW}}
     =e^{-\epsilon\|f\|^2/2}\wyl_N(f),
 \qquad f\in L_m.
\end{equation}
It remains to relate particle Wick quantization to the trace-one marginal normalization. Recall that $b_j=N^{-1/2}a_p(e_j)$, with the same scaling parameter $N$ throughout the calculation. On the $n$-particle sector, $a_p(e_j)$ acts by contraction against $e_j$ with a factor $\sqrt n$. Thus, starting from the full $N$-particle sector, $k$ successive scaled annihilations contribute the factor
\begin{equation}
 \prod_{r=0}^{k-1}\sqrt{\frac{N-r}{N}}
 =
 \sqrt{\frac{N!}{N^k(N-k)!}}.
\end{equation}
Consequently, for $N\geq k$, bosonic symmetry gives
\begin{align*}
 &(b_{j_1}\cdots b_{j_k}\Psi)
                  (x_{k+1},\ldots,x_N)\\
 &\qquad=
 \sqrt{\frac{N!}{N^k(N-k)!}}
 \int_{\R^{3k}}
 \left(\prod_{\ell=1}^k
       \overline{e_{j_\ell}(x_\ell)}\right)
 \Psi(x_1,\ldots,x_N)\,d\ux_k.
\end{align*}
Here $\Psi$ is any vector in the $N$-particle sector with
values in the field Fock space; the integral is the partial
Hilbert-space contraction, defined first for finite sums of
tensor products and then by continuity. Expanding
$b_{\OJ^{(k)}}$ in the $e_i$ coordinates and applying this
formula to both sides of each matrix element gives
\begin{equation}
 \Tr\bigl[\Gamma_N
     (b_{\OJ^{(k)}}^{\mathrm{Wick}}\otimes\wyl_N(f))\bigr]
 =\frac{N!}{N^k(N-k)!}\,
       G_N^{(k)}(\OJ^{(k)},f).
\end{equation}
In this calculation, $\OJ^{(k)}$ is extended by zero to
$\h^{\otimes k}$ as specified above; bosonic symmetry
identifies the contraction with its action on the first $k$
labels. Combining this identity with the two anti-Wick
calculations and the partial-trace definition of
$\Gamma_N^{d,m}$ yields
\begin{equation}
\begin{aligned}
 \int b_{\OJ^{(k)}}(z)&e^{2i\re\langle f,u\rangle} \,\dcd\mathfrak m_N^{d,m}\\
 &=e^{-\epsilon\|f\|^2/2} \Tr\bigl[\Gamma_N (b_{\OJ^{(k)}}^{\mathrm{Wick}}\otimes\wyl_N(f))\bigr]
                   +O_{d,k,\OJ^{(k)}}(\epsilon)\\
 &=e^{-\epsilon\|f\|^2/2}\frac{N!}{N^k(N-k)!}\,G_N^{(k)}(\OJ^{(k)},f)
                   +O_{d,k,\OJ^{(k)}}(\epsilon).
 \label{eq:joint-source-wick-normalization}
\end{aligned}
\end{equation}
For $k=0$ the factorial and Wick factors are interpreted as
one, and the error is zero. The unbounded particle polynomial causes no loss in passage
to the finite-dimensional weak limit. Indeed,
$|b_{\OJ^{(k)}}(z)|\leq\|\OJ^{(k)}\|\|z\|^{2k}$ and, for $R>0$,
\begin{equation}
 \int_{\{\|z\|>R\}}|b_{\OJ^{(k)}}(z)|
                         \,\dcd\mathfrak m_N^{d,m}
 \leq\frac{\|\OJ^{(k)}\|}{R^2}
             \int\|z\|^{2k+2}\,\dcd\mathfrak m_N^{d,m}.
\end{equation}
For fixed $d,m,k$, the moments of order $2k+2$ are
bounded uniformly in $N$ by
\eqref{eq:joint-source-husimi-high-moments}.
Together with the preceding tail estimate, this allows us
to pass to the finite-dimensional weak limit by first
inserting a continuous cutoff in $z$ and then letting
$R\to\infty$. The field exponential has modulus one,
so this step requires no additional field moments.
Moreover, for fixed $k$ and $f$,
\begin{equation}
 \frac{N!}{N^k(N-k)!}\longrightarrow1,
 \qquad
 e^{-\epsilon\|f\|^2/2}\longrightarrow1,
 \qquad \epsilon=N^{-1}\longrightarrow0.
\end{equation}
Since the $O_{d,k,\OJ^{(k)}}(\epsilon)$ error also
vanishes, we obtain, along the extracted subsequence,
\begin{equation}
 \label{eq:joint-source-characteristic-limit}
 G_N^{(k)}(\OJ^{(k)},f)
 \longrightarrow
 \int b_{\OJ^{(k)}}(\vfi)
 e^{2i\re\langle f,u\rangle}\,\dcd\nu(\vfi,u) \quad \text{as}\quad N\to\infty.
\end{equation}
At this stage, the convergence holds for
$\OJ^{(k)}=P_d^{\otimes k}\OJ^{(k)}P_d^{\otimes k}$
and $f\in L_m$. 

We extend the convergence to every compact
$\OJ^{(k)}$. Since $P_d^{\otimes k}$ converges strongly
to the identity, compactness gives
\begin{equation}
 \bigl\|
 P_d^{\otimes k}\OJ^{(k)}P_d^{\otimes k}
 -\OJ^{(k)}
 \bigr\|\longrightarrow0 \quad \text{as}\quad d\to\infty.
\end{equation}
Writing $\OJ_d^{(k)}:=P_d^{\otimes k}\OJ^{(k)}P_d^{\otimes k}$, we estimate the changes in the quantum expectation and in the limiting integral. Since $\Gamma_N$ is positive with trace one and $\wyl_N(f)$ is unitary, trace duality gives
\begin{equation}
 \left|
 G_N^{(k)}(\OJ^{(k)},f)
 -G_N^{(k)}(\OJ_d^{(k)},f)
 \right|
 \leq
 \|\Gamma_N\|_1
 \|\OJ^{(k)}-\OJ_d^{(k)}\|
 \|\wyl_N(f)\|
 =
 \|\OJ^{(k)}-\OJ_d^{(k)}\|.
\end{equation}
For the limiting integral, the definition of $b_{\OJ^{(k)}}$ gives
\begin{equation}
 \left|
 b_{\OJ^{(k)}}(\vfi)-b_{\OJ_d^{(k)}}(\vfi)
 \right|
 \leq
 \|\OJ^{(k)}-\OJ_d^{(k)}\|\,\|\vfi\|^{2k}.
\end{equation}
Since the exponential has modulus one and $\nu$ is a probability measure supported on $\{\|\vfi\|\leq1\}\times\h_f$, it follows that
\begin{align*}
 &\left|
 \int
 \left(
 b_{\OJ^{(k)}}(\vfi)-b_{\OJ_d^{(k)}}(\vfi)
 \right)
 e^{2i\re\langle f,u\rangle}\,\dcd\nu(\vfi,u)
 \right|\\
 &\qquad\leq
 \|\OJ^{(k)}-\OJ_d^{(k)}\|
 \int\|\vfi\|^{2k}\,\dcd\nu
 \leq
 \|\OJ^{(k)}-\OJ_d^{(k)}\|.
\end{align*}
For fixed $d$, the convergence in \eqref{eq:joint-source-characteristic-limit} already holds with $\OJ^{(k)}$ replaced by $\OJ_d^{(k)}$. The triangle inequality therefore yields
\begin{equation}
 \limsup_{N\to\infty}
 \left|
 G_N^{(k)}(\OJ^{(k)},f)
 -
 \int b_{\OJ^{(k)}}(\vfi)
 e^{2i\re\langle f,u\rangle}\,\dcd\nu(\vfi,u)
 \right|\leq
 2\|\OJ^{(k)}-\OJ_d^{(k)}\|.
\end{equation}
Compactness of $\OJ^{(k)}$ implies that the right-hand side tends to zero as $d\to\infty$. Thus, \eqref{eq:joint-source-characteristic-limit} holds for every compact $\OJ^{(k)}$ and every $f\in L_m$.

For the extension in the field variable, let $f,f'\in\h_f$. The Weyl relation gives
\begin{equation}
 \wyl_N(f)
 =
 e^{i\operatorname{Im}\langle f',f-f'\rangle/N}
 \wyl_N(f')\wyl_N(f-f').
\end{equation}
Using unitarity, $|e^{ia}-1|\leq|a|$ for real $a$, and the spectral-calculus estimate
\begin{equation}
 \|(e^{i\phi_N(f-f')}-\1)\Psi\|
 \leq\|\phi_N(f-f')\Psi\|,
 \qquad
 \Psi\in\cD(\phi_N(f-f')),
\end{equation}
the quadrature bound \eqref{eq:joint-source-quadrature-bound} and the uniform photon-number estimate yield
\begin{equation}
 \label{eq:joint-source-field-equicontinuity}
 \begin{aligned}
 &|G_N^{(k)}(\OJ^{(k)},f)
       -G_N^{(k)}(\OJ^{(k)},f')|\\
 &\qquad\leq
 \|\OJ^{(k)}\|
 \left(
 2\sqrt{\cM_T}\|f-f'\|
 +\frac{\|f'\|\|f-f'\|}{N}
 \right).
 \end{aligned}
\end{equation}
The same estimate holds uniformly in $t\in[0,T]$ when the suppressed time variable is restored.

On the classical side, $|b_{\OJ^{(k)}}(\vfi)|\leq\|\OJ^{(k)}\|$
almost surely. Hence,
$|e^{ia}-e^{ib}|\leq|a-b|$, Cauchy--Schwarz,
and \eqref{eq:joint-source-measure-number} give
\begin{align*}
 \Bigg|
 \int b_{\OJ^{(k)}}(\vfi)
 &\left(
 e^{2i\re\langle f,u\rangle}
 -e^{2i\re\langle f',u\rangle}
 \right)\,\dcd\nu(\vfi,u)
 \Bigg|\\
 &\qquad\leq
 2\|\OJ^{(k)}\|\|f-f'\|
 \int\|u\|\,\dcd\nu\\
 &\qquad\leq
 2\|\OJ^{(k)}\|\|f-f'\|
 \left(\int\|u\|^2\,\dcd\nu\right)^{1/2}\\
 &\qquad\leq
 2\|\OJ^{(k)}\|\sqrt{\cM_T}\|f-f'\|.
\end{align*}
Taking $f'=P_m^{\mathrm f}f$ and using the convergence
already proved for $f'\in L_m$, the triangle inequality
therefore gives
\begin{align*}
 &\limsup_{N\to\infty}
 \left|
 G_N^{(k)}(\OJ^{(k)},f)
 -
 \int b_{\OJ^{(k)}}(\vfi)
 e^{2i\re\langle f,u\rangle}\,\dcd\nu
 \right|\leq
 4\|\OJ^{(k)}\|\sqrt{\cM_T}
 \|f-P_m^{\mathrm f}f\|.
\end{align*}
Since $P_m^{\mathrm f}f\to f$ in $\h_f$, letting
$m\to\infty$ proves
\eqref{eq:joint-source-characteristic-limit}
for every compact $\OJ^{(k)}$ and every $f\in\h_f$.

We next establish the energy bound and unit-sphere
concentration of the particle marginal of $\nu$.
Let $\OT^2e_j=\lam_je_j$, where
$0<\lam_1\leq\lam_2\leq\cdots$ and
$\lam_j\to\infty$, including multiplicities.
For each fixed $d$, the operators
\begin{equation}
 P_d=\sum_{j=1}^d|e_j\rangle\langle e_j|,
 \qquad
 \OT^2P_d=\sum_{j=1}^d
 \lam_j|e_j\rangle\langle e_j|
\end{equation}
are bounded and finite rank. Their quadratic forms satisfy
\begin{align*}
 \langle\vfi,\OT^2P_d\vfi\rangle
 &=\sum_{j=1}^d
 \lam_j|\langle e_j,\vfi\rangle|^2
 =\|\OT P_d\vfi\|^2,\\
 \langle\vfi,P_d\vfi\rangle
 &=\|P_d\vfi\|^2.
\end{align*}
Positivity of $\gamma_N^{(1)}$ and the uniform energy
bound imply
\begin{align*}
 \Tr(\OT^2P_d\gamma_N^{(1)})
 &=\sum_{j=1}^d
 \lam_j\langle e_j,\gamma_N^{(1)}e_j\rangle\leq\Tr(\OT^2\gamma_N^{(1)})
 \leq C_E,
\end{align*}
and
\begin{align*}
 \Tr((\1-P_d)\gamma_N^{(1)})
 &=\sum_{j>d}
 \langle e_j,\gamma_N^{(1)}e_j\rangle\\
 &\leq\frac1{\lam_{d+1}}
 \sum_{j>d}\lam_j
 \langle e_j,\gamma_N^{(1)}e_j\rangle\\
 &\leq\frac{C_E}{\lam_{d+1}}.
\end{align*}
Since $\Tr\gamma_N^{(1)}=1$, applying
\eqref{eq:joint-source-characteristic-limit}
with $k=1$, $f=0$, and $d$ fixed gives
\begin{align*}
 \int\|\OT P_d\vfi\|^2\,\dcd\nu
 &=\lim_{N\to\infty}
 \Tr(\OT^2P_d\gamma_N^{(1)})
 \leq C_E,\\
 \int\|P_d\vfi\|^2\,\dcd\nu
 &=\lim_{N\to\infty}
 \Tr(P_d\gamma_N^{(1)})
 \geq1-\frac{C_E}{\lam_{d+1}}.
\end{align*}
Now, $\|P_d\vfi\|^2\uparrow\|\vfi\|^2$ as
$d\to\infty$. By monotone convergence, the preceding lower
bound, and the already established unit-ball support yield
\begin{equation}
 1\geq\int\|\vfi\|^2\,\dcd\nu
 =\lim_{d\to\infty}\int\|P_d\vfi\|^2\,\dcd\nu
 \geq
 \lim_{d\to\infty}
 \left(1-\frac{C_E}{\lam_{d+1}}\right)
 =1.
\end{equation}
Thus,
\begin{equation}
    \int(1-\|\vfi\|^2)\,\dcd\nu=0.
\end{equation}
Since the integrand is nonnegative almost surely,
$\|\vfi\|=1$ almost surely.

Likewise, the spectral expansion gives
\begin{equation}
 \|\OT P_d\vfi\|^2
 =\sum_{j=1}^d\lam_j|\langle e_j,\vfi\rangle|^2
 \uparrow\|\OT\vfi\|^2,
\end{equation}
where the last expression is defined to be $+\infty$
outside $\cD(\OT)$. Applying monotone convergence once
more, we conclude that $\nu(\mathbb S\h\times\h_f)=1$ and
\begin{equation}
 \label{eq:joint-source-energy-and-sphere}
 \begin{aligned}
 &\int\|\OT\vfi\|^2\,\dcd\nu
 =\lim_{d\to\infty}
 \int\|\OT P_d\vfi\|^2\,\dcd\nu
 \leq C_E.
 \end{aligned}
\end{equation}
In particular, $\vfi\in\cD(\OT)$ almost surely.

Finally, the mixed moments in \eqref{eq:joint-source-characteristic-limit}, including the case $k=0$, determine $\nu$ uniquely among particle-phase-invariant probability measures supported on the particle unit ball. Let $\widetilde\nu$ be another such measure with the same mixed moments. For fixed $d,m$, denote their pushforwards under $(\vfi,u)\mapsto(P_d\vfi,P_m^{\mathrm f}u)$ by $\nu^{d,m}$ and $\widetilde\nu^{d,m}$. Both are supported on $B_d\times L_m$, where
\begin{equation}
 B_d:=\{z\in E_d:\|z\|\leq1\}.
\end{equation}
Fix $f\in L_m$, then rank-one choices of $\OJ^{(k)}$ supported on the symmetric tensor power of $E_d$ give, up to nonzero normalization constants, every monomial $\overline z^{\,\alpha}z^\beta$ with $|\alpha|=|\beta|=k$. Thus the assumed equality of mixed moments gives
\begin{equation}
 \int \overline z^{\,\alpha}z^\beta
 e^{2i\re\langle f,u\rangle}\,\dcd\nu^{d,m}
 =
 \int \overline z^{\,\alpha}z^\beta
 e^{2i\re\langle f,u\rangle}\,\dcd\widetilde\nu^{d,m}.
\end{equation}
The constant monomial is supplied by $k=0$. If $|\alpha|\neq|\beta|$, both integrals vanish: invariance under $(z,u)\mapsto(e^{i\theta}z,u)$ multiplies either integral by $e^{i(|\beta|-|\alpha|)\theta}$ for every $\theta\in\R$, which forces it to be zero. By linearity, the two measures therefore agree on every test of the form
\begin{equation}
 p(z,\overline z)e^{2i\re\langle f,u\rangle},
\end{equation}
where $p$ is a polynomial.

The polynomials in $z,\overline z$ form a conjugation-closed algebra containing the constants and separating points of the compact set $B_d$. By the Stone--Weierstrass theorem, they are uniformly dense in $C(B_d)$. Moreover, for either measure $\rho=\nu^{d,m}$ or $\rho=\widetilde\nu^{d,m}$,
\begin{equation}
 \left|
 \int (a(z)-p(z,\overline z))
 e^{2i\re\langle f,u\rangle}\,\dcd\rho
 \right|
 \leq \|a-p\|_{C(B_d)}.
\end{equation}
Hence the two measures agree on $a(z)e^{2i\re\langle f,u\rangle}$ for every $a\in C(B_d)$. In particular, taking $a(z)=e^{2i\re\langle g,z\rangle}$ gives
\begin{equation}
 \int e^{2i\re(\langle g,z\rangle+\langle f,u\rangle)}
 \,\dcd\nu^{d,m}
 =
 \int e^{2i\re(\langle g,z\rangle+\langle f,u\rangle)}
 \,\dcd\widetilde\nu^{d,m},
\end{equation}
for every $g\in E_d$ and $f\in L_m$. Thus, their joint Fourier transforms coincide on the underlying finite-dimensional real vector space $E_d\times L_m$. Uniqueness of the Fourier transform of finite measures yields $\nu^{d,m}=\widetilde\nu^{d,m}$. Since this holds for every $d,m$, all finite-dimensional coordinate distributions agree. The orthonormal-basis coordinate maps generate the Borel sigma-algebra of the separable Hilbert product $\h\times\h_f$, hence $\nu=\widetilde\nu$.

\medskip\noindent\textbf{Step 3: One subsequence for every time.}

We now restore the time variable in $G_{N,t}^{(k)}(\OJ^{(k)},f)$. First, take $\OJ^{(k)}$ of finite rank, with its range and adjoint range contained in a finite span of eigenvectors of $\dch_k=\sum_{j=1}^k\dch_0^{(j)}$, and take $f\in\cD(\omg)$. In particular, $[\dch_k,\OJ^{(k)}]$ extends to a bounded operator. For $k=0$, the same arguments apply with the scalar test $1$ and with all particle commutators absent.

We first explain the domains in the expectation identity. For fixed $N$, the form comparison in Step~1 identifies the quadratic-form domain of $\HMN+c_*N$ with $\cD(\dcQ_N^{1/2})$. The lift of $\OJ^{(k)}$, and its adjoint, preserve this domain: on the first $k$ particle variables this follows from their finite spectral ranges; on the other variables they commute with the free operators. If the lift does not preserve full particle symmetry, replace it temporarily with its average over all particle permutations. This average has norm at most $\|\OJ^{(k)}\|$, preserves the form domain and the symmetric subspace, and has the same expectation in $\Psi_{N,t}$. The same statements hold for its adjoint.

The Weyl operators $\wyl_N(f)$ and its adjoint preserve the field-energy form domain because $f\in\cD(\sqrt\omg)$. The displacement identity gives the estimate below, first on the linear span of coherent vectors with profiles in $\cD(\omg)$, which is a core for the field-energy form and is preserved by $\wyl_N(f)$ when $f\in\cD(\omg)$. Closedness extends the estimate to $\cD(\Hf^{1/2})$:
\begin{equation}
 \|\Hf^{1/2}\wyl_N(f)\zeta\|
 \leq \|\Hf^{1/2}\zeta\|
       +N^{-1/2}\|\sqrt\omg f\|\,\|\zeta\|.
\end{equation}
Consequently, the quadratic form identity applies to the permutation-averaged observable. Since $\HMN$ and $\wyl_N(f)$ commute with particle permutations. In particular,
\begin{equation}
\begin{aligned}
 &G_{N,t}^{(k)}(\OJ^{(k)},f)
       -G_{N,s}^{(k)}(\OJ^{(k)},f)=i\int_s^t
   \langle\Psi_{N,r},
      [\HMN,\OJ^{(k)}\wyl_N(f)]\Psi_{N,r}\rangle\,dr,
\end{aligned}
\end{equation}
where the commutator is initially understood as a quadratic form. For completeness, we apply the usual Schr\"odinger identity to the spectral truncations $\mathbf1_{[0,L]}(\HMN+c_*N)\Psi_{N,r}$. They converge uniformly on bounded time intervals in the energy form norm. Indeed, for fixed $N$, spectral calculus and unitarity of the evolution give
\begin{align*}
 &\left\|
 (\HMN+c_*N)^{1/2}
 \bigl(\1-\mathbf1_{[0,L]}(\HMN+c_*N)\bigr)
 \Psi_{N,r}
 \right\|\\
 &\qquad=
 \left\|
 (\HMN+c_*N)^{1/2}
 \bigl(\1-\mathbf1_{[0,L]}(\HMN+c_*N)\bigr)
 \Psi_{N,0}
 \right\|
 \longrightarrow0.
\end{align*}
The right-hand side is independent of $r$, which proves the asserted uniform convergence. Since the observable and its adjoint are bounded in that norm, the resulting form-commutator integrals converge as $L\to\infty$. This argument uses the energy form domain to justify differentiation; the separate photon-number estimate from Step~1 will give the bounds uniform in $N$ below. The commutator decomposes as follows:
\begin{equation}
\begin{aligned}
 [\HMN,\OJ^{(k)}\wyl_N(f)]
  ={}&[\dch_k,\OJ^{(k)}]\wyl_N(f)\\
 &+\frac{\lam}{N}\sum_{1\leq i<j\leq N}
    [v(x_i-x_j),\OJ^{(k)}]\wyl_N(f)\\
 &+\eta\sum_{j=1}^k
    [\phi_N(g_{x_j}),\OJ^{(k)}]\wyl_N(f)\\
 &+\OJ^{(k)}[\Hf,\wyl_N(f)]
   +\eta\OJ^{(k)}\sum_{j=1}^N
       [\phi_N(g_{x_j}),\wyl_N(f)].
\end{aligned}
\end{equation}
Only pairs meeting the first $k$ labels contribute to the direct interaction, and there are at most $kN$ such pairs. The first two terms therefore contribute at most
\begin{equation}
 \|[\dch_k,\OJ^{(k)}]\|
  +2\lam k\|v\|_{L^\infty}\|\OJ^{(k)}\|
\end{equation}
to the absolute value of the expectation. Displacement by $if/\sqrt N$ gives
\begin{equation}
 \wyld_N(f)a(g)\wyl_N(f)
       =a(g)+\frac{i}{\sqrt N}\langle g,f\rangle\1.
\end{equation}
Taking the adjoint of this identity and adding yields
\begin{equation}
 \wyld_N(f)\phi_N(g)\wyl_N(f)
       =\phi_N(g)-\frac{2}{N}
                     \im \langle g,f\rangle\1.
\end{equation}
Applying the same displacement identity to $\Hf$ gives
\begin{equation}
 \wyld_N(f)\Hf\wyl_N(f)
   =\Hf+\phi_N(i\omg f)
                  +N^{-1}\langle f,\omg f\rangle\1.
\end{equation}
Thus
\begin{align}
 [\Hf,\wyl_N(f)]
 &=\wyl_N(f)\left(\phi_N(i\omg f)
              +N^{-1}\langle f,\omg f\rangle\1\right),
       \label{eq:joint-source-free-field-commutator}\\
 [\phi_N(g),\wyl_N(f)]
 &=-2N^{-1}\im \langle g,f\rangle
                      \wyl_N(f).
       \label{eq:joint-source-interaction-commutator}
\end{align}
In the free-field term this uses $\omg f\in\h_f$ and $(\omg f)/\sqrt\omg=\sqrt\omg f\in\h_f$; in the interaction terms it uses $\|g_x/\sqrt\omg\|=\|\widehat w/\sqrt\omg\|$, uniformly in $x$. The field-energy bounds therefore permit passage to the energy form domain. The photon-number bounds additionally make every field-vector expression below well defined on the actual states. By \eqref{eq:joint-source-quadrature-bound}, the free-field contribution is bounded by
\begin{equation}
 \|\OJ^{(k)}\|\left(
       2\|\omg f\|\sqrt{\cM_T}
          +N^{-1}\|\sqrt\omg f\|^2\right).
\end{equation}
We also need the displaced photon-number estimate
\begin{equation}\label{eq:joint-source-displaced-number}
 \left\|\left(\frac{\cN_f+1}{N}\right)^{1/2}
                   \wyl_N(f)\Psi_{N,t}\right\|
       \leq\sqrt{\cM_T}+\frac{\|f\|}{N}.
\end{equation}
To verify its scaling, we write the number form as the squared norm of the vector of scaled annihilators, together with the additional component $N^{-1/2}\Psi_{N,t}$. Under displacement, the scaled annihilator components change by $if/N$ times the unit vector $\Psi_{N,t}$. The triangle inequality in this direct-sum Hilbert space gives
\begin{equation}
 \left\|\left(\frac{\cN_f+1}{N}\right)^{1/2}
                   \wyl_N(f)\Psi_{N,t}\right\|
 \leq
 \left\|\left(\frac{\cN_f+1}{N}\right)^{1/2}
                   \Psi_{N,t}\right\|
       +\frac{\|f\|}{N}.
\end{equation}
For each $j\leq k$, move the field operator to the left in the first commutator pairing and apply Cauchy--Schwarz in the second:
\begin{align*}
 \Big|\langle\Psi_{N,t}&,
       [\phi_N(g_{x_j}),\OJ^{(k)}]
                         \wyl_N(f)\Psi_{N,t}\rangle\Big|\\
 &\quad\leq
   \|\phi_N(g_{x_j})\Psi_{N,t}\|\,
   \|\OJ^{(k)}\wyl_N(f)\Psi_{N,t}\|\\
 &\qquad\quad+
   \|(\OJ^{(k)})^*\Psi_{N,t}\|\,
   \|\phi_N(g_{x_j})\wyl_N(f)\Psi_{N,t}\|\\
 &\quad\leq
   4\|w\|_{L^2}\|\OJ^{(k)}\|
          \left(\sqrt{\cM_T}+\frac{\|f\|}{N}\right).
\end{align*}
Here, the field estimates are applied fiberwise to the particle positions, and $\OJ^{(k)}$ commutes with the photon number. Summing over $j\leq k$ and multiplying by $\eta$ gives the required bound for the third term. Finally, \eqref{eq:joint-source-interaction-commutator} gives
\begin{equation}
 \left\|\eta\,\OJ^{(k)}\sum_{j=1}^N
           [\phi_N(g_{x_j}),\wyl_N(f)]\right\|
 \leq 2\eta\|w\|_{L^2}\|\OJ^{(k)}\|\,\|f\|.
\end{equation}
The factor $N^{-1}$ in the single-label commutator cancels the $N$ labels in this last sum. Inserting these estimates into the integrated expectation identity proves
\begin{equation}\label{eq:joint-source-time-equicontinuity}
 |G_{N,t}^{(k)}(\OJ^{(k)},f)
       -G_{N,s}^{(k)}(\OJ^{(k)},f)|
     \leq C_{k,\OJ^{(k)},f,T}|t-s|,
\end{equation}
uniformly in $N\geq k$ along the retained tail of the sequence. Choose, for every $k\geq1$, a countable operator-norm dense family in $\mathcal K(\Hkb)$ consisting of the regular finite-rank tests just used, and choose a countable norm-dense subset of $\h_f$ contained in $\cD(\omg)$, including the scalar test for $k=0$. Uniform boundedness and \eqref{eq:joint-source-time-equicontinuity}, followed by Arzel\`a--Ascoli and a diagonal extraction over these countable families and over $k$, give uniform convergence on $[0,T]$ along one subsequence. For arbitrary compact tests, use
\begin{equation}
 \sup_{N,t}
 |G_{N,t}^{(k)}(\OJ^{(k)},f)
       -G_{N,t}^{(k)}(K,f)|
 \leq\|\OJ^{(k)}-K\|.
\end{equation}
Together with \eqref{eq:joint-source-field-equicontinuity}, this extends uniform convergence to every $\OJ^{(k)}\in\mathcal K(\Hkb)$ and $f\in\h_f$. Indeed, the approximation errors are uniform in $N,t$, so the approximating uniformly convergent sequences make each original sequence uniformly Cauchy. Their limits are continuous in time.

At each fixed $t$, Step~2 applied to a further subsequence represents these already prescribed limits by a measure $\nu_t$. Its uniqueness among particle-phase-invariant measures on the particle unit ball makes the representation independent of that further extraction. Along the single subsequence just obtained, we therefore have
\begin{equation}\label{eq:joint-source-uniform-characteristics}
 \sup_{0\leq t\leq T}\left|
 G_{N,t}^{(k)}(\OJ^{(k)},f)
 -\int b_{\OJ^{(k)}}(\vfi)
         e^{2i\re\langle f,u\rangle}\,d\nu_t(\vfi,u)
                   \right|\longrightarrow0,
\end{equation}
as $N\to\infty$, for every $t$, Step~2 also gives
\begin{equation}\label{eq:joint-source-joint-moment-bounds}
 \int\|\OT\vfi\|^2\,d\nu_t\leq C_E,
 \qquad
 \int\|u\|^2\,d\nu_t\leq\cM_T.
\end{equation}

We next verify measurability of this family. Temporarily regard $\nu_t$ as a probability measure on the two countable coordinate sequences, equipped with the product topology. The particle unit ball and every field ball of radius $R$ are compact in this topology: they are closed subsets of products of compact coordinate disks. The second bound in \eqref{eq:joint-source-joint-moment-bounds} gives at most $\cM_T/R^2$ mass outside the latter ball. Thus the entire family is tight on the coordinate space.

Suppose $t_n\to t$. The countable coordinate space is Polish, so the preceding uniform tightness and Prokhorov's theorem imply that every subsequence of $(\nu_{t_n})$ has a further weakly convergent subsequence. Each limit is particle-phase-invariant, supported on the particle unit ball, and has a field second moment at most $\cM_T$. The last assertion follows by applying lower semicontinuity to finite coordinate sums and then monotone convergence; in particular, the field coordinates remain square summable almost surely. For finite-coordinate particle tests and finite-coordinate field Fourier variables, \eqref{eq:joint-source-uniform-characteristics} and the continuity of its limiting functions identify the mixed moments of this measure with those of $\nu_t$. Passing these tests through weak convergence is legitimate: their particle polynomial factor is bounded on the particle unit ball, and a continuous cutoff outside that ball in the finitely many tested particle coordinates gives a bounded continuous extension to the whole coordinate space. The uniqueness argument of Step~2 therefore identifies every such limit with $\nu_t$. Tightness now implies that the entire sequence $\nu_{t_n}$ converges weakly to $\nu_t$ on the coordinate space.

This proves weak continuity in the coordinate topology and hence the Borel probability-kernel assertion. Indeed, evaluation on every coordinate-space Borel set is a measurable function of $t$; the square-summable subspace is Borel, and reconstruction in the chosen bases is Borel measurable with values in the usual Hilbert-space Borel structure. No continuity for the norm topology of the field space is required here. Taking $f=0$ in \eqref{eq:joint-source-uniform-characteristics} and using the prescribed particle marginal limits gives
\begin{equation}
 \gamma_{\infty,t}^{(k)}
  =\int|\vfi^{\otimes k}\rangle
        \langle\vfi^{\otimes k}|\,d\nu_t(\vfi,u).
\end{equation}
The uniqueness of the phase-invariant strong de Finetti measure identifies the particle marginal with $\mu_t$. At time zero, the coherent product assumption and the Weyl relations give the characteristic identity
\begin{equation}
 G_{N,0}^{(k)}(\OJ^{(k)},f)
 =\Tr(\OJ^{(k)}\gamma_{N,0}^{(k)})
       \exp\left(-\frac{\|f\|^2}{2N}
                 +2i\re\langle f,\widehat u_0\rangle\right).
\end{equation}
As $N\to\infty$, $G_{N,0}^{(k)}(\OJ^{(k)},f)$ converges to the mixed characteristic moment of $\mu_0\otimes\delta_{\widehat u_0}$. The already-proved measure uniqueness yields \eqref{eq:joint-source-initial-law}. This identifies the factorized initial law. The construction at positive times retains particle--field correlations.

\medskip\noindent\textbf{Step 4: Passing one unbounded field insertion to the limit.}
For $K\in\mathcal K(\Hkb)$ and $f\in\h_f$, put
\begin{equation}
 \begin{aligned}
 L_{N,t}(K,f)&=\langle\Psi_{N,t},K\phi_N(f)\Psi_{N,t}\rangle,\\
 L_t(K,f)&=\int b_K(\vfi)\,2\re\langle f,u\rangle\,d\nu_t.
 \end{aligned}
\end{equation}
The first expression is well-defined because
\eqref{eq:joint-source-uniform-number} places $\Psi_{N,t}$ in
$\cD(\cN_f^{1/2})$, hence in $\cD(\phi_N(f))$.
The second is absolutely integrable because $|b_K(\vfi)|\leq\|K\|$
on $\mathbb S\h$ and the field has a finite first moment.

The quantities $L_{N,t}(K,f)$ and $L_t(K,f)$ are obtained by differentiating the corresponding mixed characteristic functions with respect to a real parameter $r$ at $r=0$, while keeping $t$ fixed. Indeed, since $\wyl_N(rf)=e^{ir\phi_N(f)}$, we have
\begin{equation}
 \left.\frac{d}{dr}G_{N,t}^{(k)}(K,rf)\right|_{r=0}
 =
 i\langle\Psi_{N,t},K\phi_N(f)\Psi_{N,t}\rangle
 =
 iL_{N,t}(K,f).
\end{equation}
Similarly,
\begin{equation}
 \left.\frac{d}{dr}
 \int b_K(\vfi)e^{2ir\re\langle f,u\rangle}\,d\nu_t
 \right|_{r=0}
 =
 iL_t(K,f).
\end{equation}
The first differentiation is justified by $\Psi_{N,t}\in\cD(\phi_N(f))$, and the second by dominated convergence and the finite first field moment. To pass $N\to\infty$ in these derivatives, we establish a Taylor-remainder estimate uniform in $N$ and $t$, using only the photon-number bound already proved. For fixed $N$ and $f$, let $E$ denote the projection-valued spectral measure of the self-adjoint operator $\phi_N(f)$ on the full particle--field space; explicitly,
\begin{equation}
 E(B)=\mathbf1_B(\phi_N(f)),
 \qquad B\subset\R\ \text{Borel}.
\end{equation}
The spectral theorem gives
\begin{equation}
 \phi_N(f)=\int_{\R}x\,E(dx),
 \qquad
 \wyl_N(rf)=\int_{\R}e^{irx}\,E(dx),
 \qquad r\in\R,
\end{equation}
where the first operator integral is understood on $\cD(\phi_N(f))$. Since $K$ acts only on the particles, it commutes with every spectral projection $E(B)$. Consequently, for every Borel set $B\subset\R$,
\begin{equation}
 \bigl|\langle\Psi_{N,t},K E(B)\Psi_{N,t}\rangle\bigr|
 =\bigl|\langle E(B)\Psi_{N,t},K E(B)\Psi_{N,t}\rangle\bigr|
 \leq\|K\|\,\|E(B)\Psi_{N,t}\|^2.
\end{equation}
Applying this estimate to all finite Borel partitions of $B$ shows that
the total variation of the complex measure
$\langle\Psi_{N,t},K E(\cdot)\Psi_{N,t}\rangle$ is bounded by the positive
measure $\|K\|\langle\Psi_{N,t},E(\cdot)\Psi_{N,t}\rangle$.
The scalar inequality $|e^{irx}-1-irx|\leq r^2x^2/2$ therefore gives,
for every real $r$,
\begin{equation}\label{eq:joint-source-taylor-quantum}
\begin{aligned}
 &\bigl|G_{N,t}^{(k)}(K,rf)-G_{N,t}^{(k)}(K,0)
                -irL_{N,t}(K,f)\bigr|\\
 &\quad\leq\frac{r^2}{2}\|K\|
       \int_{\R}x^2\,\langle\Psi_{N,t},E(dx)\Psi_{N,t}\rangle\\
 &\quad=\frac{r^2}{2}\|K\|\,\|\phi_N(f)\Psi_{N,t}\|^2
 \leq2r^2\|K\|\|f\|^2\cM_T.
\end{aligned}
\end{equation}
Here the spectral second moment requires only
$\Psi_{N,t}\in\cD(\phi_N(f))$; it does not require
$\Psi_{N,t}\in\cD(\phi_N(f)^2)$. The last inequality follows from
\eqref{eq:joint-source-quadrature-bound}.
The same scalar estimate on the classical side yields
\begin{equation}
\begin{aligned}
 &\left|\int b_K(\vfi)
       \bigl(e^{2ir\re\langle f,u\rangle}-1\bigr)\,d\nu_t
                   -irL_t(K,f)\right|\\
 &\qquad \qquad \leq2r^2\|K\|\int|\langle f,u\rangle|^2\,d\nu_t
 \leq2r^2\|K\|\|f\|^2\cM_T.
\end{aligned}
\end{equation}
Subtracting the two expansions, dividing by $|r|$, and applying
\eqref{eq:joint-source-uniform-characteristics} at $rf$ and $0$, with
$r\ne0$ fixed, gives
\begin{equation}
 \limsup_{N\to\infty}\sup_{0\leq t\leq T}
 |L_{N,t}(K,f)-L_t(K,f)|
 \leq4|r|\|K\|\|f\|^2\cM_T.
\end{equation}
Letting $r\to0$ proves
\begin{equation}\label{eq:joint-source-linear-field-limit}
 \sup_{0\leq t\leq T}|L_{N,t}(K,f)-L_t(K,f)|\longrightarrow0.
\end{equation}
The uniform quadratic remainder, rather than characteristic convergence
alone, is what permits passage to the limit in this unbounded insertion.

\medskip\noindent\textbf{Step 5: Finite field modes in the source.}
Let $(e_\ell)$ be the field basis used above and

\begin{equation}
    P_m^{\mathrm f} =\sum_{\ell=1}^m|e_\ell\rangle\langle e_\ell|.
\end{equation}
Set
\begin{equation}
 c_\ell(\ux_k)=\langle e_\ell,G_{\ux_k}\rangle,
 \qquad \dcV_{N,m}^{(k)}=\phi_N(P_m^{\mathrm f}  \,  G_{\ux_k}).
\end{equation}
Each $c_\ell$ is a bounded symmetric function of $\ux_k$.
Annihilation is antilinear in its argument, while creation is linear.
In particular,
\begin{equation}
 \phi_N(ie_\ell)
 =\frac{-i a(e_\ell)+i\ad(e_\ell)}{\sqrt N}.
\end{equation}
Writing $P_m^{\mathrm f}  \, G_{\ux_k}=\sum_{\ell=1}^m c_\ell(\ux_k)e_\ell$ and separating
real and imaginary parts therefore gives
\begin{equation}\label{eq:joint-source-real-mode-expansion}
 \dcV_{N,m}^{(k)}
 =\sum_{\ell=1}^m\left(
       M_{\re (c_\ell)}\phi_N(e_\ell)
       +M_{\im (c_\ell)}\phi_N(ie_\ell)\right).
\end{equation}
Here $M_c$ denotes multiplication by $c$ on $\Hkb$.
This identity holds on $\cD(\cN_f^{1/2})$ by the field bound.
For a compact $\OJ^{(k)}$, multiplication by these bounded coefficients
shows that $[\OJ^{(k)},M_{\re (c_\ell)}]$ and
$[\OJ^{(k)},M_{\im (c_\ell)}]$ are compact. Thus
\eqref{eq:joint-source-linear-field-limit} applies to every term of
the commutator in \eqref{eq:joint-source-real-mode-expansion}.
To identify the resulting classical multiplier, note that
\begin{equation}
\begin{aligned}
 F_{P_m^{\mathrm f}  \, u}^{(k)}(\ux_k)
 &=2\re\langle G_{\ux_k},P_m^{\mathrm f}  \, u\rangle\\
   &=2\re\langle P_m^{\mathrm f}  \, G_{\ux_k},u\rangle\\
 &=\sum_{\ell=1}^m\left(
    (\re c_\ell(\ux_k))\,2\re\langle e_\ell,u\rangle
   +(\im c_\ell(\ux_k))\,2\re\langle ie_\ell,u\rangle
                         \right).
\end{aligned}
\end{equation}
We obtain, uniformly for $t\in[0,T]$,
\begin{equation}\label{eq:joint-source-finite-mode-limit}
 \langle\Psi_{N,t},[\OJ^{(k)},\dcV_{N,m}^{(k)}]\Psi_{N,t}\rangle
 \longrightarrow
 \int\langle\vfi^{\otimes k},
        [\OJ^{(k)},F_{P_m^{\mathrm f}  \, u}^{(k)}]\vfi^{\otimes k}\rangle\,d\nu_t.
\end{equation}

\medskip\noindent\textbf{Step 6: Removing the field projection after compact testing.}
Strong convergence of $P_m^{\mathrm f} $ gives $(1-P_m^{\mathrm f} )g_x\to0$ for each $x$;
it does not give convergence uniformly over $x\in\R^3$.
The following estimate removes the projection after a finite-rank
particle test has been inserted. Put
$R_m(\ux_k)=(1-P_m^{\mathrm f} )G_{\ux_k}$ and, for $p\in\Hkb$, define
\begin{equation}
 d_m(p)^2=\int|p(\ux_k)|^2\|R_m(\ux_k)\|_{\h_f}^2\,d\ux_k.
\end{equation}
Since $R_m(\ux_k)\to0$ for every $\ux_k$ and
$\|R_m(\ux_k)\|\leq k\|w\|_{L^2} $, dominated convergence proves
\begin{equation}\label{eq:joint-source-tail-domination}
 d_m(p)\longrightarrow0\qquad(\forall p\in\Hkb).
\end{equation}
First take $\OJ^{(k)}=|p\rangle\langle q|$, with $p,q\in\Hkb$.
For $\Psi\in\cD(\cN_f^{1/2})$, its partial inner product
\begin{equation}
 Z_q=\int\overline{q(\ux_k')}\Psi(\ux_k',\cdot)\,d\ux_k'
\end{equation}
is a vector in the remaining particle variables and the field, and
$\OJ^{(k)}\Psi=p\otimes Z_q$.
The bounded partial-inner-product map commutes with the bounded spectral
truncations of $(\cN_f+1)^{1/2}$. Passing to their limit shows that
$Z_q\in\cD(\cN_f^{1/2})$ and
\begin{equation}
 \left\|\left(\frac{\cN_f+1}{N}\right)^{1/2}Z_q\right\|
 \leq\|q\|
 \left\|\left(\frac{\cN_f+1}{N}\right)^{1/2}\Psi\right\|.
\end{equation}
Consequently, the fiberwise field estimate yields
\begin{equation}
\begin{aligned}
 \|\phi_N(R_m)\OJ^{(k)}\Psi\|^2
 =&\int |p(\ux_k)|^2
           \|\phi_N(R_m(\ux_k))Z_q\|^2\,d\ux_k\\
 &\leq4d_m(p)^2
       \left\|\left(\frac{\cN_f+1}{N}\right)^{1/2}Z_q\right\|^2\\
 &\leq4d_m(p)^2\|q\|^2
       \left\|\left(\frac{\cN_f+1}{N}\right)^{1/2}\Psi\right\|^2.
\end{aligned}
\end{equation}
Interchanging $p$ and $q$ gives the corresponding estimate for
$(\OJ^{(k)})^*\Psi$. In particular, both vectors belong to the domain
of the decomposable selfadjoint field operator $\phi_N(R_m)$.
The uniform bound on $R_m$ also gives
\begin{equation}
 \|\phi_N(R_m)\Psi\|
 \leq
 2k\|w\|_{L^2}
 \left\|
 \left(\frac{\cN_f+1}{N}\right)^{1/2}\Psi
 \right\|.
\end{equation}
Thus $\Psi$, $\OJ^{(k)}\Psi$, and $(\OJ^{(k)})^*\Psi$ all belong to $\cD(\phi_N(R_m))$, which justifies the following commutator pairings. For a normalized $\Psi$, its two commutator pairings satisfy
\begin{equation}
\begin{aligned}
 \langle\Psi,[\OJ^{(k)},\phi_N(R_m)]\Psi\rangle
 &=\langle\phi_N(R_m)(\OJ^{(k)})^*\Psi,\Psi\rangle
   -\langle\Psi,\phi_N(R_m)\OJ^{(k)}\Psi\rangle,\\
 \bigl|\langle\Psi,[\OJ^{(k)},\phi_N(R_m)]\Psi\rangle\bigr|
 &\leq\|\phi_N(R_m)(\OJ^{(k)})^*\Psi\|
       +\|\phi_N(R_m)\OJ^{(k)}\Psi\|.
\end{aligned}
\end{equation}
With $\Psi=\Psi_{N,t}$, this proves
\begin{equation}\label{eq:joint-source-quantum-tail}
 \bigl|\langle\Psi_{N,t},[\OJ^{(k)},\phi_N(R_m)]\Psi_{N,t}\rangle\bigr|
 \leq2\sqrt{\cM_T}
          \bigl(\|q\|d_m(p)+\|p\|d_m(q)\bigr).
\end{equation}
The right side tends to zero independently of $N$ and $t$.
Only the separate photon-number bound has been used. 

For the classical tail, the real multiplier
$D_{m,u}(\ux_k)=2\re\langle R_m(\ux_k),u\rangle$ satisfies
\begin{equation}
 \|D_{m,u}p\|^2
 \leq4\|u\|^2\int |p(\ux_k)|^2\|R_m(\ux_k)\|^2\,d\ux_k
 =4\|u\|^2d_m(p)^2.
\end{equation}
The same rank-one calculation, now with the unit vector
$\vfi^{\otimes k}$, gives
\begin{equation}
 \bigl|\langle\vfi^{\otimes k},
       [\OJ^{(k)},D_{m,u}]\vfi^{\otimes k}\rangle\bigr|
 \leq2\|u\|\bigl(\|q\|d_m(p)+\|p\|d_m(q)\bigr).
\end{equation}
After integration and Cauchy--Schwarz in $\nu_t$, we obtain
\begin{equation}\label{eq:joint-source-classical-tail}
 \left|\int\langle\vfi^{\otimes k},[\OJ^{(k)},D_{m,u}]
                            \vfi^{\otimes k}\rangle\,d\nu_t\right|
 \leq2\sqrt{\cM_T}
          \bigl(\|q\|d_m(p)+\|p\|d_m(q)\bigr).
\end{equation}
Both tail estimates extend to finite sums of rank-one operators.
Since
\begin{equation}
    \dcV_N^{(k)}-\dcV_{N,m}^{(k)}=\phi_N(R_m) \quad \text{and} \quad F_u^{(k)}-F_{P_m^{\mathrm f}  \, u}^{(k)}=D_{m,u},
\end{equation}
\eqref{eq:joint-source-finite-mode-limit} and the two tail estimates
show that, for every finite-rank $\OJ^{(k)}$,
\begin{equation}
\begin{split}
 \sup_{0\leq t\leq T}\left|
 \langle\Psi_{N,t},[\OJ^{(k)},\dcV_N^{(k)}]\Psi_{N,t}\rangle
 -\int\langle\vfi^{\otimes k},
             [\OJ^{(k)},F_u^{(k)}]\vfi^{\otimes k}\rangle\,d\nu_t
 \right|\longrightarrow0.
\end{split}
\end{equation}
Here the order of limits is first $N\to\infty$ with $m$ fixed, and
then $m\to\infty$.

\medskip\noindent\textbf{Step 7: Trace-class bounds and completion.}
The uniform bound $\|G_{\ux_k}\|\leq k\|w\|_{L^2} $ and
\eqref{eq:joint-source-quadrature-bound} give
\begin{equation}
 \|\dcV_N^{(k)}\Psi_{N,t}\|
 \leq2k\|w\|_{L^2} \sqrt{\cM_T}.
\end{equation}
In particular, the commutator defining the finite-$N$ source is the
trace-class operator
\begin{equation}
 [\dcV_N^{(k)},\Gamma_{N,t}]
 =|\dcV_N^{(k)}\Psi_{N,t}\rangle\langle\Psi_{N,t}|
  -|\Psi_{N,t}\rangle\langle\dcV_N^{(k)}\Psi_{N,t}|.
\end{equation}
The trace norm of $|p\rangle\langle q|$ equals $\|p\|\|q\|$.
Contractivity of partial trace in trace norm therefore gives
\begin{equation}\label{eq:joint-source-uniform-trace-bound}
 \|\cS_{N,t}^{(k)}\|_{\fkS_1}
 \leq2\|\dcV_N^{(k)}\Psi_{N,t}\|
 \leq4k\|w\|_{L^2} \sqrt{\cM_T}.
\end{equation}
On the classical side, the same rank-one identity, for $\|\vfi\|=1$, gives
\begin{equation}\label{eq:joint-source-bochner-bound}
 \|[F_u^{(k)},|\vfi^{\otimes k}\rangle
                         \langle\vfi^{\otimes k}|]\|_{\fkS_1}
 \leq2\|F_u^{(k)}\vfi^{\otimes k}\|\,\|\vfi^{\otimes k}\|
 \leq4k\|w\|_{L^2} \|u\|.
\end{equation}
The integrand in \eqref{eq:joint-source-limit-operator} is continuous
with values in the trace class for the strong Hilbert topologies.
Indeed, $\vfi\mapsto\vfi^{\otimes k}$ and the associated rank-one
projector are norm continuous, while \eqref{eq:joint-source-classical-field}, applied to differences of
field vectors, gives operator-norm continuity of $u\mapsto F_u^{(k)}$.
More explicitly, for $\|\vfi\|=\|\psi\|=1$,
\begin{equation}
\begin{aligned}
 &\bigl\|
 |\vfi^{\otimes k}\rangle\langle\vfi^{\otimes k}|
 -|\psi^{\otimes k}\rangle\langle\psi^{\otimes k}|
 \bigr\|_{\fkS_1}\\
 &\qquad\leq
 2\|\vfi^{\otimes k}-\psi^{\otimes k}\|
 \leq 2k\|\vfi-\psi\|.
\end{aligned}
\end{equation}
Together with $\|[B,A]\|_{\fkS_1}\leq2\|B\|\|A\|_{\fkS_1}$ for bounded $B$ and trace-class $A$, this proves the asserted joint continuity of the commutator integrand.

Since the trace class is separable, this continuous map is strongly
measurable. Its norm is integrable by
\eqref{eq:joint-source-bochner-bound} and the field moment bound.
Thus, \eqref{eq:joint-source-limit-operator} defines a Bochner integral,
and
\begin{equation}
 \|\cS_{\infty,t}^{(k)}\|_{\fkS_1}
 \leq4k\|w\|_{L^2} \int\|u\|\,d\nu_t
 \leq4k\|w\|_{L^2} \sqrt{\cM_T}.
\end{equation}
For bounded particle tests, cyclicity of the trace gives
\begin{equation}
\begin{aligned}
 &\Tr\!\left(\OJ^{(k)}
       [F_u^{(k)},|\vfi^{\otimes k}\rangle\langle\vfi^{\otimes k}|]
       \right)=\langle\vfi^{\otimes k},
            [\OJ^{(k)},F_u^{(k)}]\vfi^{\otimes k}\rangle.
\end{aligned}
\end{equation}
The finite-rank convergence proved in Step~6 consequently has the
orientation and sign asserted in \eqref{eq:joint-source-uniform-limit}. For a general compact $\OJ^{(k)}$ and any finite-rank $K$, the two
trace bounds imply
\begin{equation}
 \sup_{0\leq t\leq T}
 \left|\Tr\!\left[(\OJ^{(k)}-K)
          (\cS_{N,t}^{(k)}-\cS_{\infty,t}^{(k)})\right]\right|
 \leq8k\|w\|_{L^2} \sqrt{\cM_T}\,\|\OJ^{(k)}-K\|.
\end{equation}
Taking $N\to\infty$ for fixed $K$ and then approximating
$\OJ^{(k)}$ in operator norm by finite-rank operators proves
\eqref{eq:joint-source-uniform-limit} for every compact test.

Finally, the continuous inclusion of the trace class into the
Hilbert--Schmidt class carries the Bochner integral into the
Hilbert--Schmidt integral. Identifying Hilbert--Schmidt operators
with their $L^2(d\ux_k\,d\ux_k')$ kernels therefore gives
\eqref{eq:joint-source-limit-kernel} as an $L^2$ identity.
No pointwise convergence of the finite-$N$ kernels is needed.

For a fixed compact $\OJ^{(k)}$, the quantum and classical pairings
are measurable in time. For the classical pairing this follows from
the Borel-kernel property of $\nu_t$ and the integrable bound above.
For the quantum pairing it follows, for example, by truncating
$\dcV_N^{(k)}$ spectrally, using strong continuity of $\Psi_{N,t}$
for each bounded truncation, and then taking the pointwise limit.
Consequently, for every $h\in L^1(0,T)$,
\begin{equation}
\begin{aligned}
 &\left|\int_0^T h(t)\Tr\!\left[
       \OJ^{(k)}(\cS_{N,t}^{(k)}-\cS_{\infty,t}^{(k)})\right]\,dt\right|\\
 &\quad\leq\|h\|_{L^1(0,T)}
       \sup_{0\leq t\leq T}\left|\Tr\!\left[
       \OJ^{(k)}(\cS_{N,t}^{(k)}-\cS_{\infty,t}^{(k)})\right]\right|
 \longrightarrow0.
\end{aligned}
\end{equation}
More generally, let $\OJ_t^{(k)}$ be strongly measurable with values
in the Banach space $\mathcal K(\Hkb)$ with its operator norm,
and satisfy
$\int_0^T\|\OJ_t^{(k)}\|\,dt<\infty$.
Its approximation in this integral norm by compact-operator-valued
simple functions, followed by the uniform trace bounds, proves
\begin{equation}
 \int_0^T\Tr\!\left[
       \OJ_t^{(k)}(\cS_{N,t}^{(k)}-\cS_{\infty,t}^{(k)})\right]\,dt
 \longrightarrow0.
\end{equation}

This applies in particular to the compact tests obtained by conjugation with the free particle evolution. Indeed, let $U(t)$ denote the strongly continuous free unitary group on $\Hkb$. For a compact $\OJ^{(k)}$ and a finite-rank operator $K$,
\begin{equation}
\begin{aligned}
 &\|U(t)^*\OJ^{(k)}U(t)-U(s)^*\OJ^{(k)}U(s)\|\\
 &\qquad\leq
 2\|\OJ^{(k)}-K\|
 +\|U(t)^*KU(t)-U(s)^*KU(s)\|.
\end{aligned}
\end{equation}
For fixed finite-rank $K$, the last term tends to zero as $t\to s$ by strong continuity of $U(t)$. Approximating $\OJ^{(k)}$ in operator norm by such $K$ therefore proves operator-norm continuity of the conjugated test. Moreover,
\begin{equation}
 \|U(t)^*\OJ^{(k)}U(t)\|=\|\OJ^{(k)}\|,
\end{equation}
so these tests satisfy the required measurability and integrability conditions on $[0,T]$.
\end{proof}

\noindent\textbf{Acknowledgments:}
T.C. thanks Gian Michele Graf, Israel Michael Sigal, and Eitan Tadmor for inspiring discussions. T.C. gratefully acknowledges support from the NSF through the grant DMS-2009800 and the RTG Grant DMS-1840314 - {\em Analysis of PDE}. 
\\
\\
\textbf{Data Statement:}
Data sharing is not applicable to this article, as no datasets were generated or analyzed during the current study
\\
\\
\textbf{Conflict of Interest Statement:}
The authors have no competing interests to declare that are relevant to the content of this article.

\bibliographystyle{abbrvnat}
\bibliography{references}

@inproceedings{dietzelee2023,
  author    = {Dietze, Charlotte
               and Lee, Jinyeop},
  editor    = {Correggi, Michele
               and Falconi, Marco},
  title     = {Uniform in Time Convergence to {B}ose--{E}instein Condensation for a Weakly Interacting {B}ose Gas with an External Potential},
  booktitle = {Quantum Mathematics II},
  year      = {2023},
  publisher = {Springer Nature Singapore},
  address   = {Singapore},
  pages     = {267--311},
  isbn      = {978-981-99-5884-9}
}

@article{esy4,
  author  = {L. Erd{\"o}s and B. Schlein and H.-T. Yau},
  journal = {Ann. Math.},
  number  = {1},
  pages   = {291--370},
  title   = {Derivation of the {G}ross-{P}itaevskii equation for the dynamics of {B}ose-{E}instein condensates},
  volume  = {172},
  year    = {2010}
}

@article{esve2015,
  author  = {M. Escobedo and J. J. L. Vel{\'a}zquez},
  journal = {Invent. Math.},
  number  = {3},
  pages   = {761--847},
  title   = {Finite time blow-up and condensation for the bosonic {N}ordheim equation},
  volume  = {200},
  year    = {2015}
}

@article{frgrsc,
  author  = {J. Fr{\"o}hlich and S. Graffi and S. Schwarz},
  journal = {Comm. Math. Phys.},
  number  = {3},
  pages   = {681--697},
  title   = {Mean-field- and classical limit of many-body {S}chr{\"o}dinger dynamics for bosons},
  volume  = {271},
  year    = {2007}
}

@article{kiscst,
  author  = {K. Kirkpatrick and B. Schlein and G. Staffilani},
  journal = {Amer. J. Math.},
  number  = {1},
  pages   = {91--130},
  title   = {Derivation of the two-dimensional nonlinear {S}chr{\"o}dinger equation from many body quantum dynamics},
  volume  = {133},
  year    = {2011}
}

@book{lisesoyn,
  address   = {Basel},
  author    = {E. H. Lieb and R. Seiringer and J. P. Solovej and J. Yngvason},
  publisher = {Birkh\"auser},
  series    = {Oberwolfach seminars},
  title     = {The mathematics of the {Bose} gas and its condensation},
  volume    = {34},
  year      = {2005},
  isbn      = {9783764373368}
}

@article{liseyn,
  author  = {E. H. Lieb and R. Seiringer and J. Yngvason},
  journal = {Phys. Rev. A - Atomic, Molecular, and Optical Physics},
  number  = {4},
  pages   = {436021--4360213},
  title   = {Bosons in a trap: A rigorous derivation of the {G}ross-{P}itaevskii energy functional},
  volume  = {61},
  year    = {2000}
}

@article{liseyn2,
  author  = {E. H. Lieb and R. Seiringer and J. Yngvason},
  journal = {Comm. Math. Phys.},
  number  = {1},
  pages   = {17--31},
  title   = {A rigorous derivation of the {G}ross-{P}itaevskii energy functional for a two-dimensional {Bose} gas},
  volume  = {224},
  year    = {2001}
}

@article{narose,
  author  = {P. T. Nam and N. Rougerie and R. Seiringer},
  journal = {Analysis $\&$ PDE},
  number  = {2},
  pages   = {459--485},
  title   = {Ground states of large bosonic systems: the {G}ross-{P}itaevskii limit revisited},
  volume  = {9},
  year    = {2016}
}

@article{N,
  author        = {M. Napi{\'o}rkowski},
  journal       = {arXiv preprint},
  title         = {Dynamics of interacting bosons: a compact review},
  year          = {2021},
  archiveprefix = {arXiv},
  eprint        = {2101.04594}
}

@article{lieb2002proof,
  title     = {Proof of {{Bose}-{E}instein} condensation for dilute trapped gases},
  author    = {Lieb, Elliott H and Seiringer, Robert},
  journal   = {Physical review letters},
  volume    = {88},
  number    = {17},
  pages     = {170409},
  year      = {2002},
  publisher = {APS}
}

@article{soffer1999resonances,
  title     = {Resonances, radiation damping and instability in {H}amiltonian nonlinear wave equations},
  author    = {Soffer, Avy and Weinstein, Michael I},
  journal   = {Inventiones mathematicae},
  volume    = {136},
  number    = {1},
  pages     = {9--74},
  year      = {1999},
  publisher = {Springer}
}

@article {tadmor,
    AUTHOR = {Tadmor, Eitan},
     TITLE = {Decrease of entropy and emergence of order in collective
              dynamics},
   JOURNAL = {Commun. Contemp. Math.},
  FJOURNAL = {Communications in Contemporary Mathematics},
    VOLUME = {28},
      YEAR = {2026},
    NUMBER = {5},
     PAGES = {Paper No. 2540006, 23} 
}

@article{lewin2014derivation,
  title={Derivation of {Hartree's} theory for generic mean-field {Bose} systems},
  author={Lewin, Mathieu and Nam, Phan Th{\`a}nh and Rougerie, Nicolas},
  journal={Advances in Mathematics},
  volume={254},
  pages={570--621},
  year={2014},
  publisher={Elsevier}
}

@article{chen2019mean,
  title={Mean field dynamics of a quantum tracer particle interacting with a boson gas},
  author={Chen, Thomas and Soffer, Avy},
  journal={Journal of Functional Analysis},
  volume={276},
  number={3},
  pages={971--1006},
  year={2019},
  publisher={Elsevier}
}

@article{pickl2011simple,
  title={A simple derivation of mean field limits for quantum systems},
  author={Pickl, Peter},
  journal={Letters in Mathematical Physics},
  volume={97},
  number={2},
  pages={151--164},
  year={2011},
  publisher={Springer}
}

@article{rodnianski2009quantum,
  title={Quantum fluctuations and rate of convergence towards mean field dynamics},
  author={Rodnianski, Igor and Schlein, Benjamin},
  journal={Communications in Mathematical Physics},
  volume={291},
  number={1},
  pages={31--61},
  year={2009},
  publisher={Springer}
}

@article{knowles2010mean,
  title={Mean-field dynamics: singular potentials and rate of convergence},
  author={Knowles, Antti and Pickl, Peter},
  journal={Communications in Mathematical Physics},
  volume={298},
  pages={101--138},
  year={2010},
  publisher={Springer}
}

@article{chen2026trapped,
  title={Trapped bosons in mean field {{QED}}, nonlinear resonance cascades and dynamical {{BEC}} formation},
  author={Chen, Thomas and Mezher, Ali},
  journal={arXiv preprint arXiv:2604.11756},
  year={2026}
}

@article{ammari2014wigner,
  title={Wigner measures approach to the classical limit of the {{Nelson}} model: Convergence of dynamics and ground state energy},
  author={Ammari, Zied and Falconi, Marco},
  journal={Journal of Statistical Physics},
  volume={157},
  number={2},
  pages={330--362},
  year={2014},
  publisher={Springer}
}

@article{ammari2011mean,
  title={Mean field propagation of {Wigner} measures and {BBGKY} hierarchies for general bosonic states},
  author={Ammari, Zied and Nier, Francis},
  journal={Journal de math{\'e}matiques pures et appliqu{\'e}es},
  volume={95},
  number={6},
  pages={585--626},
  year={2011},
  publisher={Elsevier}
}

@article{chen2015unconditional,
  title={Unconditional uniqueness for the cubic {Gross-Pitaevskii} hierarchy via {Quantum} de {Finetti}},
  author={Chen, Thomas and Hainzl, Christian and Pavlovi{\'c}, Nata{\v{s}}a and Seiringer, Robert},
  journal={Communications on Pure and Applied Mathematics},
  volume={68},
  number={10},
  pages={1845--1884},
  year={2015},
  publisher={Wiley Online Library}
}

@article{hudson1976locally,
  title={Locally normal symmetric states and an analogue of de {F}inetti's theorem},
  author={Hudson, Robin L and Moody, Graham R},
  journal={Zeitschrift f{\"u}r Wahrscheinlichkeitstheorie und verwandte Gebiete},
  volume={33},
  number={4},
  pages={343--351},
  year={1976},
  publisher={Springer}
}

@article{stormer1969symmetric,
  title={Symmetric states of infinite tensor products of {C}*-algebras},
  author={St{\o}rmer, Erling},
  journal={Journal of Functional Analysis},
  volume={3},
  number={1},
  pages={48--68},
  year={1969},
  publisher={Elsevier}
}

@inproceedings{ammari2008mean,
  title={Mean field limit for bosons and infinite dimensional phase-space analysis},
  author={Ammari, Zied and Nier, Francis},
  booktitle={Annales Henri Poincar{\'e}},
  volume={9},
  pages={1503--1574},
  year={2008},
  organization={Springer}
}

@article{ammari2020well,
  title={On well-posedness for general hierarchy equations of {{G}ross--{P}itaevskii} and {H}artree type},
  author={Ammari, Zied and Liard, Quentin and Rouffort, Cl{\'e}ment},
  journal={Archive for Rational Mechanics and Analysis},
  volume={238},
  number={2},
  pages={845--900},
  year={2020},
  publisher={Springer}
}

@article{ammari2016uniqueness,
  author        = {Ammari, Zied and Liard, Quentin},
  title         = {On uniqueness of measure-valued solutions
                   to {Liouville}'s equation of {Hamiltonian} {PDEs}},
  journal       = {Discrete and Continuous Dynamical Systems},
  volume        = {38},
  pages         = {723--748},
  year          = {2018},
  doi           = {10.3934/dcds.2018032}
}

@article{yau1991relative,
  title={Relative entropy and hydrodynamics of {{G}inzburg--{L}andau} models},
  author={Yau, Horng-Tzer},
  journal={Letters in Mathematical Physics},
  volume={22},
  number={1},
  pages={63--80},
  year={1991},
  publisher={Springer}
}

@book{nielsen2001quantum,
  title={Quantum computation and quantum information},
  author={Nielsen, Michael A and Chuang, Isaac L},
  volume={2},
  year={2001},
  publisher={Cambridge university press Cambridge}
}

@article{wehrl1978general,
  title={General properties of entropy},
  author={Wehrl, Alfred},
  journal={Reviews of Modern Physics},
  volume={50},
  number={2},
  pages={221},
  year={1978},
  publisher={APS}
}

@article{carlen2010trace,
  title={Trace inequalities and quantum entropy: an introductory course},
  author={Carlen, Eric},
  journal={Entropy and the quantum},
  volume={529},
  pages={73--140},
  year={2010}
}

@article{zhang2016quantifying,
  title={Quantifying coherence in infinite-dimensional systems},
  author={Zhang, Yu-Ran and Shao, Lian-He and Li, Yongming and Fan, Heng},
  journal={Physical Review A},
  volume={93},
  number={1},
  pages={012334},
  year={2016},
  publisher={APS}
}

@article{frohlich2019microscopic,
  title={A microscopic derivation of time-dependent correlation functions of the 1{D} cubic nonlinear {S}chr{\"o}dinger  equation},
  author={Fr{\"o}hlich, J{\"u}rg and Knowles, Antti and Schlein, Benjamin and Sohinger, Vedran},
  journal={Advances in Mathematics},
  volume={353},
  pages={67--115},
  year={2019},
  publisher={Elsevier}
}

@article{penrose1956bose,
  title={{{Bose}-{E}instein} condensation and liquid helium},
  author={Penrose, Oliver and Onsager, Lars},
  journal={Physical Review},
  volume={104},
  number={3},
  pages={576},
  year={1956},
  publisher={APS}
}

@article{deuchert2019bose,
  title={{Bose}--{E}instein condensation in a dilute, trapped gas at positive temperature},
  author={Deuchert, Andreas and Seiringer, Robert and Yngvason, Jakob},
  journal={Communications in Mathematical Physics},
  volume={368},
  number={2},
  pages={723--776},
  year={2019},
  publisher={Springer}
}

@article{bach2000return,
  title={Return to equilibrium},
  author={Bach, Volker and Fr{\"o}hlich, J{\"u}rg and Sigal, Israel Michael},
  journal={Journal of Mathematical Physics},
  volume={41},
  number={6},
  pages={3985--4060},
  year={2000},
  publisher={AIP Publishing}
}

@article{diehl2008quantum,
  title={Quantum states and phases in driven open quantum systems with cold atoms},
  author={Diehl, Sebastian and Micheli, Andrea and Kantian, Adrian and Kraus, B and B{\"u}chler, Hans Peter and Zoller, Peter},
  journal={Nature Physics},
  volume={4},
  number={11},
  pages={878--883},
  year={2008},
  publisher={Nature Publishing Group UK London}
}

@article{kraus2008preparation,
  title={Preparation of entangled states by quantum {M}arkov processes},
  author={Kraus, Barbara and B{\"u}chler, Hans P and Diehl, Sebastian and Kantian, Adrian and Micheli, Andrea and Zoller, Peter},
  journal={Physical Review A—Atomic, Molecular, and Optical Physics},
  volume={78},
  number={4},
  pages={042307},
  year={2008},
  publisher={APS}
}

@article{araki1970entropy,
  title={Entropy inequalities},
  author={Araki, Huzihiro and Lieb, Elliott H},
  journal={Communications in Mathematical Physics},
  volume={18},
  number={2},
  pages={160--170},
  year={1970},
  publisher={Springer}
}

@article{lindblad1975completely,
  title={Completely positive maps and entropy inequalities},
  author={Lindblad, G{\"o}ran},
  journal={Communications in Mathematical Physics},
  volume={40},
  number={2},
  pages={147--151},
  year={1975},
  publisher={Springer}
}

@article{winter2016tight,
  title={Tight uniform continuity bounds for quantum entropies: conditional entropy, relative entropy distance and energy constraints},
  author={Winter, Andreas},
  journal={Communications in Mathematical Physics},
  volume={347},
  number={1},
  pages={291--313},
  year={2016},
  publisher={Springer}
}

@article{shirokov2010continuity,
  title={Continuity of the von {N}eumann entropy},
  author={Shirokov, ME},
  journal={Communications in Mathematical Physics},
  volume={296},
  number={3},
  pages={625--654},
  year={2010},
  publisher={Springer}
}

@article{shirokov2018adaptation,
  title={Adaptation of the {Alicki—Fannes—Winter} method for the set of states with bounded energy and its use},
  author={Shirokov, Maksim E},
  journal={Reports on Mathematical Physics},
  volume={81},
  number={1},
  pages={81--104},
  year={2018},
  publisher={Elsevier}
}

@article{tsai2002relaxation,
  title={Relaxation of excited states in nonlinear {S}chr{\"o}dinger equations},
  author={Tsai, Tai-Peng and Yau, Horng-Tzer},
  journal={International Mathematics Research Notices},
  volume={2002},
  number={31},
  pages={1629--1673},
  year={2002},
  publisher={Hindawi Publishing Corporation}
}

@article{soffer2004selection,
  title={Selection of the ground state for nonlinear {S}chr{\"o}dinger equations},
  author={Soffer, Avraham and Weinstein, Michael I},
  journal={Reviews in Mathematical Physics},
  volume={16},
  number={08},
  pages={977--1071},
  year={2004},
  publisher={World Scientific}
}

@article{gang2007relaxation,
  title={Relaxation of solitons in nonlinear {S}chr{\"o}dinger equations with potential},
  author={Gang, Zhou and Sigal, IM},
  journal={Advances in Mathematics},
  volume={216},
  number={2},
  pages={443--490},
  year={2007},
  publisher={Elsevier}
}

@inproceedings{de2013approach,
  title={Approach to ground state and time-independent photon bound for massless spin-boson models},
  author={De Roeck, Wojciech and Kupiainen, Antti},
  booktitle={Annales Henri Poincar{\'e}},
  volume={14},
  pages={253--311},
  year={2013},
  organization={Springer}
}

@inproceedings{leopold2017mean,
  title={Mean-field limits of particles in interaction with quantized radiation fields},
  author={Leopold, Nikolai and Pickl, Peter},
  booktitle={Workshop on Macroscopic Limits of Quantum Systems},
  pages={185--214},
  year={2017},
  organization={Springer}
}

@article{AIF_2026__76_2_701_0,
     author = {Ammari, Zied and Falconi, Marco and Hiroshima, Fumio},
     title = {Towards a derivation of Classical ElectroDynamics of charges and fields from {QED}},
     journal = {Annales de l'Institut Fourier},
     pages = {701--787},
     year = {2026},
     publisher = {Association des Annales de l{\textquoteright}institut Fourier},
     volume = {76},
     doi = {10.5802/aif.3720},
}

@article{cardenas2023effective,
  title={On the effective dynamics of {{Bose}-Fermi} mixtures},
  author={C{\'a}rdenas, Esteban and Miller, Joseph K and Pavlovi{\'c}, Nata{\v{s}}a},
  journal={arXiv preprint arXiv:2309.04638},
  year={2023}
}

@article{cardenas2025emergence,
  title={Emergence of fermion-mediated interactions in {{Bose}-Fermi} mixtures},
  author={C{\'a}rdenas, Esteban and Miller, Joseph K and Mitrouskas, David and Pavlovi{\'c}, Nata{\v{s}}a},
  journal={arXiv preprint arXiv:2502.18678},
  year={2025}
}

@InProceedings{kirkpatrick2021large,
  author       = {Kirkpatrick, Kay and Rademacher, Simone and Schlein, Benjamin},
  booktitle    = {Annales Henri Poincare},
  title        = {A large deviation principle in many-body quantum dynamics},
  organization = {Springer},
  pages        = {2595--2618},
  volume       = {22},
  year         = {2021},
}

@Article{petrat2016new,
  author    = {Petrat, S{\"o}ren and Pickl, Peter},
  title     = {A new method and a new scaling for deriving fermionic mean-field dynamics},
  number    = {1},
  pages     = {3},
  volume    = {19},
  journal   = {Mathematical Physics, Analysis and Geometry},
  publisher = {Springer},
  year      = {2016},
}

@Article{fantechi2025quantum,
  author    = {Fantechi, Michele and Merkli, Marco},
  title     = {Quantum systems coupled to environments via mean field interactions},
  pages     = {169981},
  volume    = {476},
  journal   = {Annals of Physics},
  publisher = {Elsevier},
  year      = {2025},
}

@Article{boers2016mean,
  author    = {Boers, Niklas and Pickl, Peter},
  title     = {On mean field limits for dynamical systems},
  number    = {1},
  pages     = {1--16},
  volume    = {164},
  journal   = {Journal of Statistical Physics},
  publisher = {Springer},
  year      = {2016},
}

@article{chen2011quintic,
  title={The quintic {NLS} as the mean field limit of a {B}oson gas with three-body interactions},
  author={Chen, Thomas and Pavlovi{\'c}, Nata{\v{s}}a},
  journal={Journal of Functional Analysis},
  volume={260},
  number={4},
  pages={959--997},
  year={2011},
  publisher={Elsevier}
}

@inproceedings{chen2014derivation,
  title={Derivation of the cubic {NLS} and {G}ross--{P}itaevskii hierarchy from manybody dynamics in d= 3 based on spacetime norms},
  author={Chen, Thomas and Pavlovi{\'c}, Nata{\v{s}}a},
  booktitle={Annales Henri Poincar{\'e}},
  volume={15},
  pages={543--588},
  year={2014},
  organization={Springer}
}

@article{chen2025derivation,
  title={Derivation of renormalized {H}artree-{F}ock-{B}ogoliubov and quantum {B}oltzmann equations in an interacting {B}ose gas},
  author={Chen, Thomas and Hott, Michael},
  journal={Journal of Functional Analysis},
  volume={289},
  number={10},
  pages={111096},
  year={2025},
  publisher={Elsevier}
}

@article{golse2017schrodinger,
  title={The {S}chr{\"o}dinger equation in the mean-field and semiclassical regime},
  author={Golse, Fran{\c{c}}ois and Paul, Thierry},
  journal={Archive for Rational Mechanics and Analysis},
  volume={223},
  pages={57--94},
  year={2017},
  publisher={Springer}
}

@article{rougerie2021scaling,
  title={Scaling limits of bosonic ground states, from many-body to non-linear {S}chr{\"o}dinger},
  author={Rougerie, Nicolas},
  journal={EMS Surveys in Mathematical Sciences},
  volume={7},
  pages={253--408},
  year={2021}
}

@article{merkli2008resonance,
  title={Resonance theory of decoherence and thermalization},
  author={Merkli, M and Sigal, IM and Berman, GP},
  journal={Annals of Physics},
  volume={323},
  number={2},
  pages={373--412},
  year={2008},
  publisher={Elsevier}
}

@article{lions1993mesures,
  title={Sur les mesures de {W}igner},
  author={Lions, Pierre-Louis and Paul, Thierry},
  journal={Revista matem{\'a}tica iberoamericana},
  volume={9},
  pages={553--618},
  year={1993}
}

@inproceedings{hasler2011ground,
  title={Ground states in the spin boson model},
  author={Hasler, David and Herbst, Ira},
  booktitle={Annales Henri Poincar{\'e}},
  volume={12},
  pages={621--677},
  year={2011},
  organization={Springer}
}

@article{shirokov2019uniform,
  title={Uniform continuity bounds for information characteristics of quantum channels depending on input dimension and on input energy},
  author={Shirokov, Maxim E},
  journal={Journal of Physics A: Mathematical and Theoretical},
  volume={52},
  number={1},
  pages={014001},
  year={2019},
  publisher={IOP Publishing}
}

@inproceedings{de2021quantum,
  title={Quantum optimal transport with quantum channels: G. {D}e {P}alma and {D}. {T}revisan},
  author={De Palma, Giacomo and Trevisan, Dario},
  booktitle={Annales Henri Poincar{\'e}},
  volume={22},
  pages={3199--3234},
  year={2021},
  organization={Springer}
}

@article{jakvsic2001entropy,
  title={On entropy production in quantum statistical mechanics},
  author={Jak{\v{s}}i{\'c}, V and Pillet, C-A},
  journal={Communications in Mathematical Physics},
  volume={217},
  number={2},
  pages={285--293},
  year={2001},
  publisher={Springer}
}

@article{jakvsic2002non,
  title={Non-Equilibrium Steady States of Finite Quantum Systems Coupled to Thermal Reservoirs},
  author={Jak{\v{s}}i{\'c}, Vojkan and Pillet, C-A},
  journal={Communications in mathematical physics},
  volume={226},
  number={1},
  pages={131--162},
  year={2002},
  publisher={Springer}
}

@article{carlen2012bounds,
  title={Bounds for entanglement via an extension of strong subadditivity of entropy},
  author={Carlen, Eric A and Lieb, Elliott H},
  journal={Letters in mathematical physics},
  volume={101},
  number={1},
  pages={1--11},
  year={2012},
  publisher={Springer}
}

@article{frohlich2004statistical,
  title={Statistical mechanics of thermodynamic processes},
  author={Fr{\"o}hlich, J{\"u}rg and Merkli, Marco and Schwarz, Simon and Ueltschi, Daniel},
  journal={arXiv preprint math-ph/0410013},
  year={2004}
}

@book{kato1966perturbation,
  title={Perturbation theory for linear operators},
  author={Kato, Tosio and Kat{\aa}o, Tosio},
  volume={132},
  year={1966},
  publisher={Springer}
}

@article{chen2008cauchy,
  title={On the {C}auchy problem for focusing and defocusing {G}ross-{P}itaevskii hierarchies},
  author={Chen, Thomas and Pavlovi{\'c}, Nata{\v{s}}a},
  journal={arXiv preprint arXiv:0812.3463},
  year={2008}
}

@inproceedings{chen2010energy,
  title={Energy conservation and blowup of solutions for focusing {G}ross--{P}itaevskii hierarchies},
  author={Chen, Thomas and Pavlovi{\'c}, Nata{\v{s}}a and Tzirakis, Nikolaos},
  booktitle={Annales de l'IHP Analyse non lin{\'e}aire},
  volume={27},
  pages={1271--1290},
  year={2010}
}
\end{document}